\documentclass[11pt]{article}
\usepackage[utf8]{inputenc}
\usepackage{graphicx}
\usepackage{epstopdf}
\usepackage{ifthen}
\usepackage{amssymb}
\usepackage{lineno}
\usepackage{enumerate}
\usepackage{fancyhdr}
\usepackage{amsmath}
\usepackage{amsfonts}
\usepackage{a4wide}
\usepackage{color}
\usepackage{hyperref}

\newcommand{\PDC}[1]{\widehat{{#1}}}

\newtheorem{theorem}{Theorem}
\newtheorem{lemma}{Lemma}
\newtheorem{corollary}{Corollary}

\newtheorem{claim}{Claim}
\newtheorem{definition}{Definition}

\newenvironment{proof}{\noindent\textbf{Proof.}}{{}\hfill$\Box$\\}

  \newenvironment{proofclaim}{\noindent \emph{Proof.}\ }{\hfill
    $\Diamond$\vspace{1em}}

\newcounter{sclaim}

\newenvironment{sclaim}[1][]%
{\refstepcounter{sclaim}\vspace{1ex}\noindent{\it  (\arabic{sclaim})  {#1}{}}\it}{\vspace{1ex}}

\newenvironment{proofsclaim}[1][]%
	{\noindent {}{#1}{}}{ This proves~(\arabic{sclaim}).\vspace{1ex}}

\newcommand{\Out}{\normalfont{Out}}
\newcommand{\Dual}{\normalfont{Dual}}
\newcommand{\mb}[1]{\mathbb{#1}}
\newcommand{\mc}[1]{\mathcal{#1}}

\newcommand{\ccw}{counterclockwise }

\newcommand{\cw}{clockwise }
\newcommand{\cww}{clockwise}

\title{A bijection for essentially 4-connected toroidal
  triangulations\thanks{This work was
      supported by the grant EGOS ANR-12-JS02-002-01 and GATO
      ANR-16-CE40-0009-01.}}

\author{Nicolas Bonichon, Benjamin Lévêque}

\begin{document}

\maketitle
%\tableofcontents

\begin{abstract}
  Transversal structures (also known as regular edge labelings) are
  combinatorial structures defined over 4-connected plane
  triangulations with quadrangular outer-face.  They have been
  intensively studied and used for many applications (drawing
  algorithm, random generation, enumeration\dots). In this paper we
  introduce and study a generalization of these objects for the
  toroidal case.  Contrary to what happens in the plane, the set of
  toroidal transversal structures of a given toroidal triangulation is
  partitioned into several distributive lattices. We exhibit a subset
  of toroidal transversal structures, called balanced, and show that
  it forms a single distributive lattice. Then, using the minimal
  element of the lattice, we are able to enumerate bijectively
  essentially 4-connected toroidal triangulations.
\end{abstract}

\section{Introduction}

A graph embedded on a surface is called a \emph{map} on this surface
if all its faces are homeomorphic to open disks. Maps are considered
up to homeomorphism.  A map is a \emph{triangulation} if all its faces
have size three.  Given a graph embedded on a surface, a
\emph{contractible loop} is an edge enclosing a region homeomorphic to
an open disk.  Two edges of an embedded graph are called
\emph{homotopic multiple edges} if they have the same extremities and
their union encloses a region homeomorphic to an open disk.  In this
paper, we restrict ourselves to graphs embedded on surfaces that do not
have contractible loops nor homotopic multiple edges. Note that this
is a weaker assumption, than the graph being \emph{simple}, i.e. not
having loops nor multiple edges.  In this paper we distinguish cycles
from closed walk as cycles have no repeated vertices.  A
\emph{contractible cycle} is a cycle enclosing a region homeomorphic
to an open disk.  A \emph{triangle} (resp. \emph{quadrangle}) of a map
is a closed walk of length three (resp. four) that delimits on one
side a region homeomorphic to an open disk. This region is called the
\emph{interior} of the triangle (resp. quadrangle). Note that a
triangle is not necessarily a face of the map as its interior may be
not empty. Note also that a triangle is not necessarily a cycle since
non-contractible loops are allowed.  A \emph{unicellular map} is a map
with only one face, which corresponds to the natural generalization of
planar trees when going to higher genus, see~\cite{CMS09,Cha11}.

In this paper we consider finite maps. We denote by $n$ be the number
of vertices and $m$ the number of edges of a graph. Given a graph
embedded on a surface, we use $f$ for the number of faces.  Euler's
formula says that any map on an orientable surface of genus $g$
satisfies $n-m+f=2-2g$. In particular, the plane is the surface of
genus $0$, the torus the surface of genus $1$, the double torus the
surface of genus $2$, etc.  By Euler's formula, a toroidal
triangulation with $n$ vertices has exactly $3n$ edges and $2n$ faces.

The universal cover of the torus is a surjective mapping $p$ from the
plane to the torus that is locally a homeomorphism.  If the torus is
represented by a hexagon (or parallelogram) in the plane whose opposite sides are
pairwise identified, then the universal cover of the torus is obtained
by replicating the hexagon (or parallelogram) to tile the plane.

A graph is \emph{$k$-connected} if it has at least $k+1$ vertices and
if it stays connected after removing any $k-1$ vertices.  Extending
the notion of essentially 2-connectedness defined in \cite{MR98}, we
say that a toroidal map $G$ is \emph{essentially $k$-connected} if its
universal cover $G^\infty$ is $k$-connected (note that this is
different from $G$ being $k$-connected).  This paper is focused on the
study of essentially 4-connected toroidal triangulations via 
generalizing transversal structures to the toroidal case.

Transversal structures are originally defined on $4$-connected planar
triangulations with four vertices on the outer face.  They have been
introduced by Kant et He~\cite{KH97} (under the name \emph{regular
  edge labelings}) for graph drawing applications of planar
maps~\cite{KH97,Fus09}. Deep combinatorial properties of these objects
have been studied later by Fusy~\cite{Fus09} with numerous other
applications like encoding, enumeration, random generation, etc.
Indeed, in the planar case, transversal structures are strongly
related to a more general object called $\alpha$-orientations by
Felsner~\cite{Fel04}. Consider a graph $G$, with vertex set $V$, and a
function $\alpha:V\to \mb{N}$. An orientation of $G$ is an
\emph{$\alpha$-orientation} if, for every vertex $v\in V$, its
outdegree $d^+(v)$ equals $\alpha(v)$. For a fix planar map $G$ and
function $\alpha$, the set of $\alpha$-orientations of $G$ carries a
structure of distributive lattice (see \cite{Fel04} and older related
results~\cite{Pro93,Oss94}) In the planar case, there is a bijection
between transversal structure of a planar map and $4$-orientations of
the corresponding angle map. Thus the set of transversal structures of
a given planar map also carries a structure of distributive lattice
whose minimal element plays a crucial role for bijection purpose.

 In the toroidal case, things are more complicated since the bijection
 of transversal structures with $4$-orientations is not valid anymore.
 Moreover the set of $\alpha$-orientations of a given toroidal map is
 now partitioned into several distributive lattices (see~\cite{Pro93,GKL15})
 contrarily to the planar case where there is only one lattice and
 thus only one minimal element.  Similar issues appear in the study of
 Schnyder woods and corresponding $3$-orientations of toroidal
 triangulations. In a series of papers~\cite{GL13,GKL15,DGL15} (see
 also the HDR manuscript of the second author~\cite{LevHDR} which
 present these three papers in a unified way), these problems are
 solved by highlighting a particular global property, called ``balanced''
 in~\cite{LevHDR}, that a $3$-orientation may have.

 By following the same guidelines here, we are able to identify, in
 Section~\ref{sec:balanced}, a similar ``balanced'' property for
 $4$-orientations of the angle map. These so-called balanced
 $4$-orientations form the core object of study of this paper. Whereas
 not all $4$-orientations corresponds to transversal structures, we
 show in Section~\ref{sec:char} that all balanced ones correspond to
 transversal structures.  The existence of balanced objects for
 essentially 4-connected toroidal triangulations is proved in
 Section~\ref{sec:existence} by edge contraction. The set of
 $4$-orientations of the angle map of a given essentially
 $4$-connected toroidal triangulation is partitioned into distributive
 lattices but all the balanced $4$-orientations are contained in the
 same lattice, as shown in Section~\ref{sec:latticemain}.  The minimal
 element of this ``balanced lattice'' as some important properties
 that are used in Section~\ref{sec:bijmain} to obtain a bijection
 between essentially $4$-connected toroidal triangulations and some
 toroidal unicellular maps.  Then this bijection is used in
 Section~\ref{sec:enumeration} to enumerate essentially $4$-connected
 toroidal triangulations.

\section{Angle map, transversal structure,  balanced property and
  universal cover}
\label{sec:balanced}

\subsection{Angle map and balanced 4-orientations}
\label{sec:anglemap}
Consider a toroidal triangulation $G$. The \emph{angle map} $A(G)$ of
$G$ is the bipartite map obtained from a simultaneous embedding of
vertices of $G$ and $G^*$ such that vertices of $G^*$ are embedded
inside faces of $G$ and vice-versa, and for each angle of a vertex $v$
incident to a face $v^*$ there is an edge between $v$ and $v^*$.  Hence, $A(G)$ is a
bipartite map with one part consisting of \emph{primal-vertices} and
the other part consisting of \emph{dual-vertices}. Each dual-vertex has
degree three and each face of $A(G)$ is a quadrangle that consists of two
primal-vertices and two dual-vertices. 

Figure~\ref{fig:anglegraph} gives an example of a toroidal
triangulation and its angle map, primal-vertices are black and
dual-vertices are white (this serves as a convention for the entire
paper).

\begin{figure}[!ht]
\center
\includegraphics[scale=0.4]{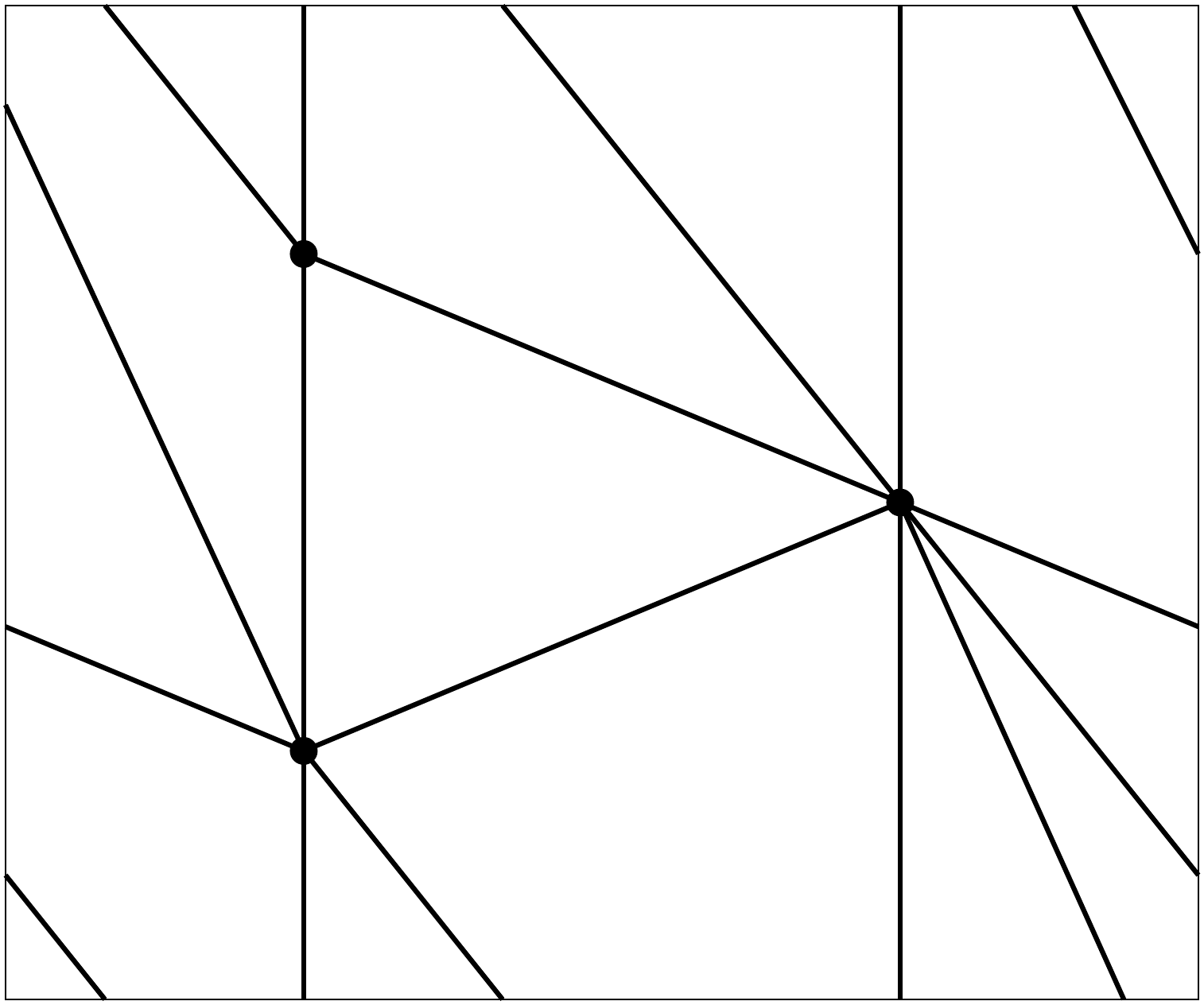} \ \ \ \ 
\includegraphics[scale=0.4]{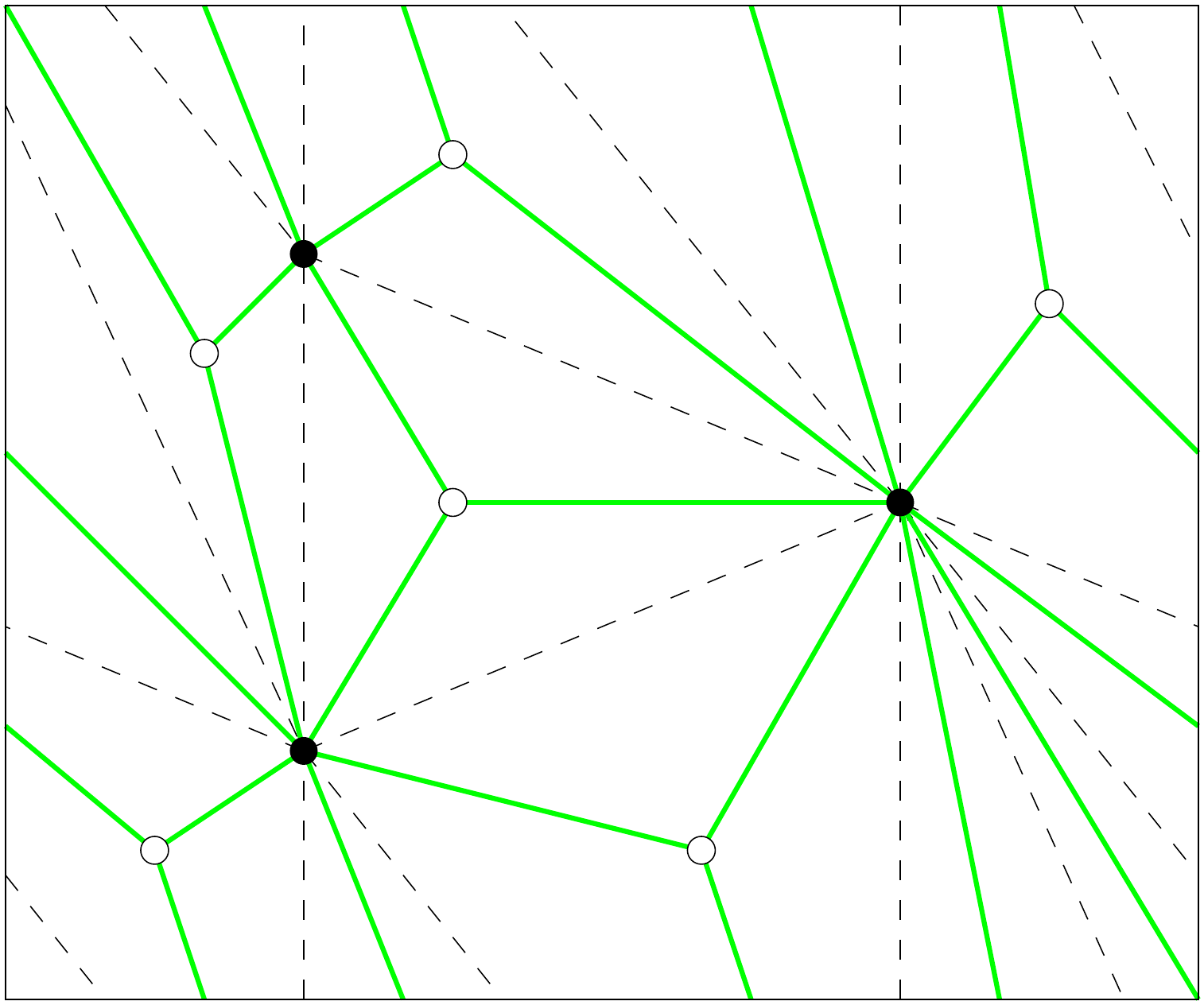}
\caption{A toroidal triangulation and its angle map.}
\label{fig:anglegraph}
\end{figure}

An orientation of the edges of $A(G)$ is called a
\emph{$4$-orientation} if every primal-vertex has outdegree exactly
$4$ and every dual-vertex has outdegree exactly $1$.  Euler's formula
says that for a toroidal triangulation we have $2n=f$, so the number
of edges of the angle map is $3f=4n+f$. Thus Euler's formula is
``compatible'' with existence of $4$-orientations for angle maps of
toroidal triangulations ($4n$ outgoing edges for primal-vertices and
$f$ outgoing edges for dual-vertices.

Consider an orientation of the edges of $A(G)$ and a cycle $C$ of $G$
together with a direction of traversal.  We define $\gamma(C)$ by:
$$\gamma (C) = \#\ \text{edges of $A(G)$ leaving $C$ on its right} -
\#\ \text{edges of $A(G)$ leaving $C$ on its left}.$$

Then we can define balanced orientations:

\begin{definition}[Balanced $4$-orientation]\ \\
A $4$-orientation of $A(G)$ is \emph{balanced} if every
non-contractible cycle $C$ of $G$ satisfies $\gamma(C)=0$.
\end{definition}

Figure~\ref{fig:4orbalanced} gives two examples of $4$-orientations of
the same angle map of a toroidal triangulation. On the left example,
the vertical loop of the triangulation, with upward direction of
traversal, has $\gamma= 2$, thus the orientation is not balanced. On
the right example, one can check that $\gamma=0$ for any
non-contractible cycle (note that we prove latter that it suffices to
check that $\gamma$ equals $0$ for a vertical cycle and a horizontal
cycle to be balanced, see Lemma~\ref{lem:gamma0all}).

\begin{figure}[!ht]
\center
\begin{tabular}{cc}
\includegraphics[scale=0.4]{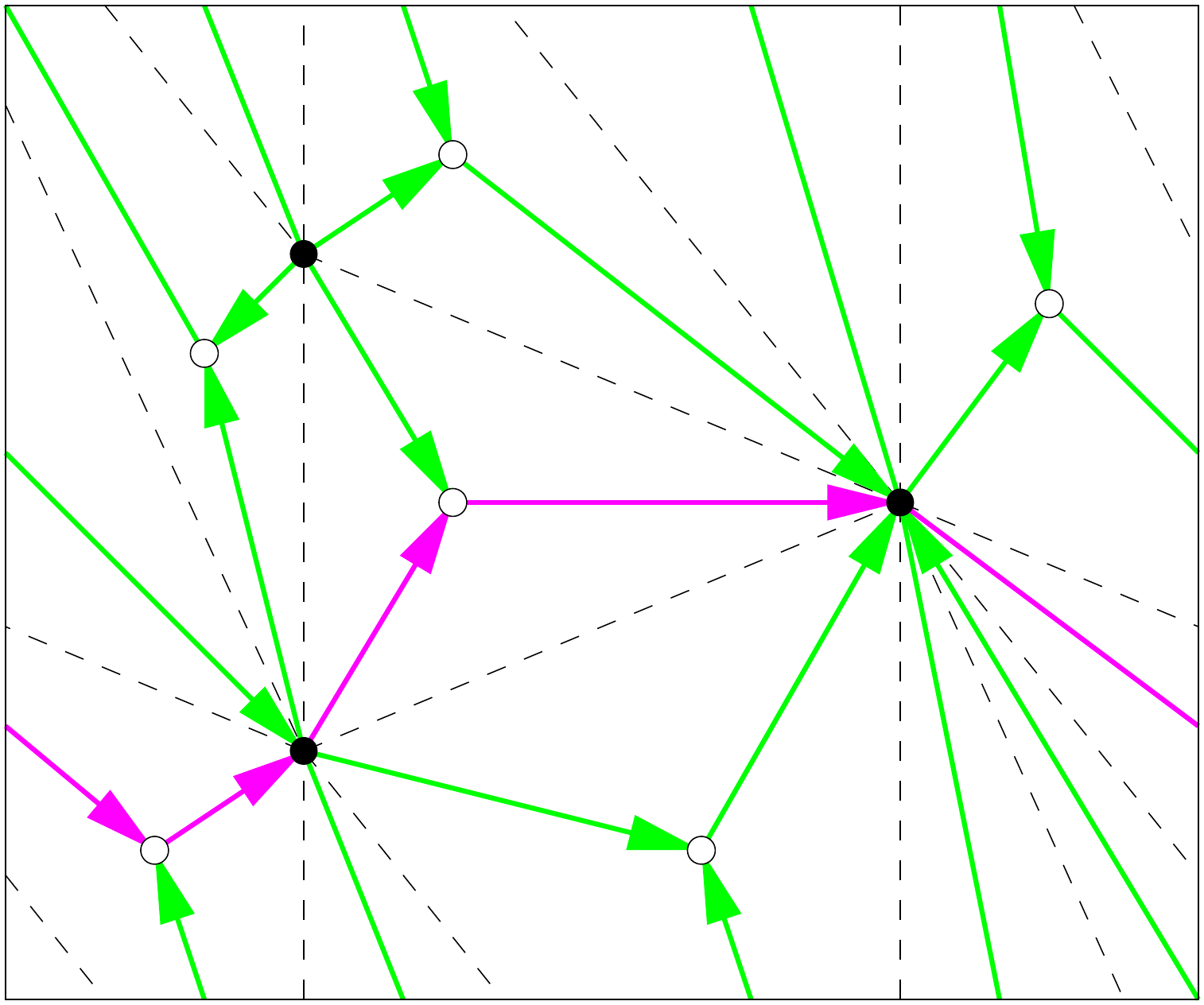} \ \ & \ \ 
\includegraphics[scale=0.4]{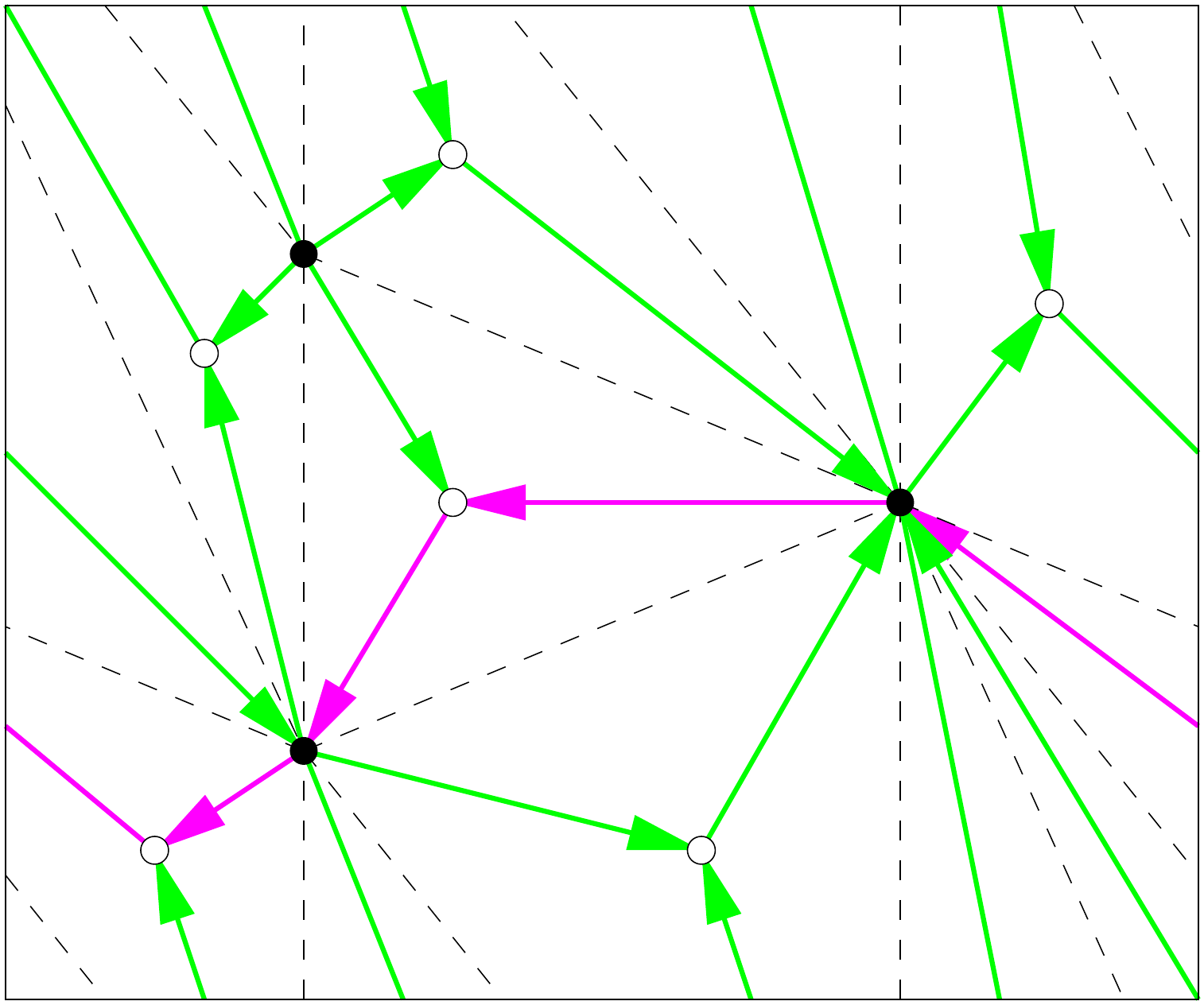}\\
  Non-balanced \ \ & \ \ Balanced \\ 
\end{tabular}
\caption{Two different $4$-orientations of the
  angle map of a toroidal triangulation, exactly one of which is
  balanced. One is obtain from the other by flipping the magenta cycle.}
\label{fig:4orbalanced}
\end{figure}

Balanced $4$-orientations are the main ingredient of this paper. Among
all, we show that an essentially $4$-connected toroidal triangulation
admits a balanced $4$-orientation of its angle map, and we exhibit the
structure of distributive lattice of the set of all these
balanced orientations.

In the next section we show how $4$-orientations are  related to
transversal structures.

\subsection{Balanced transversal structures}
\label{sec:ts}

Transversal structures have been defined
originally in the planar case (see~\cite{KH97,Fus09}) and we propose the
following generalization to the toroidal case. 

First we define the following local rule:

\begin{definition}[Transversal structure local property]
  \label{def:localproperty}\ \\
  Given a map $G$, a vertex $v$ and an orientation and coloring of the
  edges incident to $v$ with the colors blue and red, we say that $v$
  satisfies the \emph{transversal structure local property} (or
  \emph{local property} for short) if the edges around $v$ form in
  counterclockwise order a non-empty interval of outgoing edges of
  color blue, a non-empty interval of outgoing edges of color red, a
  non-empty interval of incoming edges of color blue, a non-empty
  interval of incoming edges of color red (see Figure~\ref{fig:ruleST}).
%  where the correspondence between blue/red and colors 0/1,
%  respectively, serves as a convention for the entire paper).
% $0$, $1$, we say that $v$
% satisfies the \emph{transversal structure local property} (or
% \emph{local property} for short) if the edges around $v$ form in
% couterclockwise order a non-empty interval of outgoing edges of color
% $0$, a non-empty interval of outgoing edges of color $1$, a non-empty
% interval of incoming edges of color $1$, a non-empty interval of
% incoming edges of color $0$ (see Figure~\ref{fig:ruleST} where the
% correspondence between blue/red and colors 0/1, respectively, serves
% as a convention for the entire paper).
\end{definition}

\begin{figure}[!ht]
\center
\includegraphics[scale=0.4]{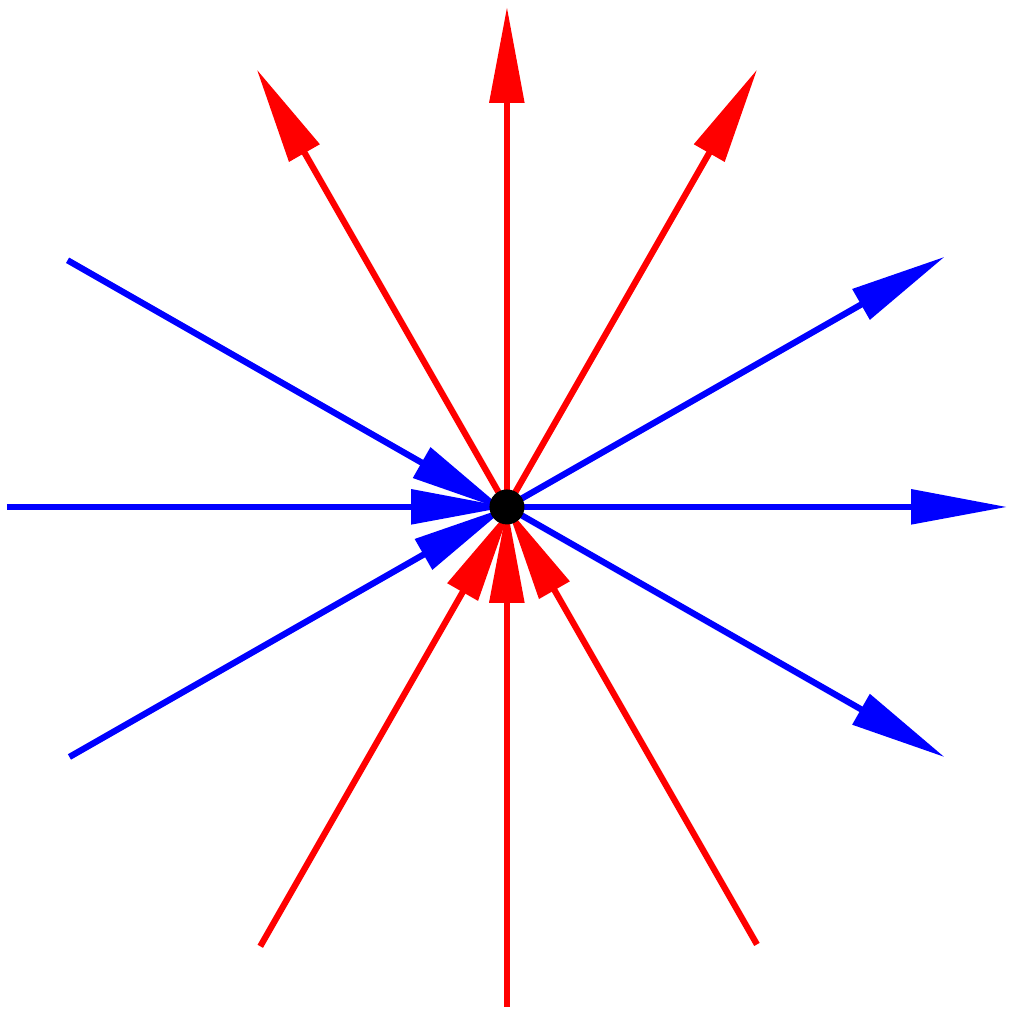} 
\caption{The (transversal structure) local property.}
\label{fig:ruleST}
\end{figure}

Then the definition of toroidal transversal structure is the
following:

\begin{definition}[Toroidal transversal structure]
  \label{def:TTS}\ \\
  Given a toroidal map $G$, a \emph{toroidal transversal structure} of
  $G$ is an orientation and coloring of the edges of $G$ with the
  colors blue and red %$0$, $1$
  where every vertex satisfies the transversal structure local
  property.
\end{definition}

See Figure~\ref{fig:balancedTTS} for an example of a toroidal
transversal structure of the triangulation of
Figure~\ref{fig:anglegraph}.
 
\begin{figure}[!ht]
\center
\includegraphics[scale=0.5]{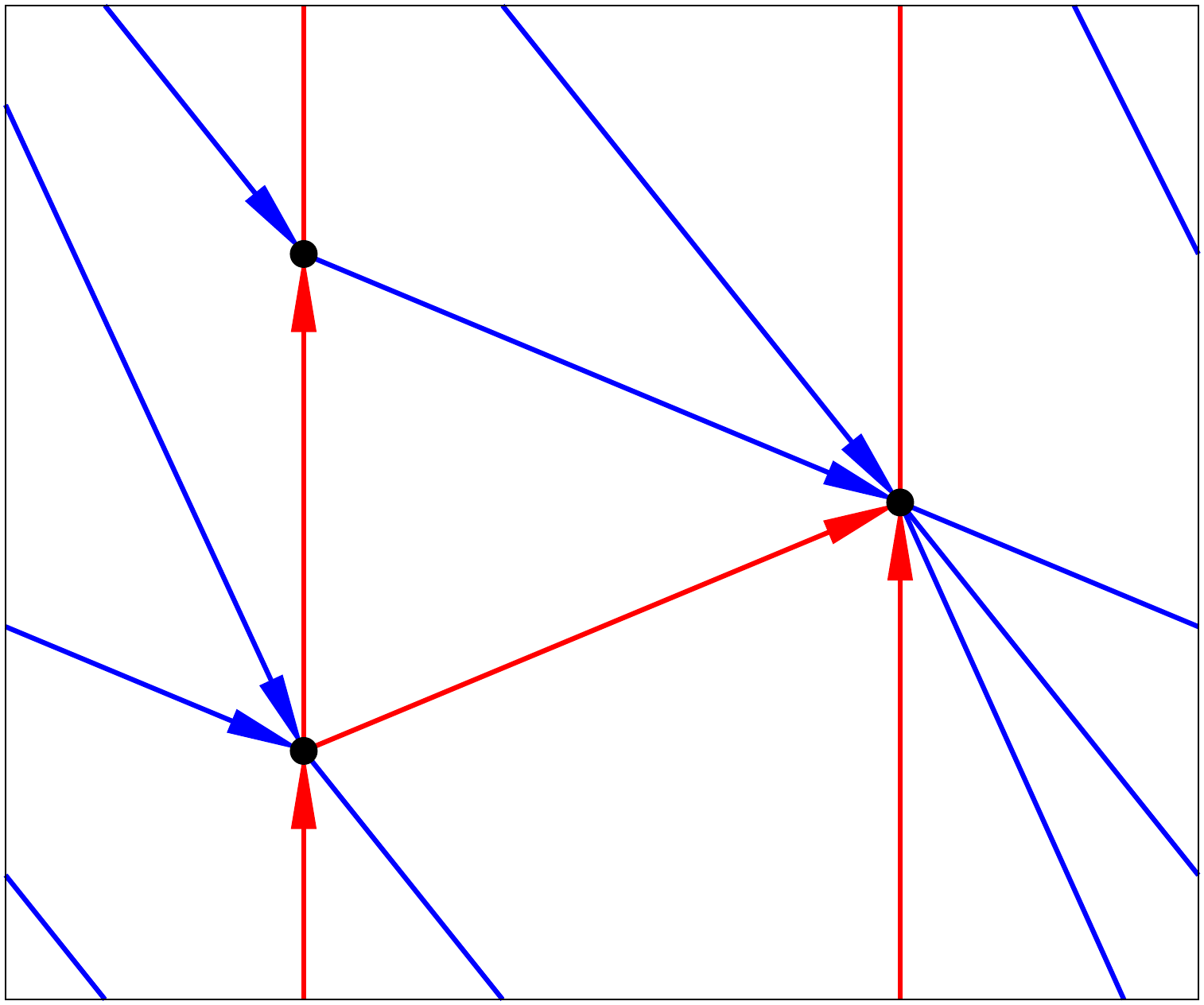}
\caption{Example of a toroidal transversal structure.}
\label{fig:balancedTTS}
\end{figure}

From a toroidal transversal structure of a toroidal triangulation $G$,
one can deduce a $4$-orientation of its angle map $A(G)$ by the
following rule applied around each primal-vertex (see
Figure~\ref{fig:ruleangle}): an edge $e$ of $A(G)$ is oriented toward
its primal-vertex if the two primal edges around $e$ share the same
color otherwise $e$ is oriented toward its dual-vertex. The fact that
primal-vertices of $A(G)$ gets outdegree $4$ is clear by the
definition of transversal structure. The fact that dual-vertices gets
outdegree $1$ is due to the property that, by the local rule, all
(triangular) faces of $G$ looks like one of Figure~\ref{fig:ruleface}
where the four cases are symmetric by rotation of the order (outgoing
blue, outgoing red, incoming blue, incoming red).

\begin{figure}[!ht]
\center
%\includegraphics[scale=0.4]{rule-tts.eps} 
%\ \ \ \
\includegraphics[scale=0.4]{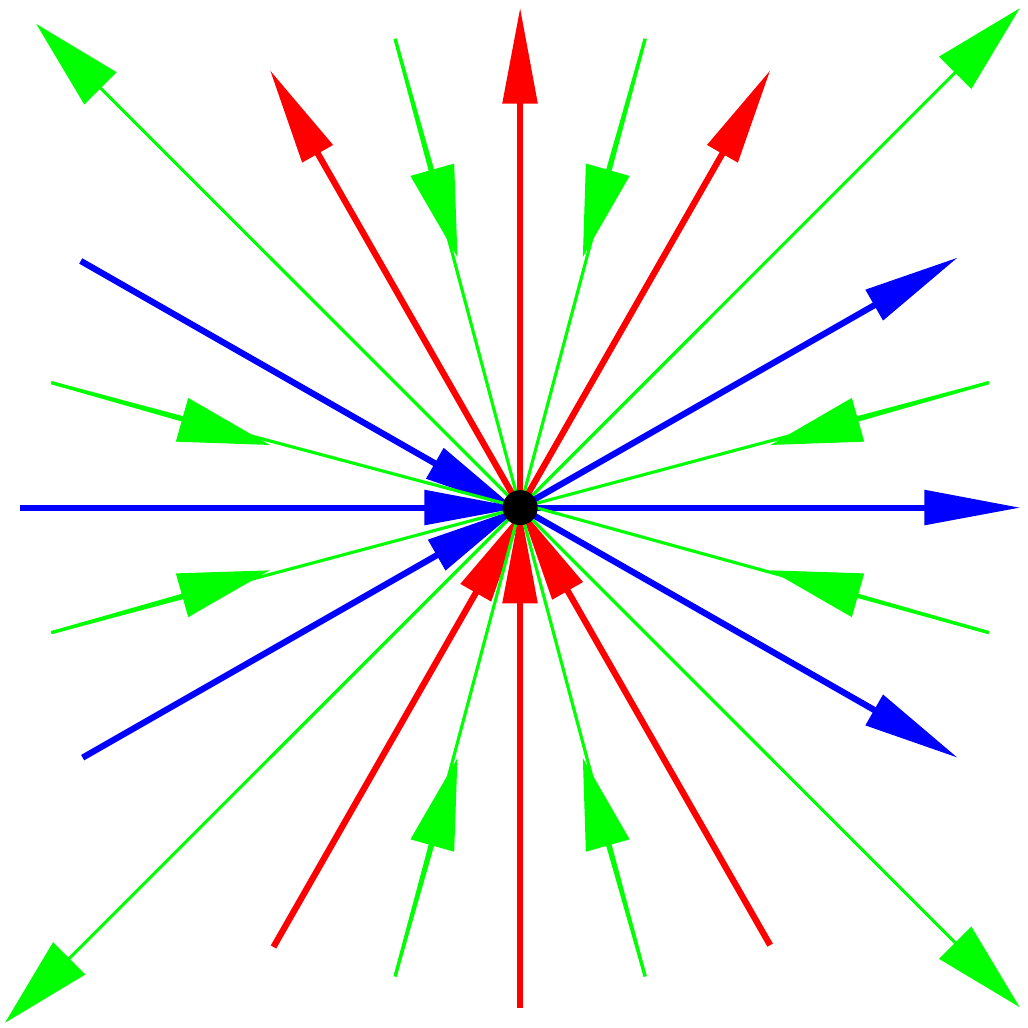}
\caption{Orientation of the angle map corresponding to a transversal
  structure.}
\label{fig:ruleangle}
\end{figure}

\begin{figure}[!ht]
\center
\includegraphics[scale=0.4]{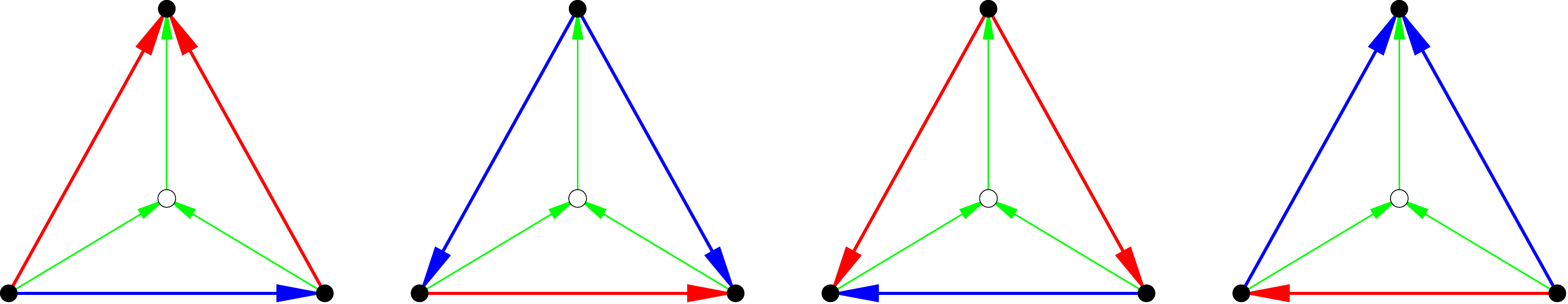} 
\caption{The four possible faces in a transversal structure and the
  corresponding orientation of the angle map.}
\label{fig:ruleface}
\end{figure}

The $4$-orientation on the right of Figure~\ref{fig:4orbalanced} is
the one obtained from Figure~\ref{fig:balancedTTS} by the rule of
Figure~\ref{fig:ruleangle}.

In the plane, there is a bijection between transversal structures of a
map and $4$-orientations of its angle map (see~\cite{Fus09}). This is
not true in the toroidal case. For example, there is no transversal
structure associated to the (non-balanced) $4$-orientation of the left
of Figure~\ref{fig:4orbalanced}.

Like it has been done in~\cite[Theorems 3.7]{GKL15} (see
also~\cite[Section 4.2]{LevHDR}) for toroidal Schnyder woods, it is
possible to characterize which $4$-orientations of the angle map of a
toroidal triangulation corresponds a transversal structure. This is
done in Section~\ref{sec:char}. A consequence of such a
characterization (see Corollary~\ref{cor:bal4orTS}) is that if a
$4$-orientation is balanced, then it corresponds to a transversal
structure.

So the balanced property is a sufficient condition to
corresponds to a transversal structure. Note that it is not a
necessary condition. Figure~\ref{fig:nonbalancedTTS} gives an example
of a transversal structure of a toroidal triangulation whose
corresponding $4$-orientation of its angle map is not balanced. The
horizontal cycle (with direction of traversal from right to left) has $\gamma= 8$.

\begin{figure}[!ht]
\center
\includegraphics[scale=0.4]{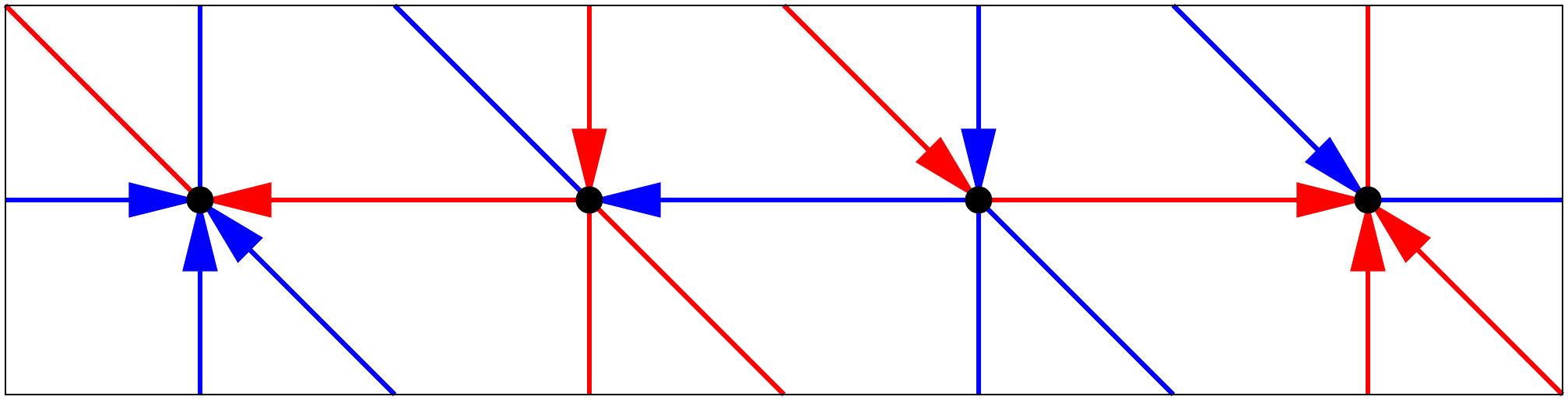} 

\ \\

\includegraphics[scale=0.4]{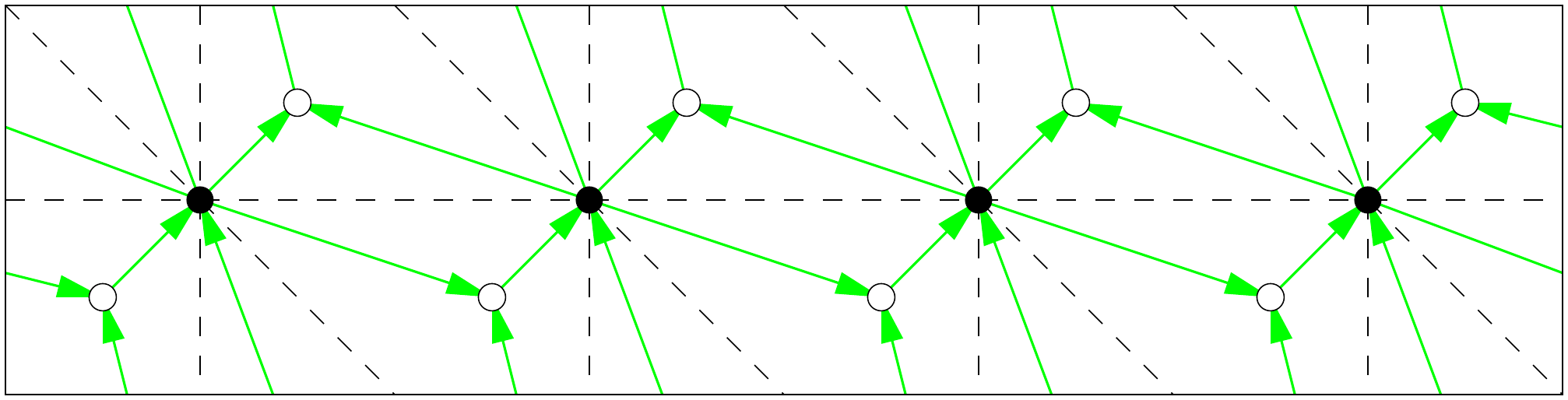}
\caption{Example of a toroidal transversal structure whose
  corresponding $4$-orientation of  angle map is not balanced.}
\label{fig:nonbalancedTTS}
\end{figure}

Note also that in the plane, transversal structures can be defined by
omitting the orientation of the edges in the local property since
there is a bijection with $4$-orientations of the angle map. Again,
this is not the case in the
torus. Figure~\ref{fig:redbluenotorientable} gives an example of a
blue/red coloring of the edges of a toroidal triangulation satisfying
the local rule of transversal structure, without the orientation of
the edges. It is not possible to orient the edges so this coloring
becomes a toroidal transversal structure. The corresponding
orientation of the angle map is still a $4$-orientation. Note that by
 gluing two copies of this example, one obtains
Figure~\ref{fig:nonbalancedTTS} that becomes orientable.

\begin{figure}[!ht]
\center
\includegraphics[scale=0.4]{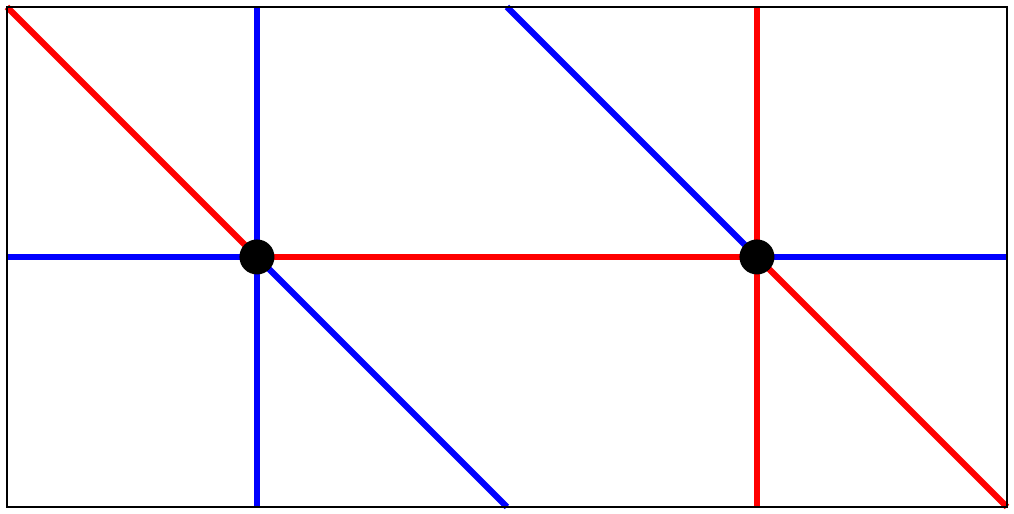} 

\ \\

\includegraphics[scale=0.4]{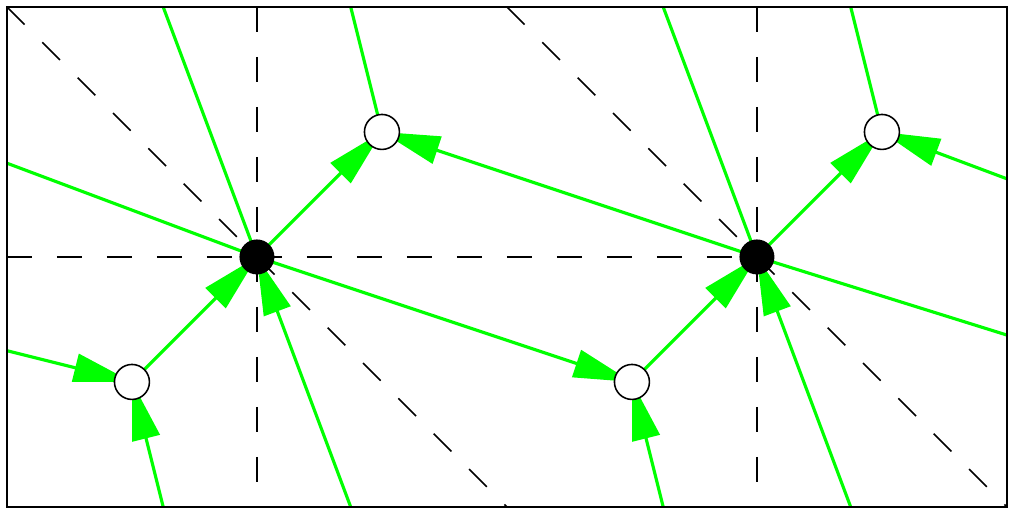}
\caption{Example of a blue/red coloring of the edges of a toroidal
  triangulation satisfying the local rule of transversal structure,
  without the orientation of the edges, and that is not a transversal
  structure.}
\label{fig:redbluenotorientable}
\end{figure}

We give the following definition of \emph{balanced} for toroidal transversal
structure:

\begin{definition}[Balanced toroidal transversal structure]\ \\
  A toroidal transversal structure is \emph{balanced} if its
corresponding $4$-orientation of  angle map is balanced.
\end{definition}

Figure~\ref{fig:balancedTTS} gives an example of a balanced toroidal
transversal structure. The corresponding $4$-orientation of the angle
map is the balanced $4$-orientation of the right of Figure~\ref{fig:4orbalanced}.

In section~\ref{sec:existence} we prove the existence of balanced
toroidal transversal structure for essentially 4-connected toroidal
triangulations. This implies the existence of balanced
$4$-orientations for their angle maps.

\subsection{Transversal structures in the universal cover}
\label{sec:universal}
% The universal cover of the torus is a surjective mapping $p$ from the
% plane to the surface that is locally a homeomorphism.  If the torus is
% represented by a hexagon in the plane whose opposite sides are
% pairwise identified, then the universal cover of the torus is obtained
% by replicating the hexagon to tile the plane.

Consider a toroidal map $G$ and its universal cover
$G^\infty$. Note that $G$ does not have contractible loops nor
homotopic multiple edges if and only if $G^\infty$ is simple.

We need the following lemma from~\cite{GKL15}:

\begin{lemma}[{\cite[Lemma~2.8]{GKL15}}]
\label{lem:finitecc}
Suppose that for a finite set of vertices $X$ of $G^\infty$, the graph
$G^\infty\setminus X$ is not connected. Then $G^\infty\setminus X$ has
a finite connected component.
\end{lemma}

Suppose now that $G$ is a toroidal triangulation  given with a
transversal structure. Consider the natural extension of the
transversal structure of $G$ to $G^\infty$, where an edge of
$G^\infty$ receive the orientation and color of the corresponding edge
in $G$.  Let $G^\infty_B$, $G^\infty_R$ be the directed subgraphs of
$G^\infty$ induced by the edges of color blue and red, respectively.
The graphs $(G^\infty_B)^{-1}$ and $(G^\infty_R)^{-1}$ are the graphs obtained
from $G^\infty_B$ and $G^\infty_R$ by reversing all their edges.  Similarly to what
happens for Schnyder woods (see~\cite[Lemma~6]{LevHDR}) we have the following property:

\begin{lemma}
  \label{lem:nodirectedcycle}
  The graphs $G^\infty_B\cup G^\infty_R$ and
  $G^\infty_B\cup (G^\infty_R)^{-1}$ contain no directed cycle.
\end{lemma}

\begin{proof}
  Let us prove the property for $G^\infty_B\cup G^\infty_R$, the proof
  is similar for $G^\infty_B\cup (G^\infty_R)^{-1}$.  Suppose by
  contradiction that there is a directed cycle in
  $G^\infty_B\cup G^\infty_R$.  Let $C$ be such a cycle containing the
  minimum number of faces in the finite map $D$ with border
  $C$. Suppose without loss of generality that $C$ turns around $D$
  counterclockwisely. By the transversal structure local property,
  every vertex of $D$ has at least one outgoing edge of color red in
  $D$.  So there is a cycle of color red in $D$ and this cycle is $C$
  by minimality of $C$.  Every vertex of $D$ has at least one incoming
  edge of color blue in $D$. So, again by minimality of $C$, the cycle
  $C$ is a cycle of color blue. This contradicts the fact that edges
  of $G^\infty$ have a unique color.
\end{proof}

For a vertex $v$ of $G^\infty$, we define  $P_0(v)$ (resp. $P_1(v)$, $P_2(v)$, $P_3(v)$)  the subgraph of
$G^\infty$ obtained by keeping all the edges that are on an oriented
path  of
$G^\infty_B$ (resp. $G^\infty_R$, $(G^\infty_B)^{-1}$, $(G^\infty_R)^{-1}$)
starting at $v$.
Then we have the following lemma: 

\begin{lemma}
  \label{lem:nocommongeneral}
  For every vertex $v$ and $i,j\in\{0,1,2,3\}$, $i\neq j$, the two
  subgraphs $P_{i}(v)$ and $P_{j}(v)$  of $G^\infty$ have $v$ as only common
  vertex.
\end{lemma}

\begin{proof}
  If $P_{i}(v)$ and $P_{j}(v)$ intersect on two vertices, then
   $G^\infty_B\cup G^\infty_R$ or
  $G^\infty_B\cup (G^\infty_R)^{-1}$ contains a directed cycle,
  contradicting Lemma~\ref{lem:nodirectedcycle}.
\end{proof}

Now we can prove that the existence of a transversal structure for a
toroidal triangulation implies the $4$-connectedness of its universal
cover:

\begin{lemma}
\label{lem:e4c}
If a toroidal triangulation admits a toroidal transversal structure,
then it is essentially 4-connected.
\end{lemma}

\begin{proof}
  Suppose by contradiction that there exists three vertices $x,y,z$ of
  $G^\infty$ such that $G'=G^\infty\setminus\{x,y,z\}$ is not
  connected. Then, by Lemma~\ref{lem:finitecc}, the graph $G'$ has a
  finite connected component $R$. Let $v$ be a vertex of $R$. By
  Lemma~\ref{lem:nodirectedcycle}, for $i\in\{0,1,2,3\}$, the infinite
  and acyclic graph $P_{i}(v)$ does not lie in $R$ so it intersects
  one of $x,y,z$. So for two distinct $i,j$, the two graphs $P_{i}(v)$
  and $P_{j}(v)$ intersect on a vertex distinct from $v$, a
  contradiction to Lemma~\ref{lem:nocommongeneral}.
\end{proof}

In Section~\ref{sec:existence}, we prove the converse of
Lemma~\ref{lem:e4c} (see Theorem~\ref{th:existence}).

 A \emph{separating
  triangle}  of a map is a triangle whose interior is non empty.  We have the
following equivalence:

\begin{lemma}
\label{lem:e4ciffnst}
A toroidal triangulation is essentially $4$-connected if and only if
its universal cover has no
separating triangle.
\end{lemma}

\begin{proof}
  ($\Longrightarrow$) Consider an essentially $4$-connected toroidal
  triangulation $G$. So $G^\infty$ is $4$-connected. If
  $G^\infty$ has a separating triangle, then, the three vertices of the
  triangle form a contradiction to the $4$-connectedness of
  $G^\infty$. So $G^\infty$ has no separating triangle.

  ($\Longleftarrow$) Consider a toroidal triangulation $G$ such that
  $G^\infty$ has no separating triangle. Suppose  by contradiction that
  $G^\infty$ is not $4$-connected. Then there exists a set of $3$
  vertices $X=\{x,y,z\}$ such that $G^\infty\setminus X$ is not
  connected. By Lemma~\ref{lem:finitecc}, the graph $G^\infty\setminus X$
  has a finite connected component $R$. Let $F$ be the face of
  $G^\infty\setminus R$ containing $R$.  If $F$ has length $1$ or $2$
  then $G^\infty$ is not simple, a contradiction. If $F$ has size $3$,
  then $F$ is a separating triangle of $G^\infty$, a contradiction. So
  $F$ has size at least $4$.  Then there exists a vertex $v$ in
  $F\setminus X$. There is no edges between $v$ and $R$ and thus in
  $G^\infty$ the face incident to $v$ and $R$ has length strictly more
  than $3$, a contradiction of $G$ being a triangulation.
\end{proof}

We say that a quadrangle is \emph{maximal} (by inclusion) if its
interior is not strictly contained in the interior of another
quadrangle. 

\begin{lemma}
\label{lem:uniqueQuadrangle}
Consider an essentially $4$-connected toroidal triangulation $G$ and
an edge $e$ of $G$. Then there is a unique maximal quadrangle of
$G$ whose interior contains $e$.
\end{lemma}

\begin{proof}
Since $G$ is a toroidal triangulation, $e$ is clearly contained in the
interior of the quadrangle bordering its two incident faces. So $e$ is
contained in a maximal quadrangle.

Suppose by contradiction that there exist two distinct maximal
quadrangles $Q,Q'$ such that their interiors contain $e$.  Let $R,R'$
denote the interior of $Q,Q'$ respectively. The two region $R,R'$ are
distinct, both contain the two faces incident to $e$ plus some other
faces.  If there is an edge in $R$ (resp. $R'$) connecting two
opposite vertices of $Q$ (resp. $Q'$), then $G^\infty$ contains a
separating triangle, a contradiction of $G$ being essentially
$4$-connected and Lemma~\ref{lem:e4ciffnst}.  Then since there is no
homotopic multiple edges in $G$, there is at least one or two vertices
of $Q$ (resp $Q'$) in the interior of $Q'$ (resp. $Q$).  Thus the
border of the union of $R$ and $R'$ has size less or equal to four, a
contradiction to the maximality of $Q,Q'$ or of $G$ being an
essentially $4$-connected triangulation.
 \end{proof}

\section{Characterization of  orientations corresponding to transversal structures}
\label{sec:char}

\subsection{Transversal structure labeling}
\label{sec:tslab}
We need the following equivalent definition of toroidal transversal structures:

\begin{definition}[Toroidal transversal structure labeling]
  \label{def:tslab}\ \\
  Given a toroidal map $G$, a \emph{toroidal transversal structure labeling}
  (or \emph{TTS-labeling} for short) of $G$ is a labeling of the
  half-edges of $G$ with integers $0,1,2,3$ (considered modulo $4$)
  such that each edge is labeled with two integers that differ
  exactly by $(2 \bmod 4)$ and around each vertex the labeling form in
  counterclockwise order four non-empty intervals of $0$, $1$, $2$,
  $3$.
\end{definition}

Consider a toroidal map $G$.  The mapping of
Figure~\ref{fig:edgelabeling}, where an outgoing half-edge blue is
labeled $0$, an outgoing half-edge red is labeled $1$, an incoming
half-edge blue is labeled $2$, and an incoming half-edge red is
labeled $3$, shows how to see a toroidal transversal structure of $G$
as a TTS-labeling of $G$ and vice-versa. The two objects are indeed the
same.

\begin{figure}[!ht]
\center
\includegraphics[scale=0.4]{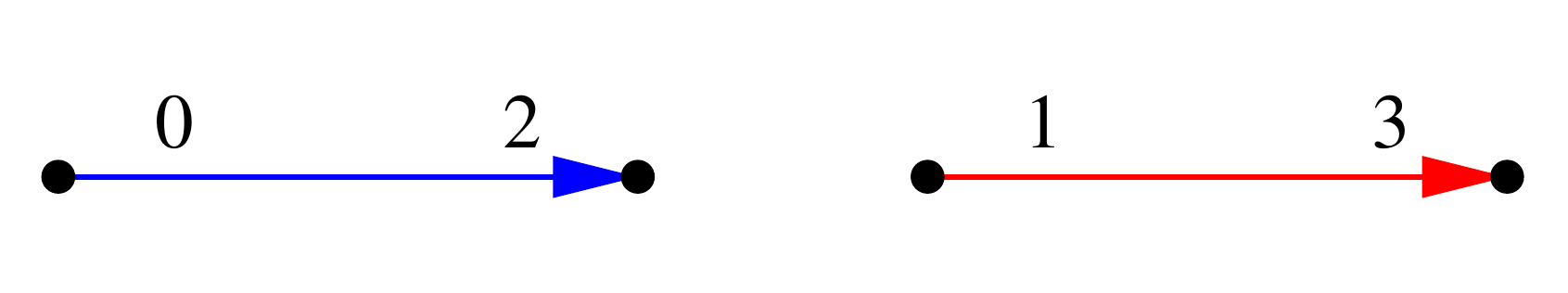}
\caption{Mapping between $TTS$-labelings and transversal structures.}
\label{fig:edgelabeling}
\end{figure}

Figure~\ref{fig:example-labeling} shows the TTS-labeling corresponding
to the transversal structure of Figure~\ref{fig:balancedTTS}. This
labeling is also represented on the corresponding orientation of its
angle map.

\begin{figure}[!ht]
\center
\includegraphics[scale=0.4]{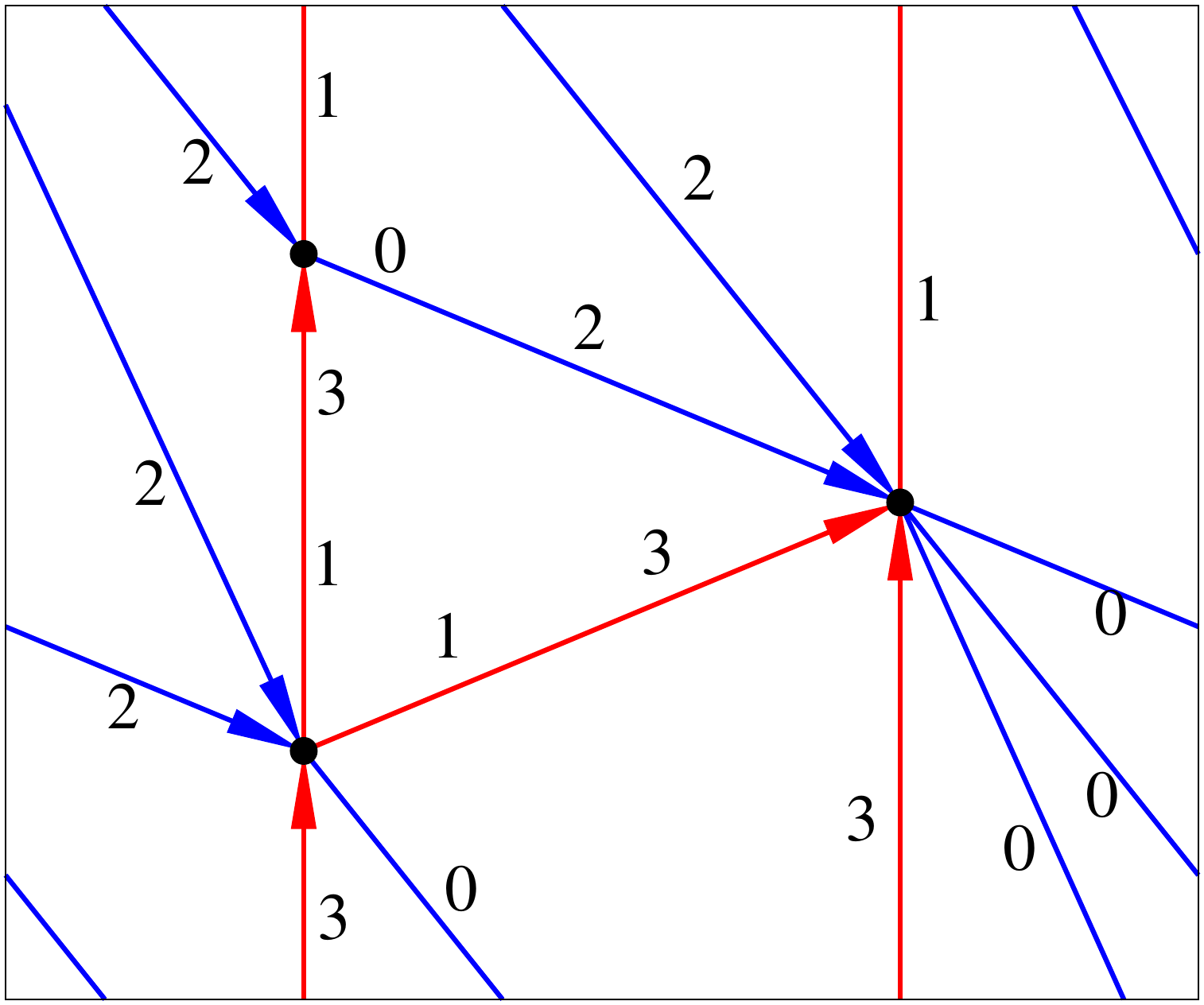} \ \ \ \
\includegraphics[scale=0.4]{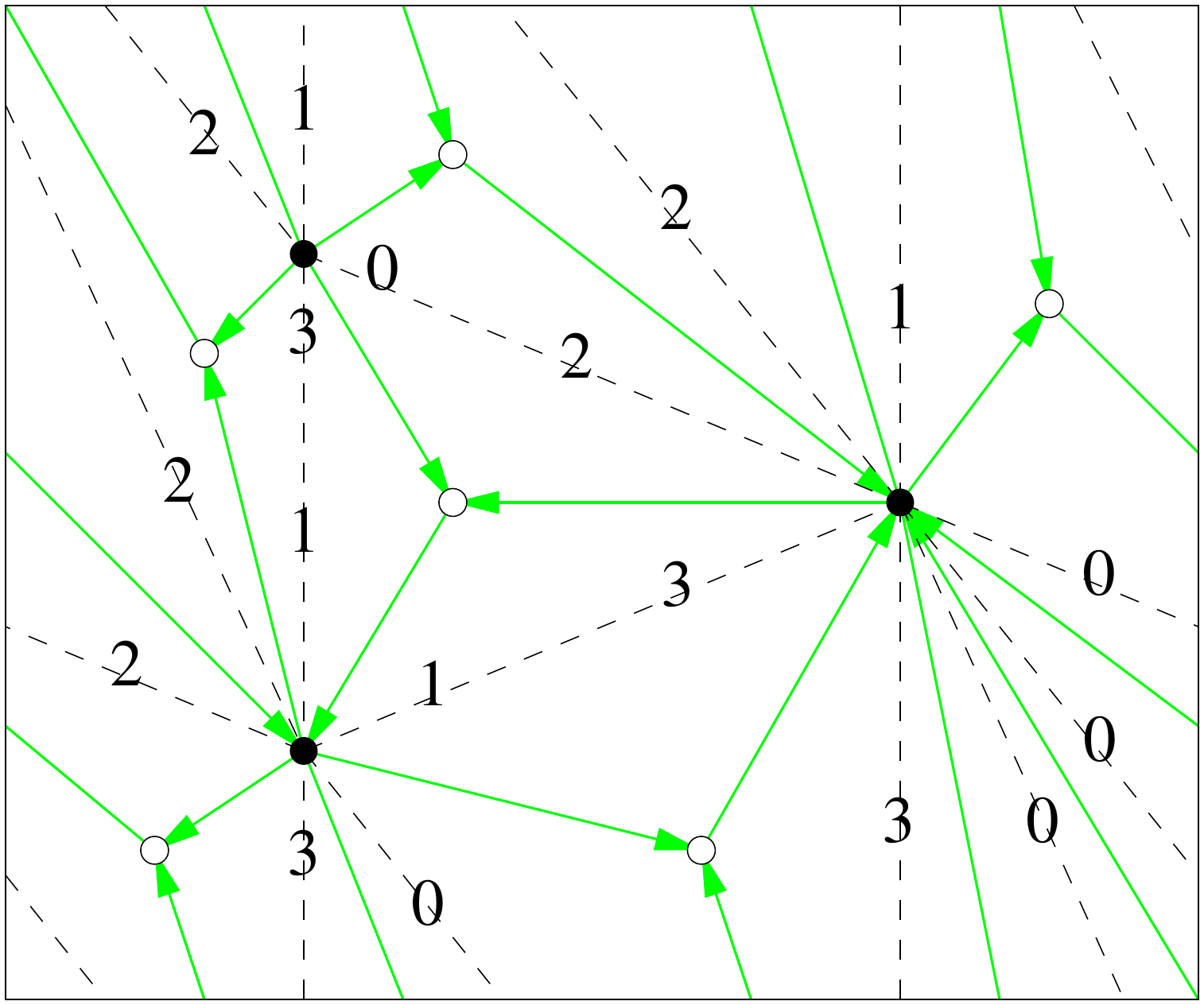}
\caption{Labeling of the half-edges of a transversal structure and
  in the  angle map.}
\label{fig:example-labeling}
\end{figure}

We say that a 4-orientation of $A(G)$ \emph{admits a
  TTS-labeling} if there is a labeling of the angles of the
primal-vertices of $A(G)$ such that this labeling corresponds to a
transversal structure of $G$  (as on
Figure~\ref{fig:example-labeling}).

The goal of this section is to
characterize which are the 4-orientations that admit TTS-labelings.
For that purpose we have to introduce some more formalism, 
similarly to what is done for Schnyder woods (see~\cite{GKL15,LevHDR}).

\subsection{A bit of homology}
\label{sec:homology}

We need a bit of surface homology of general
maps which we  discuss now. 

Consider a map $G=(V,E)$, on an orientable surface of genus $g$, given
with an arbitrary orientation of its edges. This fixed arbitrary
orientation is implicit in all the paper and is used to handle flows.
A \emph{flow} $\phi$ on $G$ is a vector in $\mb{Z}^{E}$. For any
$e\in E$, we denote by $\phi_e$ the coordinate $e$ of $\phi$.

A \emph{walk} $W$ of $G$ is a sequence of edges with a direction of
traversal such that the ending point of an edge is the starting point
of the next edge.  A walk is \emph{closed} if the start and end
vertices coincide.  A walk has a \emph{characteristic flow} $\phi(W)$
defined by:

$$\phi(W)_e=\#\text{times }W\text{ traverses } e \text{ forward} - \#\text{times
}\\
W\text{ traverses } e \text{ backward}.$$

This definition naturally extends to sets of walks.  From now on we
consider that a set of walks and its characteristic flow are the same
object and by abuse of notation we can write $W$ instead of
$\phi(W)$. We do the same for \emph{oriented subgraphs}, i.e.,
subgraphs that can be seen as a set of walks of unit length.

A \emph{facial walk} is a closed walk bounding a face.  Let $\mc{F}$
be the set of counterclockwise facial walks and let
$\mb{F}={<}\phi(\mc{F}){>}$ be the subgroup of $\mb{Z}^E$ generated by
$\mc{F}$.  Two flows $\phi, \phi'$ are \emph{homologous} if
$\phi -\phi' \in \mb{F}$. 
%They are \emph{reversely homologous} if
%$\phi +\phi' \in \mb{F}$. 
They are \emph{weakly homologous} if $\phi -\phi' \in \mb{F}$ or
$\phi +\phi' \in \mb{F}$.  We say that a flow $\phi$ is $0$-homologous
if it is homologous to the zero flow, i.e.  $\phi \in \mb{F}$.

Let $\mc{W}$ be the set of closed walks and let
$\mb{W}={<}\phi(\mc{W}){>}$ be the subgroup of $\mb{Z}^E$ generated by
$\mc{W}$.  The group $H(G)=\mb{W}/\mb{F}$ is the \emph{first homology
  group} of $G$. It is well-known that $H(G)$ only depends on the
genus of the map, and actually it is isomorphic to $\mb{Z}^{2g}$.

A set $\{B_1,\ldots,B_{2g}\}$ of (closed) walks of $G$ is said to be a
\emph{basis for the homology} if the equivalence classes of their
characteristic vectors $([\phi(B_1)],\ldots,[\phi(B_{2g})])$ generate
$H(G)$.  Then for any closed walk $W$ of $G$, we have
$W=\sum_{F\in\mc{F}}\lambda_FF+\sum_{1\leq i\leq 2g}\mu_iB_i$ for some
$\lambda\in\mathbb{Z}^\mc{F},\mu\in\mathbb{Z}^{2g}$. Moreover one of the
$\lambda_F$ can be set to zero (and then all the other coefficients
are unique).  

For any map, there exists a set of cycles that forms a basis for the
homology and it is computationally easy to build. A possible way to do
this is by considering a spanning tree $T$ of $G$, and a spanning tree
$T^*$ of $G^*$ that contains no edges dual to $T$.  By Euler's
formula, there are exactly $2g$ edges in $G$ that are not in $T$ nor
dual to edges of $T^*$. Each of these $2g$ edges forms a unique cycle
with $T$. It is not hard to see that this set of cycles, given with
any direction of traversal, forms a basis for the homology. Moreover,
note that the intersection of any pair of these  cycles is
either a single vertex or a common path.

The edges of the dual map $G^*$ of $G$ are oriented such that the dual
$e^*$ of an edge $e$ of $G$ goes from the face on the right of $e$ to
the face on the left of $e$.  Let $\mc{F}^*$ be the set of
counterclockwise facial walks of $G^*$.  Consider
$\{B^*_1,\ldots,B^*_{2g}\}$ a set of closed walks of $G^*$ that form a
basis for the homology. Let $p$ and $d$ be flows of $G$ and $G^*$,
respectively. We define the following:
$$\beta(p,d)=\sum_{e\in G}p_e d_{e^*}.$$ 

Note that $\beta$ is a
bilinear function.  We need the following lemma from~\cite{GKL15}:

\begin{lemma}[{\cite[Lemma~3.1]{GKL15}}]
\label{lm:homologous}
Given  two flows $\phi,\phi'$ of $G$, the following properties are
equivalent to each other:

\begin{enumerate}
\item The two flows $\phi, \phi'$ are homologous.
\item For any closed walk $W$ of $G^*$ we have
  $\beta(\phi,W)=\beta(\phi',W)$.
\item For any $F\in \mc{F^*}$, we have $\beta(\phi,F)=\beta(\phi',F)$, and,
  for any $1\le i\le 2g$, we have $\beta(\phi,B^*_i)=\beta(\phi',B^*_i)$.
\end{enumerate}
\end{lemma}

\subsection{The angle-dual-completion}

Consider a toroidal triangulation $G$. 
The angle-dual-completion $\PDC{A(G)}$ of $G$ is the map obtained from
simultaneously embedding $A(G)$ and $G^*$ and subdividing each edge of
$G^*$ by adding a vertex at its intersection with the corresponding
primal-edge of $G$ (see
Figure~\ref{fig:example-labeling-completion}). In $\PDC{A(G)}$ there
are three types of vertices called \emph{primal-},
\emph{dual-} and \emph{edge-vertices}, represented, respectively, in black, white, and
gray on the figures. There are  two types of edges called \emph{angle-} and
\emph{dual-edges}. Each angle-edge is between a primal- and
dual-vertex. Each dual-edge is between a dual- and an
edge-vertex. Since $G$ is a triangulation, each
dual-vertex is incident to three angle-edges and three dual-edges. Each
edge-vertex is incident to two dual-edges. Each
face of $\PDC{A(G)}$ represents a half-edge of $G$ and is a quadrangle
incident to one primal-vertex, two dual-vertices and one edge-vertex.

Given an orientation of the angle map $A(G)$, this orientation
naturally extends to an orientation of the angle-dual-completion
$\PDC{A(G)}$ where angle-edges get the orientation they have in $A(G)$
and dual-edges are oriented from the edge-vertex to the dual-vertex. A
4-orientation of $\PDC{A(G)}$ is an orientation of its edges that
corresponds to a 4-orientation of $A(G)$, i.e. primal-vertices have
outdegree exactly $4$, dual-vertices have out-degree exactly $1$ and
edge-vertices have outdegree exactly $2$.

A TTS-labeling of $G$ can be represented on $\PDC{A(G)}$ by putting labels into
faces of $\PDC{A(G)}$ (see
Figure~\ref{fig:example-labeling-completion}).  When crossing an angle-edge
that is incoming for a primal-vertex, the label does not change. When crossing an
angle-edge that is outgoing for a primal-vertex, the label changes by
$\pm 1$ depending on the orientation of this angle-edge: from left to right
$(+ 1 \bmod 4)$ or right to left $(- 1 \bmod 4)$.  When crossing a
dual-edge the label changes by $\pm 2$, and the orientation is not
relevant since $-2 \bmod 4=+2 \bmod 4$. 

\begin{figure}[!ht]
\center
\includegraphics[scale=0.4]{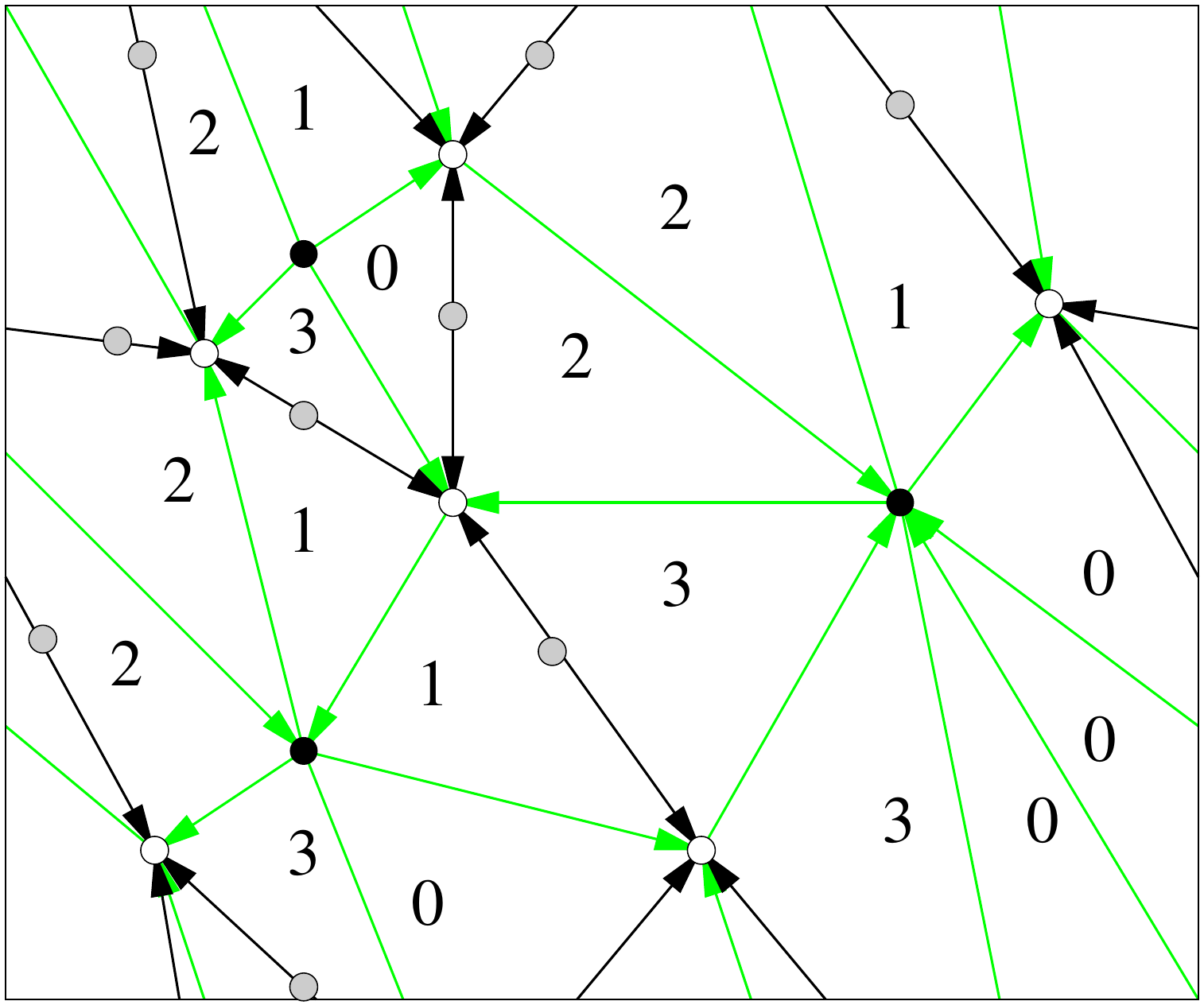}
\caption{Orientation and labeling of the angle-dual-completion
  corresponding to Figure~\ref{fig:example-labeling}.}
\label{fig:example-labeling-completion}
\end{figure}

Let $\Out$ be the set of
edges of $\PDC{A(G)}$ which are going from a primal-vertex to a
dual-vertex. We call these edges \emph{out-edges} of $\PDC{A(G)}$.
Let $\Dual$ be the set of dual-edges of $\PDC{A(G)}$.  For $\phi$ a
flow of the dual of the angle-dual-completion $\PDC{A(G)}^*$,
% (see
%Figure~\ref{fig:dual-adc})
 we
define $\delta(\phi)=\beta(\Out,\phi)+2\beta(\Dual,\phi)$.  More
intuitively, if $W$ is a walk of $\PDC{A(G)}^*$, then:
$$
\begin{array}{ll}
  \delta(W)  = &  \ \ \#\text{out-edges crossing }W\text{
                 from left to right}\\
               & -\#\text{out-edges crossing }W\text{ from right to
                 left}\\  
                &  + \ 2\times \#\text{dual-edges crossing }W\text{
                 from left to right}\\
               & - \ 2 \times\#\text{dual-edges crossing }W\text{ from right to
                 left}\\  
\end{array}
$$

The bilinearity of $\beta$ implies the linearity of $\delta$.

The following lemma gives a necessary and sufficient condition
for a 4-orientation of the angle map to admit a TTS-labeling.

\begin{lemma}
 \label{lem:charforall}
 A 4-orientation of ${A(G)}$ admits a TTS-labeling if and only if
 any closed walk $W$ of $\PDC{A(G)}^*$ satisfies
 $\delta (W) = 0 \bmod 4$.
\end{lemma}

\begin{proof} $(\Longrightarrow)$ Consider a TTS-labeling $\ell$ of
  $A(G)$.  The definition of $\delta$ is such that $\delta$ modulo $4$
  counts the variation of the labels when going from one face of
  $\PDC{A(G)}$ to another face of $\PDC{A(G)}$.  Thus for any walk $W$
  of $\PDC{A(G)}^*$ from a face $F$ to a face $F'$, the value of
  $\delta(W) \bmod 4$ is equal to $\ell(F')-\ell(F)\bmod 4$.  Thus if
  $W$ is a closed walk then $\delta(W)= 0 \bmod 4$.

$(\Longleftarrow)$ Consider a 4-orientation of $\PDC{A(G)}$ such that
any closed
walk $W$ of $\PDC{A(G)}^*$ satisfies $\delta (W) = 0 \bmod 4$.  Pick any
face $F_0$ of $\PDC{A(G)}$ and label it $0$. Consider any face $F$ of
$\PDC{A(G)}$ and a path $P$ of $\PDC{A(G)}^*$ from $F_0$ to $F$. Label $F$
with the value $\delta (P)\bmod4$. Note that the label of $F$ is
independent from the choice of $P$ as for any two paths $P_1, P_2$
going from $F_0$ to $F$, we have $\delta (P_1) =\delta (P_2) \bmod4$
since $\delta (P_1 - P_2) = 0 \bmod4$ as $P_1- P_2$ is a closed walk.

Consider a primal-vertex $v$ of $\PDC{A(G)}$.  By assumption
$d^+(v)=4$ so the labels around $v$ form in \ccw order four non-empty
intervals of $0$, $1$, $2$, $3$.  Moreover, the labels of the two faces
incident to an edge-vertex differ by $(2 \bmod 4)$. So the obtained
labeling corresponds to a TTS-labeling of $G$.
\end{proof}

In the next section we study properties of $\delta$ w.r.t.~homology in
order to simplify the condition of Lemma~\ref{lem:charforall} that
concerns any closed walk of $\PDC{A(G)}^*$. We also replace the
condition on $\delta$ to a condition on $\gamma$ that is simpler to
handle.

% \begin{figure}[!ht]
% \center
% %\includegraphics[scale=0.4]{angle-graph-st-0.eps}
% %\includegraphics[scale=0.4]{angle-completion-1.eps} \ \ \ \
% %\includegraphics[scale=0.4]{angle-completion-2.eps}
% \includegraphics[scale=0.4]{dual-adc.eps}
% \caption{$\PDC{A(G)}^*$, the dual of the angle-dual-completion
%   corresponding to Figure~\ref{fig:example-labeling-completion}. The graph $G$  is represented by black vertices and dotted lines.}
% \label{fig:dual-adc}
% \end{figure}

\subsection{Characterization theorem}
\label{sec:characterization}

Consider a toroidal triangulation $G$.  Let $\widehat{\mc{F}^*}$ be
the set of counterclockwise facial walks of the angle-dual-completion
$\PDC{A(G)}^*$.

We have the following lemmas:

\begin{lemma}
\label{lem:facedelta0}
In a $4$-orientation of $\PDC{A(G)}$, any $F\in\widehat{\mc{F}^*}$
satisfies $\delta(F)=0\bmod4$.
\end{lemma}
\begin{proof}
  If $F$ corresponds to a primal-vertex $v$ of $\PDC{A(G)}$, then $v$ has
  outdegree exactly $4$. So $\delta(F)=4=0\bmod4$.

 If $F$ corresponds to a dual-vertex $v$ of $\PDC{A(G)}$, then $v$ is
 incident to three angle-edges, exactly two of which are incoming (and
 thus in Out), and
 incident to three incoming dual-edges. So $\delta(F)= -2\times 1 - 3 \times 2
 = -8=0\bmod4$.

 If $F$ corresponds to an edge-vertex $v$ of $\PDC{A(G)}$, then $v$ is
 incident to two outgoing dual-edges. So $\delta(F)=2\times 2=4=
0\bmod4$.
\end{proof}

\begin{lemma}
\label{lem:basedelta}
In a $4$-orientation of $\PDC{A(G)}$, if $\{B_1,B_{2}\}$ is a pair of
cycles of $\PDC{A(G)}^*$, given with a direction of traversal, that
forms a basis for the homology, then for any closed walk $W$ of
$\PDC{A(G)}^*$ homologous to $\mu_1 B_1 + \mu_{2} B_{2}$,
$\mu\in\mathbb{Z}^{2}$ we have
$\delta(W)= \mu_1 \delta(B_1) + \mu_{2} \delta(B_{2}) \bmod 4$.
\end{lemma}

\begin{proof} We have
  $W=\sum_{F\in\widehat{\mc{F}^*}}\lambda_FF+\mu_1B_1+\mu_2B_2$
  for some $\lambda\in\mathbb{Z}^{f}$, $\mu\in\mathbb{Z}^{2}$.  Then by linearity of $\delta$
  and Lemma~\ref{lem:facedelta0}, the lemma follows.
\end{proof}

Lemma~\ref{lem:basedelta} can be used to simplify the condition of
Lemma~\ref{lem:charforall} and show that if $\{B_1,B_{2}\}$ is a pair
of cycles of $\PDC{A(G)}^*$ that forms a basis for the homology, then
a $4$-orientation of $\PDC{A(G)}$ admits a TTS-labeling if and only
if $\delta(B_{i})= 0 \bmod4$, for $i\in\{1,2\}$. We prefer to
formulate such a result with function $\gamma$ that is simpler to handle
(see Theorem~\ref{th:characterizationgamma}).

Let $C$ be a cycle of $G$ with a direction of traversal. Let
$W_L(C)$ be the closed walk of $\PDC{A(G)}^*$ just on the left of $C$
and going in the same direction as $C$.  Note that since the faces of
$\PDC{A(G)}^*$ have exactly one incident vertex that is a
primal-vertex, the walk $W_L(C)$ is in fact a cycle of
$\PDC{A(G)}^*$. Similarly, let $W_R(C)$ be the cycle of $\PDC{A(G)}^*$
just on the right of $C$ and going in the same direction as $C$.

% \comment{A mettre plus tard ?}
% \begin{lemma}
% \label{lem:gammaequaldelta}
% Consider a $4$-orientation of ${A(G)}$ and 
% a cycle $C$ of $G$,
% % or $G^*$,
% then $\gamma(C) = \delta (W_L(C)) + \delta (W_R(C))$.
% \end{lemma}

% \begin{proof}
%   We consider the different cases that can occur. An angle-edge of
%   $\PDC{A(G)}$ that is entering a (primal-)vertex of $C$, is not
%   counting in either $\gamma(C),\delta (W_L(C)), \delta (W_R(C))$.  An
%   angle-edge of $\PDC{A(G)}$ that is leaving a (primal-)vertex of $C$
%   from its right side (resp. left side) is counting $+1$ (resp. $-1$)
%   for $\gamma(C)$ and $\delta (W_R(C))$ (resp.  $\gamma(C)$ and
%   $\delta (W_L(C))$). For each edge of $C$ there is a
%   corresponding edge-vertex in $\PDC{A(G)}$, that is incident to two
%   dual-edges of $\PDC{A(G)}$, one that is crossing $\delta (W_L(C))$
%   from right to left and one crossing $\delta (W_RL(C))$ from left to
%   right. So they compensate each other and finally, we obtain
%   $\gamma(C)=\delta(W_L(C))+\delta(W_R(C))$.
% \end{proof}

\begin{lemma}
%\label{lem:deltagammaequ0mod3}
\label{lem:gammaequaldelta}
Consider a $4$-orientation of ${A(G)}$ and 
a cycle $C$ of $G$, then we have:
$$\gamma(C) = \delta (W_L(C)) + \delta (W_R(C))$$
$$\delta (W_L(C))=0 \bmod4 \iff %\ \ \text{and}\ \ \delta (W_R(C))=0 \bmod4 \iff
\gamma(C) = 0 \bmod8
.$$ 
\end{lemma}

\begin{proof}
  Let $x_R$ (resp. $x_L$) be the number of edges of $A(G)$ leaving $C$
  on its right (resp. left). So $\gamma(C) = x_R-x_L$. Let $k$ be the
  number of vertices of $C$. Since we are considering a
  $4$-orientation of $A(G)$, we have $x_R+x_L=4k$. 
  Moreover, an edge of $A(G)$ leaving $C$ on its right (resp. left) is
  counting $+1$ (resp. $-1$) for $\delta (W_R(C))$ (resp.
  $\delta (W_L(C))$). For each edge of $C$ there is a corresponding
  edge-vertex in $\PDC{A(G)}$, that is incident to two dual-edges of
  $\PDC{A(G)}$, one that is crossing $\delta (W_L(C))$ from right to
  left, counting $-2$, and one crossing $\delta (W_R(C))$ from left to
  right, counting $+2$.  So $\delta (W_L(C))=-2k-x_L$ and
  $\delta (W_R(C))=2k+x_R$.

Combining these equalities, one obtain: 
$\gamma(C) = \delta (W_L(C)) + \delta (W_R(C))$,
 $\gamma(C)=2\delta (W_L(C))+8k$,  $\delta (W_L(C))=\gamma(C)/2-4k$.
% and
%$\delta (W_r(C))=\gamma(C)/2+4k$. 
Then clearly, $\delta (W_L(C))=0 \bmod 4$ implies
$\gamma(C)=0 \bmod8$, and $\gamma(C)=0 \bmod8$ implies
$\delta (W_L(C))=0 \bmod4$.% $\delta (W_R(C))=0 \bmod4$.
\end{proof}

Finally, we have the following theorem which characterizes  the
4-orientations that admit TTS-labelings:

\begin{theorem}
\label{th:characterizationgamma}
Consider a toroidal triangulation $G$. Let $\{B_1,B_{2}\}$ be a pair
of cycles of $G$, given with a direction of traversal, that forms a
basis for the homology.  A $4$-orientation of $A(G)$ admits a
toroidal transversal structure labeling if and only if
$\gamma(B_{1})= 0 \bmod 8$ and $\gamma(B_{2})= 0 \bmod 8$.
\end{theorem}

\begin{proof}
  $(\Longrightarrow)$
By Lemma~\ref{lem:charforall}, we have
  $\delta (W) = 0 \bmod4$ for any closed walk $W$ of $\PDC{A(G)}^*$. So
  we have $\delta(W_L(B_1)), \delta(W_L(B_{2}))$,
%  $\delta(W_R(B_1)), \ldots, \delta(W_R(B_{2g}))$ 
are both equal to
  $0 \bmod4$.  Thus, by Lemma~\ref{lem:gammaequaldelta}, we have
  $\gamma(B_{i}) = 0 \bmod8$, for $i\in\{1,2\}$.

  $(\Longleftarrow)$ Suppose that $\gamma(B_{i}) = 0 \bmod4$, for
  $i\in\{1,2\}$.  By Lemma~\ref{lem:gammaequaldelta}, we have
  $\delta(W_L(B_{i}))= 0 \bmod4$, for $i\in\{1,2\}$. Moreover
  $\{W_L(B_1),W_L(B_{2})\}$ forms a basis for the homology.  So by
  Lemma~\ref{lem:basedelta}, $\delta (W) = 0 \bmod4$ for any closed
  walk $W$ of $\PDC{A(G)}^*$. So the orientation admits a TTS-labeling
  by Lemma~\ref{lem:charforall}.
\end{proof}

The $4$-orientation of the toroidal triangulation on the left of
Figure~\ref{fig:4orbalanced} is an example where some non-contractible
cycles have value $\gamma$ not equal to $0 \bmod 8$. The vertical loop
of the triangulation, with upward direction of traversal, has
$\gamma= 2$.  Thus by Theorem~\ref{th:characterizationgamma}, this
orientation does not correspond to a transversal structure. Whereas,
on the right example, one can check that $\gamma=0$ for a vertical
cycle and a horizontal one, thus  this orientation corresponds
to a transversal structure (represented on
Figure~\ref{fig:balancedTTS}).

A consequence of Theorem~\ref{th:characterizationgamma} is that any
balanced 4-orientation of the angle graph $A(G)$ of a toroidal
triangulation $G$ admits a TTS-labeling and thus is the
$4$-orientation corresponding to a transversal structure of $G$.

\begin{corollary}
\label{cor:bal4orTS}
  Any balanced 4-orientation of $A(G)$ is the
$4$-orientation corresponding to a (balanced) transversal structure of $G$.
\end{corollary}

Note again that there are transversal structures whose corresponding
4-orientations are not balanced, thus for which  $\gamma= 0 \bmod 8$ for
every non-contractible cycles, but not exactly $0$ for some of
them. Such an example is given on Figure~\ref{fig:nonbalancedTTS}.

\section{Existence of balanced transversal structures}
\label{sec:existence}
In this section we prove existence of balanced transversal structures
for essentially 4-connected triangulations by contracting edges until
we obtain a triangulation with just one vertex. This is done by
preserving the property that the triangulation is essentially
4-connected. The toroidal triangulation on one vertex is represented
on Figure~\ref{fig:1vertex} with a balanced transversal structure and
the corresponding angle map. Then the graph can be decontracted step
by step to obtain a balanced transversal structures of the original
triangulation.

\begin{figure}[!ht]
\center
\includegraphics[scale=0.4]{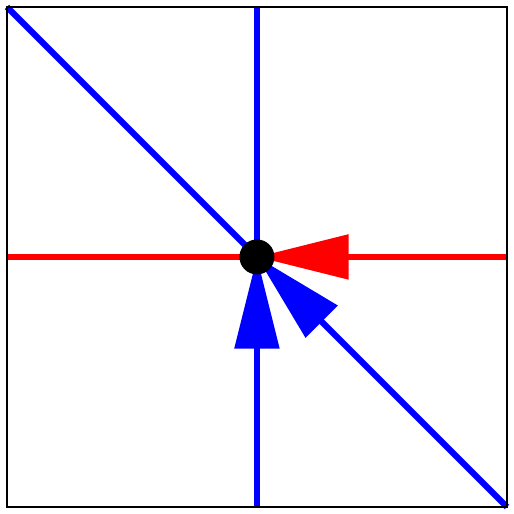} 
\ \ \ \
\includegraphics[scale=0.4]{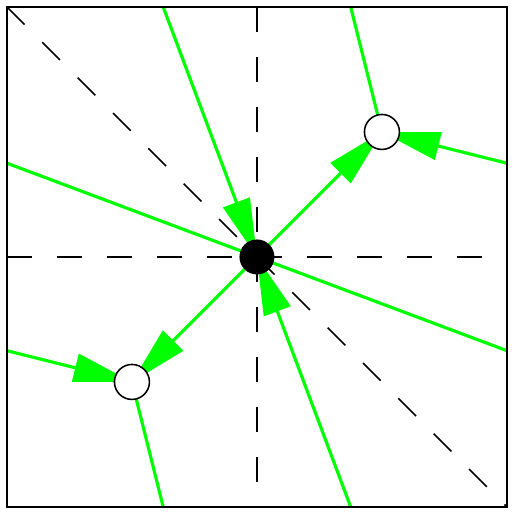}
\caption{Example of a balanced transversal
structure of the toroidal triangulation on one vertex.}
\label{fig:1vertex}
\end{figure}

\subsection{Contraction  preserving
  ``essentially 4-connected''}
\label{sec:contraction}

Given a toroidal triangulation $G$, the \emph{contraction} of a
non-loop-edge $e$ of $G$ is the operation consisting of continuously
contracting $e$ until merging its two ends. We note $G/e$ the obtained
map.  On Figure~\ref{fig:contraction-tri} the contraction of an edge
$e$ is represented.  Note that only one edge of each multiple edges
that is created is preserved (edge $e_{wx}$ and $e_{wy}$ on the
figure).

\begin{figure}[!ht]
\center
\begin{tabular}{ccc} 
\includegraphics[scale=0.4]{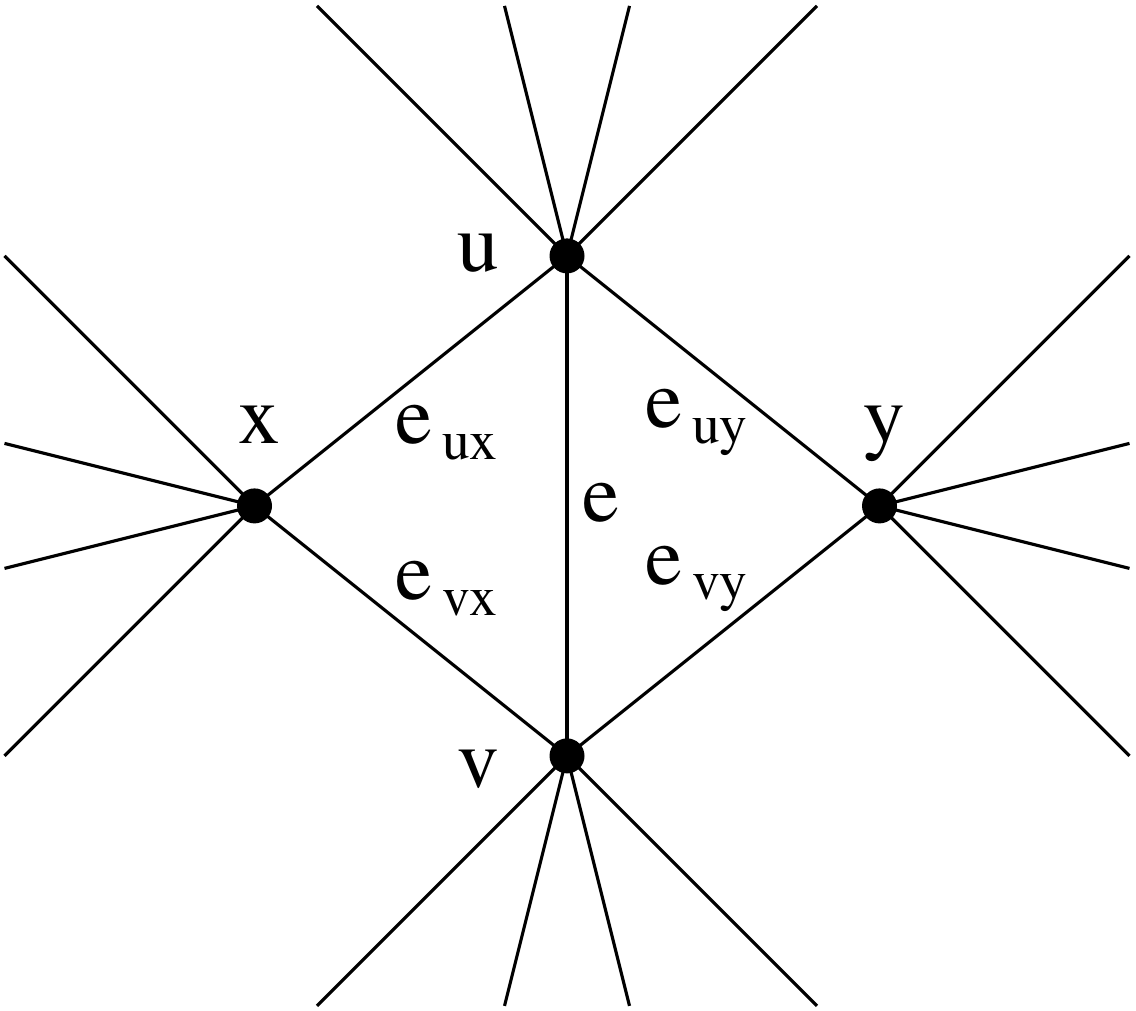}
& \hspace{3em} &
\includegraphics[scale=0.4]{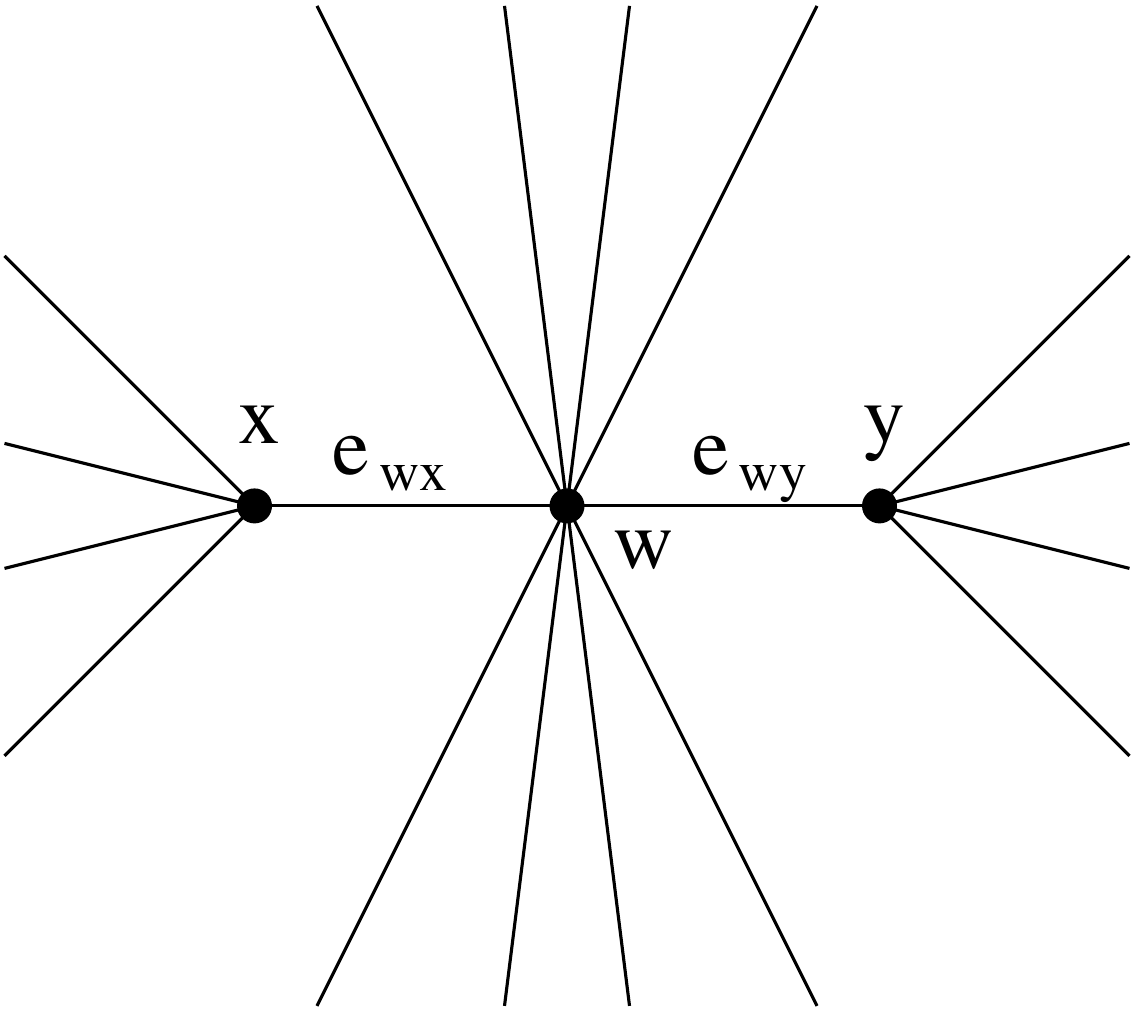}\\
& & \\
$G$& &$G/e$ \\
\end{tabular}
\caption{The contraction operation for triangulations}
\label{fig:contraction-tri}
\end{figure}

Note that the contraction operation is also defined when some vertices
are identified: $x=u$ and $y=v$ (the case represented on
Figure~\ref{fig:contraction-loop-tri}), or $x=v$ and $y=u$
(corresponding to the symmetric case with a diagonal in the other
direction).

\begin{figure}[!ht]
\center
\begin{tabular}{ccc} 
\includegraphics[scale=0.4]{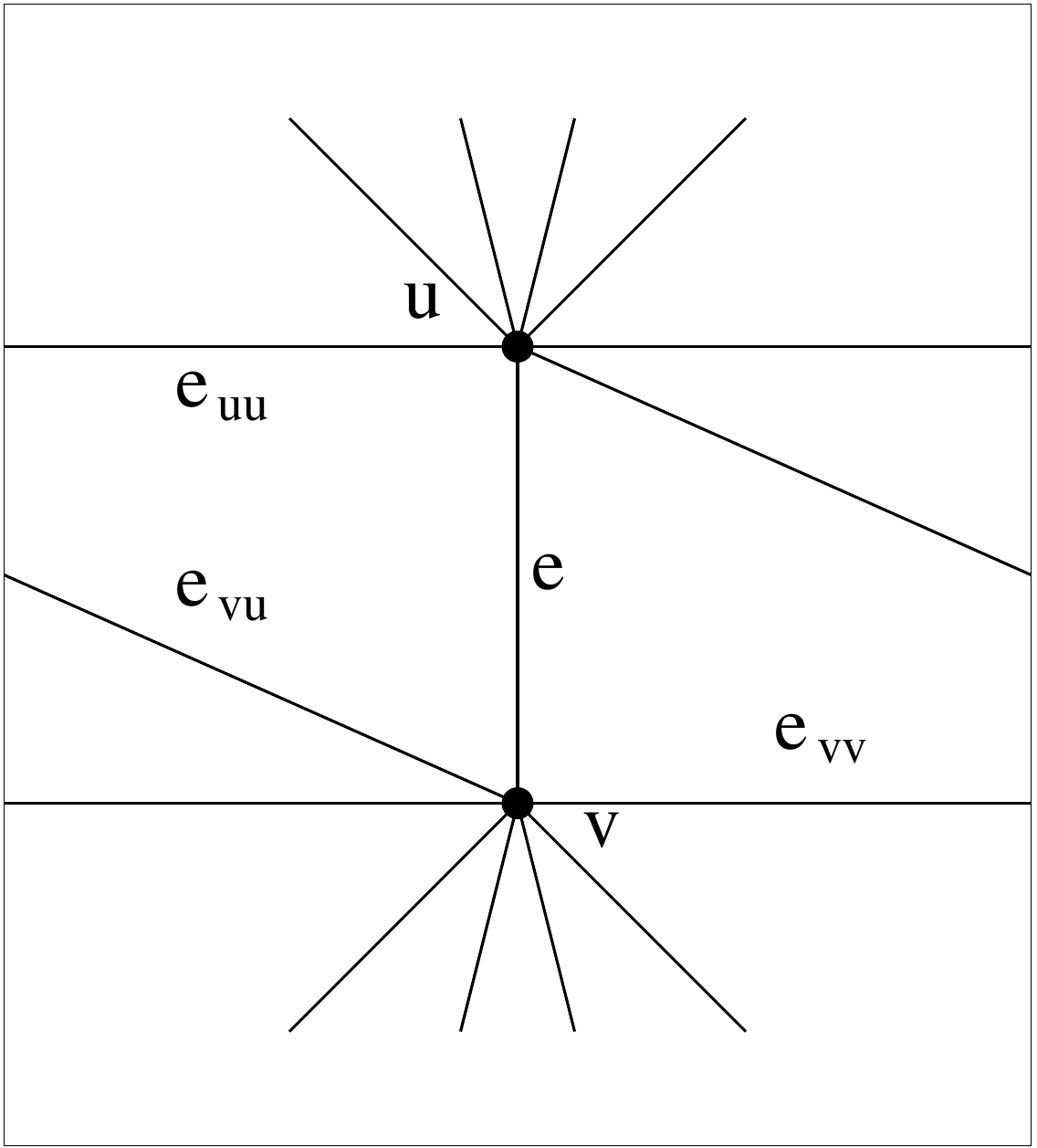}
& \hspace{3em} &
\includegraphics[scale=0.4]{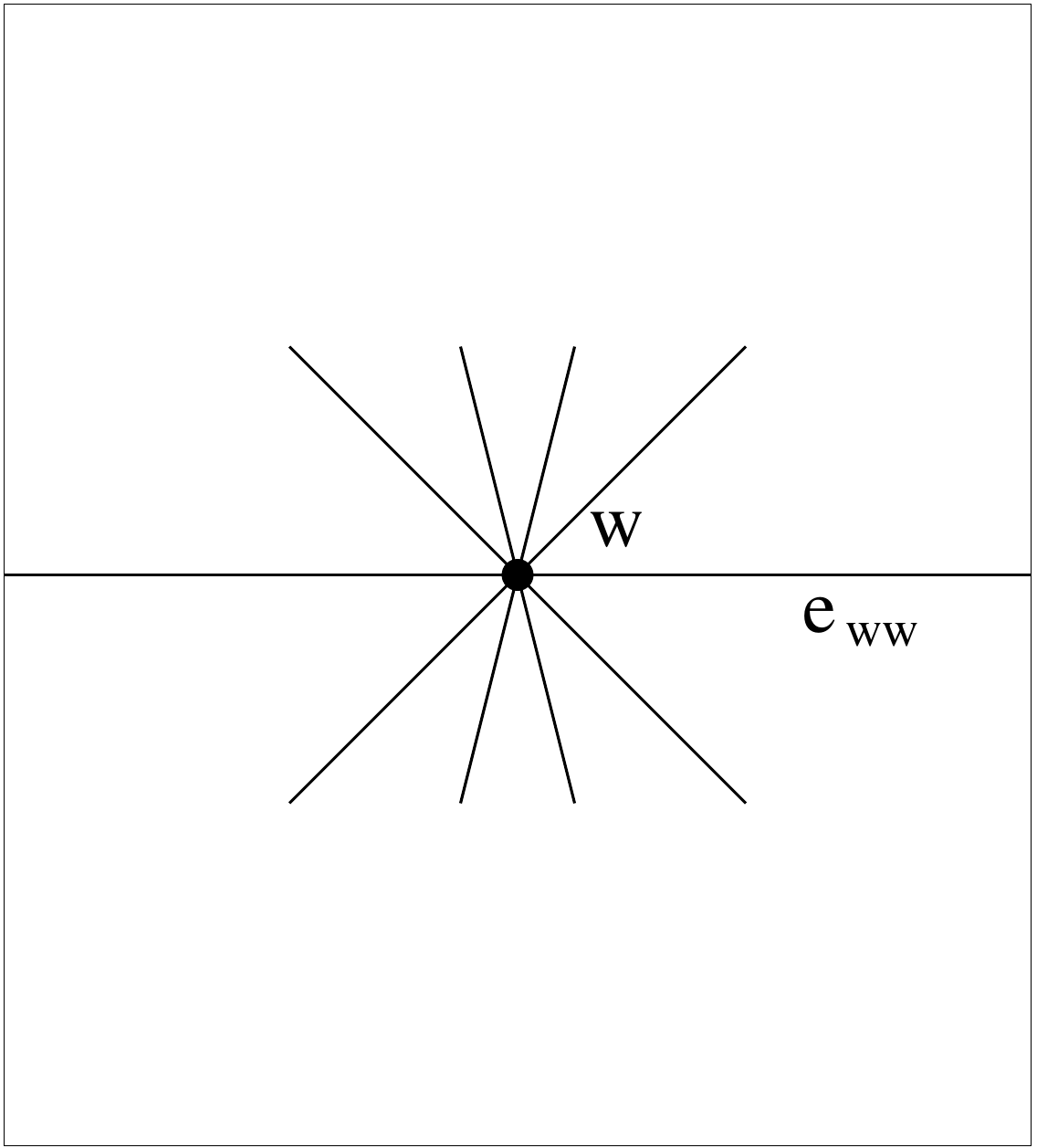}\\
& & \\
$G$& &$G/e$ \\
\end{tabular}
\caption{The contraction operation when some vertices are identified.}
\label{fig:contraction-loop-tri}
\end{figure}

In~\cite{Moh96} it is proved that in a toroidal triangulation (with no
contractible loop nor homotopic multiple edges) with at least two
vertices, one can find an edge whose contraction preserves the fact
that the map is a toroidal triangulation (with no contractible loop
nor homotopic multiple edges). Here we also need to show that we can
preserve the fact of being essentially 4-connected during contraction.
We say that a non-loop edge $e$ of an essentially 4-connected toroidal
triangulation $G$ is \emph{contractible} if $G/e$ is an essentially
4-connected toroidal triangulation.  We have the following lemma:

\begin{lemma}
\label{lem:contraction}
An essentially 4-connected toroidal triangulation with at
  least two vertices has a contractible edge.
\end{lemma}

\begin{proof} For $k\geq 3$, a \emph{separating k-walk} is a closed
  walk of size $k$ that delimits on one side a region homeomorphic to an
  open disk containing at least one vertex. This region is called the
  \emph{interior} of the separating k-walk. A separating $3$-walk is a
  separating triangle and we call a separating $4$-walk a
  \emph{separating quadrangle}.

  Let $G$ be an essentially 4-connected toroidal triangulation with at
  least two vertices. By Lemma~\ref{lem:e4ciffnst}, the map $G$ has no
  contractible loop, no homotopic multiple edges and no separating
  triangle. Consider a non-loop edge $e$ of $G$. The contracted graph
  $G/e$ is an essentially 4-connected toroidal triangulation if and
  only if $G/e$ has no contractible loop, no homotopic multiple edges
  and no separating triangle.

  Since $G$ has no homotopic multiple edges, the contraction of $e$
  cannot create a contractible loop.  Since $G$ has no separating
  triangle, the only way to create a pair of homotopic multiple edges
  in $G/e$ is if $e$ appears twice on a separating quadrangle such
  that each extremity of $e$ is incident to a non-contractible loop as
  depicted on Figure~\ref{fig:sepwalk}.a (where the dashed region
  represents the interior of the separating quadrangle).  There are two
  ways to create a separating triangle in $G/e$: either $e$ appears
  once on a separating quadrangle as depicted on
  Figure~\ref{fig:sepwalk}.b, where some vertices may be identified
  but not edges, or $e$ appears twice on a separating 5-walk such that
  one extremity of $e$ is incident to a non-contractible loop and the
  other extremity is incident to edges distinct from $e$ forming a
  non-contractible cycle of size two as depicted on
  Figure~\ref{fig:sepwalk}.c.

\begin{figure}[!ht]
\center
\begin{tabular}{cccc} 
Case a & \includegraphics[scale=0.4]{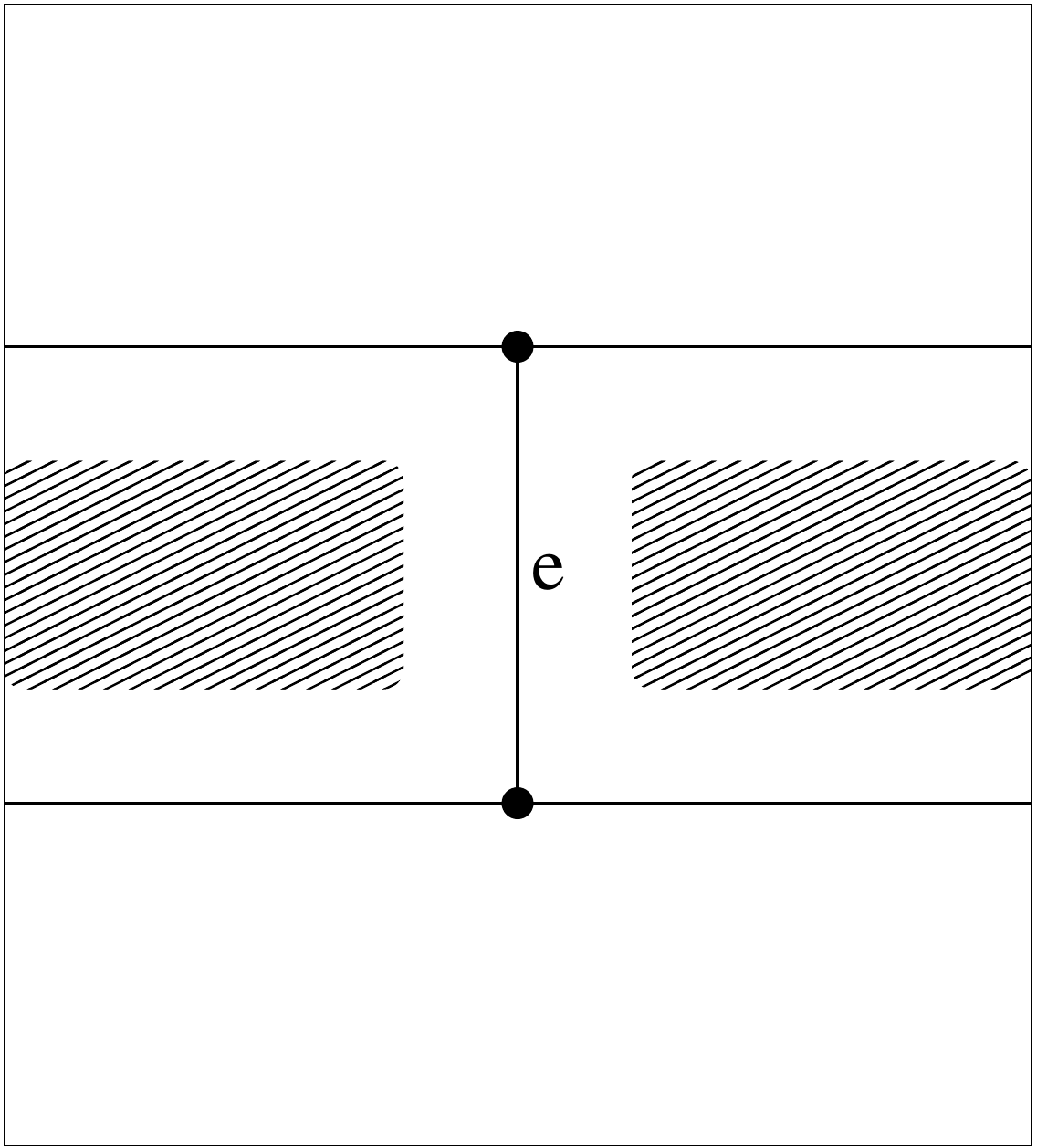}
& \hspace{3em} &
\includegraphics[scale=0.4]{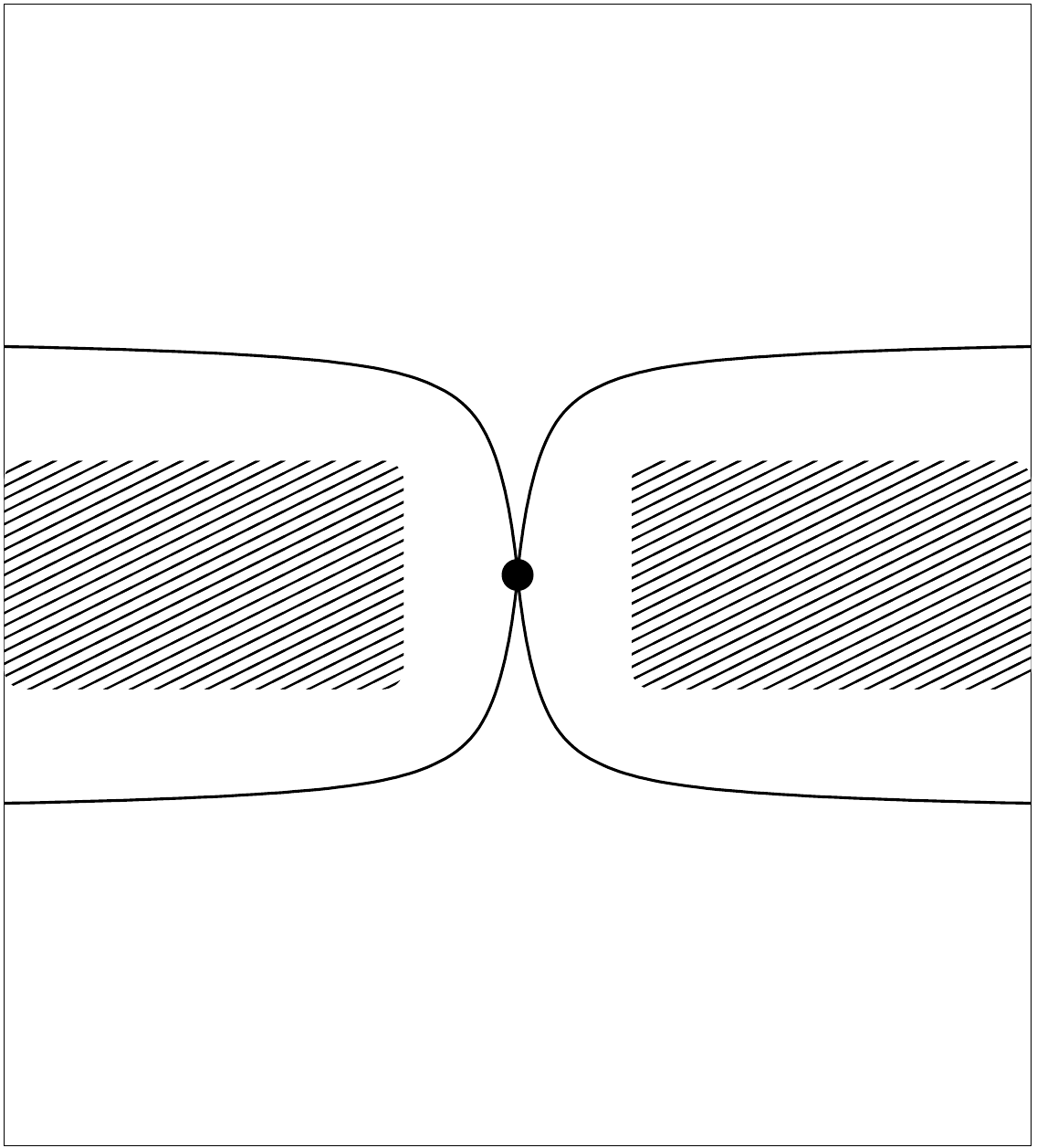}\\
& & & \\
Case b & \includegraphics[scale=0.4]{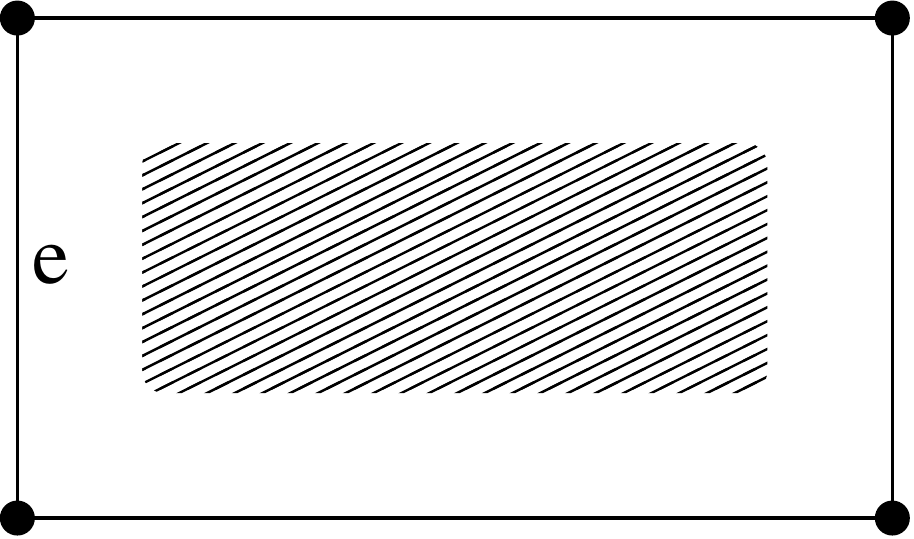}
& \hspace{3em} &
\includegraphics[scale=0.4]{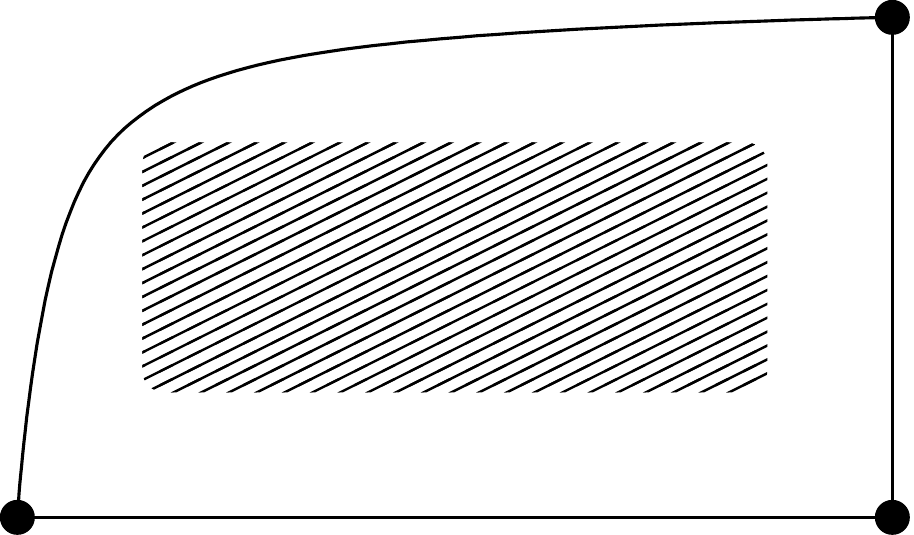}\\
& & & \\
Case c & \includegraphics[scale=0.4]{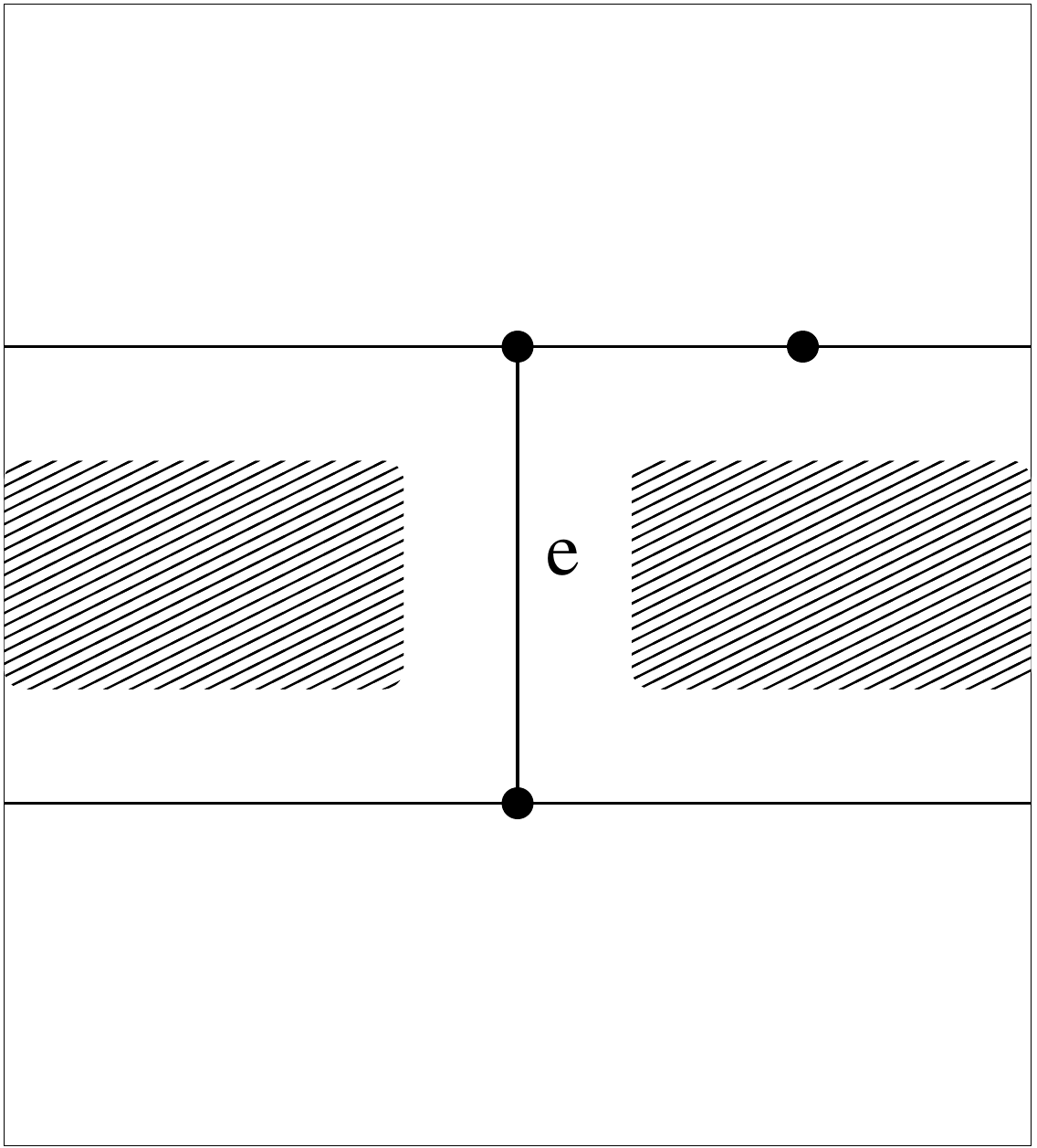}
& \hspace{3em} &
\includegraphics[scale=0.4]{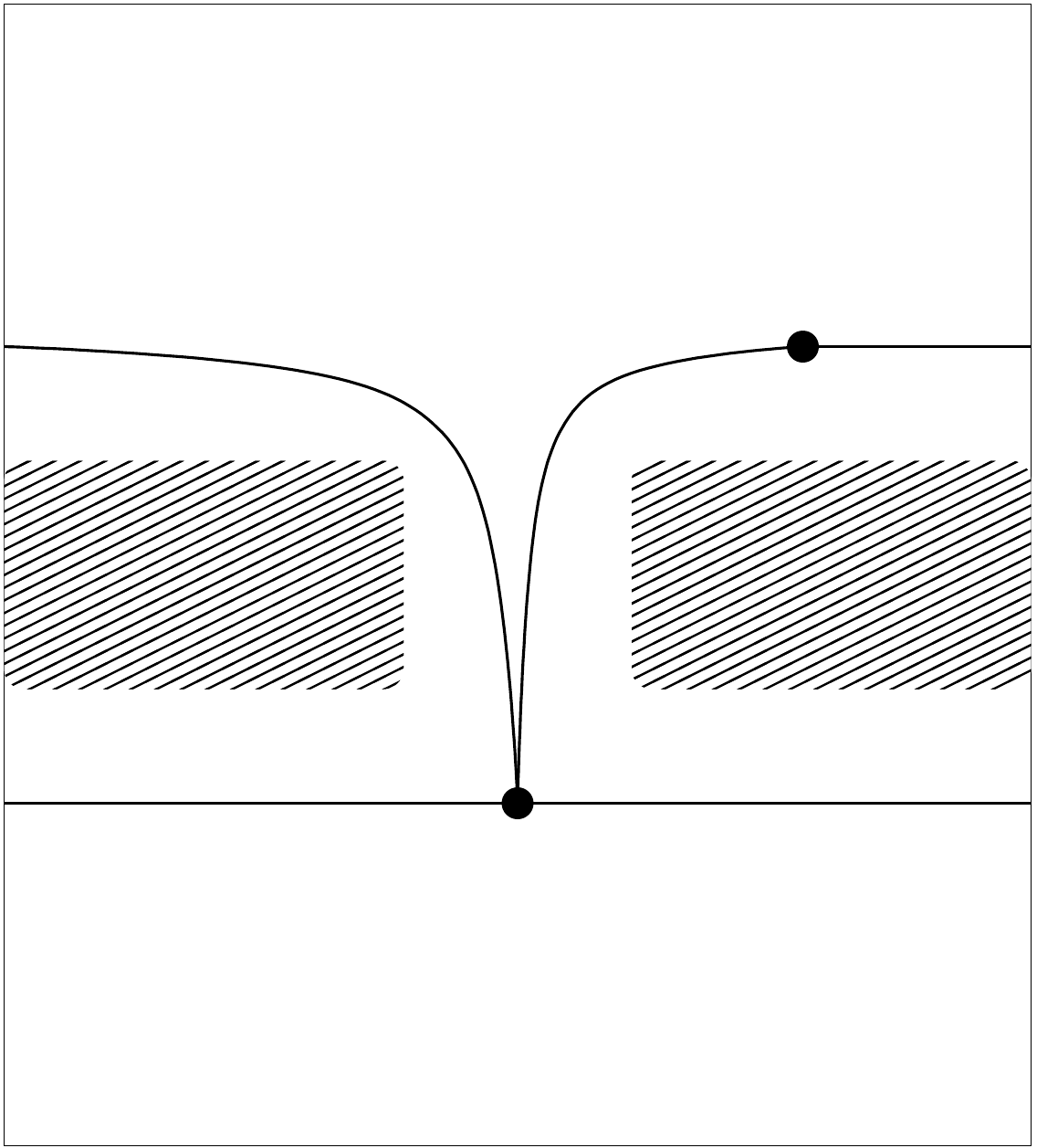}\\
&& & \\
& $G$& &$G/e$ \\
\end{tabular}
\caption{Contraction of an edge $e$ creating  a pair of homotopic
  multiple edges or a separating triangle.}
\label{fig:sepwalk}
\end{figure}

 We consider two cases whether there are separating quadrangles
  in $G$ or not.

  \begin{itemize}
  \item \emph{$G$ has some separating quadrangles:} 

    An \emph{inner chord} of a separating quadrangle $Q$ is an edge between
    its vertices that lie in the interior of the separating quadrangle.  We
    claim that a separating quadrangle $Q$ of $G$ has no inner
    chord. Suppose by contradiction that such a chord exists. Since
    there is no pair of homotopic multiple edges, this chord is
    between ``opposite'' vertices of $Q$. Thus it
    partitions the interior of $Q$ into two triangles. These two
    triangles are not separating by assumption on $G$ and thus the
    quadrangle $Q$ is not separating either, a contradiction.

    Let $Q$ be a maximal separating quadrangle.  Suppose by
    contradiction that there exists a separating quadrangle $Q'$ of
    $G$ distinct from $Q$ whose interior $R'$ intersects $R$ and is
    not included in $R$. By maximality of $Q$, we also have $R$ is not
    included in $R'$. As observed previously, $Q$ and $Q'$ have no
    inner chord.  So there is at least one or two vertices of $Q$
    (resp $Q'$) in the interior of $Q'$ (resp. $Q$).  Thus the border
    of the union of $R$ and $R'$ has size less or equal to four, a
    contradiction to the maximality of $Q$ or of $G$ being an
    essentially $4$-connected triangulation.  So a separating
    quadrangle of $G$ whose interior intersects $R$ has its interior
    included in $R$.

%Let $R$ be the interior of
%    $W$. 

%Let $a,b,c,d$ denote the vertices of $W$ in consecutive order
%    (some vertices and edges of $W$ may be identified).

    Let $G'$ be the map obtained from $G$ by keeping all the vertices
    and edges in $R$, including $Q$. The vertices and edges appearing
    several times on $Q$ are duplicated so $G'$ is a planar map. Then
    $G'$ is a 4-connected planar map in which every inner
    face is a triangle and the outer face is a quadrangle.  Let
    $a,b,c,d$ denote the outer vertices of $G'$ in \ccw order.  We
    denote also $a,b,c,d$ the corresponding vertices of $G$. Note
    that in $G$ some of these vertices might be identified.  We
    consider two cases, whether, in $G'$, there exists an inner vertex
    incident to at least three outer vertices or not.

    \begin{itemize}
    \item \emph{In $G'$, there exists an inner vertex $v$ that
    is incident to at least three outer vertices:}

  W.l.o.g., we may
    assume that $v$ is incident to $a,b,c$ in $G'$ with edges
    $e_a,e_b,e_c$ respectively.  We prove that $e_b$ is contractible
    in $G$. 

    Suppose by contradiction that $e_b$ belongs to a separating
    quadrangle $Q'$ of $G$. Then $Q'$ is distinct from $Q$ and its
    interior $R'$ of $Q'$ intersects $R$. Thus by the above remark,
    $R'$ is included in $R$.  Then $Q'$ has an inner chord $e_a$ or
    $e_c$, a contradiction.

    Suppose that $e_b$ appears twice on a separating 5-walk $W$ as
    depicted on Figure~\ref{fig:sepwalk}.c. Then, one of the extremity
    of $e_b$ is incident to a non-contractible loop $\ell$ of $W$ and
    this extremity cannot lie inside $R$ so it is $b$. So $v$ is
    incident to two edges $e_1,e_2$, distinct from $e_b$, so that the
    5-walk $W$ is the sequence of edges $\ell, e_b, e_1,e_2, e_b$.
    Then, in $G'$, the edges $e_1,e_2$ are incident to two distinct
    vertices of $\{a,c,d\}$ that are identified in $G$ so that
    $e_1,e_2$ form a non-contractible cycle of size two of $G$.

    Suppose, by contradiction, that $e_1,e_2$ are incident to
    ``consecutive'' vertices of $Q$, then w.l.o.g., we can assume that
    $e_1,e_2$ are incident to $a$ and $d$ that are identified in $G$.
    If $b,c$ are also identified in $G$, then we are in the situation
    of Figure~\ref{fig:walkstar}.a, with $W$ represented in magenta.
    Then, the interior of the separating 5-walk $W$ is partitioned
    into triangles whose interiors are empty since $G$ is essentially
    4-connected. So the interior of the separating 5-walk $W$ contains
    no vertices, a contradiction. If $b,c$ are not identified, then we
    are in the situation of Figure~\ref{fig:walkstar}.b. Then the two
    loops of the figure plus $e$ form a quadrangle whose interior
    strictly contains the interior of $Q$, a contradiction to the
    maximality of $Q$.

\begin{figure}[!ht]
\center
\begin{tabular}{ccc} 
\includegraphics[scale=0.4]{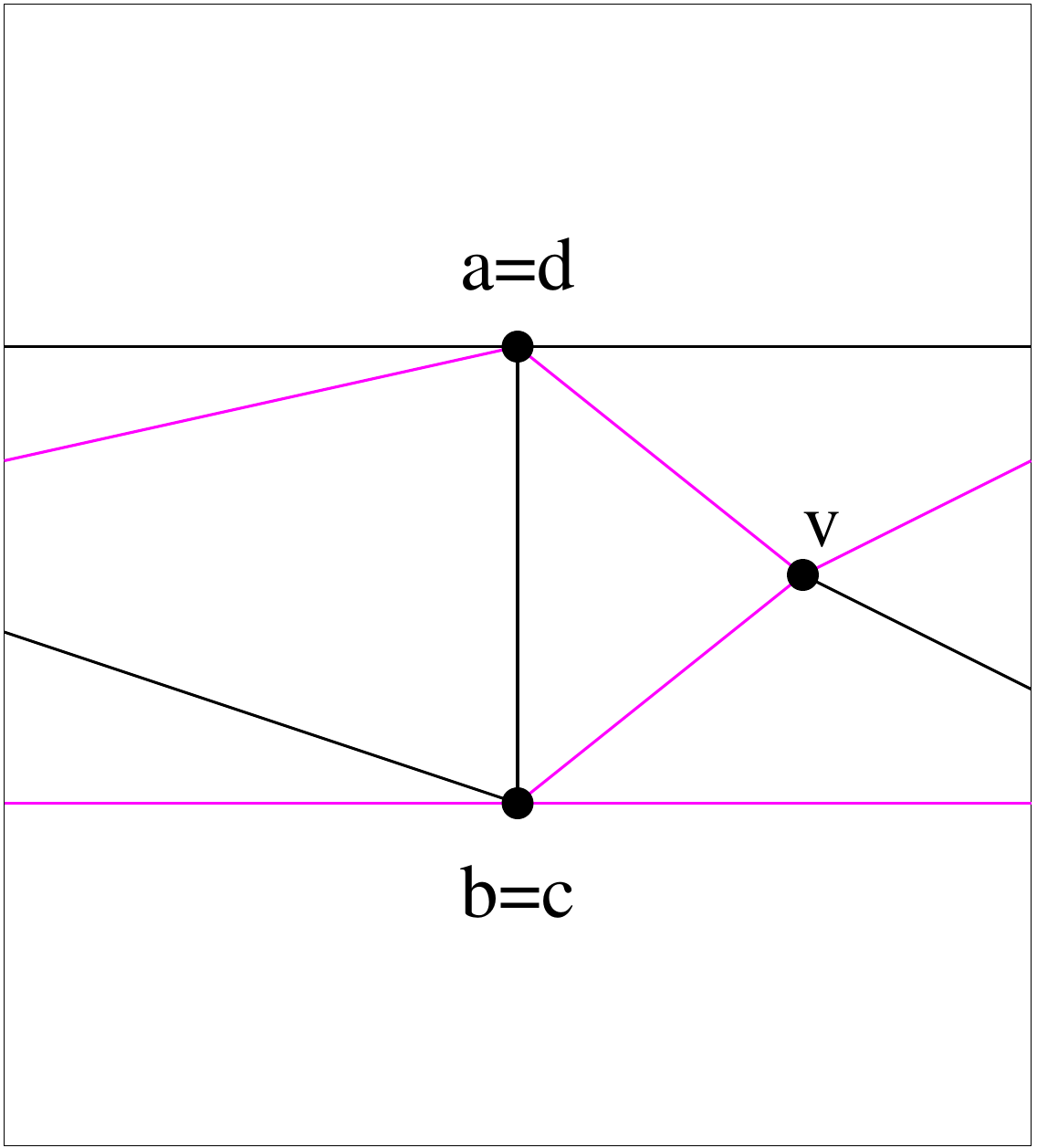} \  &\  
\includegraphics[scale=0.4]{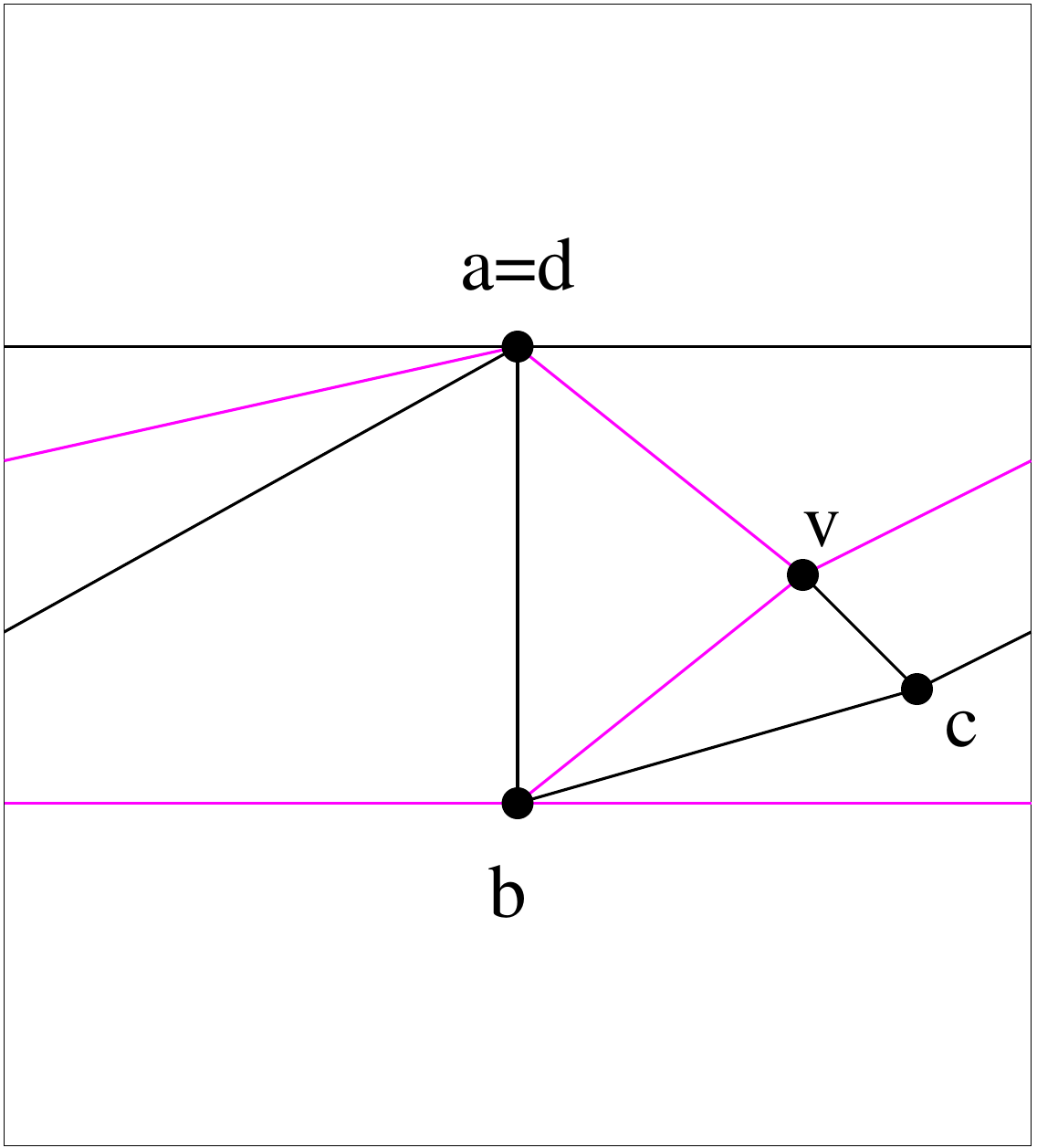} \  &\ 
\includegraphics[scale=0.4]{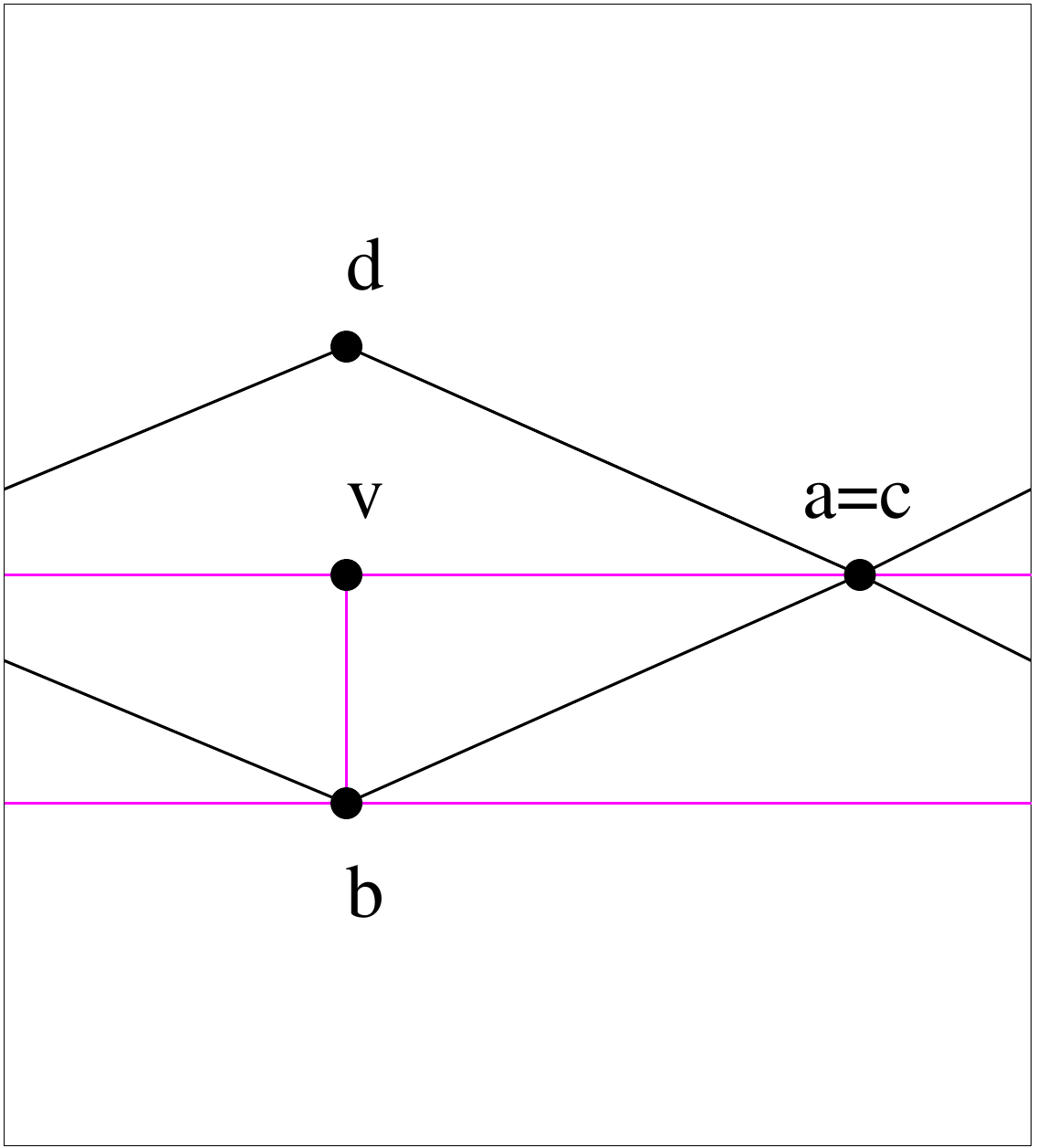} \\
& & \\
Case a \  &\  Case b \  &\  Case c \\
\end{tabular}
\caption{Cases of Lemma~\ref{lem:contraction}.}
\label{fig:walkstar}
\end{figure}

So, $e_1,e_2$ are incident to ``opposite'' vertices of $Q$.  These two
vertices are $a$ and $c$ and they are identified in $G$. Then we are
in the situation of Figure~\ref{fig:walkstar}.c
. Then, the interior of
the separating 5-walk $W$ is partitioned into triangles whose interior
are empty since $G$ is essentially 4-connected. So the interior of the
separating 5-walk $W$ contains no vertices, a contradiction.

To conclude, $e_b$ does not belong to a separating quadrangle, nor
appears twice on a separating 5-walk, so $e_b$ is contractible in $G$.

\item \emph{In $G'$, all inner vertices are incident to at most two
    outer vertices:}

  Kant and He proved~\cite[Lemma 3.1]{KH97} that $G'$ contains an
  internal edge $e$ such that $G'/e$ is 4-connected.  Let us show $e$
  is contractible in $G$.

    Suppose by contradiction that $e$ belongs to a separating quadrangle
    $Q'$ of $G$. Then the interior $R'$ of $Q'$ intersects $R$ and
    thus is included in $R$ by above remark. But then $G'/e$ is not
    4-connected, a contradiction.

    Suppose by contradiction that $e'$ appears twice on a separating
    5-walk as depicted on Figure~\ref{fig:sepwalk}.c. Then, in $G$,
    one extremity $u$ of $e'$ is incident to a non-contractible loop
    and the other extremity $v$ of $e'$ is incident to two edges
    $e_1,e_2$, distinct from $e'$, forming a non-contractible cycle of
    size two. Thus $u$ is not an inner vertex of $G'$ and the two
    extremities of $e_1,e_2$ also. So $v$ is incident to three outer
    vertices of $G'$, a contradiction.

    To conclude, $e$ does not belong to a separating quadrangle, nor
    appears twice on a separating 5-walk, so $e$ is contractible in
    $G$.

    \end{itemize}
  \item \emph{$G$ has no separating quadrangle:}

    Consider a non-loop edge $e$ of $G$. If $e$ is contractible we are
    done, so we can assume that $e$ is not contractible. Then, since
    there is no separating quadrangle, we have $e$ that appears twice
    on a separating 5-walk $W$ as depicted on
    Figure~\ref{fig:sepwalk}.c. More precisely, one  extremity
    $u$ of $e$ is incident to a non-contractible loop $\ell$ and the
    other extremity $v$ of $e$ is incident to two edges $e_1,e_2$,
    distinct from $e$, forming a non-contractible cycle of size two of
    $G$. Let $R$ be the interior of the separating 5-walk $W$. We
    consider two cases whether $v$ has some neighbors in the strict
    interior of $R$ or not.

    Suppose by contradiction that $v$ has no neighbors in the strict
    interior of $R$. Then either $v$ has some incident edges inside
    $R$ or not. Suppose first that $v$ has some incident edges inside
    $R$. Then, since $v$ has no neighbors in the strict interior of
    $R$, we have that $v$ is incident to $u$ with an edge in the
    strict interior of $R$, as depicted on
    Figure~\ref{fig:sep5diag}.a. Then since there is no separating
    triangle, nor separating quadrangle, the region $R$ contains no
    vertices, a contradiction. Suppose now that $v$ has no incident
    edge inside $R$. Then, Since $G$ is a triangulation, vertex $u$
    must be incident twice to the third vertex $w$ of the $5$-walk
    $w$, as depicted on Figure~\ref{fig:sep5diag}.b. Again since there
    is no separating triangle, the region $R$ contains no vertices, a
    contradiction.

\begin{figure}[!ht]
\center
\begin{tabular}{ccc} 
\includegraphics[scale=0.4]{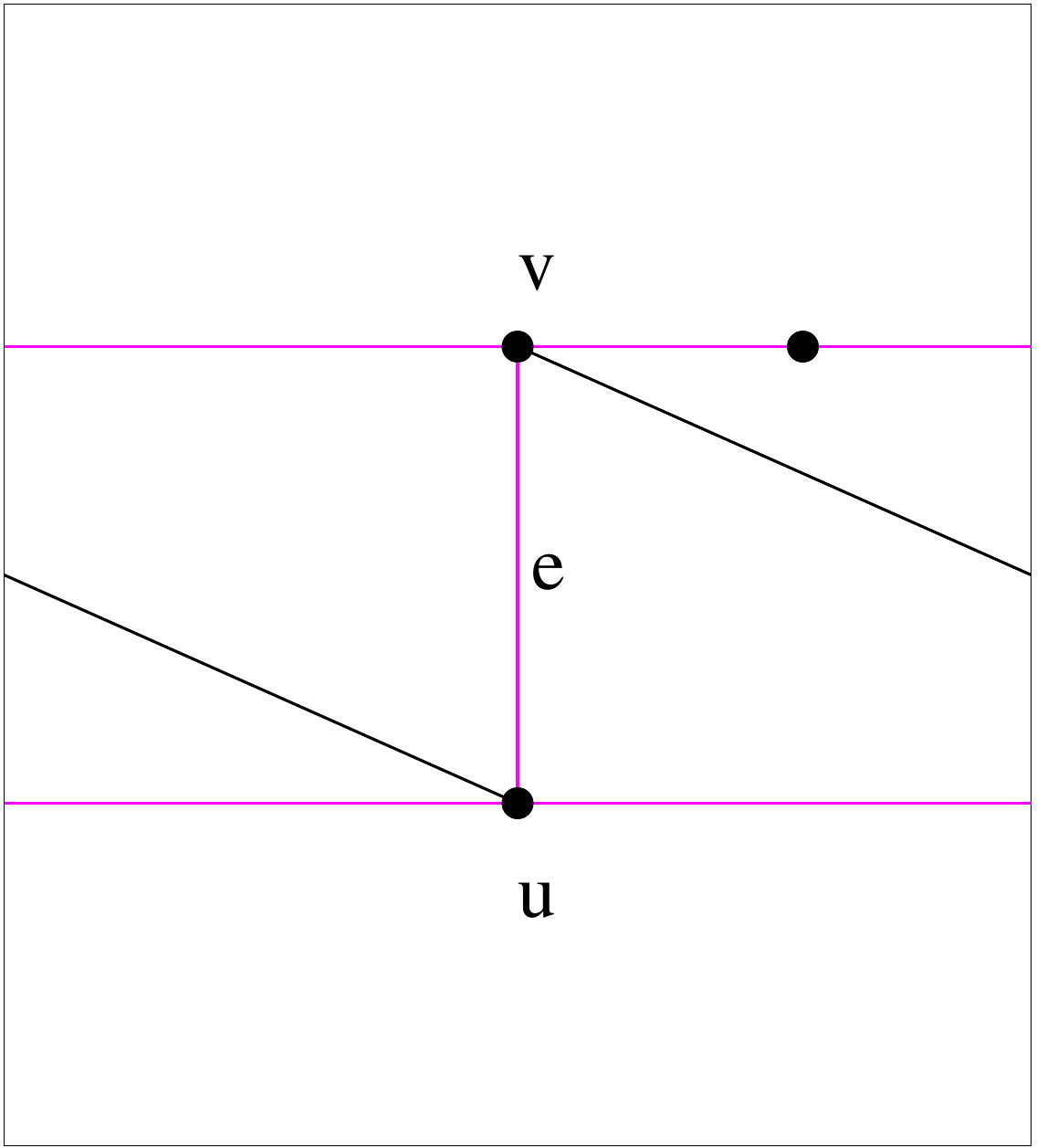} \  &\  
\includegraphics[scale=0.4]{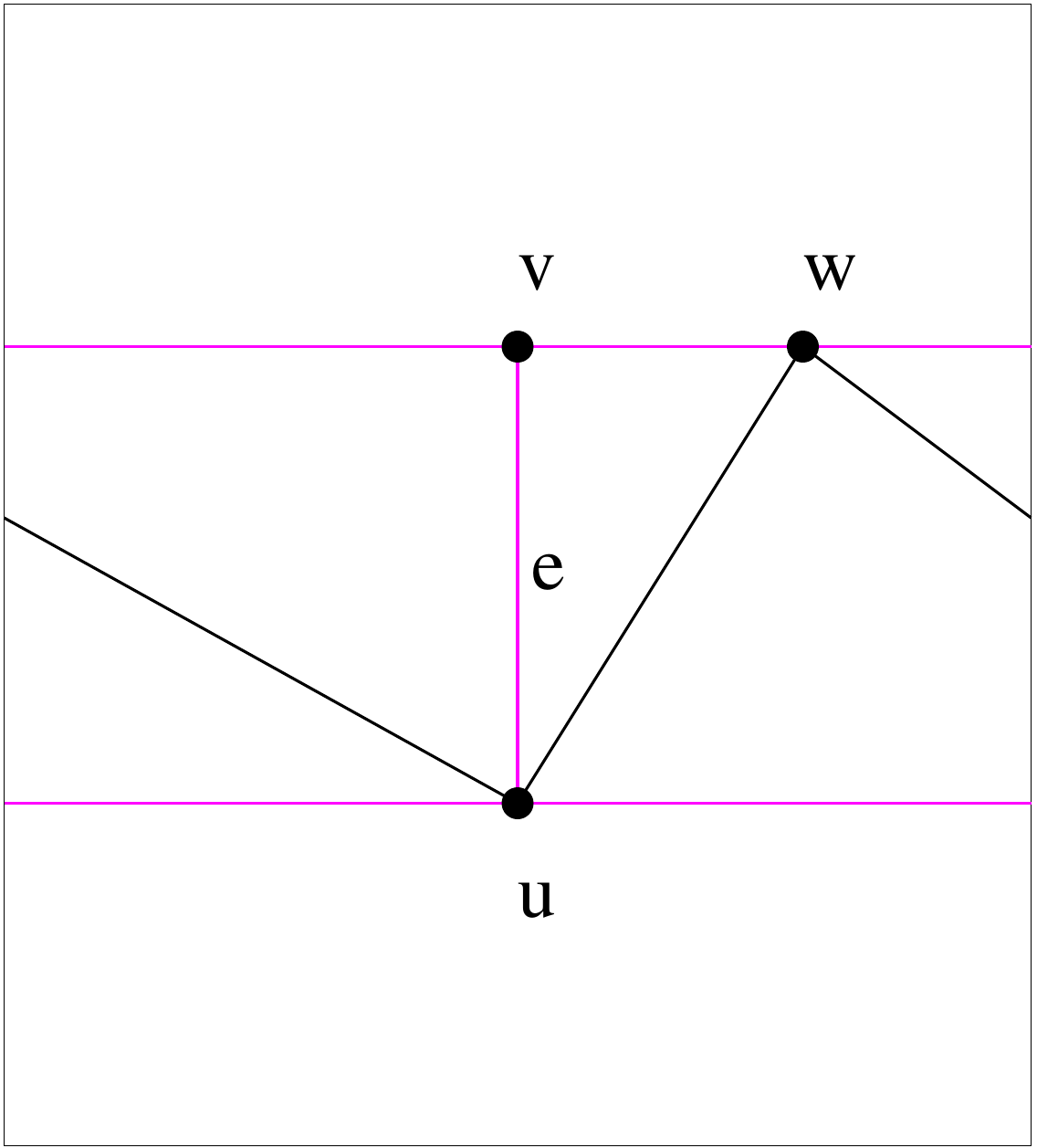} \ &\ 
\includegraphics[scale=0.4]{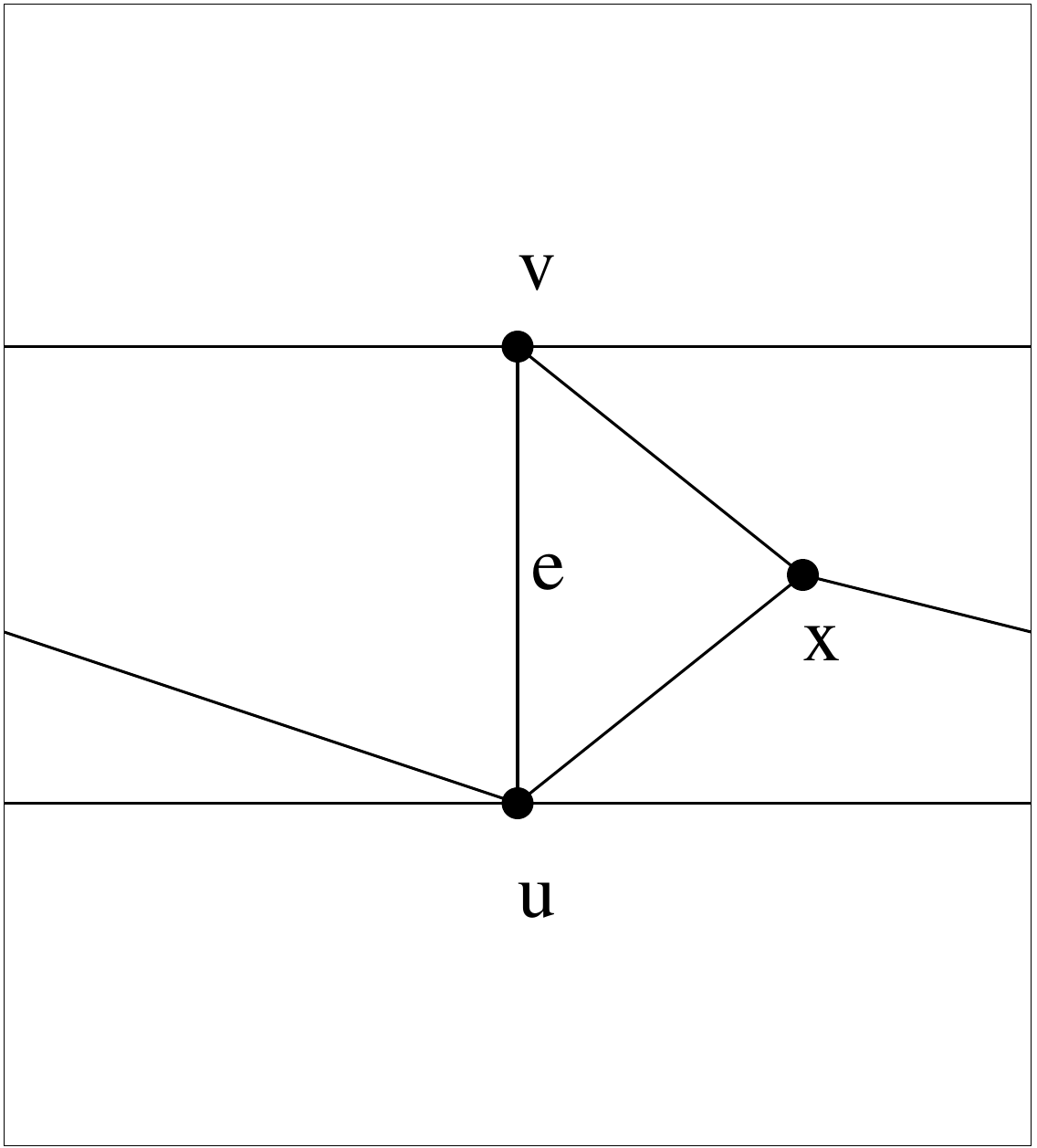} \\
&  & \\
Case a \  &\  Case b \ &\  Case c \\
\end{tabular}
\caption{Cases of Lemma~\ref{lem:contraction}.}
\label{fig:sep5diag}
\end{figure}

So $v$ has a neighbor $x$ in the strict interior of $R$.  Let $e'$ be
the edge between $v,x$. If $e'$ appears twice on a separating 5-walk
$W$, then $x$ is incident twice to $u$ to form a non-contractible
cycle of size two and $v$ is incident to a non-contractible loop
$\ell'$, as depicted on Figure~\ref{fig:sep5diag}.c. Then
$e,\ell,\ell'$ forms a separating quadrangle, a contradiction. So $e'$
is contractible.
\end{itemize}
\end{proof}

\subsection{Balanced properties and homology}

In Section~\ref{sec:char}, we have proved some properties of $\gamma$
(or $\delta$) w.r.t. homology. The obtained equalities where
conditioned by a ``modulo''. In next lemma we prove some properties of
$\gamma$ w.r.t. a basis for the homology with exact equality to obtain
a simple condition to prove that a $4$-orientation is balanced (see
Lemma~\ref{lem:gamma0all}). 

Consider an essentially 4-connected toroidal triangulation $G$ and its
angle map $A(G)$.  

\begin{lemma}
  \label{lem:gammahomology} 
  Consider a $4$-orientation of $A(G)$, a non-contractible
  cycle $C$ of $G$, given with a direction of traversal, and a
  basis for the homology $\{B_1,B_2\}$ of $G$, such that $B_1,B_2$ are
  non-contractible cycles whose intersection is a single vertex or a
  common path. If $C$ is homologous to $k_1 B_1 + k_2B_{2}$, then
  $\gamma(C)=k_1\,\gamma (B_1) + k_2\,\gamma (B_2)$.
 \end{lemma}
\begin{proof}
  Let $v$ be a vertex in the intersection of $B_1,B_2$ such that, if
  this intersection is a common path, then $v$ is one of the
  extremities of this path and let $u$ be the other extremity.
  Consider a drawing of ${G}^\infty$ obtained by replicating a flat
  representation of ${G}$ to tile the plane.  Let $v_0$ be a copy of
  $v$ in ${G}^\infty$.  Consider the walk $W$ starting from $v_0$ and
  following $k_1$ times the edges corresponding to $B_1$ and then
  $k_2$ times the edges corresponding to $B_2$ (we are going backward
  if $k_i$ is negative). This walk ends at a copy $v_1$ of $v$.  Since
  $C$ is non-contractible we have $k_1$ or $k_2$ not equal to $0$ and
  thus $v_1$ is distinct from $v_0$.  Let $W^\infty$ be the infinite
  walk obtained by replicating $W$ (forward and backward) from $v_0$.
  Note that there might be some repetition of vertices in $W^\infty$
  if the intersection of $B_1,B_2$ is a path. But in that case, by the
  choice of $B_1,B_2$ (i.e. whose intersection is a single vertex or a common path), we have that $W^\infty$ is almost a path,
  except maybe at all the transitions from ``$k_1 {B_1}$'' to
  ``$k_2B_{2}$'', or at all the transitions from ``$k_2 {B_2}$'' to
  ``$k_1B_{1}$'', where it can go back and forth a path
  corresponding to the intersection of $B_1$ and $B_2$. The existence
  or not of such ``back and forth'' parts depends on the signs of
  $k_1, k_2$ and the way $B_1,B_2$ are going through their common
  path. Figure~\ref{fig:replicating} gives an example of this
  construction with $(k_1,k_2)=(1,1)$ and $(k_1,k_2)=(1,-1)$ when
  $B_1,B_2$ intersects on a path and are oriented the same way along
  this path as on Figure~\ref{fig:replicating0}.

\begin{figure}[!ht]
\center
\includegraphics[scale=0.3]{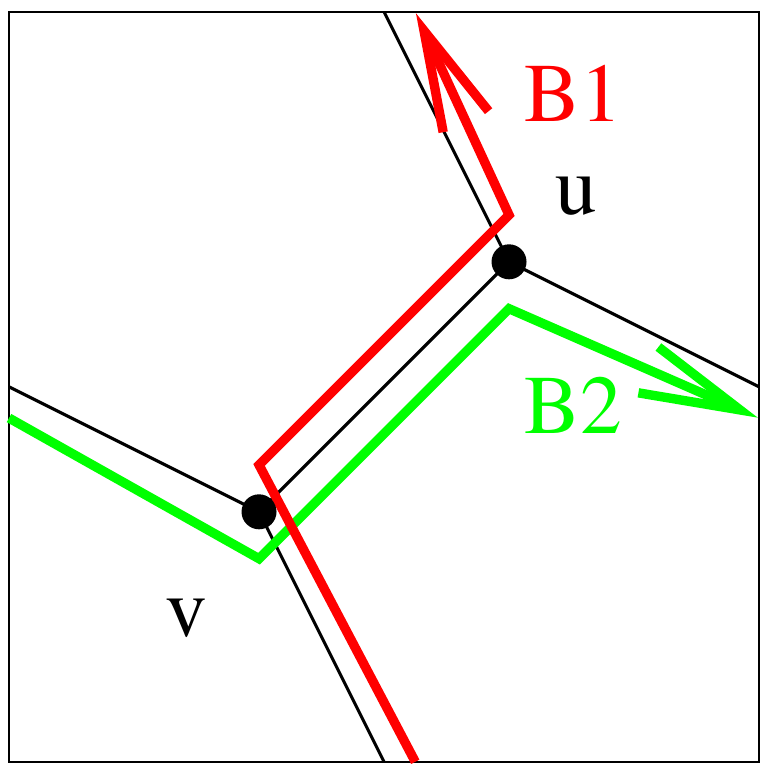}
\caption{Intersection of the basis.}
\label{fig:replicating0}
\end{figure}

\begin{figure}[!ht]
\center
\begin{tabular}{cc}
\includegraphics[scale=0.3]{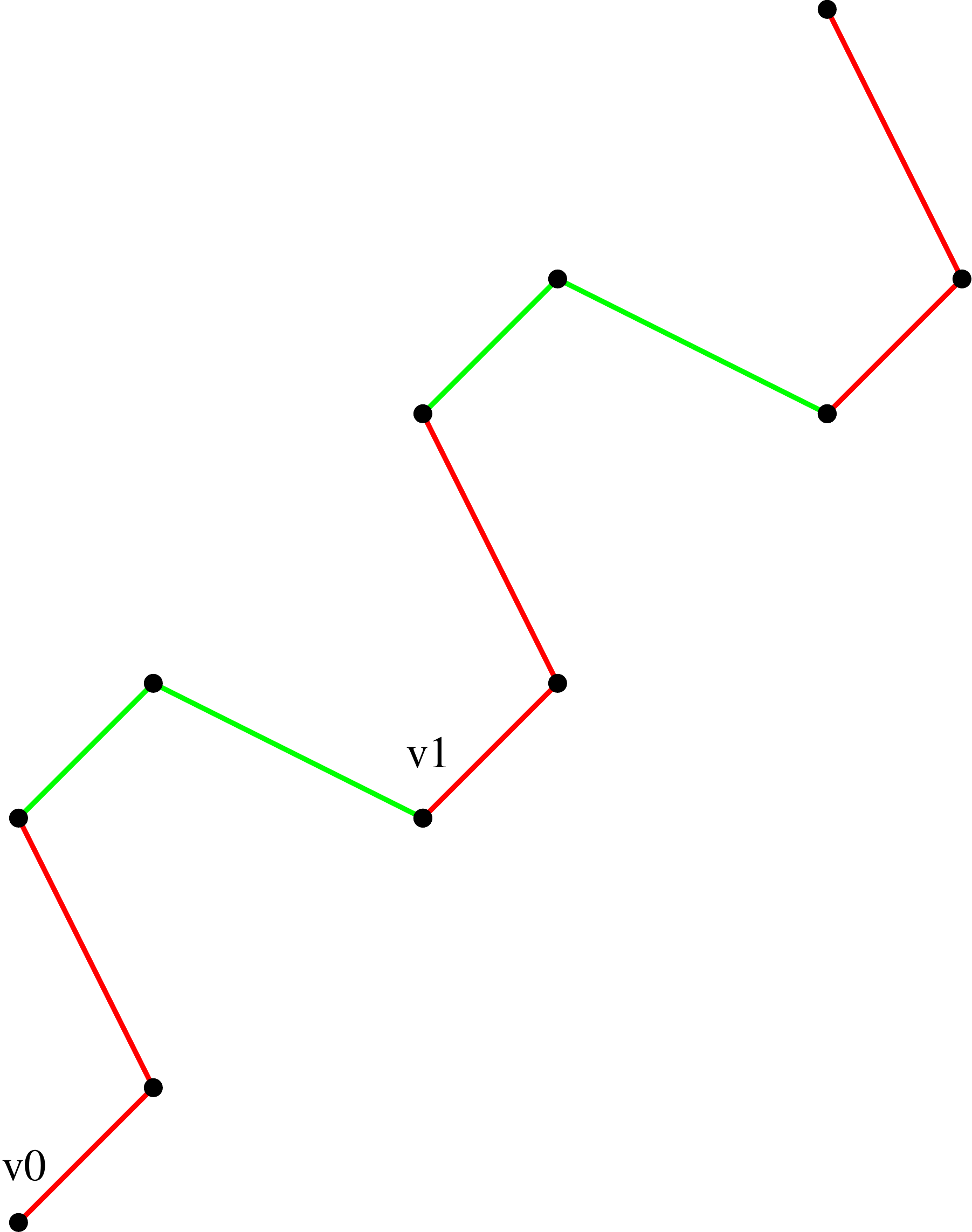} \ \ \ & \ \ \
\includegraphics[scale=0.3]{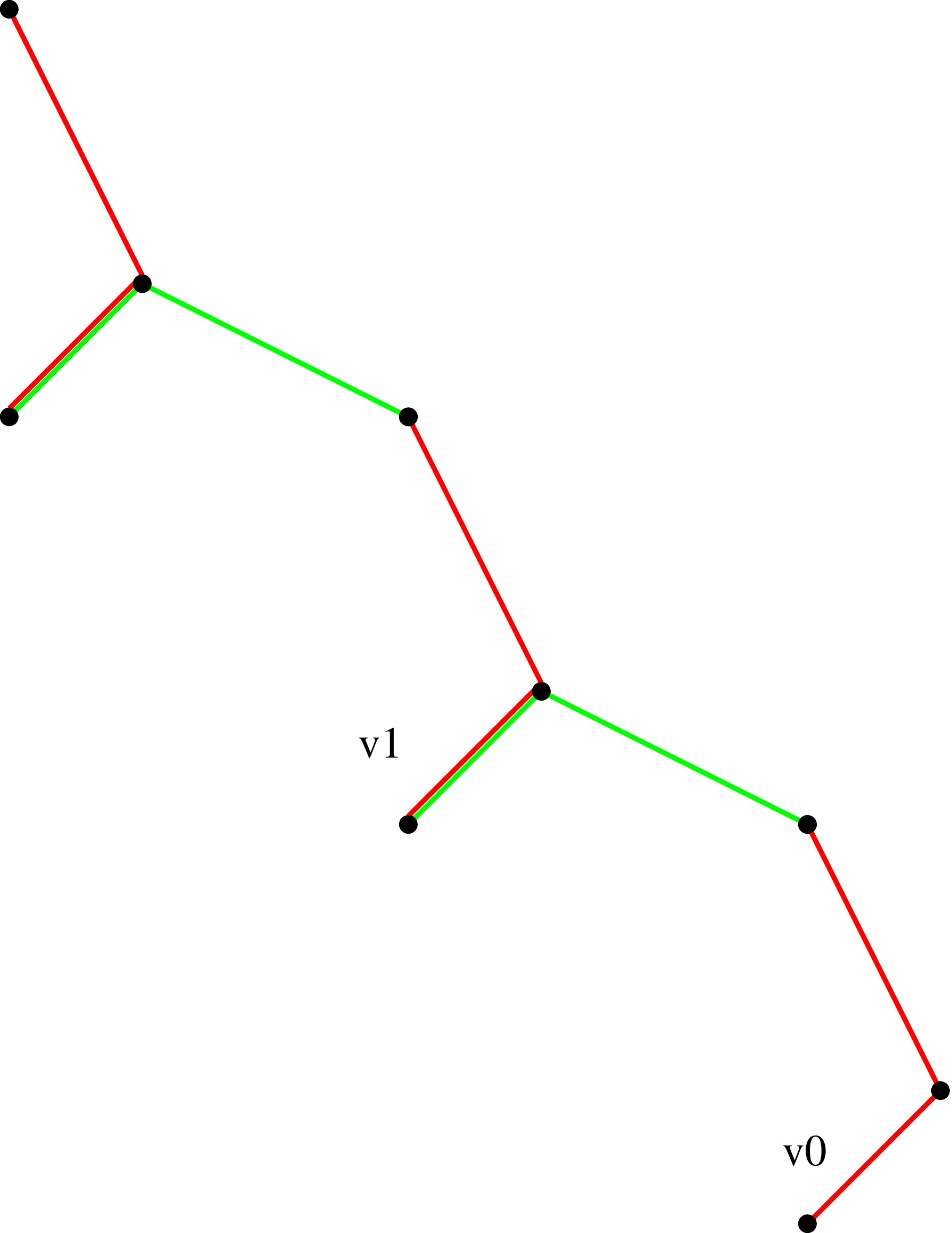} \\
$(k_1,k_2)=(1,1)$ \ \ \ & \ \ \ $(k_1,k_2)=(1,-1)$\\
\end{tabular}
\caption{Replicating ``$k_1 {B_1}$'' and ``$k_2B_{2}$'' in the universal cover.}
\label{fig:replicating}
\end{figure}

   We ``simplify'' $W^\infty$ by removing all the parts that consists
   of going back and forth along a path (if any) and call $B^\infty$
   the obtained walk that is now without repetition of vertices. By
   the choice of $v$, we have that $B^\infty$ goes through copies of
   $v$. If $v_0,v_1$ are no more a vertex along $B^\infty$, because of
   a simplification at the transition from ``$k_2 {B_2}$'' to
   ``$k_1B_{1}$'', then we replace $v_0$ and $v_1$ by the next copies
   of $v$ along $W^\infty$, i.e. at the transition from
   ``$k_1 {B_1}$'' to ``$k_2B_{2}$''.

Since ${C}$ is homologous to $k_1 {B_1} + k_{2}{B_{2}}$,
we can find an infinite path $C^\infty$, that corresponds to copies of
$C$ replicated, that does not intersect $B^\infty$ and situated on the
right side of $B^\infty$. Now we can find a copy $B'^\infty$ of
$B^\infty$, such that $C^\infty$ lies between $B^\infty$ and
$B'^\infty$ without intersecting them. Choose two copies $v'_0,v'_1$
of $v_0,v_1$ on $B'^\infty$ such that the vectors $v_0v_1$ and
$v'_0v'_1$  are equal.

Let $R_0$ be the region bounded by $B^\infty,B'^\infty$.  Let $R_1$
(resp. $R_2$) be the subregion of $R_0$ delimited by $B^\infty$ and
$C^\infty$ (resp. by $C^\infty$ and $B'^\infty$).  We consider
$R_0,R_1,R_2$ as cylinders, where part of the lines $(v_0,v'_0),(v_1,v_1')$
 are identified. Let $B,B',C'$ be the cycles of $R_0$
corresponding to $B^\infty, B'^\infty, C^\infty$ respectively.

 Let $x$ (resp. $y$) be the number of edges of $A(G)^\infty$
  leaving $B$ (resp. $B'$) in $R_0$.  Let $x'$ (resp. $y'$) be the
  number of edges of $A(G)^\infty$ leaving $C'$ on its right
  (resp. left) side in $R_0$.  We have $C'$ corresponds to exactly one copy of
$C$, so $\gamma(C)=x'-y'$. Similarly, we have $B$ and $B'$ that
almost corresponds to $k_1$ copies of $B_1$ followed by $k_2$ copies
of $B_2$, except the fact that we may have removed a back and forth
part (if any).  In any case we have the following:

\begin{claim}
\label{cl:computegamma}
  $k_1\,\gamma
(B_1) + k_2\,\gamma (B_2)=x-y$
\end{claim}

 \begin{proofclaim}
   We prove the case where the common intersection of $B_1,B_2$ is a
   path (if the intersection is a single vertex, the proof is even
   simpler). We assume, w.l.o.g, by eventually reversing one of $B_1$
   or $B_2$, that $B_1,B_2$ are oriented the same way along their
   intersection, so we are in the situation of
   Figure~\ref{fig:replicating0}.

Figure~\ref{fig:computegamma1} shows
   how to compute $k_1\,\gamma (B_1) + k_2\,\gamma (B_2)+y-x$ when
   $(k_1,k_2)=(1,1)$. Then, one can check that
   each outgoing edge of the angle graph is counted exactly the same number
   of time positively and negatively. So everything compensates and we
   obtain $k_1\,\gamma (B_1) + k_2\,\gamma (B_2)+y-x=0$.

\begin{figure}[!ht]
\center
\begin{tabular}{cccccc}
\includegraphics[scale=0.27]{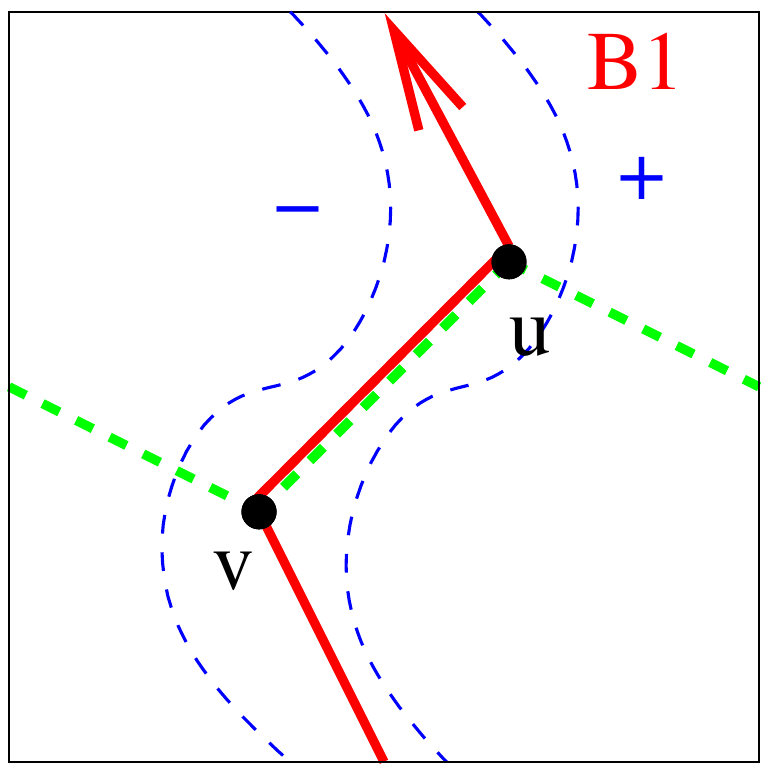}
&&\includegraphics[scale=0.27]{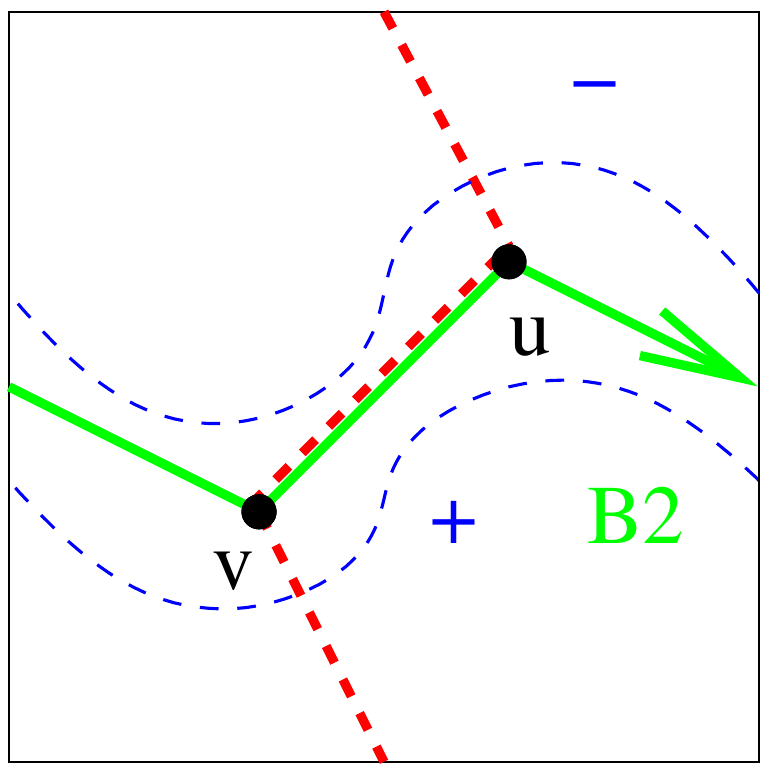}
&&\includegraphics[scale=0.27]{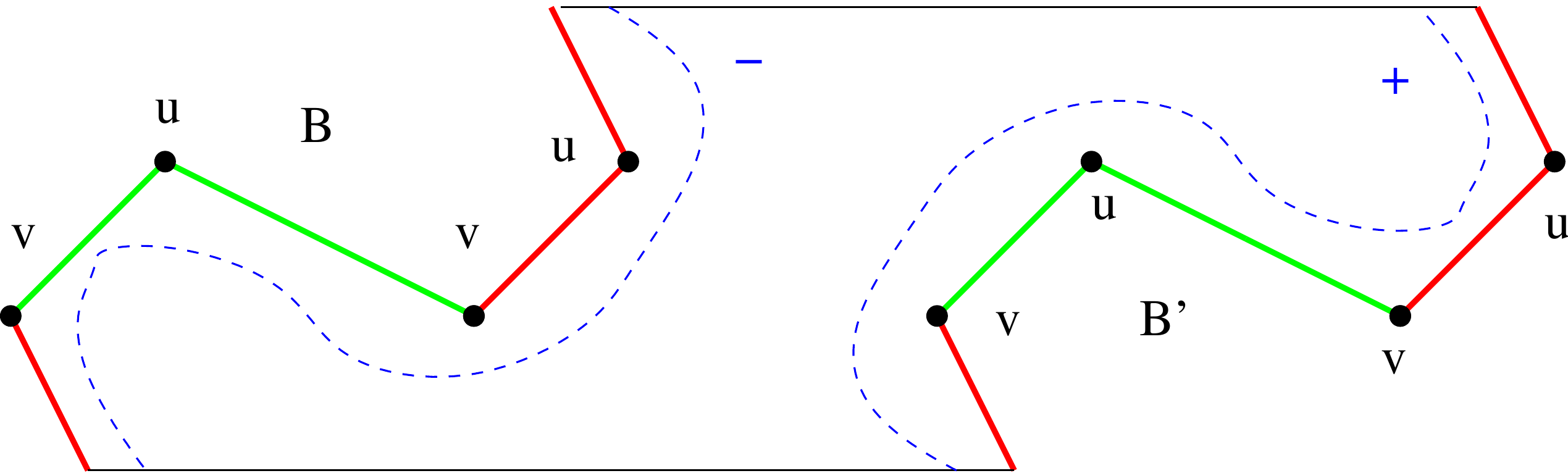} &\\
$\gamma(B_1)$&$+$ & $\gamma (B_2)$ &$+$& $(y-x)$ &$=0$
\end{tabular}
\caption{Case $(k_1,k_2)=(1,1)$.}
\label{fig:computegamma1}
\end{figure}

Figure~\ref{fig:computegamma2} shows how to compute
$k_1\,\gamma (B_1) + k_2\,\gamma (B_2)+y-x$ when
$(k_1,k_2)=(1,-1)$. As above, most of the things compensate but, in
the end, we obtain
$k_1\,\gamma (B_1) + k_2\,\gamma
(B_2)+y-x=d_{A(G)}^+(u)-d_{A(G)}^+(v)$,
as depicted on the figure.  Since the number of outgoing edges of
$A(G)$ around each vertex is equal to $4$, we have again the
conclusion $k_1\,\gamma (B_1) + k_2\,\gamma (B_2)+y-x=0$.

\begin{figure}[!ht]
\center
\begin{tabular}{ccccc}
\includegraphics[scale=0.3]{gammabasis-5}
&&\includegraphics[scale=0.3]{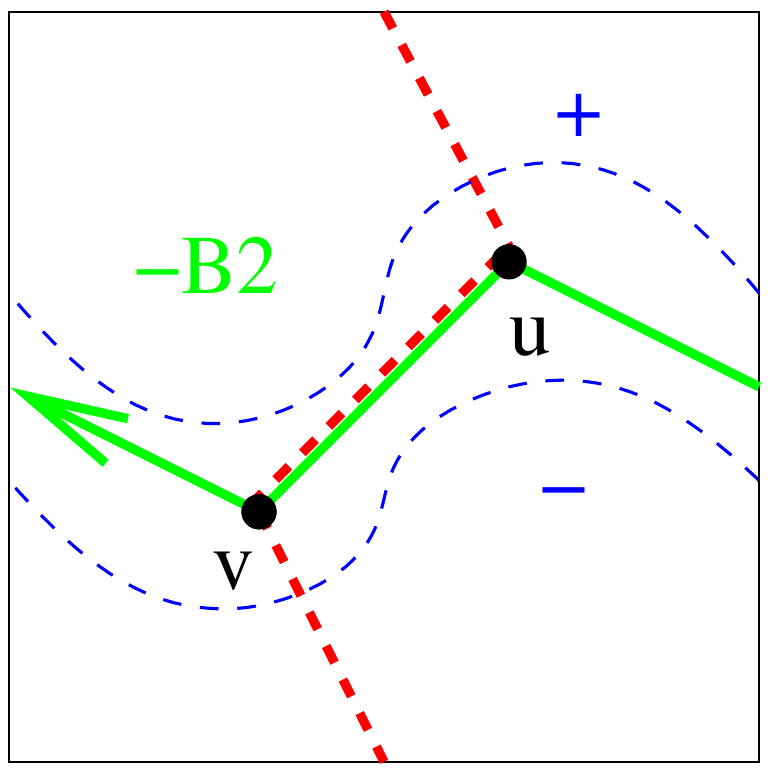}
&&\includegraphics[scale=0.3]{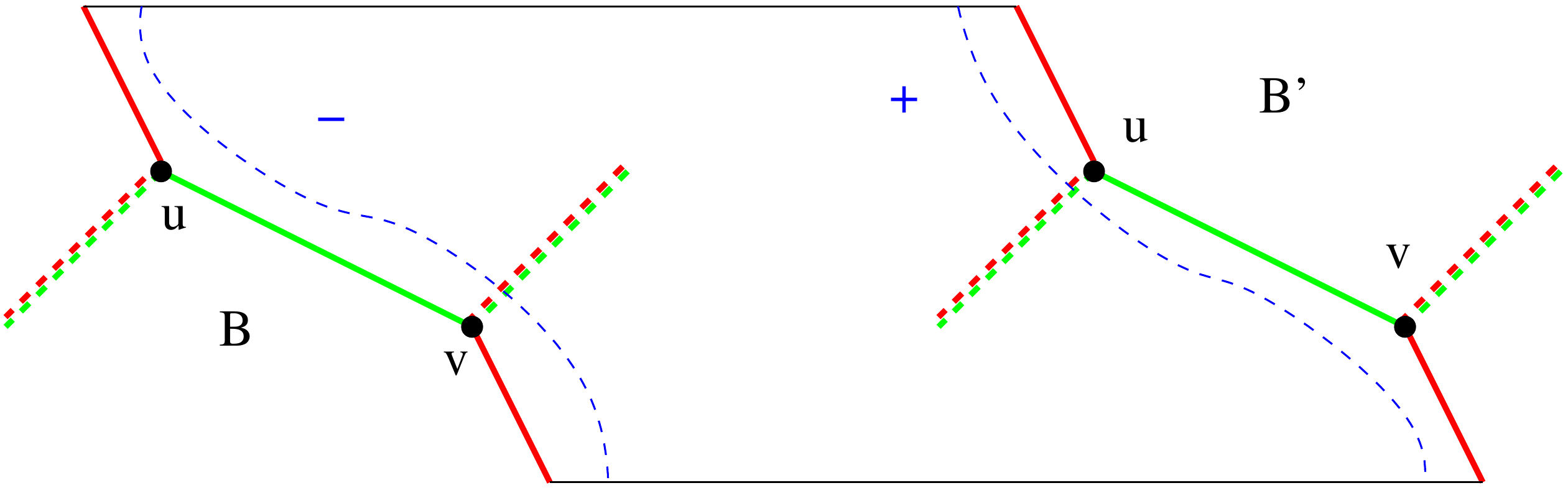}\\
$\gamma(B_1)$&$+$ & $(-\gamma (B_2))$ &$+$& $(y-x)$
\end{tabular}

\includegraphics[scale=0.3]{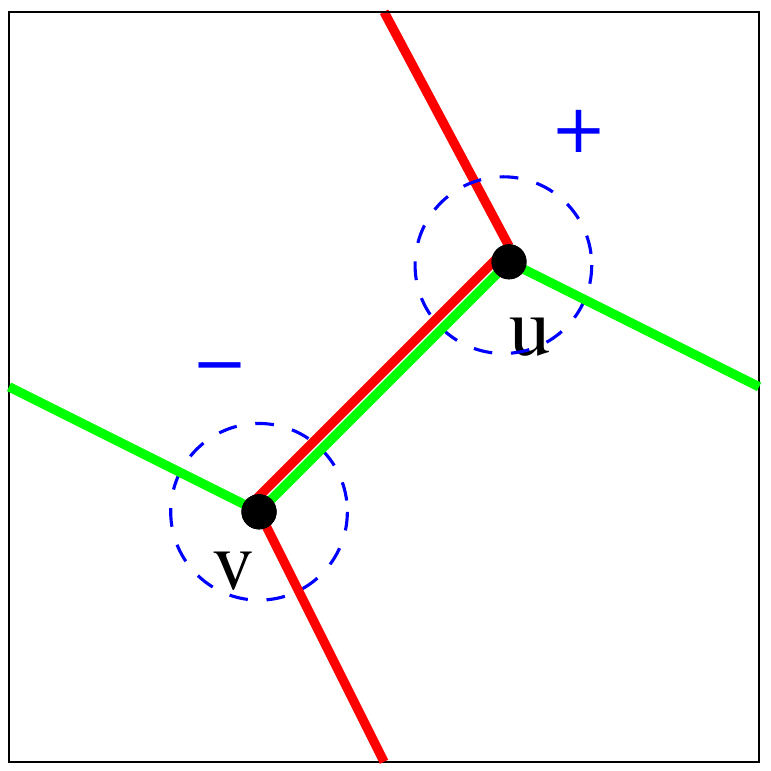}\\
$=d_{A(G)}^+(u)-d_{A(G)}^+(v)=0$
\caption{Case $(k_1,k_2)=(1,-1)$.}
\label{fig:computegamma2}
\end{figure}

One can easily be
convinced that when $|k_1|\geq 1$ and $|k_2|\geq 1$ then the same arguments
 apply. The only difference is that the red or green part of the
figures in the universal cover would be
longer (with  repetitions of $B_1$ and $B_2$). This parts being very
``clean'', they do not affect the way we compute the  equality.
Finally, if one of $k_1$ or $k_2$ is equal to zero, the analysis
is  simpler and the conclusion still holds.
 \end{proofclaim}

 For $i\in\{1,2\}$, let $H_i$ be the cylinder map made of all the
 vertices and edges of ${G}^\infty$ that are in the cylinder region
 $R_i$.  Let $k$ (resp. $k'$) be the length of $B$ (resp. $C'$).  Let
 $n_1,m_1,f_1$ be respectively the number of vertices, edges and faces
 of $H_1$.

  The number of edges of $A(G)^\infty$ in $H_1$ is equal to
  $3f_1$. Since we are considering a $4$-orientation of $A(G)$, these
  edges are decomposed into: the outgoing edges from inner vertices of
  $H_1$ (primal-vertices have outdegree $4$ in $A(G)^\infty$, so there
  are $4(n_1-k-k')$ such edges), the outgoing edges from outer
  vertices of $H_1$ (there are $x+y'$ such edges), and the outgoing
  edges from faces of $H_1$ (faces have outdegree $1$ in
  $A(G)^\infty$, so there are $f_1$ such edges). Finally, we have
  $3f_1=4(n_1-k-k')+x+y'+f_1$. Combining this with Euler's formula,
  $n_1-m_1+f_1=0$, and the fact that all the faces of $H_1$ are
  triangles, i.e. $2m_1=3f_1+k+k'$, one obtains that $x+y'=2(k+k')$.
  Similarly, by considering $H_2$, one obtain that
  $x'+y=2(k+k')$. Thus finally $x+y'=x'+y$ and thus
  $\gamma(C)=k_1\,\gamma (B_1) + k_2\,\gamma (B_2)$ by
  Claim~\ref{cl:computegamma}.
\end{proof}

Lemma~\ref{lem:gammahomology} implies the following:

\begin{lemma}
\label{lem:gamma0all}
In a 4-orientation of $A(G)$, if for two non-contractible not weakly
homologous cycles $C,C'$ of ${G}$, we have
$\gamma(C)=\gamma(C')=0$, then the $4$-orientation of $A(G)$
is balanced.
\end{lemma}

\begin{proof}
  Consider two non-contractible not weakly homologous cycles $C,C'$ of
  ${G}$ such that $\gamma(C)=\gamma(C')=0$. Consider an homology-basis
  $\{B_1,B_2\}$ of $G$, such that $B_1,B_2$ are non-contractible
  cycles whose intersection is a single vertex or a path (see
  Section~\ref{sec:homology} for discussion on existence of such a
  basis). Let $k_1,k_2,k'_1,k'_2\in \mathbb Z$, such that $C$
  (resp. $C'$) is homologous to $k_1 B_1+k_2 B_2$ (resp.
  $k'_1 B_1+k'_2 B_2$).  Since $C$ is non-contractible we have
  $(k_1,k_2)\neq (0,0)$. By eventually exchanging $B_1,B_2$, we can
  assume, w.l.o.g., that $k_1\neq 0$.  By
  Lemma~\ref{lem:gammahomology}, we have
  $k_1 \gamma(B_1)+k_2 \gamma(B_2)=\gamma(C)=0=\gamma(C')=k'_1
  \gamma(B_1)+k'_2 \gamma(B_2)$.
  So $\gamma(B_1)=(-k_2/k_1) \gamma(B_2)$ and thus
  $(-k_2 k'_1/k_1 + k'_2)\gamma(B_2)=0$. So $k'_2=k_2 k'_1/k_1$ or
  $\gamma(B_2)=0$.  Suppose by contradiction, that
  $\gamma(B_2)\neq 0$.  Then
  $(k'_1,k'_2)= \frac{k'_1}{k_1} (k_1,k_2)$, and $C'$ is homologous to
  $\frac{k'_1}{k_1} C$.  Since $C$ and $C'$ are both non-contractible
  cycles, it is not possible that one is homologous to a multiple of
  the other, with a multiple different from $-1,1$. So $C, C'$ are
  weakly homologous, a contradiction.  So $\gamma(B_2)= 0$ and thus
  $\gamma(B_1)=0$.  Then by Lemma~\ref{lem:gammahomology}, any
  non-contractible cycle of ${G}$, have $\gamma$ equal to $0$.  Thus
  the $4$-orientation is balanced
\end{proof}

\subsection{Decontraction preserving ``balanced''}

The goal of this section is to prove the following lemma:

\begin{lemma}
  \label{lem:decontraction}
  If $G$ is a toroidal triangulation given with a non-loop edge $e$
  whose extremities are of degree at least four and such that $G/ e$
  admits a balanced transversal structure, then $G$ admits a balanced
  transversal structure.
\end{lemma}

\begin{proof} Let $G'=G/e$ and consider a balanced transversal
  structure $G'$.  We show how to extend the balanced transversal
  structure of $G'$ to a balanced transversal structure of $G$.  Let
  $u,v$ be the two extremities of $e$ and $x,y$ the two vertices of
  $G$ such that the two faces incident to $e$ are $A=u,v,x$ and
  $B=v,u,y$ in clockwise order (see Figure~\ref{fig:contraction-tri}).
  Note that $u$ and $v$ are distinct by definition of edge contraction
  but that $x$ and $y$ are not necessarily distinct, nor distinct from
  $u$ and $v$. Let $w$ be the vertex of $G'$ resulting from the
  contraction of $e$. Let $e_{wx}, e_{wy}$ be the two edges of $G'$
  represented on Figure~\ref{fig:contraction-tri} (these edges are
  identified and form a loop on
  Figure~\ref{fig:contraction-loop-tri}).

  There are different cases to consider, corresponding to the
  different possible orientations and colorings of the edges $e_{wx}$
  and $e_{wy}$ in $G'$. By symmetry, there are just three cases to
  consider for Figure~\ref{fig:contraction-tri}: edges $e_{wx}$ and
  $e_{wy}$ may be in consecutive, same or opposite intervals,
  w.r.t.~the four intervals of the local property around $w$. When
  $e_{wx}, e_{wy}$ are identified as in
  Figure~\ref{fig:contraction-loop-tri}, these ``two'' edges are
  necessarily in opposite intervals, and there is just one case to
  consider.  So only the four cases represented on the left side of
  Figure~\ref{fig:decontractrule} by case $x.0$ for
   $x\in\{a,b,c,d\}$ have to be considered.

\begin{figure}[!ht]
\center
\begin{tabular}{ccc}
  \includegraphics[scale=0.4]{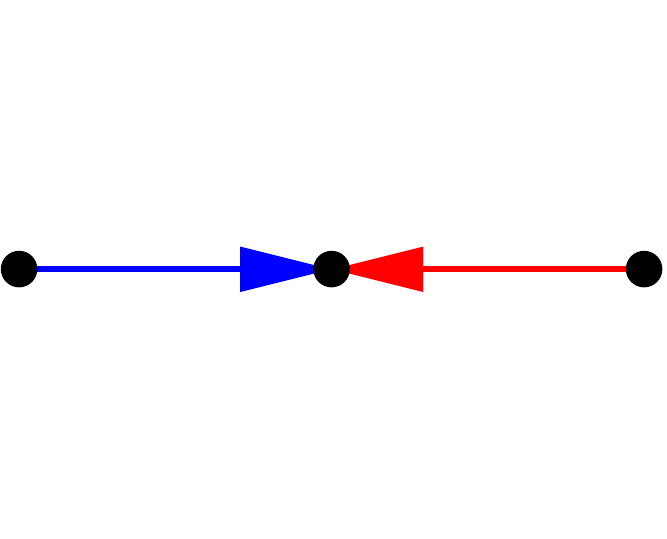} \ \ & \ \
  \includegraphics[scale=0.4]{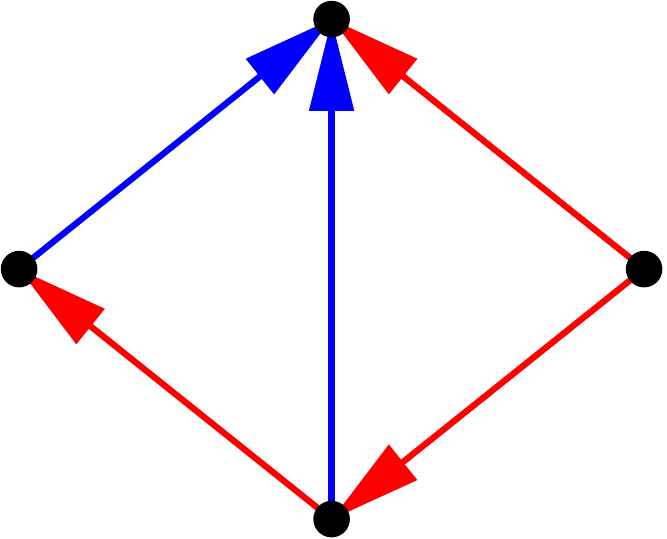} \ \ & \ \
  \includegraphics[scale=0.4]{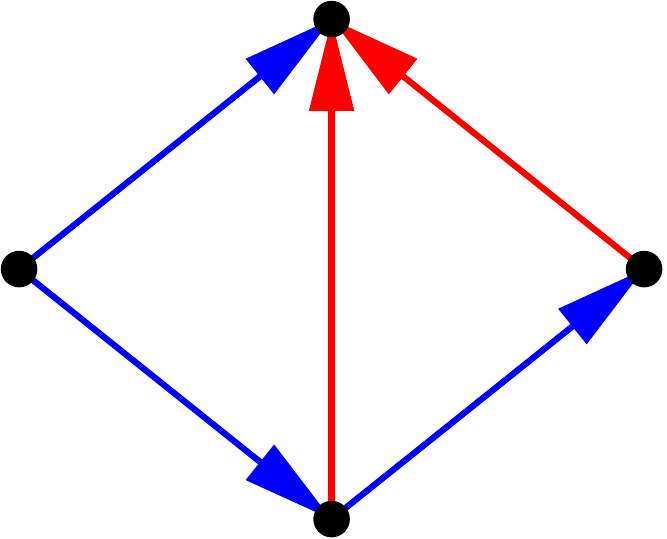} \\
a.0 \ \ & \ \ a.1 \ \ & \ \ a.2 \\ 
\\
 \includegraphics[scale=0.4]{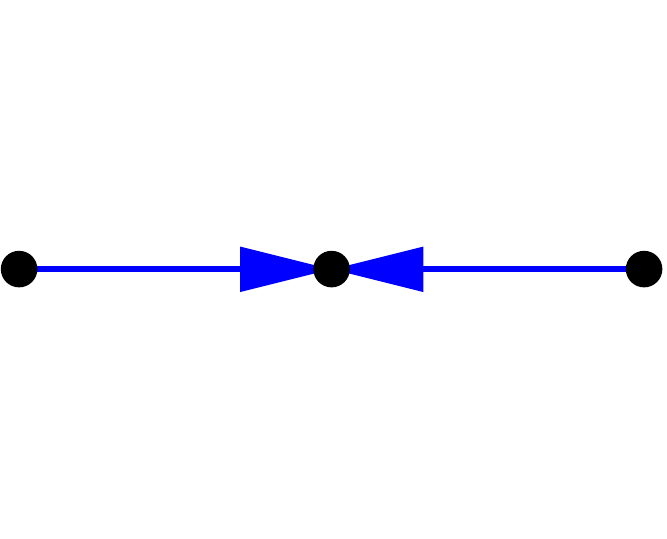} \ \ & \ \
  \includegraphics[scale=0.4]{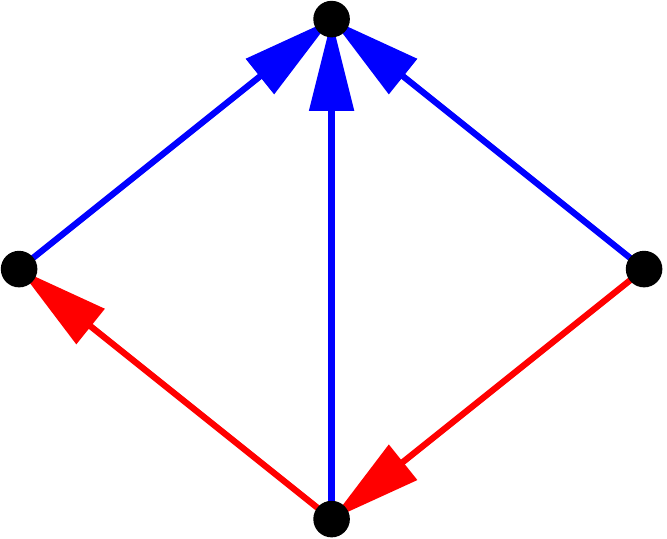} \ \ & \ \
  \includegraphics[scale=0.4]{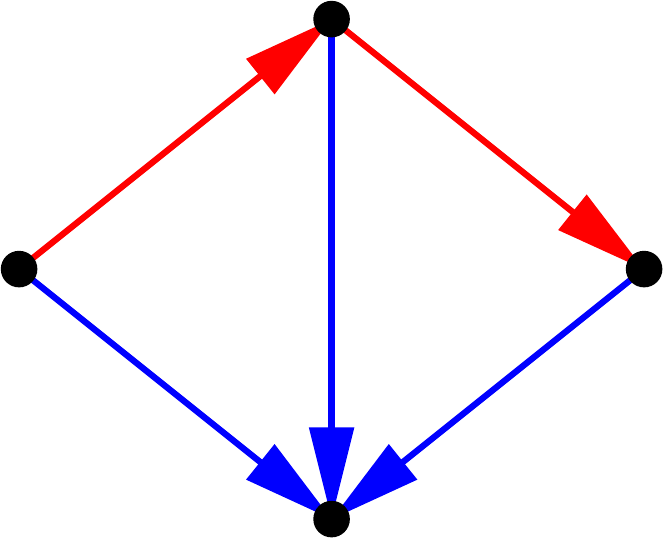} 
\\
b.0 \ \ & \ \ b.1 \ \ & \ \  b.2 \\ 
\\
  \includegraphics[scale=0.4]{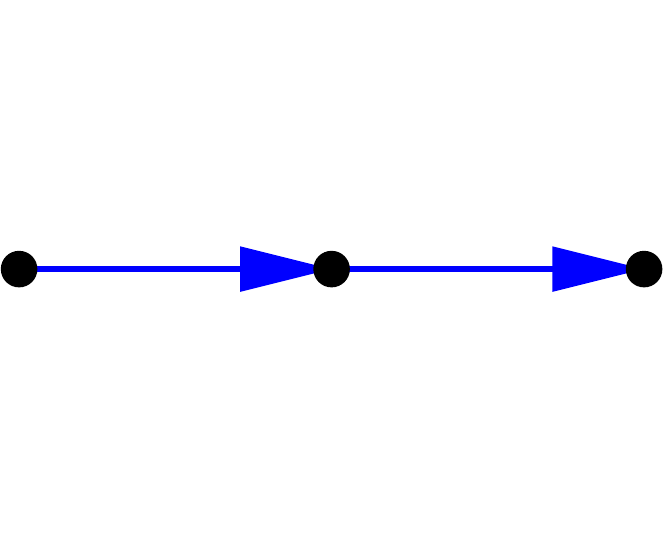} \ \ & \ \
  \includegraphics[scale=0.4]{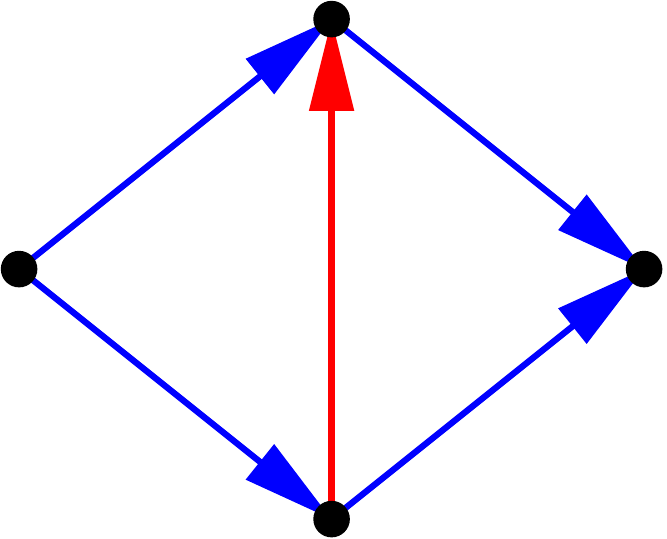} \ \ & \ \
 \\
c.0 \ \ & \ \ c.1 \ \ & \ \  \\ 
\\
  \includegraphics[scale=0.4]{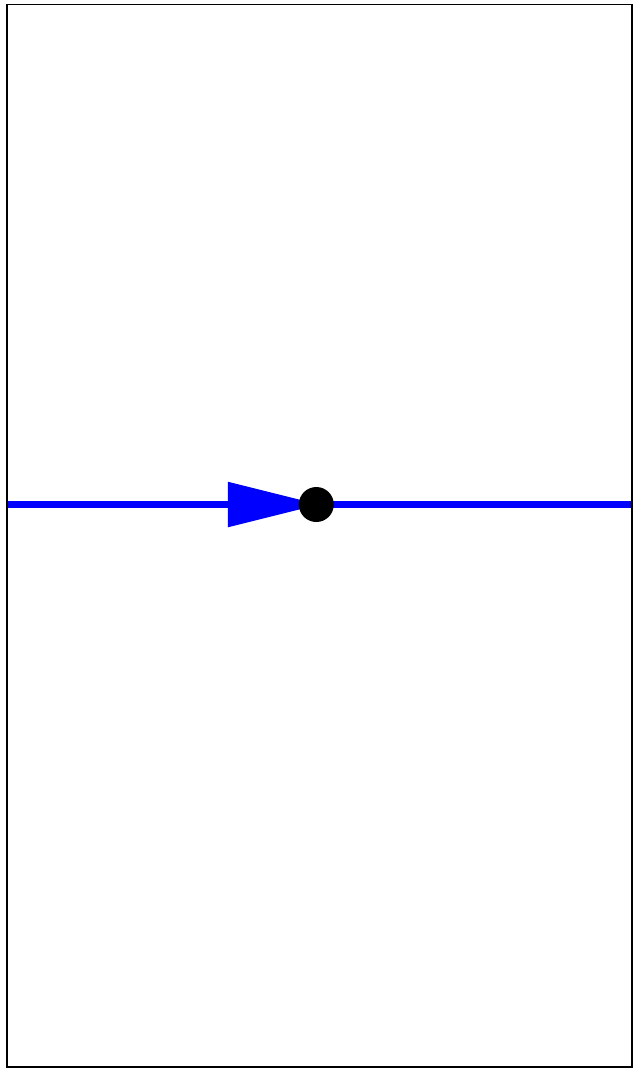} \ \ & \ \
  \includegraphics[scale=0.4]{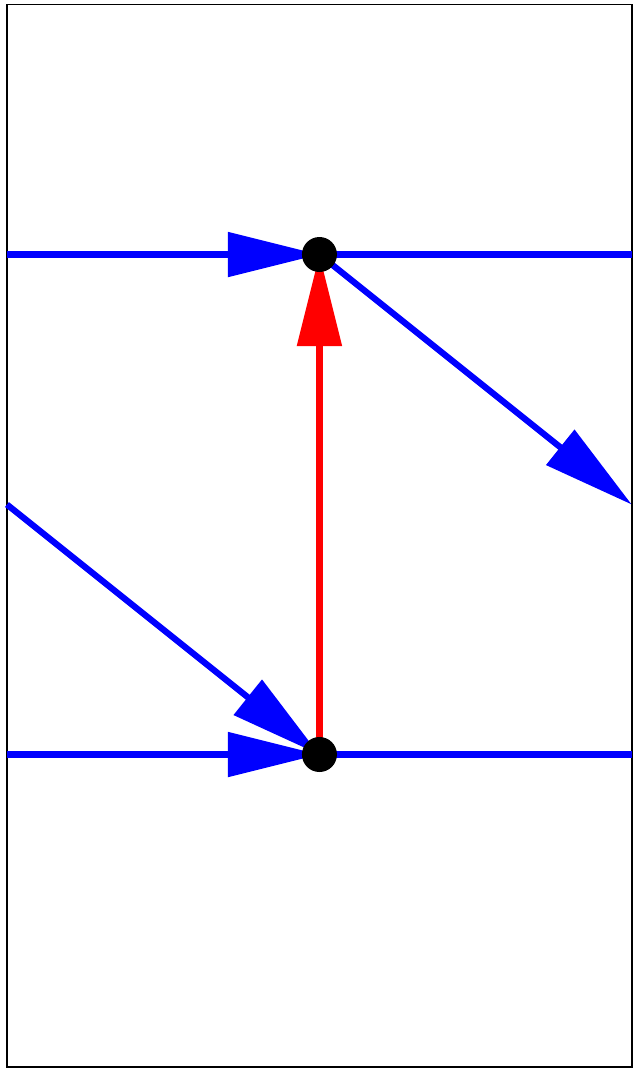} \ \ & \ \
 \\
d.0 \ \ & \ \ d.1 \ \ & \ \  \\ 
\\
\end{tabular}
\caption{Decontraction rules of the transversal structure.}
\label{fig:decontractrule}
\end{figure}

In each case $x.0$, we prove that one can color and orient the edges
of $G$ to obtain a balanced transversal structure of $G$. For that
purpose, just the edges of $G$ that are labeled on
Figures~\ref{fig:contraction-tri} and~\ref{fig:contraction-loop-tri}
have to be specified, all the other edges of $G$ keep the orientation
and coloring that they have in the balanced transversal structure of
$G'$.  For each case $x.0$, with $x\in\{a,b,c,d\}$, the orientation
and coloring of the labeled edges of $G$ are represented on the right
side of Figure~\ref{fig:decontractrule} by case $x.i$, $i\geq 1$. For
the first two cases $x\in\{a,b\}$, we might have to choose between
cases $x.1$ and $x.2$, to orient and color $G$, depending on a case
analysis explained below. For the other cases, $x\in\{c,d\}$, there is
just one coloring and orientation of $G$ to consider: case $x.1$.

The \emph{sector} $]e_1,e_2[$ of a vertex $w$, for $e_1$ and $e_2$ two edges
incident to $w$, is the counterclockwise sector of $w$ between $e_1$
and $e_2$, excluding the edges $e_1$ and $e_2$.

Let us consider the different possible orientations and
colorings of edges $e_{wx}$ and $e_{wy}$.

\begin{itemize}
\item \emph{$e_{wx}$ and $e_{wy}$ are not identified and  in
    consecutive intervals:}

  W.l.o.g., we might assume that we are in case a.0 of
  Figure~\ref{fig:decontractrule}, i.e. $e_{wx}$ (resp. $e_{wy}$) is a
  blue (resp. red) edge entering $w$.

  Since all vertices of $G$ have degree at least $4$, we have that $v$
  is incident to at least one edge in the sector $]e_{vx},e_{vy}[$ of
  $v$. So $w$ is incident to at least one edge in the sector
  $]e_{wx},e_{wy}[$ of $w$. Such an edge can be blue or red in the
  transversal structure of $G'$. Depending on if there is such a blue
  or red edge we apply coloring a.1 or a.2 to $G$, as explain below.

  W.l.o.g., we can assume that there is a blue edge incident to $w$ in
  the sector $]e_{wx},e_{wy}[$. By the local rule, this edge is
  entering $w$. Moreover the edge incident to $x$ and just after
  $e_{wx}$ in clockwise order around $x$ is entering $x$ in color
  red. So we are in the situation depicted on the left of
  Figure~\ref{fig:a-0-more}. Apply the coloring a.1 to $G$ as depicted
  on the right of Figure~\ref{fig:a-0-more}.  One can easily check
  that the local property is satisfied around every vertex of $G$ (for
  that purpose one just as to check the local property around
  $u,v,x,y$). Thus we obtain a  transversal structure of $G$.

\begin{figure}[!ht]
\center
\includegraphics[scale=0.4]{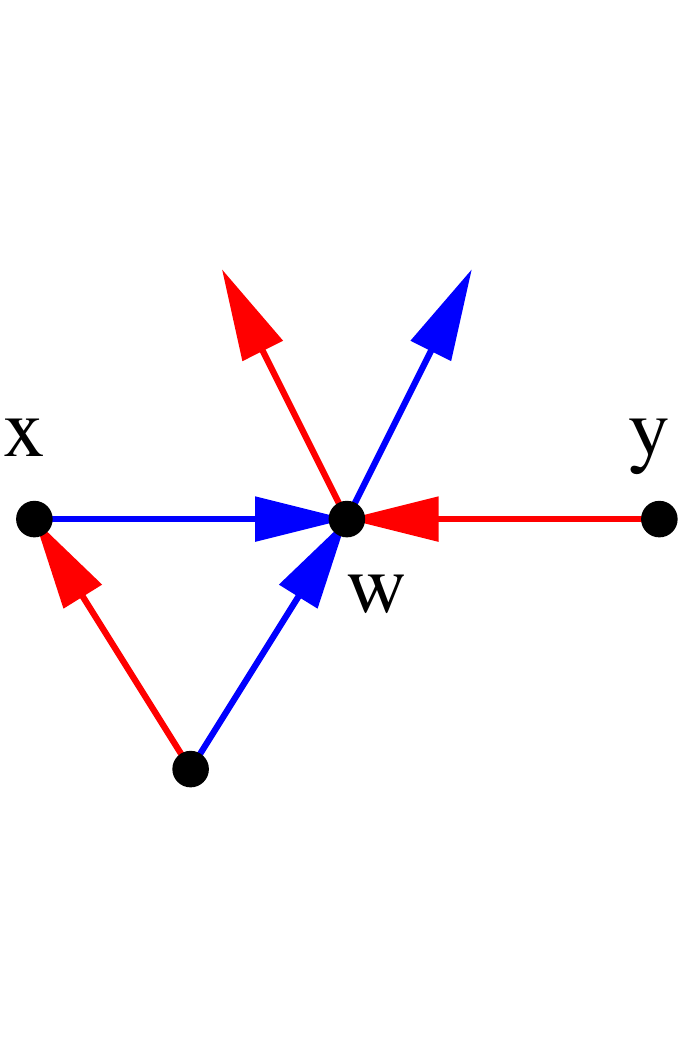} \ \ \ \ \ \
\includegraphics[scale=0.4]{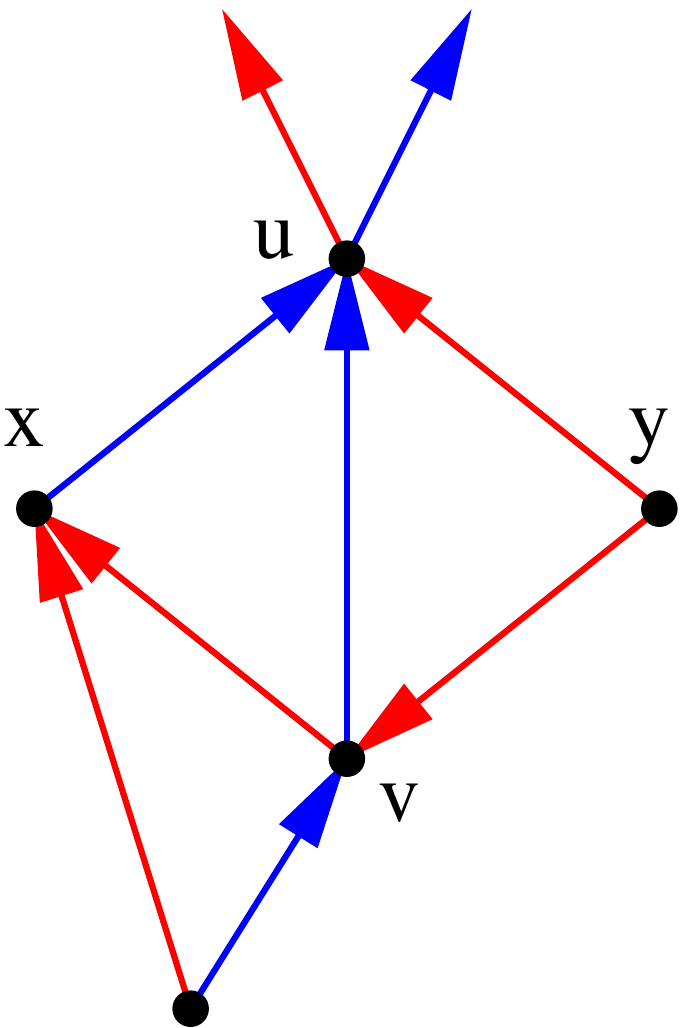}
\caption{Decontraction of case a.0 when there is a blue
  edge entering $w$ by below.}
\label{fig:a-0-more}
\end{figure}

It remains to prove that the obtained transversal structure is
balanced. For that purpose consider two non-contractible not weakly
homologous cycles $(B'_1,B'_2)$ of $G'$. Since the transversal
structure of $G'$ is balanced, we have $\gamma(B'_1)=\gamma(B'_2)=0$
by definition of balanced property. If $B'_i$ does not intersect $w$,
then it is not affected by the decontraction, let $B_i=B'_i$, so that
$B_i$ is a cycle of $G$ with same homology and same value $\gamma$ as
$B'_i$. If $B'_i$ intersects $w$, then one can consider where $B'_i$
is entering and leaving the contracted region and replace $B'_i$ by a
cycle $B_i$ of $G$ with the same homology and same value $\gamma$ as
$B'_i$. This property illustrated on Figure~\ref{fig:a-0-angle}, on
the left part we consider an example of a cycle $B'_i$ and on the
right part we give a corresponding cycle $B_i$ having the same
homology and $\gamma$ as $B'_i$. To be completely convinced that this
transformation works, one may consider all the different possibility
to enter and leave the contracted region and check that one can find a
corresponding cycle of $G$ with same $\gamma$. Finally, by this
method, we obtain two non-contractible not weakly homologous cycles
$B_1,B_2$ of $G$ with $\gamma(B_1)=\gamma(B_2)=0$, so by
Lemma~\ref{lem:gamma0all}, the obtained transversal structure is
balanced.

\begin{figure}[!ht]
\center
\includegraphics[scale=0.4]{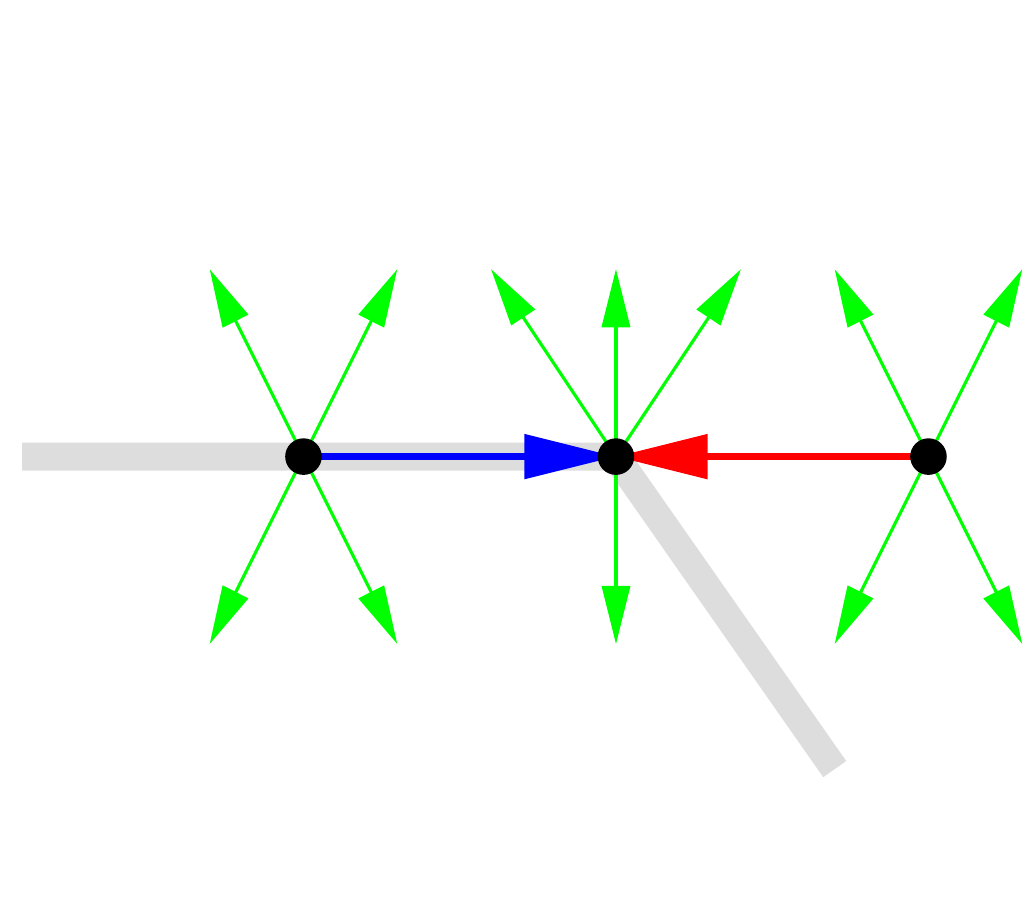} \ \ \ \ \ \
\includegraphics[scale=0.4]{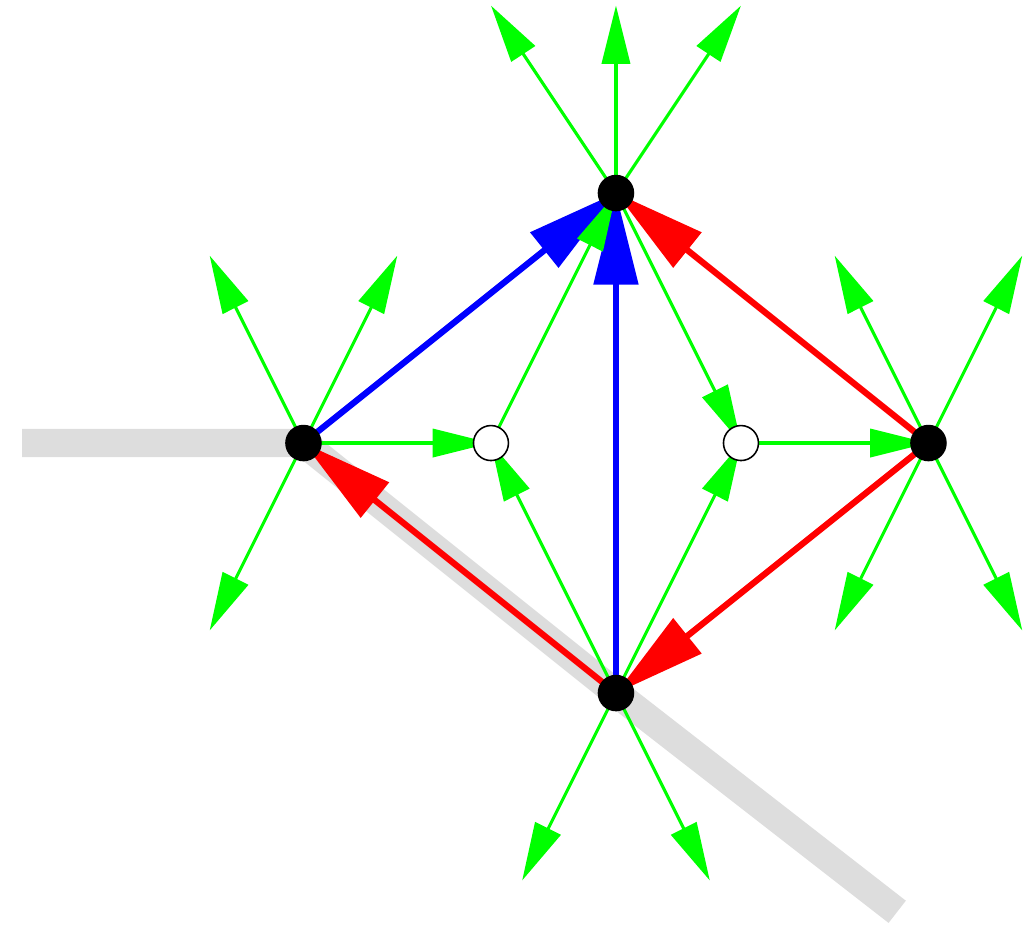}
\caption{Decontraction of case a.0 preserving $\gamma$.}
\label{fig:a-0-angle}
\end{figure}

\item \emph{$e_{wx}$ and $e_{wy}$ are not identified and  in
    the same interval:}

  W.l.o.g., we might assume that we are in case b.0 of
  Figure~\ref{fig:decontractrule}, i.e. $e_{wx}$ and $e_{wy}$ are blue
  edges entering $w$.  Depending on if the outgoing blue interval of
  $w$ is above or below $w$ we apply coloring b.1 or b.2 to $G$, as
  explain below.

  W.l.o.g., we can assume that the outgoing blue edges incident to $w$
  are above $w$. Since all vertices of $G$ have degree at least $4$,
  we have that $v$ is incident to at least one edge in the sector
  $]e_{vx},e_{vy}[$ of $v$. So $w$ is incident to at least one edge in
  the sector $]e_{wx},e_{wy}[$ of $w$. By the local rule, we are in
  the situation depicted on the left of
  Figure~\ref{fig:b-0-more}. Apply the coloring b.1 to $G$ as depicted
  on the right of Figure~\ref{fig:b-0-more}.  One can easily check
  that the local property is satisfied around every vertex of $G$ (for
  that purpose one just as to check the local property around
  $u,v,x,y$). Thus we obtain a transversal structure of $G$.

\begin{figure}[!ht]
\center
\includegraphics[scale=0.4]{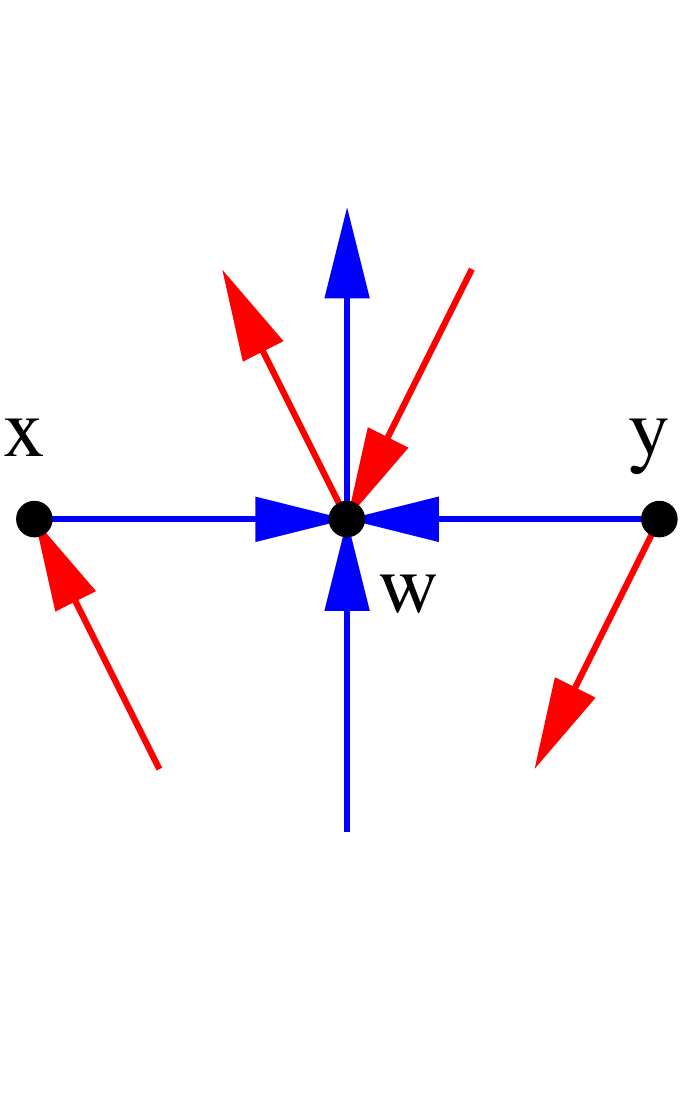} \ \ \ \ \ \
\includegraphics[scale=0.4]{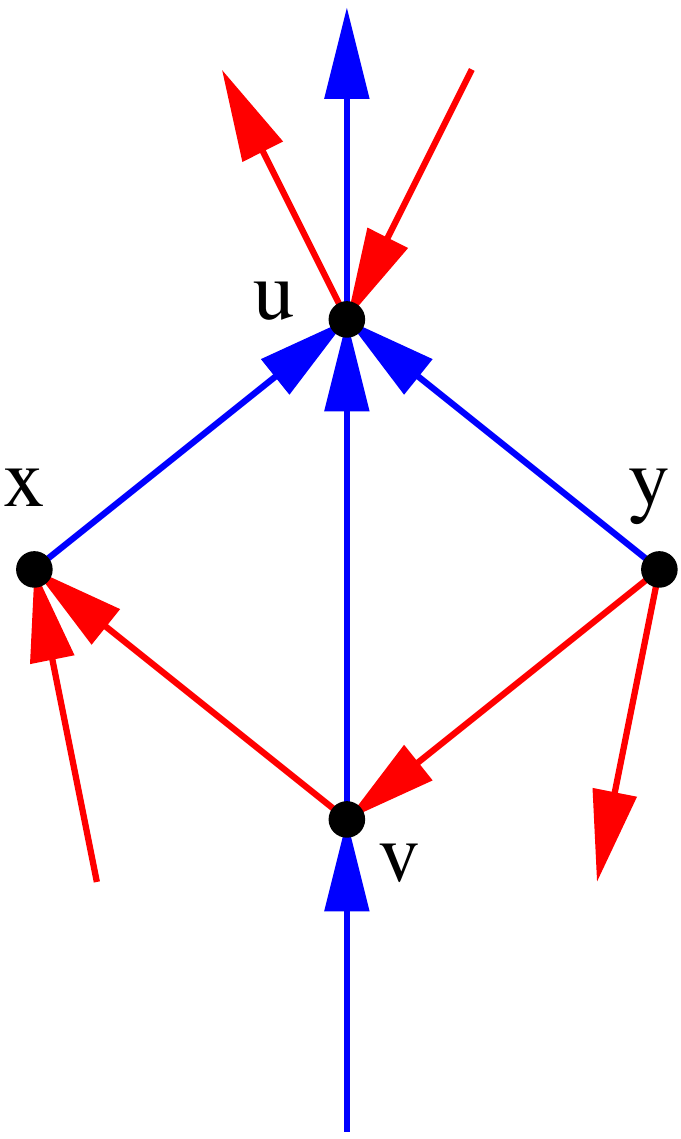}
\caption{Decontraction of case b.0 when there is a blue
  edge leaving $w$ by above.}
\label{fig:b-0-more}
\end{figure}

Similarly as the previous case the balanced property is preserved (see
Figure~\ref{fig:b-0-angle} for an example). So we obtain a balanced transversal
structure of $G$.

\begin{figure}[!ht]
\center
\includegraphics[scale=0.4]{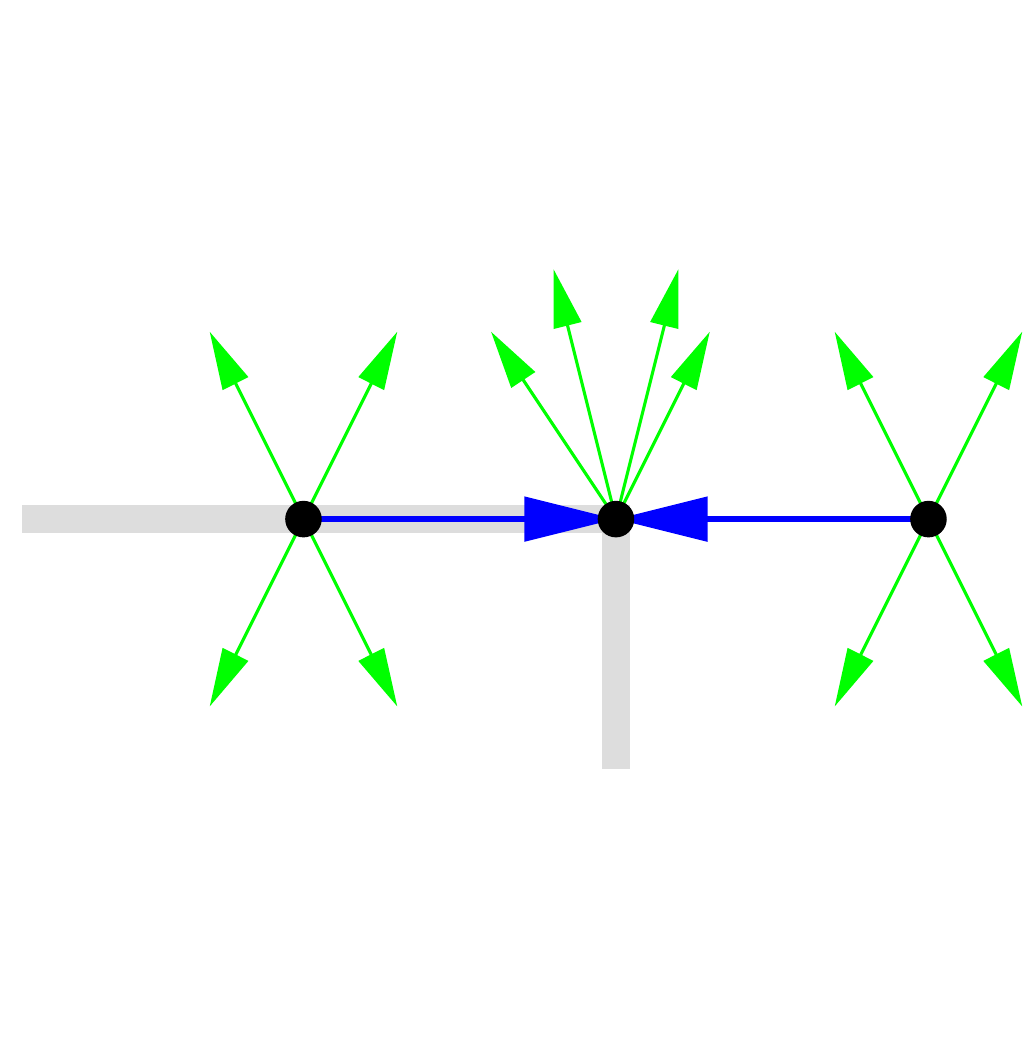} \ \ \ \ \ \
\includegraphics[scale=0.4]{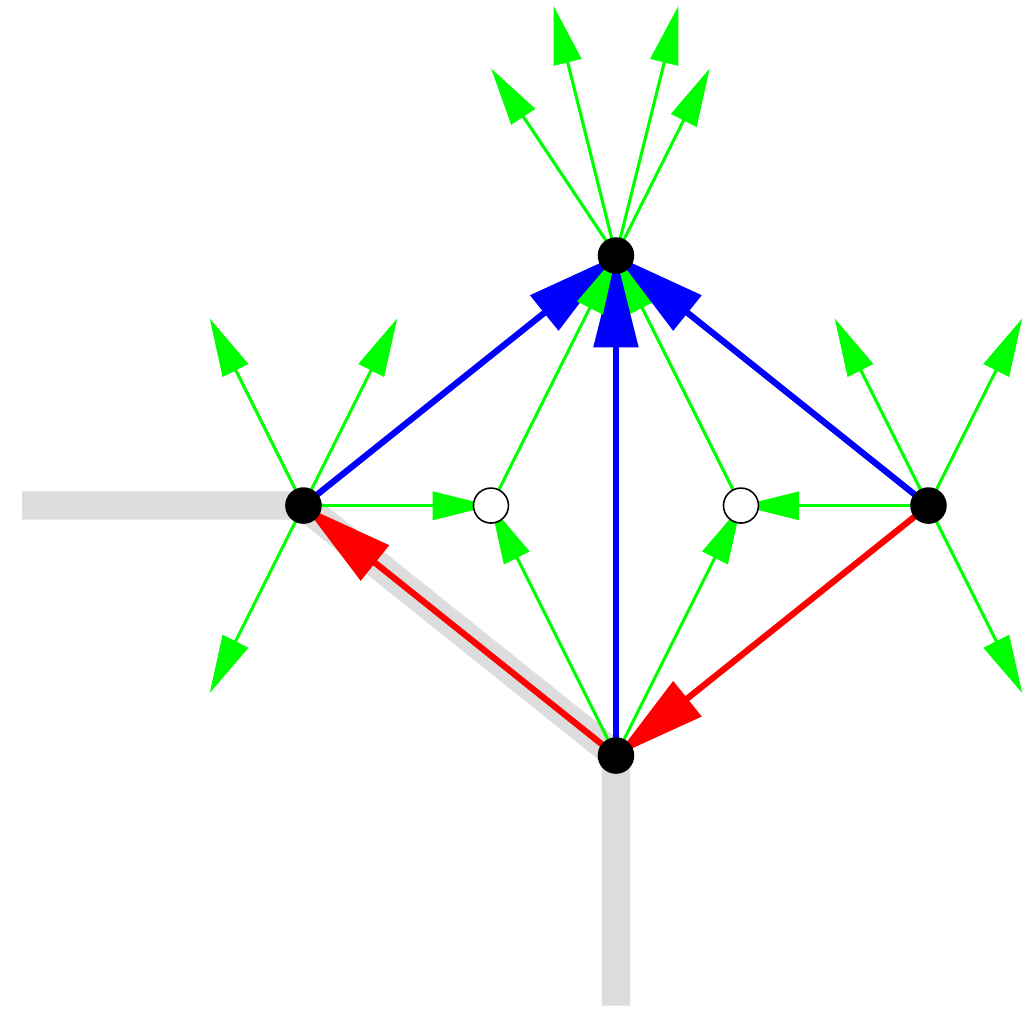}
\caption{Decontraction of case b.0 preserving $\gamma$.}
\label{fig:b-0-angle}
\end{figure}

\item \emph{$e_{wx}$ and $e_{wy}$ are not identified and  in
    opposite intervals:}

  W.l.o.g., we might assume that we are in case c.0 of
  Figure~\ref{fig:decontractrule}, i.e. $e_{wx}$ (resp. $e_{wy}$) is
  entering (resp. leaving) $w$ in color blue.  We apply coloring c.1
  to $G$. One can easily check that the local property is satisfied
  around every vertex of $G$ thus we obtain a transversal structure of
  $G$. Similarly as before the balanced property is preserved and we
  obtain a balanced transversal structure of $G$.

\item \emph{$e_{wx}$ and $e_{wy}$ are identified:}

  W.l.o.g., we might assume that we are in case d.0 of
  Figure~\ref{fig:decontractrule}.  We apply coloring d.1
  to $G$. One can easily check that the local property is satisfied
  around every vertex of $G$ thus we obtain a transversal structure of
  $G$. Similarly as before the balanced property is preserved and we
  obtain a balanced transversal structure of $G$.

\end{itemize}

For each different possible orientations and
colorings of edges $e_{wx}$ and $e_{wy}$ we are able to extend the
balanced transversal structure of $G'$ to $G$ and thus obtain the result.
\end{proof}

\subsection{Existence theorem}
\label{sec:existencemainproof}

We are now able to prove the  existence of balanced transversal
structure:

\begin{theorem}
\label{th:existence}
A toroidal triangulation admits a balanced transversal structure if
and only if it is essentially 4-connected.
\end{theorem}

\begin{proof}
($\Longrightarrow$) Clear by Lemma~\ref{lem:e4c}.

($\Longleftarrow$) Let $G$ be an essentially 4-connected toroidal
triangulation.  By Lemma~\ref{lem:contraction}, it can be contracted
to a map on one vertex by keeping the map an essentially 4-connected
toroidal triangulation. Since during the contraction process, the
universal cover is 4-connected, all the vertices have degree at least
$4$.  The toroidal triangulation on one vertex admits a transversal
structure. One such example is given on Figure~\ref{fig:1vertex} where
one can check that all non-contractible cycles (the three loops) have
value $\gamma=0$, and so the transversal structure is balanced. Then, by
Lemma~\ref{lem:decontraction} applied successively, one can decontract
this balanced transversal structure to obtain a balanced transversal
structure of $G$.
\end{proof}

\section{Distributive lattice of balanced 4-orientations}
\label{sec:latticemain}

\subsection{Transformations between balanced 4-orientations}

Consider an essentially 4-connected toroidal triangulation $G$ and its
angle map $A(G)$.  In~\cite[Theorem~4.7]{GKL15}, it is proved that the
set of homologous orientations of a given map on an orientable surface
carries a structure of distributive lattice. We want to use such a
result for the set of balanced $4$-orientations of $A(G)$. For that
purpose we prove in this section that balanced $4$-orientations are
homologous to each other, i.e. the set of edges that have to be
reversed to transform one balanced $4$-orientation into another is a
0-homologous oriented subgraph.

If $D,D'$ are two orientations of $A(G)$, let $D\setminus D'$ denote
the subgraph of $D$ induced by the edges that are not oriented as in
$D'$.
We have the following:

\begin{lemma}
\label{lem:balancediffhomolog}
Let $D$ be a balanced $4$-orientation of $A(G)$. An orientation $D'$ of
$A(G)$ is a balanced $4$-orientation if and only if $D,D'$ are homologous
(i.e. $D\setminus D'$ is $0$-homologous).
\end{lemma}

\begin{proof}  
  Let $T=D\setminus D'$. Let $\Out$ (resp. $\Out'$) be the set of
  edges of $D$ (resp. $D'$) which are going from a primal-vertex to a
  dual-vertex.  Note that an edge of $T$ is either in $\Out$ or in
  $\Out'$, so $\phi(T)=\phi(\Out)- \phi(\Out')$. Consider two cycles
  $B_1,B_2$ of $G$, given with a direction of traversal, that form a
  basis for the homology. We also denote by $D,D'$ the orientation of
  $\PDC{A(G)}$ corresponding to $D,D'$. Then we denote $\gamma_D$,
  $\delta_D$, $\gamma_{D'}$, $\delta_{D'}$, the function $\gamma$ and
  $\delta$ computed in $D$ and $D'$ respectively (see terminology of
  Section~\ref{sec:char}).

  $(\Longrightarrow)$ Suppose $D'$ is a balanced $4$-orientation of
  $A(G)$. Since $D,D'$ are both $4$-orientations of $A(G)$, they have
  the same outdegree for every vertex of $A(G)$. So we have that $T$
  is Eulerian.  Let $\widehat{\mc{F}}^*$ be the set of counterclockwise facial
  walks of $\PDC{A(G)}^*$, so for any $F\in \widehat{\mc{F^*}}$, we have
  $\beta(T,F)=0$.  Moreover, for $i\in\{1,2\}$, consider the region
  $R_i$ between $W_L(B_i)$ and $W_R(B_i)$ containing $B_i$. Since $T$
  is Eulerian, it is going in and out of $R_i$ the same number of
  times. So $\beta(T,W_L(B_i)-W_R(B_i))=0$ and by linearity of
  function $\beta$ we obtain $\beta(T,W_L(B_i))=\beta(T,W_R(B_i))$.
  Since $D,D'$ are balanced, we have
  $\gamma_D(B_i)=\gamma_{D'}(B_i)=0$. So by
  Lemma~\ref{lem:gammaequaldelta},
  $\delta_D(W_L(B_i))+\delta_D(W_R(B_i))=\delta_{D'}(W_L(B_i))+\delta_{D'}(W_R(B_i))$. Thus
  $\beta(T,W_L(B_i)+W_R(B_i))=\beta(Out-Out',W_L(B_i)+W_R(B_i))=\delta_D(W_L(B_i))+\delta_D(W_R(B_i))-\delta_{D'}(W_L(B_i))-\delta_{D'}(W_R(B_i))=0$.
  By linearity of function $\beta$ we obtain
  $\beta(T,W_L(B_i)=-\beta(T,W_R(B_i))$.  By combining this with the
  above equality, we obtain $\beta(T,W_L(B_i))=\beta(T,W_R(B_i))=0$
  for $i\in\{1,2\}$.  Since $\{W_L(B_1),W_L(B_2)\}$ form a basis for the
  homology of $\PDC{A(G)}^*$, we obtain, by Lemma~\ref{lm:homologous}, that
  $T$ is 0-homologous and so $D,D'$ are homologous to each other.

  $(\Longleftarrow)$ Suppose that $D,D'$ are homologous, i.e. $T$ is
  $0$-homologous.  Then $T$ is in particular Eulerian, so $D'$ as the
  same outdegrees as $D$. So $D'$ is a $4$-orientation of $A(G)$.  By
  Lemma~\ref{lm:homologous}, for $i\in\{1,2\}$, we have
  $\beta(T,W_L(B_i))=\beta(T,W_R(B_i))=0$. Thus
  $\delta_{D}(W_L(B_i))=\beta(Out,W_L(B_i)) +2\beta(\Dual,W_L(B_i))=
  \beta(Out',W_L(B_i))+2\beta(\Dual,W_L(B_i))=\delta_{D'}(W_L(B_i))$
  and
  $\delta_{D}(W_R(B_i))=\beta(Out,W_R(B_i)) +2\beta(\Dual,W_R(B_i))=
  \beta(Out',W_R(B_i))+2\beta(\Dual,W_R(B_i))=\delta_{D'}(W_R(B_i))$.
  So by Lemma~\ref{lem:gammaequaldelta},
  $\gamma_{D}(B_i)=\delta_{D}(W_L(B_i))+\delta_{D}(W_R(B_i))=\delta_{D'}(W_L(B_i))+\delta_{D'}(W_R(B_i))=\gamma_{D'}(B_i)$. Since
  $D$ is balanced, we have $\gamma_D(B_i)=0$ and so
  $\gamma_{D'}(B_i)=0$.  Then, by Lemma~\ref{lem:gamma0all}, we have
  $D'$ is a balanced $4$-orientation of $A(G)$.
\end{proof}

\subsection{Distributive lattice of homologous orientations}
\label{sec:lattice}

Consider a partial order $\leq$ on a set $S$.  Given two elements
$x,y$ of $S$, let $m(x,y)$ (resp. $M(x,y)$) be the set of elements $z$
of $S$ such that $z\leq x$ and $z\leq y$ (resp. $z\geq x$ and
$z\geq y$).  If $m(x,y)$ (resp. $M(x,y)$) is not empty and admits a
unique maximal (resp. minimal) element, we say that $x$ and $y$ admit
a \emph{meet} (resp. a \emph{join}), noted $x\vee y$ (resp.
$x\wedge y$).  Then $(S,\leq)$ is a \emph{lattice} if any pair of
elements of $S$ admits a meet and a join. Thus in particular a lattice
has a unique maximal (resp. minimal) element.  A lattice is
\emph{distributive} if the two operators $\vee$ and $\wedge$ are
distributive on each other.

Consider an essentially 4-connected toroidal triangulation $G$ and its
angle map $A(G)$.  By Theorem~\ref{th:existence}, $G$ admits a
balanced transversal structure, thus $A(G)$ admits a balanced
$4$-orientation.  Let $D_0$ be a particular balanced $4$-orientation
of $A(G)$.  Let $\mathcal B(A(G),D_0)$ be the set of all the
orientations of $A(G)$ homologous to $D_0$.  A general result
of~\cite[Theorem~4.7]{GKL15} concerning the lattice structure of
homologous orientations implies that $\mathcal B(A(G),D_0)$ carries a
structure of a distributive lattice.

Note that by Lemma~\ref{lem:balancediffhomolog}, the set
$\mathcal B(A(G),D_0)$ is exactly the set of all balanced
$4$-orientations of $A(G)$. Thus we can simplify the notations and
denote $\mathcal B(A(G))$ the set $\mathcal B(A(G),D_0)$ as it does
not depend on the choice of $D_0$.

We give below some terminology and results from~\cite{GKL15} adapted
to our settings in order to describe the lattice properly. We need to
define an order on $\mathcal B(A(G))$ for that purpose. Fix an
arbitrary face $f_0$ of $A(G)$ and let $F_0$ be its counterclockwise
facial walk. Note that fixing a face $f_0$ of $A(G)$ corresponds to
fixing an edge $e_0$ of $G$.  Let $\mc{F}$ be the set of
counterclockwise facial walks of $A(G)$ and
$\mc{F}'=\mc{F}\setminus \{F_0\}$.  Note that
$\phi(F_0)=-\sum_{F\in\mc{F}'}\phi(F)$. Since the characteristic flows
of $\mc{F}'$ are linearly independent, any oriented subgraph of $A(G)$
has at most one representation as a combination of characteristic
flows of $\mc{F}'$. Moreover the $0$-homologous oriented subgraphs of
$A(G)$ are precisely the oriented subgraph that have such a
representation.  We say that a $0$-homologous oriented subgraph $T$ of
$A(G)$ is \emph{counterclockwise} (resp. \emph{clockwise})
w.r.t.~$f_0$ if its characteristic flow can be written as a
combination with positive (resp. negative) coefficients of
characteristic flows of $\mc{F}'$,
i.e. $\phi(T)=\sum_{F\in\mc{F}'}\lambda_F\phi(F)$, with
$\lambda\in\mathbb{N}^{|\mc{F}'|}$
(resp. $-\lambda\in\mathbb{N}^{|\mc{F}'|}$).  Given two orientations
$D,D'$, of $A(G)$ we set $D\leq_{f_0} D'$ if and only if
$D\setminus D'$ is counterclockwise.  Then we have the following
theorem:

\begin{theorem}[\cite{GKL15}]
\label{th:lattice}
$(\mathcal B(A(G)),\leq_{f_0})$ is a distributive lattice.
\end{theorem}

To define the elementary flips that generates the lattice.  We start
by reducing the graph $A(G)$. We call an edge of $A(G)$ \emph{rigid
  with respect to $\mathcal B(A(G))$} if it has the same orientation
in all elements of $\mathcal B(A(G))$. Rigid edges do not play a role
for the structure of $\mathcal B(A(G))$. We delete them from $A(G)$
and call the obtained embedded graph the \emph{reduced angle graph},
noted $\widetilde{A(G)}$.  Note that, this graph is embedded but it is
not necessarily a map, as some faces may not be homeomorphic to open
disks. Note that if all the edges are rigid, then
$|\mathcal B(A(G))|=1$ and $\widetilde{A(G)}$ has no edges. We have
the following lemma concerning rigid edges:

\begin{lemma}[{\cite[Lemma~4.8]{GKL15}}]
\label{lem:non-rigid}
Given an edge $e$ of $A(G)$, the following are equivalent:
\begin{enumerate}
\item $e$ is non-rigid
\item  $e$ is contained in a
$0$-homologous oriented subgraph of $D_0$
\item  $e$ is contained in a
$0$-homologous oriented subgraph of any element of $\mathcal B(A(G))$
\end{enumerate}
\end{lemma}

By Lemma~\ref{lem:non-rigid}, one can build $\widetilde{A(G)}$ by
keeping only the edges that are contained in a $0$-homologous oriented
subgraph of $D_0$.  Note that this implies that all the edges of
$\widetilde{A(G)}$ are incident to two distinct faces of
$\widetilde{A(G)}$.  Denote by $\widetilde{\mathcal{F}}$ the set of
oriented subgraphs of $\widetilde{A(G)}$ corresponding to the
boundaries of faces of $\widetilde{A(G)}$ considered counterclockwise.
Note that any $\widetilde{F}\in \widetilde{\mathcal{F}}$ is
$0$-homologous and so its characteristic flow has a unique way to be
written as a combination of characteristic flows of
$\mc{F}'$. Moreover this combination can be written
$\phi(\widetilde{F})=\sum_{F\in X_{\widetilde{F}}}\phi(F)$, for
$X_{\widetilde{F}}\subseteq\mc{F}'$. Let $\widetilde{f}_0$ be the face
of $\widetilde{A(G)}$ containing $f_0$ and $\widetilde{F}_0$ be the
element of $\widetilde{\mathcal{F}}$ corresponding to the boundary of
$\widetilde{f}_0$.  Let
$\widetilde{\mathcal{F}}'=\widetilde{\mathcal{F}}\setminus
\{\widetilde{F}_0\}$.
The elements of $\widetilde{\mathcal{F}}'$ are precisely the
elementary flips which suffice to generate the entire distributive
lattice $(\mathcal B(A(G)),\leq_{f_0})$, i.e.
the  Hasse diagram $\mathcal{H}$ of the lattice has
vertex set
$\mathcal B(A(G))$ and there is an oriented edge from $D_1$ to $D_2$ in
$\mathcal{H}$ (with $D_1\leq_{f_0}D_2$) if and only if
$D_1\setminus D_2\in \widetilde{\mathcal{F}}'$.  

% We define the directed graph $\mathcal{H}$ with vertex set
% $\mathcal B(A(G))$.  There is an oriented edge from $D_1$ to $D_2$ in
% $\mathcal{H}$ (with $D_1\leq_{f_0}D_2$) if and only if
% $D_1\setminus D_2\in \widetilde{\mathcal{F}}'$.  
% Then $\mathcal{H}$ is the Hasse diagram of a distributive lattice (see~\cite{GKL15}).

% %We define the label
% %of that edge as $c(D_1,D_2)=D_1\setminus D_2$. 
% We have the following:
% %results from~\cite{GKL15}:

% \begin{theorem}[\cite{GKL15}]
%   \label{lem:fulfillshasse}
%   $\mathcal{H}$ is the Hasse diagram of a distributive lattice.
% \end{theorem}

%We continue to investigate further the set
%$O(G,D_0)$.

Moreover, we have:

\begin{lemma}[{\cite[Proposition~4.13]{GKL15}}]
\label{lem:necessary}
For every element $ \widetilde{F}\in \widetilde{\mathcal{F}}$, there
exists $D$ in $\mathcal B(A(G))$ such that $\widetilde{F}$ is an oriented
subgraph of $D$.
\end{lemma}

By Lemma~\ref{lem:necessary}, for every element
$ \widetilde{F}\in \widetilde{\mathcal{F}}'$ there exists $D$ in
$\mathcal B(A(G))$ such that $\widetilde{F}$ is an oriented subgraph of
$D$. Thus there exists $D'$ such that $\widetilde{F}=D\setminus D'$
and $D,D'$ are linked in $\mathcal{H}$.  Thus
$\widetilde{\mathcal{F}}'$ is a minimal set that generates the
lattice.

Let $D_{\max}$ (resp. $D_{\min}$) be the maximal (resp. minimal)
element of the lattice $(\mathcal B(A(G)),\leq_{f_0})$. Then we have
the following lemmas:

\begin{lemma}[{\cite[Proposition~4.14]{GKL15}}]
\label{lem:maxtilde}
  $\widetilde{F}_0$ (resp. $-\widetilde{F}_0$) is an oriented subgraph
  of $D_{\max}$ (resp. $D_{\min}$). 
\end{lemma}

\begin{lemma}[{\cite[Proposition~4.15]{GKL15}}]
\label{prop:maximal}
$D_{\max}$ (resp. $D_{\min}$) contains no counterclockwise
(resp. clockwise) non-empty $0$-homologous oriented subgraph w.r.t.~$f_0$.
\end{lemma}

Note that in the definition of counterclockwise (resp. clockwise)
non-empty $0$-homologous oriented subgraph, used in
Lemma~\ref{prop:maximal}, the sum is taken over elements of $\mc{F}'$
and thus does not use $F_0$. In particular, $D_{\max}$
(resp. $D_{\min}$) may contain regions whose boundary is oriented
counterclockwise (resp. clockwise) according to the interior of the
region but then such a region contains $f_0$ in its interior.

Note that, assuming that an element of $\mathcal B(A(G))$ is given,
there is a generic method to compute in linear time the minimal
balanced element $D_{\min}$ of $(\mathcal B(A(G)),\leq_{f_0})$
(see~\cite[last paragraph of Section 8]{DGL15}.  This minimal element
plays the role of a canonical orientation and is particularly
interesting for bijection purpose as shown in
Section~\ref{sec:bijmain}.

\subsection{Faces of the reduced angle graph}

Previous section is about the general situation of the lattice
structure of homologous orientations and is more or less a copy/paste
from~\cite{GKL15} of the terminology and results that we need
here. Now we study in more detail this lattice w.r.t.~the balanced
property as done in~\cite[Section~10]{DGL15} for Schnyder woods and
$3$-orientations (see also~\cite[Section~8.4]{LevHDR}).

Consider an essentially 4-connected toroidal triangulation $G$ and its
angle map $A(G)$.  We consider the terminology of the previous
section and assume that there exists a balanced $4$-orientation $D_0$ of
$A(G)$.

We say that a walk $W$ of $A(G)$ is a \emph{$4$-disk} if it is a face
of $A(G)$ (see the left of Figure~\ref{fig:48disk}).  We say that a
walk $W$ of $A(G)$ is a \emph{$8$-disk} if it has size $8$, encloses a
region $R$ homeomorphic to an open disk and the dual-vertices of $W$
have their edge not on $W$ that is inside $R$ (see the right of
Figure~\ref{fig:48disk}). Finally, we say that a walk $W$ of $A(G)$ is
a \emph{$\{4,8\}$-disk} if it is either a $4$-disk or a $8$-disk.

\begin{figure}[!ht]
\center
\begin{tabular}{cc}
\includegraphics[scale=0.4]{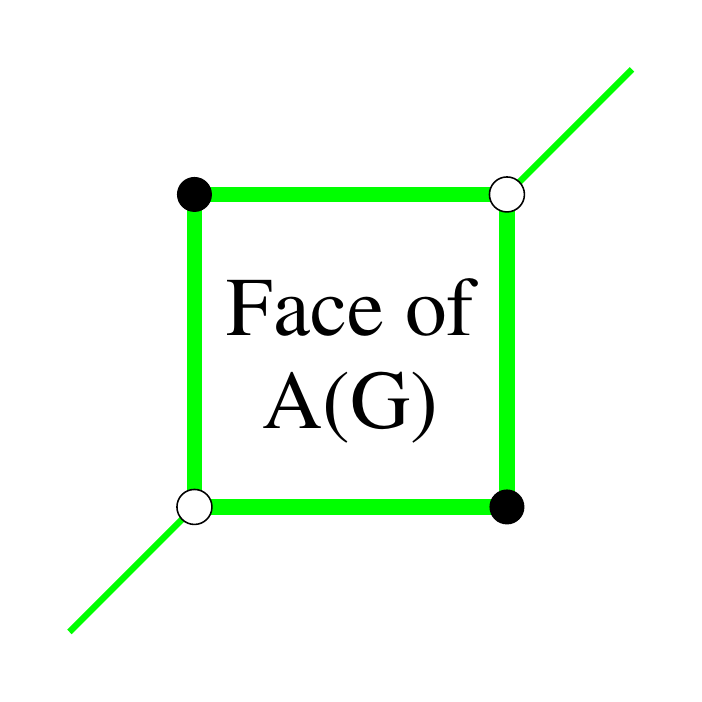} \ \ \ \ &\ \ \ \ 
\includegraphics[scale=0.4]{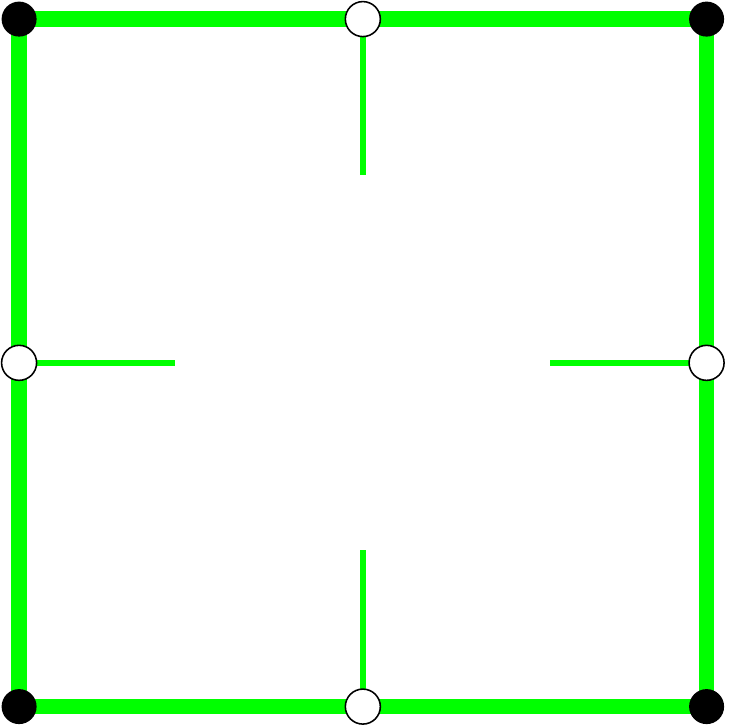} \\
$4$-disk \ \ \ \ &\ \ \ \ $8$-disk \\
\end{tabular}
\caption{The $\{4,8\}$-disks of $A(G)$.}
\label{fig:48disk}
\end{figure}

Suppose that $A(G)$ is given with a 4-orientation.  For an edge $e_0$
of $A(G)$ we define the \emph{left walk} (resp. \emph{right walk})
from $e_0$ as the sequence of edges $W=(e_i)_{i\geq 0}$ of $A(G)$
obtained by the following: if $e_i$ is entering a primal-vertex $v$,
then $e_{i+1}$ is the first outgoing edge while going \cw (resp. \ccw)
around $v$ from $e_i$, and if $e_i$ is entering a dual-vertex $v^*$,
then $e_{i+1}$ is the only outgoing edge of $v^*$. A closed left/right
walk is a left/right walk that is repeating periodically on itself,
i.e. a finite sequence of edges $W=(e_i)_{0\leq i\leq 2k-1}$, with
$k> 0$, such that its repetition is a left/right walk.  We have the
following lemma concerning closed left/right walks in balanced
$4$-orientations:

\begin{lemma}
\label{lem:facialwalktilde}
In a balanced $4$-orientation of $A(G)$, a closed left (resp. right)
walk $W$ of $A(G)$ encloses a region homeomorphic to an open disk on
its left (resp. right) side. Moreover, the border of this region is a
$\{4,8\}$-disk.
\end{lemma}

\begin{proof} 
  Consider a closed left walk $W=(e_i)_{0\leq i\leq 2k-1}$ of $A(G)$, with
  $k> 0$. W.l.o.g., we may assume that all the $e_i$ are distinct,
  i.e. there is no strict subwalk of $W$ that is also a closed left
  walk.  Note that $W$ cannot cross itself otherwise it is not a left
  walk. However $W$ may have repeated vertices but in that case it
  intersects itself tangentially on the right side.

  Suppose by contradiction that there is an oriented subwalk $W'$ of
  $W$, that forms a cycle $C$ enclosing a region $R$ on its right side
  that is homeomorphic to an open disk.  Let $v$ be the starting and
  ending vertex of $W'$. Note that we do not consider that $W'$ is a
  strict subwalk of $W$, so we might have $W'=W$.  Consider the graph
  $H$ obtained from $A(G)$ by keeping all the vertices and edges that
  lie in the region $R$, including $W'$. Since $W$ can intersect
  itself only tangentially on the right side, we have that $H$ is a
  bipartite planar map whose outer face boundary is $W'$.  The inner
  faces of $H$ are quadrangles. Let $2k'$ be the length of $W'$.  Let
  $n',m',f'$ be the number of vertices, edges and faces of $H$.  By
  Euler's formula, $n'-m'+f'=2$. All the inner faces have size $4$ and
  the outer face has size $2k'$, so $2m'=4(f'-1)+2k'$.  Combining the
  two equalities gives $m'=2n'-k'-2$.  Let $n_p'$ (resp. $n'_d$) be
  the number of inner primal-vertices (resp. inner dual-vertices) of
  $H$. So $n'=2k'+n_p'+n_d'$ and thus $m'=2n_p'+2n_d'+3k'-2$.  Since
  $W'$ is a subwalk of a left walk, all primal-vertices of $H$, except
  $v$ (if it is a primal-vertex), have their incident outer edges in
  $H$. Since $W'$ is following oriented edges, if $v$ is a primal
  vertex, it has at least one outgoing edge in $H$.  Since we are
  considering a $4$-orientation of $A(G)$, we have
  $m'\geq 4(k'-1)+1+4n_p'+n_d'$. By counting the edges of $H$ incident
  to dual-vertices, we have $m'\geq 2k'+3n_d'$.  Combining the three
  (in)equalities of $m'$, gives
  $2(2n_p'+2n_d'+3k'-2)\geq (4(k'-1)+1+ 4n_p'+n_d')+(2k'+3n_d')$, a
  contradiction. So there is no oriented subwalk of $W$, that forms a
  cycle enclosing an open disk on its right side.

We now claim the following:

\begin{claim}
\label{cl:leftsidedisk}
  The left side of $W$ encloses a region  homeomorphic to an open disk
\end{claim}

\begin{proofclaim}
  We consider two cases depending on the fact that $W$ is a cycle
  (i.e. with no repetition of vertices) or not.

\begin{itemize}
\item \emph{$W$ is a cycle} 

  Suppose by contradiction that $W$ is a non-contractible cycle. For
  each dual-vertex of $W$, there is an edge of $G$ between its
  neighbors in $W$. This edge might be either on the left or right side
  of $W$.  Consider the cycle $C$ of $G$ made of all these
  edges. Since we are considering a balanced $4$-orientation of
  $A(G)$, we have $\gamma(C)=0$ and thus there is exactly $2k$
  outgoing edges of $A(G)$ that are incident to the left side of
  $C$. There is no incident outgoing edge of $A(G)$ on the left side
  of $W$. So the outgoing edges that are on the left side of $C$ are
  exactly the $2k$ edges of $W$. So $C$ is completely on the right
  side of $W$ and all the edges of $W$ are outgoing for
  primal-vertices. Thus all dual-vertices have an outgoing edge on the
  left side of $W$, a contradiction.  Thus $W$ is a contractible
  cycle. 

  As explained above, the contractible cycle $W$ does not enclose a
  region homeomorphic to an open disk on its right side. So $W$
  encloses a region homeomorphic to an open disk on its left side, as
  claimed.

\item \emph{$W$ is not a cycle} 

  Since $W$ cannot cross itself nor intersect itself tangentially on
  the left side, it has to intersect tangentially on the right side.
  Such an intersection on a vertex $v$ is depicted on
  Figure~\ref{fig:righttangent}.(a). The edges of $W$ incident to $v$
  are noted as on the figure, $a,b,c,d$, where $W$ is going
  periodically through $a,b,c,d$ in this order.  The (green) subwalk of
  $W$ from $a$ to $b$ does not enclose a region homeomorphic to an
  open disk on its right side. So we are not in the case depicted on
  Figure~\ref{fig:righttangent}.(b). Moreover if this (green) subwalk
  encloses a region homeomorphic to an open disk on its left side,
  then this region contains the (red) subwalk of $W$ from $c$ to $d$,
  see Figure~\ref{fig:righttangent}.(c). Since $W$ cannot cross
  itself, this (red) subwalk necessarily encloses a region
  homeomorphic to an open disk on its right side, a contradiction. So
  the (green) subwalk of $W$ starting from $a$ has to form a
  non-contractible curve before reaching $b$. Similarly for the (red)
  subwalk starting from $c$ and reaching $d$. Since $W$ is a left-walk
  and cannot cross itself, we are, w.l.o.g., in the situation of
  Figure~\ref{fig:righttangent}.(d) (with possibly more tangent
  intersections on the right side). In any case, the left side of $W$
  encloses a region homeomorphic to an open disk.

 \begin{figure}[!ht]
 \center
 \begin{tabular}{cc}
 \includegraphics[scale=0.3]{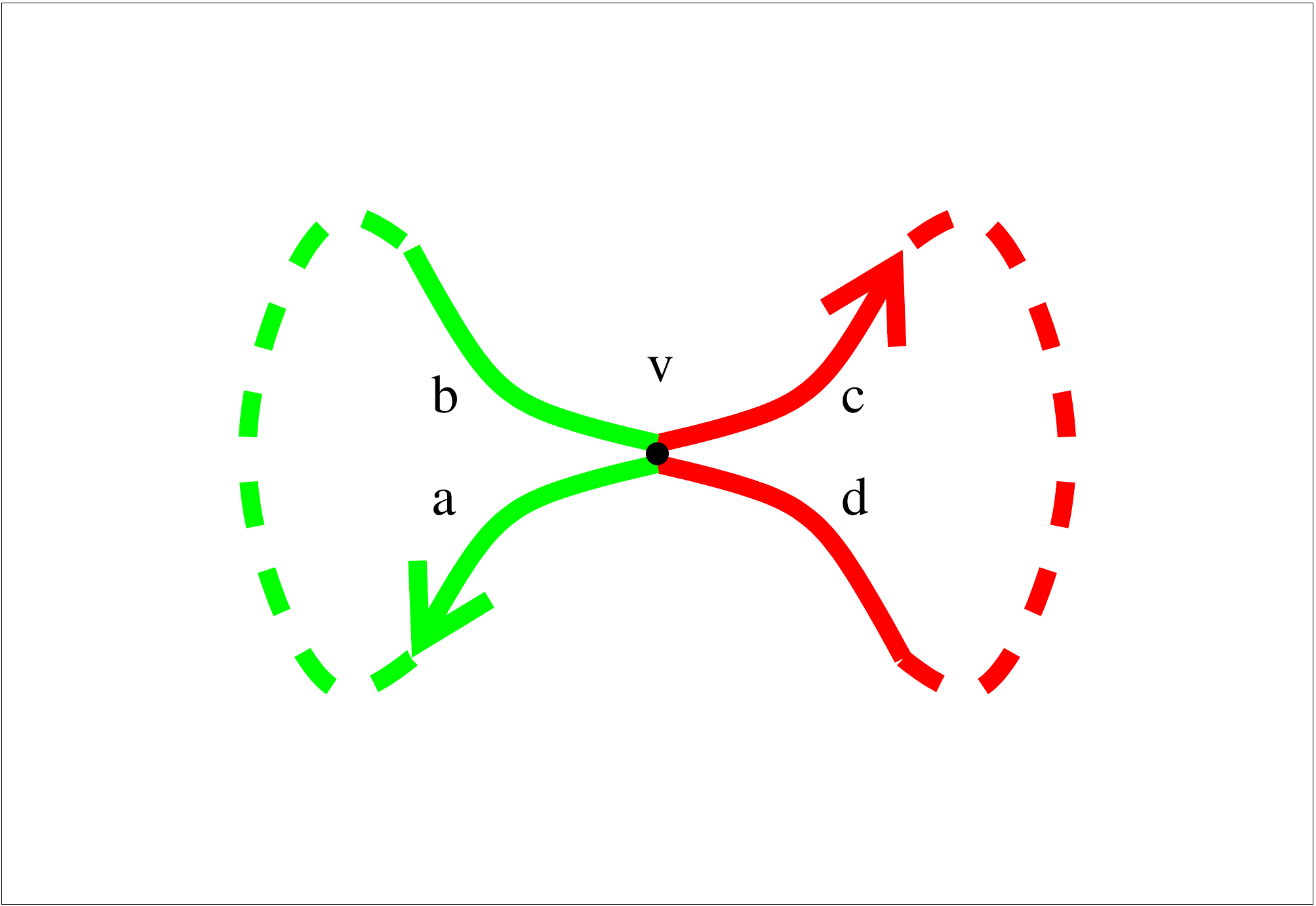} \ \ & \ \ 
 \includegraphics[scale=0.3]{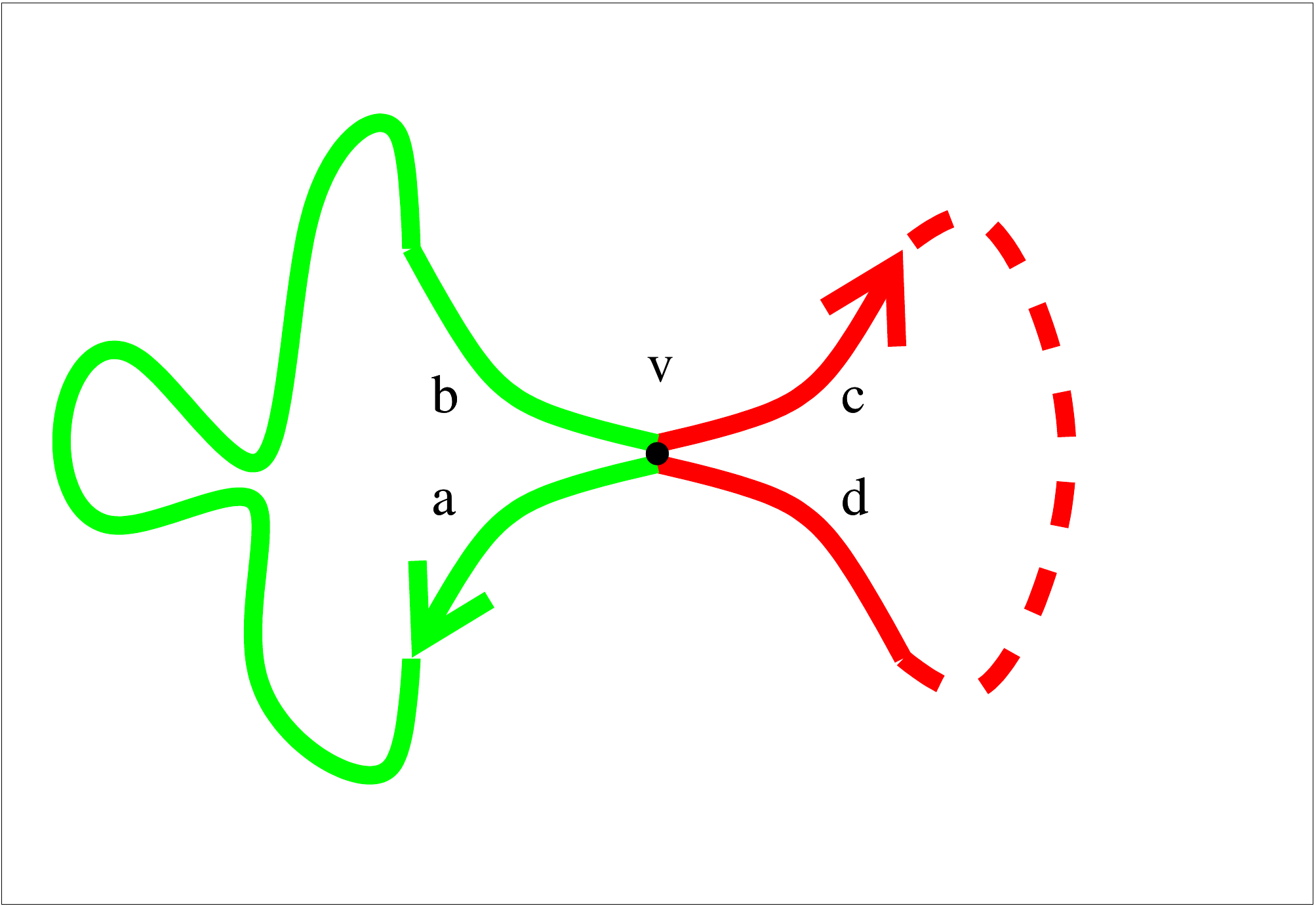}\\
(a) \ \ & \ \ (b)\\
& \\
 \includegraphics[scale=0.3]{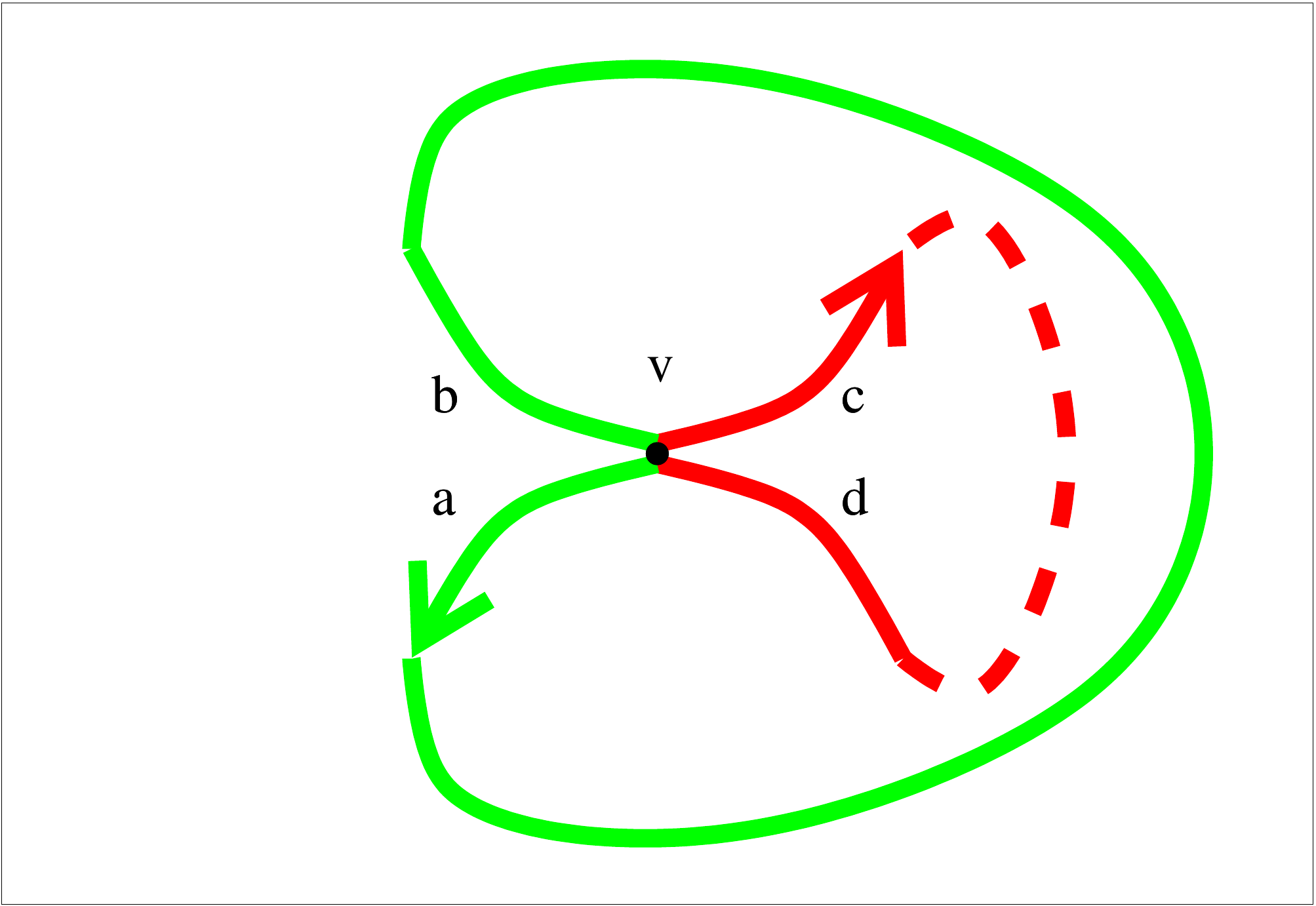}\ \ & \ \ 
 \includegraphics[scale=0.3]{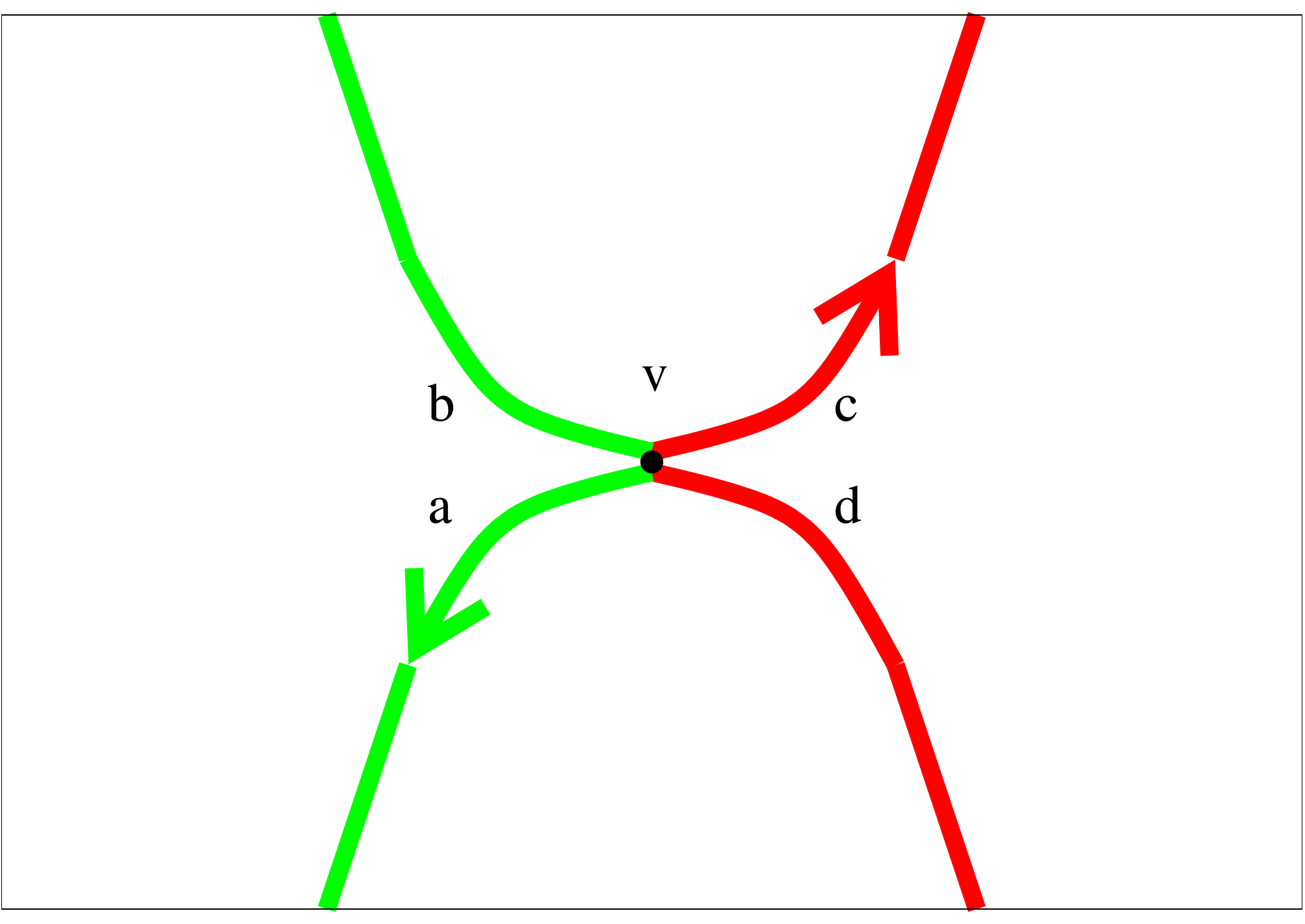}\\
(c) \ \ & \ \ (d)\\
 \end{tabular}
\caption{Case analysis for the proof of Claim~\ref{cl:leftsidedisk}.}
 \label{fig:righttangent}
 \end{figure}

\end{itemize}

\end{proofclaim}

By Claim~\ref{cl:leftsidedisk}, the left side of $W$ encloses a region
$R$ homeomorphic to an open disk.  Consider the graph $H$ obtained
from $A(G)$ by keeping only the vertices and edges that lie in $R$,
including $W$. The vertices of $W$ appearing several times on the
border of $R$ are duplicated, so $H$ is a bipartite planar map.  The
inner faces of $H$ are quadrangles.   As above, let $n',m',f'$ be the number of
vertices, edges and faces of $H$ and $n_p'$ (resp. $n'_d$) its number
of inner primal-vertices (resp. inner dual-vertices). So as above, one
obtain the first equality $m'=2n_p'+2n_d'+3k-2$.  There is no incident
outgoing edge of $A(G)$ on the left side of $W$. So all inner edges of
$H$ are outgoing for inner vertices of $H$.  Since we are considering
a $4$-orientation of $A(G)$, we have $m'=2k + 4n_p'+n_d'$.  Outer
dual-vertices of $H$ might be of degree $2$ or $3$ in $H$. Let $x$ be
the number of outer dual-vertices of $H$ of degree $3$ in $H$, so by
counting the edges of $H$ incident to dual-vertices we have
$m'=2k+x+3n_d'$.  Combining the three equalities of $m'$, gives
$2(2n_p'+2n_d'+3k-2)=(2k + 4n_p'+n_d')+(2k+x+3n_d')$, so $x=2k-4$.
Since $k\geq x \geq 0$, the only possible values are
$(k,x)\in\{(2,0),(3,2),(4,4)\}$.

If $(k,x)=(2,0)$, then $W$ has size four, its two dual-vertices are of
degree $2$ in $H$ so they have their edge not on $W$ that is outside
$R$. If $W$ is not a face of $A(G)$, then there are two distinct edges
between the primal-vertices of $W$ inside $R$, forming a pair of
homotopic multiple edges, a contradiction. So $W$ is a face of $A(G)$
and thus a $4$-disk.  If $(k,x)=(3,2)$, then $W$ has size six, with
two dual-vertices of degree $3$ in $H$ and $G^\infty$ has a separating
triangle, a contradiction to Lemma~\ref{lem:e4ciffnst}.  If
$(k,x)=(4,4)$, then $W$ has size eight, with its four dual-vertices of
degree $3$ in $H$. So $W$ is a $8$-disk.

The proof is similar for closed right walks.
\end{proof}

The boundary of a face of $\widetilde{A(G)}$ may be composed of several
closed walks. Let us call \emph{quasi-contractible} the faces of
$\widetilde{A(G)}$ that are homeomorphic to an open disk or to an open disk with
punctures.  Note that such a face may have several boundaries (if
there is some punctures), but exactly one of these boundaries encloses
the face. Let us call \emph{outer facial walk} this special
boundary. Then we have the following:

\begin{lemma}
\label{lem:contractibletilde}
All the faces of $\widetilde{A(G)}$ are quasi-contractible and their
outer facial walk is a $\{4,8\}$-disk.
\end{lemma}

\begin{proof}
Consider a face $\widetilde{f}$ of
  $\widetilde{A(G)}$. Let $\widetilde{F}$ be the
  element of $\widetilde{\mathcal{F}}$ corresponding to the boundary
  of $\widetilde{f}$. By Lemma~\ref{lem:necessary}, there exists
  $D\in \mathcal B(A(G))$ such that $\widetilde{F}$ is an oriented subgraph
  of $D$. 
%Recall thatBy Lemma~\ref{lem:balancediffhomolog}, we have that $D$ is a
%  balanced orientation of $A(G)$.

  All the faces of $A(G)$ form a $4$-disk. Thus either $\widetilde{f}$ is
 a face of $A(G)$ and we are done or $\widetilde{f}$ contains in its interior at least one edge
  of $A(G)$.  Start from  such edge $e_0$ and consider the
  left-walk $W=(e_i)_{i\geq 0}$ of $D$ from $e_0$.  Suppose that for
  $i\geq 0$, edge $e_i$ is entering a vertex $v$ that is on the border
  of $\widetilde{f}$.  Recall that by definition $\widetilde{F}$ is
  oriented \ccw according to its interior, so either $e_{i+1}$ is in
  the interior of $\widetilde{f}$ or $e_{i+1}$ is on the border of
  $\widetilde{f}$. Thus $W$ cannot leave $\widetilde{f}$ and its border.

  Since $A(G)$ has a finite number of edges, some edges are used
  several times in $W$.  Consider a minimal subsequence
  $W'=e_k, \ldots, e_\ell$ such that no edge appears twice and
  $e_k=e_{\ell+1}$.  Thus $W$ ends periodically on $W'$ that is a
  closed left walk. By Lemma~\ref{lem:facialwalktilde}, $W'$ encloses
  a region $R$ homeomorphic to an open disk on its left
  side. Moreover, the border of this region is a $\{4,8\}$-disk.  Thus
  $W'$ is a $0$-homologous oriented subgraph of $D$. So all its edges
  are non-rigid by Lemma~\ref{lem:non-rigid}. So all the edges of $W'$
  are part of the border of $\widetilde{f}$. Since $\widetilde{F}$ is
  oriented \ccw according to its interior, the region $R$ contains
  $\widetilde{f}$. So $\widetilde{f}$ is quasi-contractible and $W'$
  is its outer facial walk and a $\{4,8\}$-disk.
\end{proof}

% Now we can apply Lemma~\ref{lem:contractibletilde} to the extremal
% elements of the lattice.

% \begin{lemma}
% \label{lem:trianglef0}
% In $D_{\max}$ (resp. $D_{\min}$) there is a \ccw (resp. \cww)
% $\{4,8\}$-disk containing $f_0$, and a \cw (resp. \ccww)
% $\{4,8\}$-disk not containing $f_0$.
% \end{lemma}

% \begin{proof}
%   By Lemma~\ref{lem:contractibletilde}, $\widetilde{f}_0$ is
%   quasi-contractible and its outer facial walk $W$ is a
%   $\{4,8\}$-disk.  By Lemma~\ref{lem:maxtilde}, $\widetilde{F}_0$ is
%   an oriented subgraph of $D_{\max}$.  Thus $W$ is oriented \ccw
%    and contains $f_0$.  Since $\widetilde{A(G)}$
%   has some edges, we have $|S_0(G)|\geq 2$ so $D_{\max}$ has at least
%   one neighbor below in the Hasse diagram of the lattice, thus the
%   second part of the lemma is clear.  Similarly for $D_{\min}$.
% \end{proof}

A simple counting argument gives the
following lemma (see Figure~\ref{fig:8diskin}):

\begin{figure}[!ht]
\center
\includegraphics[scale=0.4]{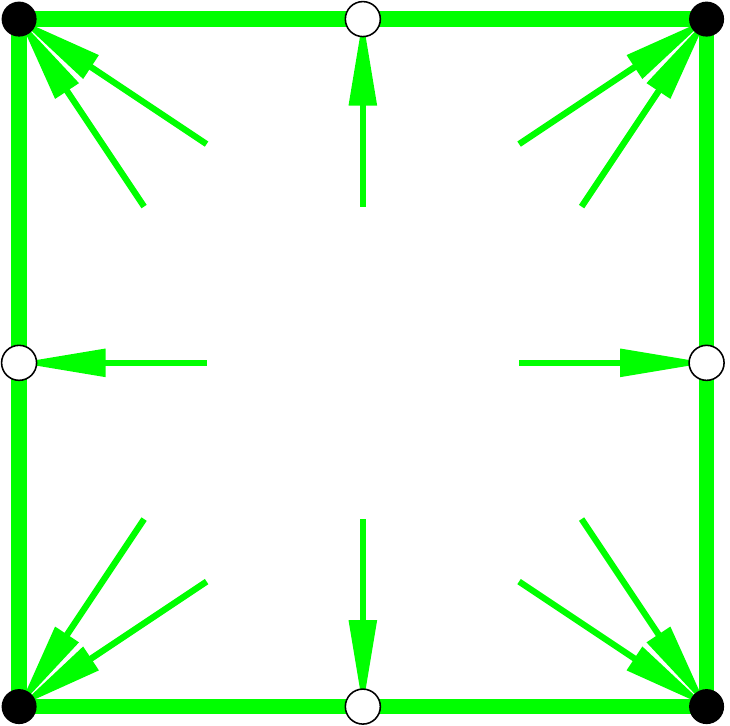}
\caption{Orientation of inner edges incident to a $8$-disk.}
\label{fig:8diskin}
\end{figure}

\begin{lemma}
\label{lem:8diskin}
In a $4$-orientation of $A(G)$, the edges that are
in the interior of a $8$-disk  and incident to it are
entering it. 
\end{lemma}

\begin{proof}
  Consider a $8$-disk $W$ of $A(G)$.  Consider the graph $H$ obtained
  from $A(G)$ by keeping only the vertices and edges that lie in $W$
  and its interior. The vertices of $W$ appearing several times on $W$
  are duplicated, so $H$ is a bipartite planar map.  
Let $x$ be the number of inner-edges of $H$ that are incident to
its outer-face and directed toward the interior. We want to prove that
$x=0$.

Let $n',m',f'$ be the number of vertices, edges and faces of $H$. By
Euler's formula, $n'-m'+f'=2$. All the inner faces have size $4$ and
the outer face has size $8$, so $2m'=4(f'-1)+8$.  Combining the two
equalities gives $m'=2n'-6$.  Let $n_p'$ (resp. $n'_d$) be the number
of inner primal-vertices (resp. inner dual-vertices) of $H$. So
$n'=n_p'+n_d'+8$ and thus $m'=2n_p'+2n_d'+10$.  Since we are
considering a $4$-orientation of $A(G)$, we have $m'=4n_p'+n_d'+x+8$.
By counting the edges of $H$ incident to dual-vertices we have
$m'=3n_d'+12$.  Combining the three equalities of $m'$, gives
$2(2n_p'+2n_d'+10)=(4n_p'+n_d'+x+8)+(3n_d'+12)$, so $x=0$.

\end{proof}

We say that a $\{4,8\}$-disk of $A(G)$ is \emph{maximal} (by
inclusion) if its interior is not strictly contained in the interior
of another $\{4,8\}$-disk of $A(G)$.  
%For bijective purpose, we need
%to strengthen Lemma~\ref{lem:trianglef0} with the following:

%\comment{Statement}
\begin{lemma}
\label{lem:maxdiskroot}
There is a unique maximal $\{4,8\}$-disk of $A(G)$ containing
$f_0$ and it is oriented \ccw (resp. \cww) in $D_{\max}$
(resp. $D_{\min}$).

% . Moreover, we have:

% \begin{itemize}
% \item 
% $W$ it is oriented \ccw (resp. \cww) in $D_{\max}$
% (resp. $D_{\min}$). 
% \item All the edges of $W$ are in
% $\widetilde{A(G)}$.
% \item All the faces of $A(G)$ that are in the
% interior and incident to $W$ lies in the same face of
% $\widetilde{A(G)}$.
% \end{itemize}
\end{lemma} 

\begin{proof}
  By Lemma~\ref{lem:contractibletilde}, $\widetilde{f}_0$ is
  quasi-contractible and its outer facial walk is a $\{4,8\}$-disk. So
  there is a $\{4,8\}$-disk containing $f_0$.  Let $W$ be a maximal
  $\{4,8\}$-disk containing $f_0$.  By Lemma~\ref{lem:8diskin}, if $W$
  is a $8$-disk, then, for any $4$-orientation of $A(G)$, the edges of
  $A(G)$ that are in the interior of $W$ and incident to it are
  entering it. If $W$ is a $4$-disk, then there is no edge of $A(G)$
  in the interior of $W$.  So all the edges in the interior of $W$ and
  incident to it are rigid edges, i.e. these edges are not in
  $\widetilde{A(G)}$.  So there is a face $\widetilde{f}$ of
  $\widetilde{A(G)}$ containing all the faces $F_W$ of $A(G)$ that are
  in the interior of $W$ and incident to it. Note that there might be
  some punctures in $\widetilde{f}$, so $\widetilde{f}$ does not
  necessarily contain all the faces of $A(G)$ that are in the interior
  of $W$.  By Lemma~\ref{lem:contractibletilde}, $\widetilde{f}$ is
  quasi-contractible and its outer facial walk is a $\{4,8\}$-disk. By
  maximality of $W$, the $\{4,8\}$-disk $W$ is the only $\{4,8\}$-disk
  of $A(G)$ containing the faces $F_W$. So the outer facial walk of
  $\widetilde{f}$ is $W$ and all the edges of $W$ are non-rigid,
  i.e. these edges are in $\widetilde{A(G)}$.

  Suppose by contradiction that there exists another maximal
    $\{4,8\}$-disk $W'$ containing $f_0$ that is distinct from
  $W$. As for $W$, all the edges of $W'$ are in $\widetilde{A(G)}$.
  The interiors of $W$ and $W'$ have to be distinct, not included one
  into each other, but intersecting. Then at least one edge of $W'$
  has to be in the interior of $W$ and incident to it, a contradiction
  to the fact that these edges are not in $\widetilde{A(G)}$. So $W$
  is the unique maximal $\{4,8\}$-disk containing $f_0$.

  We now prove the second part of the lemma for $D_{\max}$ (the proof
  is similar for $D_{\min}$).  Let $\widetilde{F}$ be the element of
  $\widetilde{\mathcal{F}}$ corresponding to the boundary of
  $\widetilde{f}$. We consider two cases depending on the fact that
  $\widetilde{f}$ is equal to $\widetilde{f}_0$ or not, i.e.
  $\widetilde{f}=\widetilde{f}_0$ or $\widetilde{f}$ has some
  punctures, one of which contains $\widetilde{f}_0$.

  \begin{itemize}
  \item \emph{$\widetilde{f}=\widetilde{f}_0$:} 
By Lemma~\ref{lem:maxtilde},
  we have $\widetilde{F}=\widetilde{F}_0$ is an oriented subgraph of
  $D_{\max}$ and thus $W$ is oriented \ccw w.r.t.~its interior.

\item \emph{$\widetilde{f}\neq\widetilde{f}_0$:} 
%Suppose by
%  contradiction that in $D_{\max}$, the $\{4,8\}$-disk $W$ is not
%  oriented \ccw w.r.t. to its interior.  
By Lemma~\ref{lem:necessary},
  there exists an element of $\mathcal B(A(G))$ for which
  $\widetilde{F}$ is an oriented subgraph.
Let $D$ be such an element, chosen such that
  $\widetilde{F}$ is not an oriented subgraph of any orientation,
  distinct from $D$, that is on  oriented paths from $D$ to
  $D_{\max}$ in the Hasse diagram of $(\mathcal B(A(G)),\leq_{f_0})$.
The $\{4,8\}$-disk $W$ is
  oriented \ccw in $D$.
% so $D$ is distinct from $D_{\max}$. 
% Moreover by
%  assumption there is no element $D'$ of $\mathcal B(A(G))$, distinct
%  from $D$, such that $\widetilde{F}$ is an oriented subgraph of $D'$
%  and $D\leq_{f_0} D' \leq_{f_0} D_{\max}$.
Recall that $\widetilde{f}$ has at least one
  puncture containing $\widetilde{f}_0$. 

Let $D'$ be the orientation obtained from $D$ by reversing all the
  edges of $\widetilde{f}$ that are not on its outer facial walk,
  i.e. obtained by reversing the border of all the punctures. 
So the $\{4,8\}$-disk $W$ is still
  oriented \ccw in $D'$.
  We claim that $D, D'$ are such that $D\leq_{f_0} D'$.  Indeed, let
  $T=D\setminus D'$ and $X$ denote the set of all the elements of
  $\mathcal{F}'$ that corresponds to faces of $A(G)$ that are not in
  the punctures of $\widetilde{f}$.
 Then we have
  $\phi(T)= \sum_{{F}\in X}\phi({F})$ and $X$ is a subset of
  ${\mathcal{F}}'$. So $D\leq_{f_0} D'$.
  Consider $D'=D_0,\ldots,D_k=D_{\max}$, with $k\geq 0$, the elements
  of $\mathcal B(A(G))$ on an oriented path from $D'$ to $D_{\max}$ in the Hasse
  diagram of $(\mathcal B(A(G)),\leq_{f_0})$. 

  Suppose by contradiction that the $\{4,8\}$-disk $W$ is not oriented
  \ccw in $D_{\max}$. Let $1\leq i \leq k$ be the minimal integer such
  that the $\{4,8\}$-disk $W$ is not oriented \ccw in $D_i$. Thus the
  $\{4,8\}$-disk $W$ is oriented \ccw in $D_{i-1}$ but not in
  $D_{i}$. Since $D_{i-1}$ and $D_{i}$ are linked in the Hasse
  diagram, we have
  $D_{i-1}\setminus D_{i} \in \widetilde{\mathcal{F}}'$. Let
  $\widetilde{F}'\in\widetilde{\mathcal{F}}'$ be such that
  $\widetilde{F}'=D_{i-1}\setminus D_{i}$. By assumption on $D$ we
  have $\widetilde{F}'$ is distinct from $\widetilde{F}$. Moreover,
  since $W$ is oriented \ccw in $D_{i-1}$, we have that
  $\widetilde{F}'$ is distinct from all the elements of
  $\widetilde{\mathcal{F}}'$ corresponding to faces that are incident
  to $W$ and not in its interior. So $\widetilde{F}'$ is disjoint from
  $W$ and $W$ has the same orientation in $D_{i-1}$ and $D_{i}$, a
  contradiction.  So $W$ is oriented \ccw in $D_{\max}$.
  \end{itemize}
\end{proof}
\subsection{Example of  a balanced lattice}

Consider the essentially $4$-connected toroidal triangulation $G$ of
Figure~\ref{fig:anglegraph} and its angle map $A(G)$.  One example of
a balanced $4$-orientation of $A(G)$ is given on the right of
Figure~\ref{fig:4orbalanced}, we call it $D_0$ in this section.  By
Lemma~\ref{lem:non-rigid}, an edge of $A(G)$ is non-rigid if and only
if if is contained in a $0$-homologous oriented subgraph of $D_0$. So
with this rule, one can build the reduced angle graph
$\widetilde{A(G)}$ depicted on Figure~\ref{fig:latticereduced}.  One
can check that Lemma~\ref{lem:contractibletilde} is satisfied since
 the faces are made of one 8-disk and some 4-disks.
We choose arbitrarily a special face $f_0$ of $\widetilde{A(G)}$ as
depicted on the figure. 

\begin{figure}[!ht]
  \center
\includegraphics[scale=0.4]{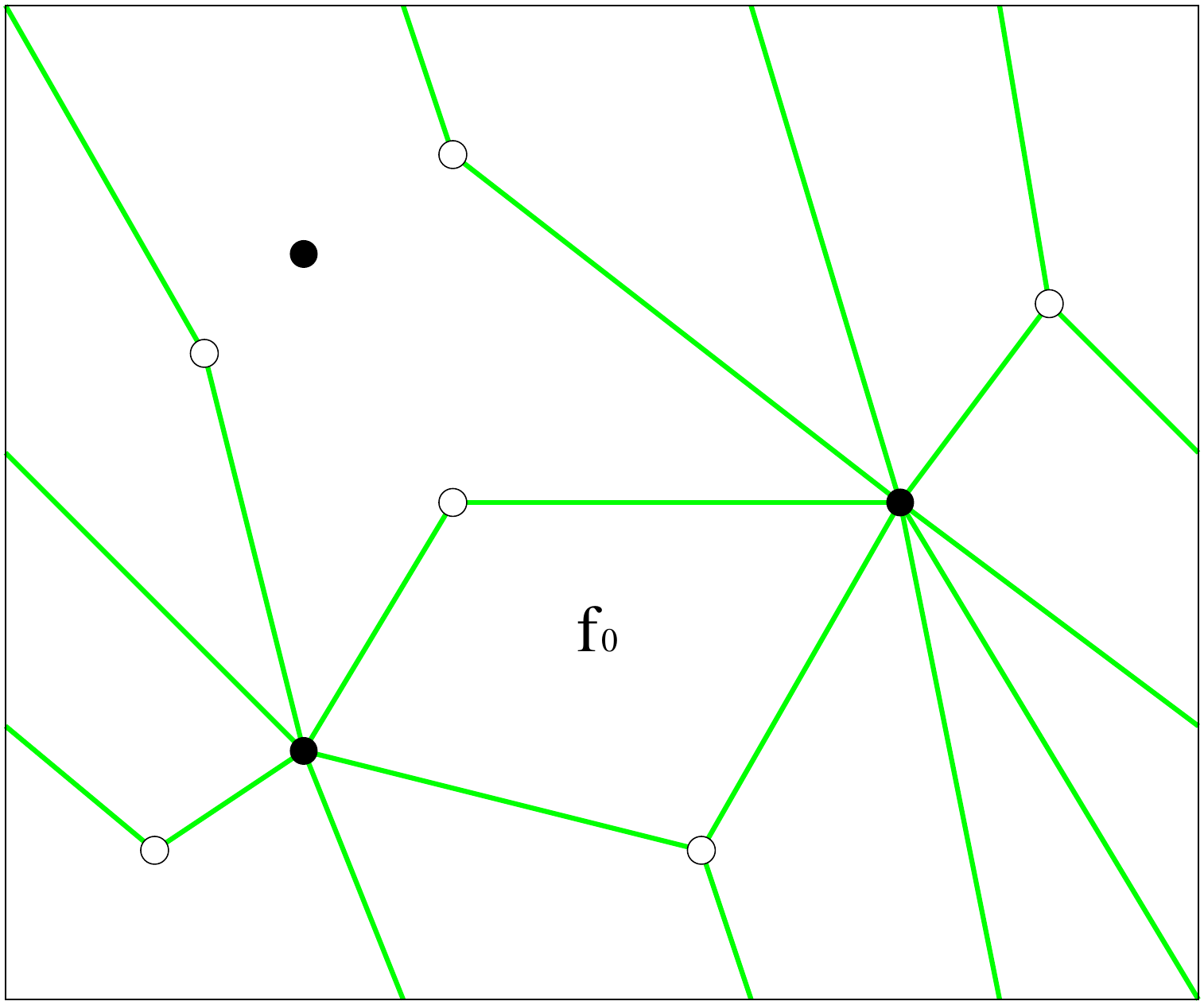}
  \caption{The reduced angle graph of the triangulation of Figure~\ref{fig:anglegraph}.}
  \label{fig:latticereduced}
\end{figure}

The set of all
orientations of $A(G)$ that are homologous to $D_0$ is exactly the set
$\mathcal B(A(G))$ of all balanced $4$-orientations of $A(G)$ by
Lemma~\ref{lem:balancediffhomolog}.  Moreover, we have that
$(\mathcal B(A(G)),\leq_{f_0})$ is a distributive lattice by
Theorem~\ref{th:lattice}. The Hasse diagram of this lattice is
represented on the left of Figure~\ref{fig:lattice}.  Each node of the diagram is
a balanced $4$-orientation of $A(G)$ and  black edges are the edges of
the diagram. 

The orientation on the
left of Figure~\ref{fig:4orbalanced} is not in the diagram since it is
not balanced.
The orientation $D_0$, on the
right of Figure~\ref{fig:4orbalanced}, is the
second one starting from the top. 
The other orientations of the diagram
are obtained from $D_0$ by flipping oriented faces of the reduced
angle graph $\widetilde{A(G)}$,
except $f_0$.  

When a face of the reduced angle graph is oriented this is represented by a
circle. The circle is black when it corresponds to the face containing
$f_0$. The circle is magenta if the boundary of the corresponding face
of $\widetilde{A(G)}$ is oriented \ccw and cyan otherwise. For the
face of $\widetilde{A(G)}$ that is a 8-disk, we represent the circle
around the unique vertex that is in the interior of this 8-disk.

An edge in the Hasse diagram from $D$ to $D'$ (with $D\leq D'$)
corresponds to a face of $\widetilde{A(G)}$ oriented \ccw in $D$ whose
edges are reversed to form a face oriented \cw in $D'$, i.e. a magenta
circle replaced by a cyan circle. The outdegree of a node is its
number of magenta circle and its indegree is its number of cyan
circle.  By Lemma~\ref{lem:necessary}, all the faces of
$\widetilde{A(G)}$ have a circle at least once. The special face is not
allowed to be flipped and, by Lemma~\ref{lem:maxtilde}, it is oriented
\ccw in the maximal element of the lattice and \cw in the minimal
element.  By Lemma~\ref{prop:maximal}, the maximal (resp. minimal)
element contains no other faces of $\widetilde{A(G)}$ oriented \ccw
(resp. \cww), indeed it contains only cyan (resp. magenta) circles and
one black.  One can play with the black circle and see which are the
orientations of the lattice that are in correspondence by flipping the
face $f_0$.

\begin{figure}[!ht]
  \centering
\includegraphics[scale=0.175]{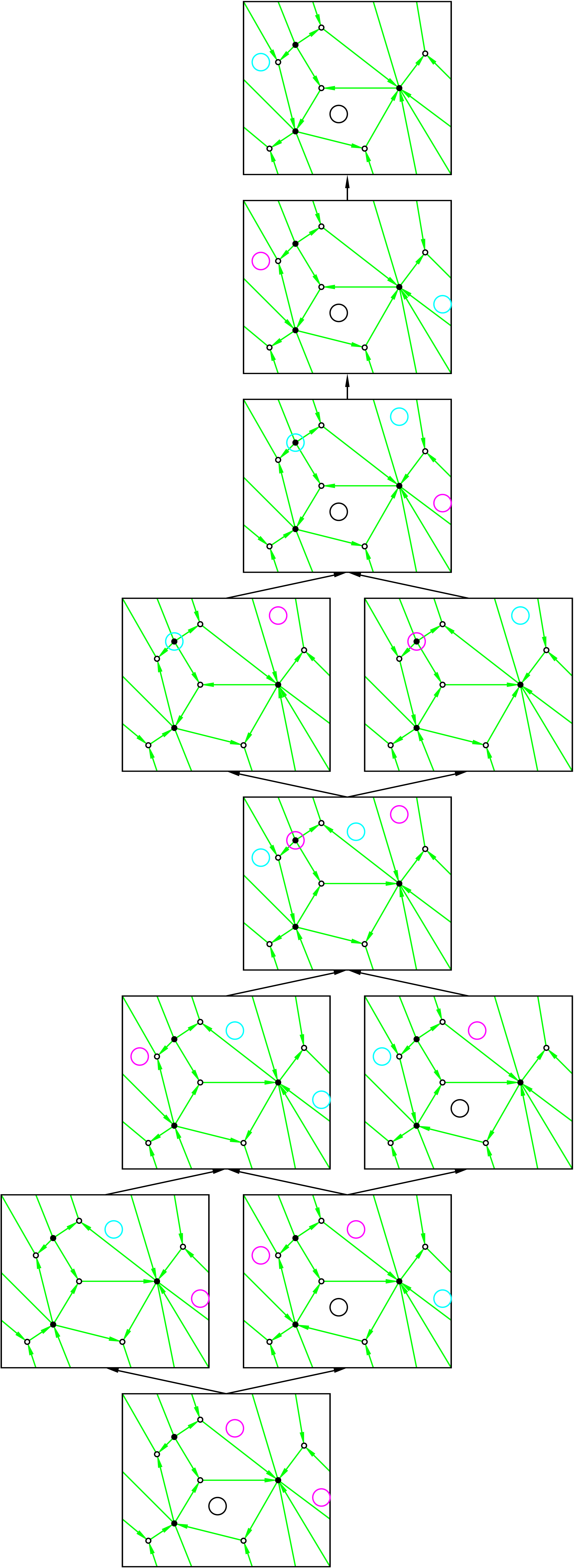}
\includegraphics[scale=0.175]{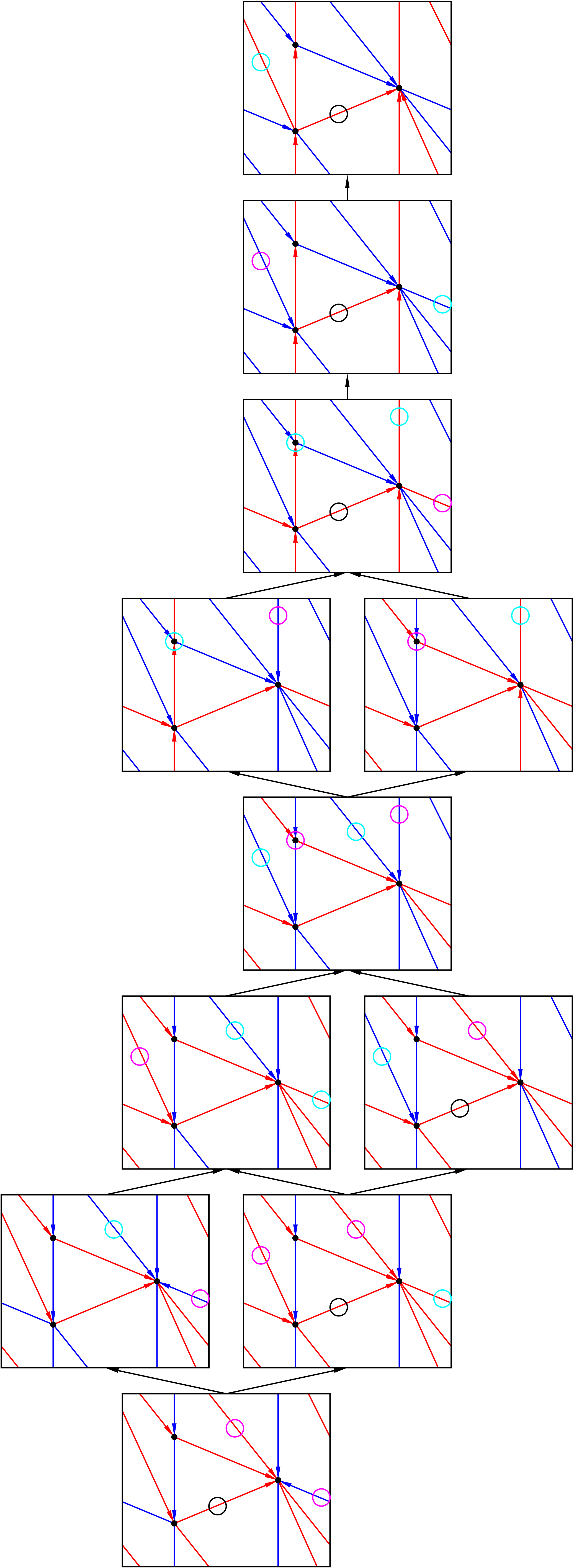}
\caption{The Hasse diagram of the distributive lattice of
  the balanced 4-orientations of the angle map of an essentially
  $4$-connected toroidal triangulation.}
  \label{fig:lattice}
\end{figure}

All the $4$-orientations of the diagram are balanced so they
correspond to transversal structures by
Corollary~\ref{cor:bal4orTS}. These transversal structures are
represented on the right of Figure~\ref{fig:lattice}. The lattice may
have been defined directly on the transversal structures using the
same transformations as in the planar case (see~\cite[Figure~6 and
Theorem~2]{Fus09}). But we prefer to present this by considering
$\alpha$-orientations here since it is a more general framework that
also enables to use directly results from~\cite{GKL15} without
re-proving the lattice structure.

\section{Bijection with unicellular mobiles}
\label{sec:bijmain}
\subsection{From essentially 4-connected toroidal triangulations to  mobiles}
\label{sec:e4c2mob}

Consider an essentially 4-connected toroidal triangulation $G$ and its
angle map $A(G)$. In order to use the lattice structure on the
(non-empty) set
$\mathcal B(A(G))$ we need to choose a particular face of
$A(G)$, i.e. a particular edge of $G$.  This choice has to be done
appropriately so that the minimal element of the lattice have some
interesting properties for the bijection. For that purpose, we
have to consider quadrangles of $G$.
 We choose a half-edge $h_0$ of $G$ that is in the
interior and incident to a maximal quadrangle of $G$. We call $h_0$
the \emph{root half-edge} of $G$.  The vertex $v_0$ of $G$ incident to $h_0$
is called the \emph{root vertex}.  The face  $f_0$ of
$A(G)$ containing $h_0$ is called the \emph{root face}. Consider the order
$\leq_{f_0}$ define on $\mathcal B(A(G))$ in
Section~\ref{sec:lattice}. By Theorem~\ref{th:lattice},
$(\mathcal B(A(G)),\leq_{f_0})$ forms a distributive lattice. Thus we
can consider the minimal balanced $4$-orientation $D_{\min}$ of
$\mathcal B(A(G))$ w.r.t.~$f_0$. Since there is no ambiguity, we may
also say that $D_{\min}$ is minimal w.r.t.~$h_0$.

By
Corollary~\ref{cor:bal4orTS}, $D_{\min}$ corresponds to a transversal
structure of $G$ and admits a TTS-labeling (see
Section~\ref{sec:tslab}). By convention, the transversal structure of
$G$ associated to $D_{\min}$ that we consider is the one where $h_0$
is an outgoing half-edge of color blue, i.e. in the TTS-labeling the
half-edge is labeled $0$.

Let us associate to any $4$-orientation $D$ of $A(G)$ a particular
graph $M$ embedded on the torus, called \emph{mobile associated to}
$D$. The vertex set of $M$ is the same as $G$. Its edge set is
composed of some edges of $G$ plus some half-edges that are incident
to only one vertex. We sometimes call the edges of $M$ full-edges to
avoid confusion with half-edges of $M$. Moreover we see the full-edges
of $M$ as two half-edges of $M$ that meet at the middle of the
edge. Then the set of half-edges of $M$ is defined by the following
rule: a half-edge $h$ of $G$, incident to a vertex $v$ of $G$, is an
half-edge of $M$ if and only if the edge of $D$ just after $h$ in \cw
order around $v$ is outgoing (see Figure~\ref{fig:mobile-rule}). If
$G$ is rooted on a particular half-edge $h_0$, then the \emph{extended
  mobile} $M^+$ is obtained from $M$ by adding the root half-edge
$h_0$ if not already in $M$. If the two half-edges of the same edge of
$G$ are in $M$ (resp. $M^+$), then they meet in order to form a
full-edge of $M$ (resp. $M^+$). The half-edges of $M$ (resp. $M^+$)
that are not part of a full-edge of $M$ (resp. $M^+$) are called
\emph{stems} and they are presented by an arrow on the figures.

\begin{figure}[!ht]
\center
\includegraphics[scale=0.5]{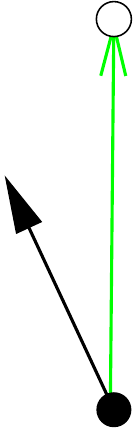}
\caption{Rule for  half-edges of the mobile.}
\label{fig:mobile-rule}
\end{figure}

A balanced transversal structure of $K_7$ is given on
Figure~\ref{fig:k7-mobile} with the corresponding balanced $4$-orientation
of its angle map that is minimal w.r.t.~the barred half-edge. The
extended mobile associated to this orientation is represented twice,
once with the angle map and once alone as a hexagon whose opposite
sides are identified to form a toroidal  map. The vertices are labeled
from $1$ to $7$, and the root half-edge is represented in magenta on
the figures.

\begin{figure}[!ht]
\center
\includegraphics[scale=0.34]{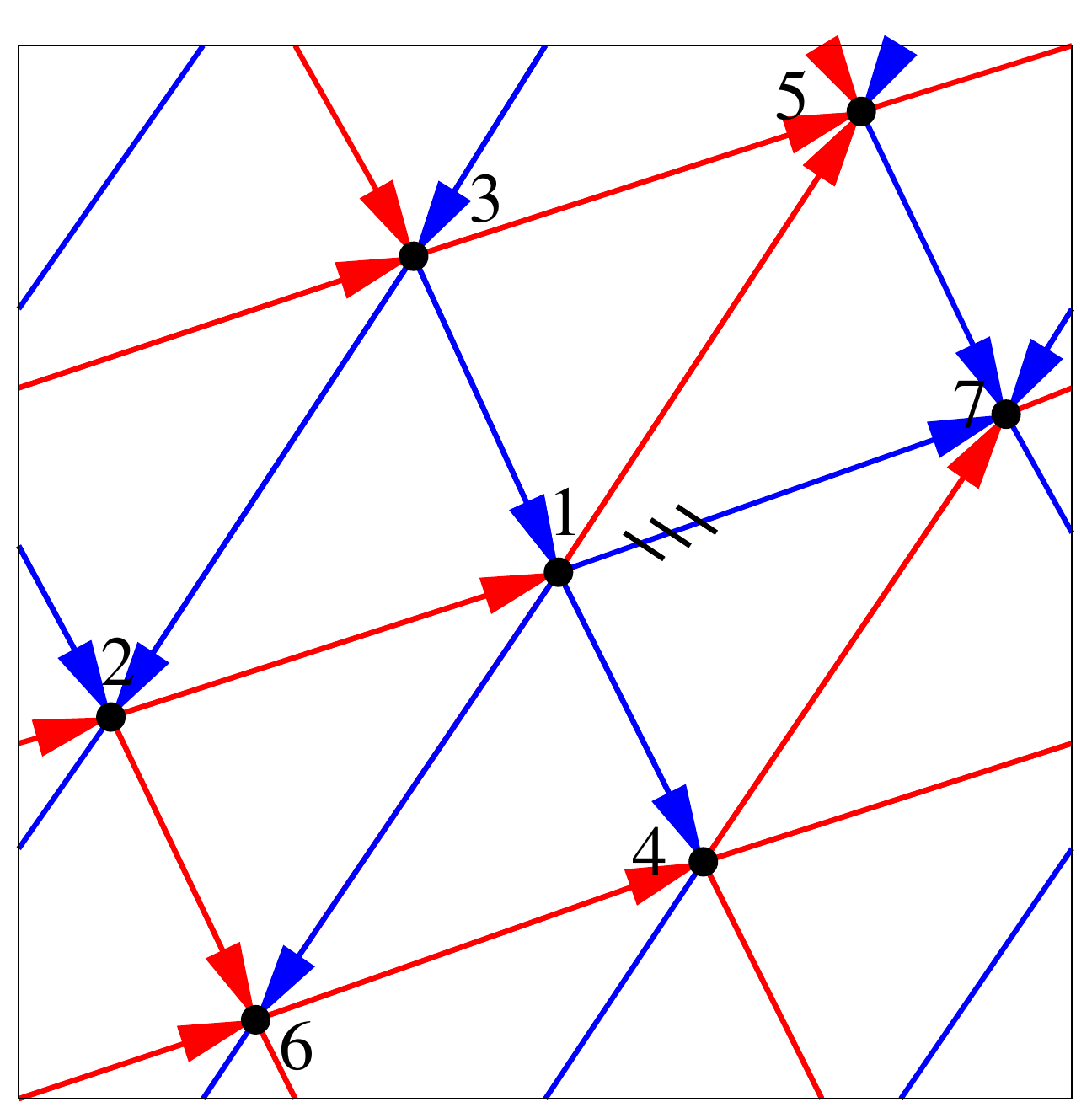} \ \
\includegraphics[scale=0.34]{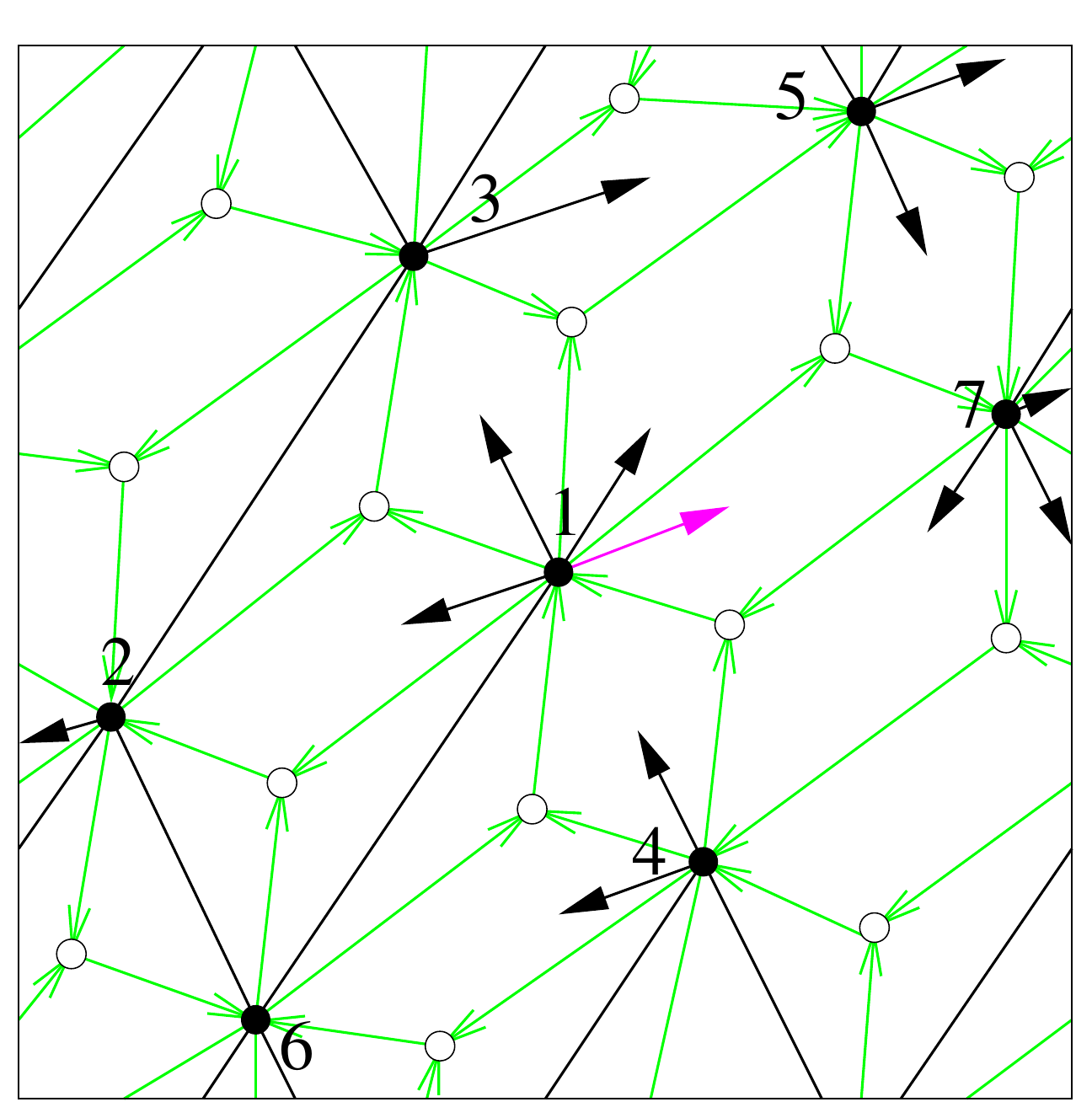} \ \
\includegraphics[scale=0.34]{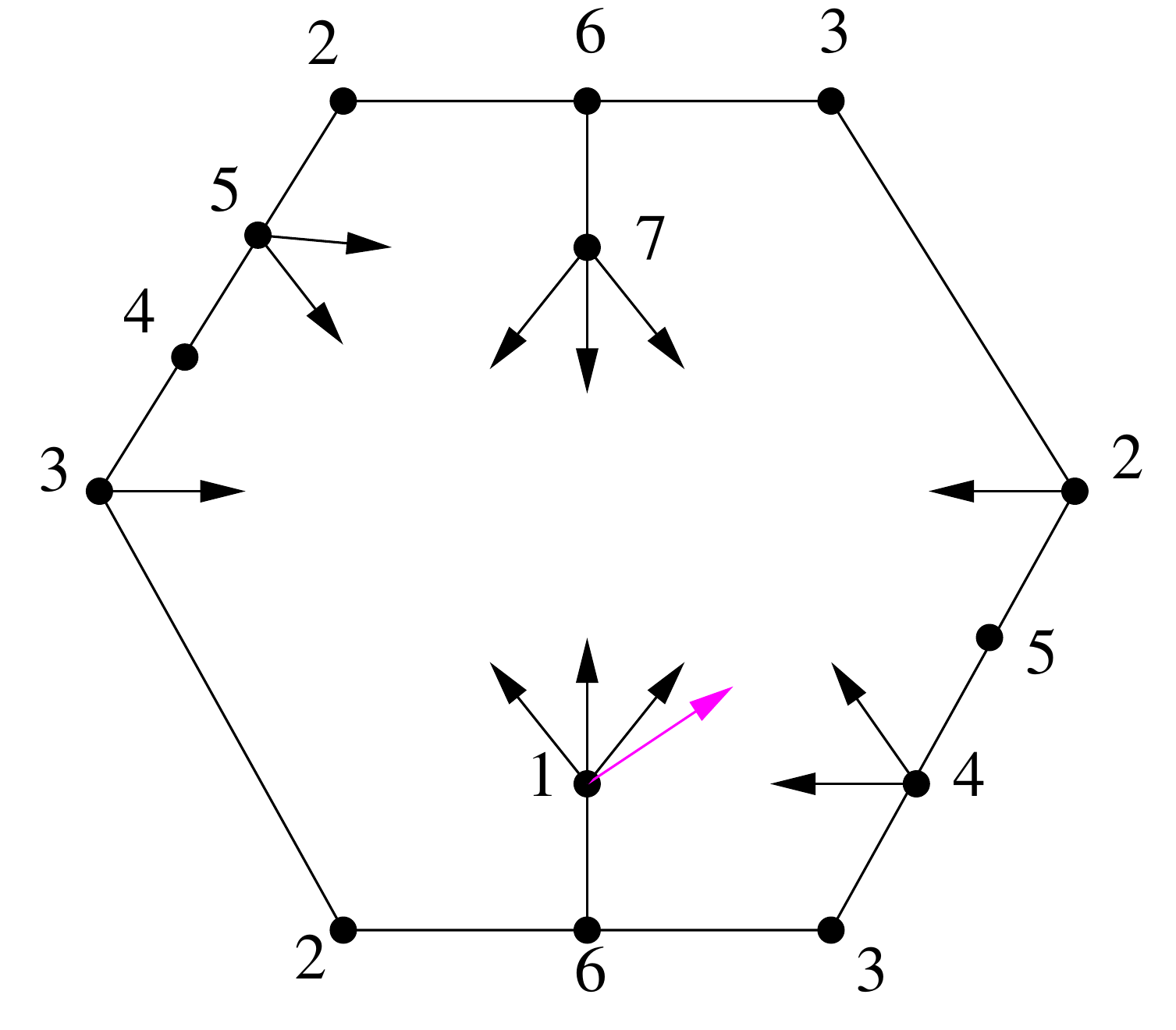}
\caption{Balanced transversal structure of $K_7$, given with the
  corresponding balanced $4$-orientation of its angle map, that is
  minimal w.r.t.~the barred half-edge, and the corresponding
  extended mobile, represented as an hexagon whose opposite sides are
  identified.}
\label{fig:k7-mobile}
\end{figure}

Note that the mobile can be computed directly from the transversal
structure by considering the following set of half-edges for $M$: an
half-edge $h$ of $G$, incident to a vertex $v$ of $G$, is a half-edge
of $M$ if and only if it is the last edge of an interval (outgoing
blue, outgoing red, incoming blue, incoming red) of $v$ in \cw order
around $v$. This point of view corresponds more to the
planar study of transversal structure from~\cite{Fus09}. But the rule
of Figure~\ref{fig:mobile-rule} corresponds to a more general
framework to construct so-called ``mobile'' that can be applied to any
orientation (see~\cite{BF12,BF12b,BC11}) and not only to transversal
structure. 

Part of the TTS-labeling of $D_{min}$ can be represented on the mobile
$M$ by keeping only the labels that are on half-edges of $M$. By the
mobile rule, one can note that a mobile-labeling satisfies the
following properties: the four labels that appear around each vertex
are exactly $0,1,2,3$ in \ccw order and the two labels that appear on
each edge differ exactly by $(2\bmod 4)$ (see
Figure~\ref{fig:labelmobile} where the TTS-labeling is represented on
the transversal structure and on the corresponding mobile).

\begin{figure}[!ht]
\center
\includegraphics[scale=0.34]{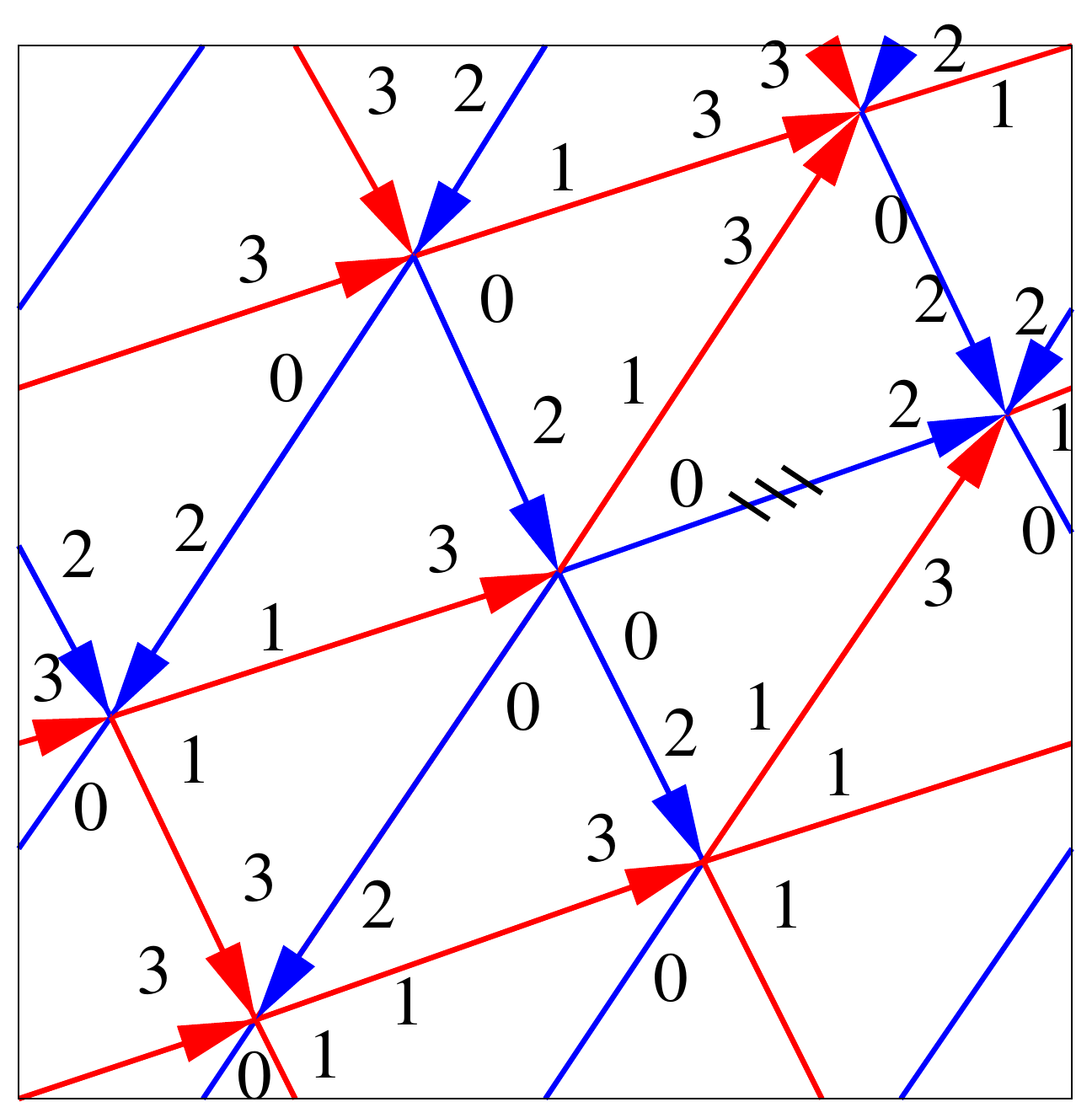} \ \ \ \ 
\includegraphics[scale=0.34]{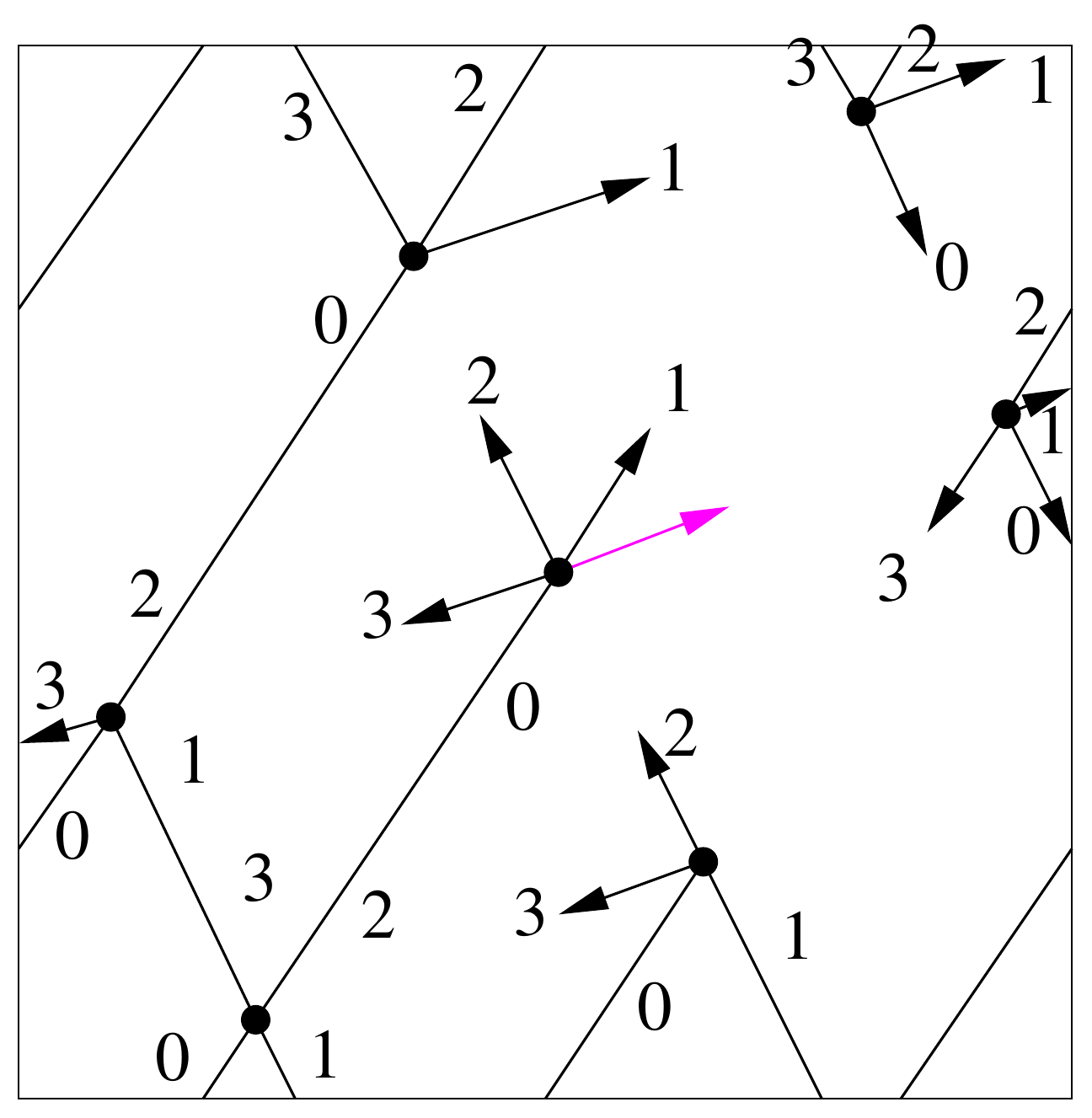}% \ \
\caption{TTS-labeling represented on the mobile.}
\label{fig:labelmobile}
\end{figure}

We say that an edge $e$ of $G$ is \emph{covered} by $M$ (resp. $M^+$)
if there is at least one half-edge of $e$ in $M$ (resp. $M^+$). We say that a
vertex $v$ of $G$ is \emph{covered} by $M$ (resp. $M^+$) if there is
at least one  half-edge incident to $v$ in $M$ (resp. $M^+$). 

The main result of
this section is the following theorem:

\begin{theorem}
\label{th:unicellular}
Consider an essentially 4-connected toroidal triangulation $G$, and a
root half-edge $h_0$ of $G$ that is in the interior and incident to a
maximal quadrangle of $G$. Then the extended mobile $M^+$ associated
to the minimal balanced $4$-orientation of $A(G)$ w.r.t.~$h_0$ is a
toroidal unicellular map covering all the vertices and edges of $G$.
Moreover, either $h_0$ is a stem of $M^+$ or its removal creates two
connected components, one of which is a tree.
\end{theorem}

\begin{proof}
  Consider the minimal balanced $4$-orientation $D_{\min}$ of $A(G)$
  w.r.t.~$h_0$, the associated mobile $M$ and the extended mobile
  $M^+$. We consider the superposition of $D_{\min}$ and $M^+$ (see
  the middle of Figure~\ref{fig:k7-mobile}).

  Let us first prove that $M$ has a unique face.  Consider a
  particular face $F$ of $M$. Note that this face is not necessarily
  homeomorphic to an open disk (it can be homeomorphic to a torus, a
  cylinder, a disk, with punctures) so the border of $F$ can be made
  of several closed walk of $M$.  By definition of the mobile $M$,
  each occurrence of a vertex $v$ on the border of $F$ has an incident
  edge $e$ of $D_{\min}$ that is outgoing in the interior of $F$ and
  such that there is no other edge of $D_{\min}$ incident to $v$
  between $e$ and the border of $F$ while going \ccw around $v$ from
  $e$ (see rule of Figure~\ref{fig:mobile-rule}).

  Similarly as in the proof of Lemma~\ref{lem:contractibletilde},
  start from any edge $e_0$ of $D_{\min}$ inside $F$ and consider the
  right-walk $W=(e_i)_{i\geq 0}$ of $D_{\min}$. By previous paragraph,
  each time a vertex of the border of $F$ is reached by $W$, the
  ``right'' outgoing edge puts $W$ back inside $F$, so $W$ cannot
  leave $F$.  Since $A(G)$ has a finite number of edges, some edges
  are used several times in $W$.  Consider a minimal subsequence
  $W'=e_k, \ldots, e_\ell$ such that no edge appears twice and
  $e_k=e_{\ell+1}$.  Thus $W$ ends periodically on the sequence of
  edges $e_k, \ldots, e_\ell$.  So, by
  Lemma~\ref{lem:facialwalktilde}, the right side of $W'$ encloses a
  region $R$ homeomorphic to an open disk and $W'$ is a
  $\{4,8\}$-disk.  Let $f_0$ be the root face of $A(G)$, i.e. the face
  of $A(G)$ containing $h_0$.  By Lemma~\ref{prop:maximal}, $D_{\min}$
  contains no clockwise non-empty $0$-homologous oriented subgraph
  w.r.t.~$f_0$ (see definition in Section~\ref{sec:lattice}). Since
  $W'$ is going \cw around $R$ according to the interior of $R$, we
  have that $R$ contains $f_0$.

  The edges of $G$ ``around'' a $\{4,8\}$-disk of $A(G)$ form a
  quadrangle as depicted by the bold black edges of
  Figure~\ref{fig:48disk-sq}.  Let $Q$ be the quadrangle of $G$
  ``around'' the $\{4,8\}$-disk $W'$.  Recall that the root half-edge
  $h_0$ is in the interior and incident to a maximal quadrangle. Thus
  the interior of this maximal quadrangle contains $Q$ and $h_0$ is
  one of the thin black half-edges of Figure~\ref{fig:48disk-sq}.

  % The edges of $G$ ``around'' a $\{4,8\}$-disk of $A(G)$ form a
  % quadrangle as depicted by the bold black edges of
  % Figure~\ref{fig:48disk-sq}.  Let $Q$ be the quadrangle of $G$
  % ``around'' the $\{4,8\}$-disk $W'$.  Recall that the root half-edge
  % $h_0$ is in the interior and incident to a maximal quadrangle.  By
  % Lemma~\ref{lem:uniqueQuadrangle}, $h_0$ is in the interior of a
  % unique maximal quadrangle of $G$. So this quadrangle is $Q$ and
  % $h_0$ is incident to $Q$. So $h_0$ is one of the thin black
  % half-edges of Figure~\ref{fig:48disk-sq}.

\begin{figure}[!ht]
\center
\begin{tabular}{cc}
\includegraphics[scale=0.4]{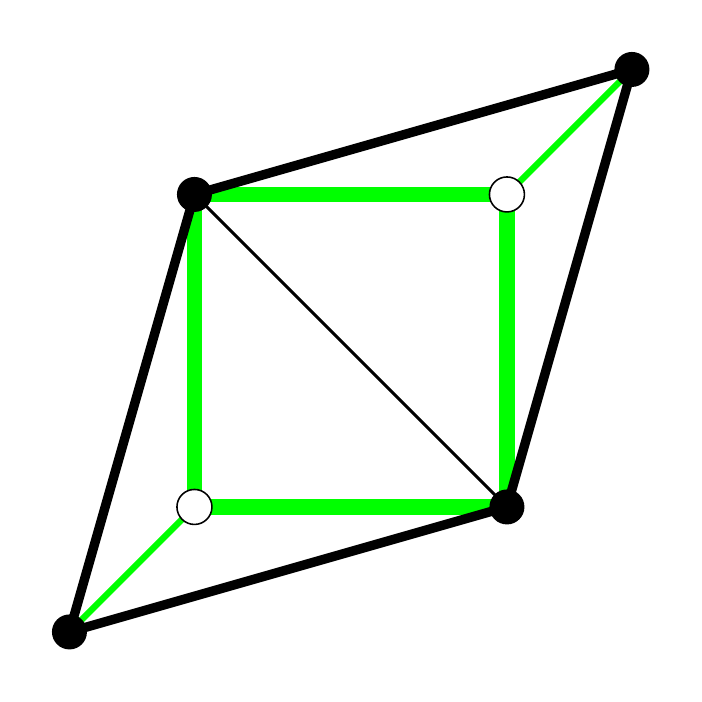} \ \ \ \ &\ \ \ \ 
\includegraphics[scale=0.4]{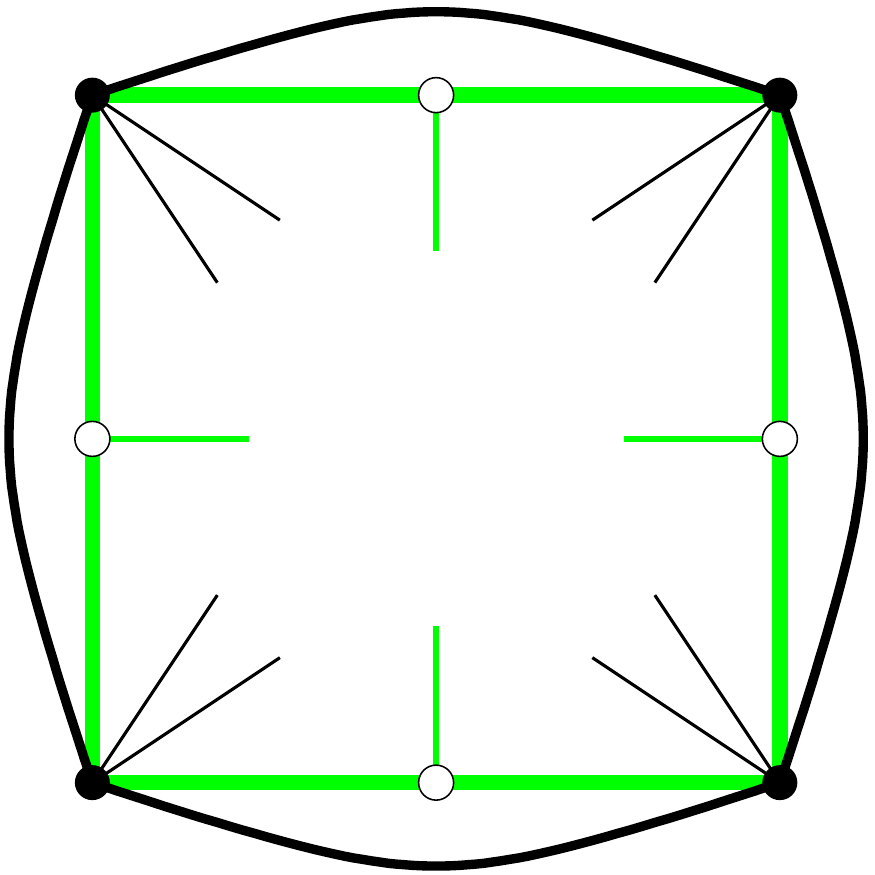} \\
$4$-disk \ \ \ \ &\ \ \ \ $8$-disk \\
\end{tabular}
\caption{The quadrangle of $G$ around a $\{4,8\}$-disk of $A(G)$.}
\label{fig:48disk-sq}
\end{figure}

Lemma~\ref{lem:8diskin} shows that all the edges of $A(G)$ that are in
the interior of a $8$-disk of $A(G)$ and incident to it are entering
it. Thus the orientation of the $\{4,8\}$-disk $W'$ and of the edges
in its interior and incident to it are as depicted on
Figure~\ref{fig:48disk-mobile}. Then by the definition of the mobile
$M$ (see rule of Figure~\ref{fig:mobile-rule}), there is no half-edge of $M$
in the interior of $Q$ and incident to $Q$. Thus $h_0$ is not in
$M$. So $h_0$ is in the strict interior of $F$ and $F$ is the unique
face of $M$. Moreover $M^+$ has strictly one more half-edge than $M$.

\begin{figure}[!ht]
\center
\begin{tabular}{cc}
\includegraphics[scale=0.4]{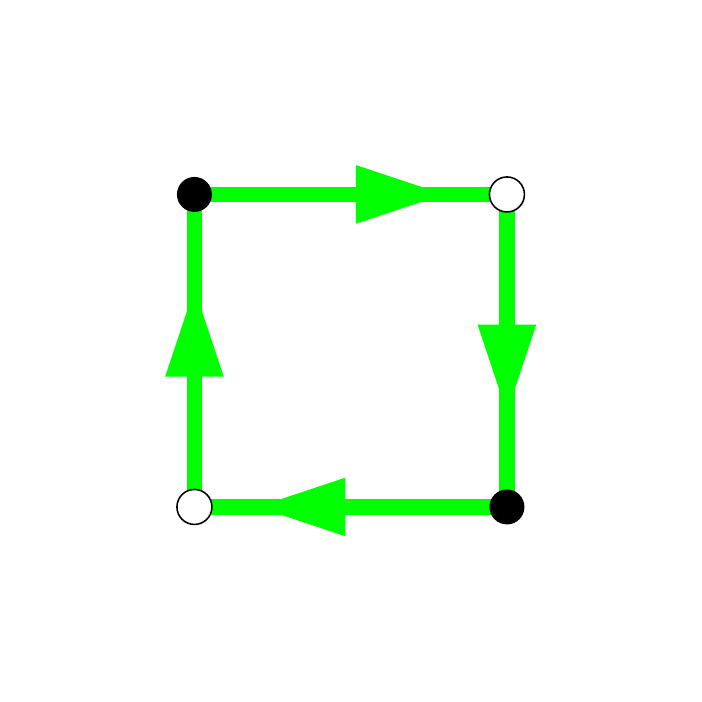} \ \ \ \ &\ \ \ \ 
\includegraphics[scale=0.4]{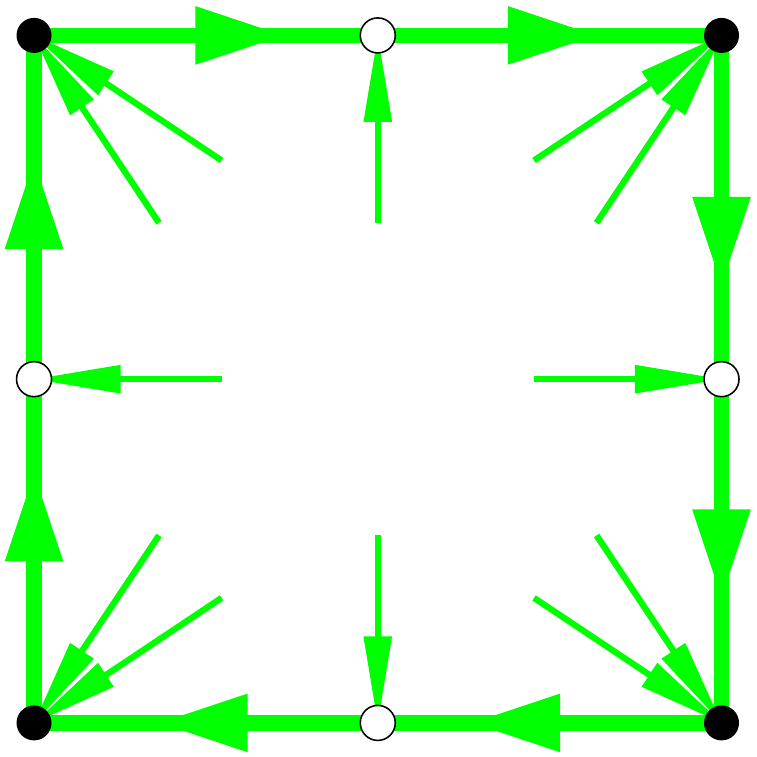} \\
$4$-disk \ \ \ \ &\ \ \ \ $8$-disk \\
\end{tabular}
\caption{Orientation of the $\{4,8\}$-disk.}
\label{fig:48disk-mobile}
\end{figure}

The number of half-edges of $M$ is equal to $4n$ (one half-edge for
each outgoing edge of the $4$-orientation $D_{\min}$ of $A(G)$).  Thus
the number of half-edges of $M^+$ is equal to $4n+1$. The toroidal
triangulation $G$ has exactly $3n$ edges. So $M^+$ has at least $n+1$
full-edges. Since $M^+$ is a graph embedded on the torus with $n$
vertices, if it has strictly more than $n+1$ edges, then it does not
have a unique face. So $M^+$ has exactly $n+1$ edges and it is a
unicellular map covering all the vertices. The number of distinct
edges of $G$ covered by $M^+$ is exactly $(4n+1)-(n+1)=3n$. So $M^+$
is covering all the edges of $G$.

Since there is no half-edge of $M$ in the interior of $Q$ and incident
to $Q$. We have that either $W'$ is a $4$-disk and $h_0$ is a stem of
$M^+$ or $W'$ is a $8$-disk and the removal of $h_0$ from $M^+$
creates two connected components, one of which is a tree.
\end{proof}

By Lemma~\ref{lem:uniqueQuadrangle}, there is a unique maximal
quadrangle containing the root half-edge, that we call the \emph{root
  quadrangle}.
 
The example of $K_7$ of Figure~\ref{fig:k7-mobile}, is an example
where the $\{4,8\}$-disk inside the root quadrangle is a
$4$-disk. There is no vertices in the strict interior of the root
quadrangle and the root half-edge $h_0$ of $M^+$ (in magenta) is not
part of a full-edge of $M^+$. 

When the root quadrangle has some vertices in its interior, then the
$\{4,8\}$-disk  inside the root quadrangle is in fact a $8$-disk
and the part of the mobile $M$ inside this root quadrangle is a tree
(exactly like in the planar case, see~\cite{Fus09}). In $M^+$ this
tree is connected to the ``toroidal'' part of $M$ that is external to
the root quadrangle with the addition of the half-edge $h_0$ added to
$M^+$.

Figure~\ref{fig:disk8mobile} is an example of an essentially
$4$-connected toroidal triangulation with some nested quadrangles. The
barred half-edge is the root half-edge. It is chosen inside a non
empty root quadrangle. There are also non empty quadrangles outside
the root quadrangle.  The triangulation is given with a balanced
transversal structure whose corresponding orientation of the angle
graph (not represented) is the minimal balanced $4$-orientation w.r.t.~the
barred half-edge. The corresponding extended mobile is
given. One can see that Theorem~\ref{th:unicellular} is
satisfied, i.e. the extended mobile is a unicellular map covering all
the vertices and edges. The magenta half-edge, corresponds to the root
half-edge and links the two connected part of the mobile, one of which
is a tree.

\begin{figure}[!ht]
\center
\includegraphics[scale=0.34]{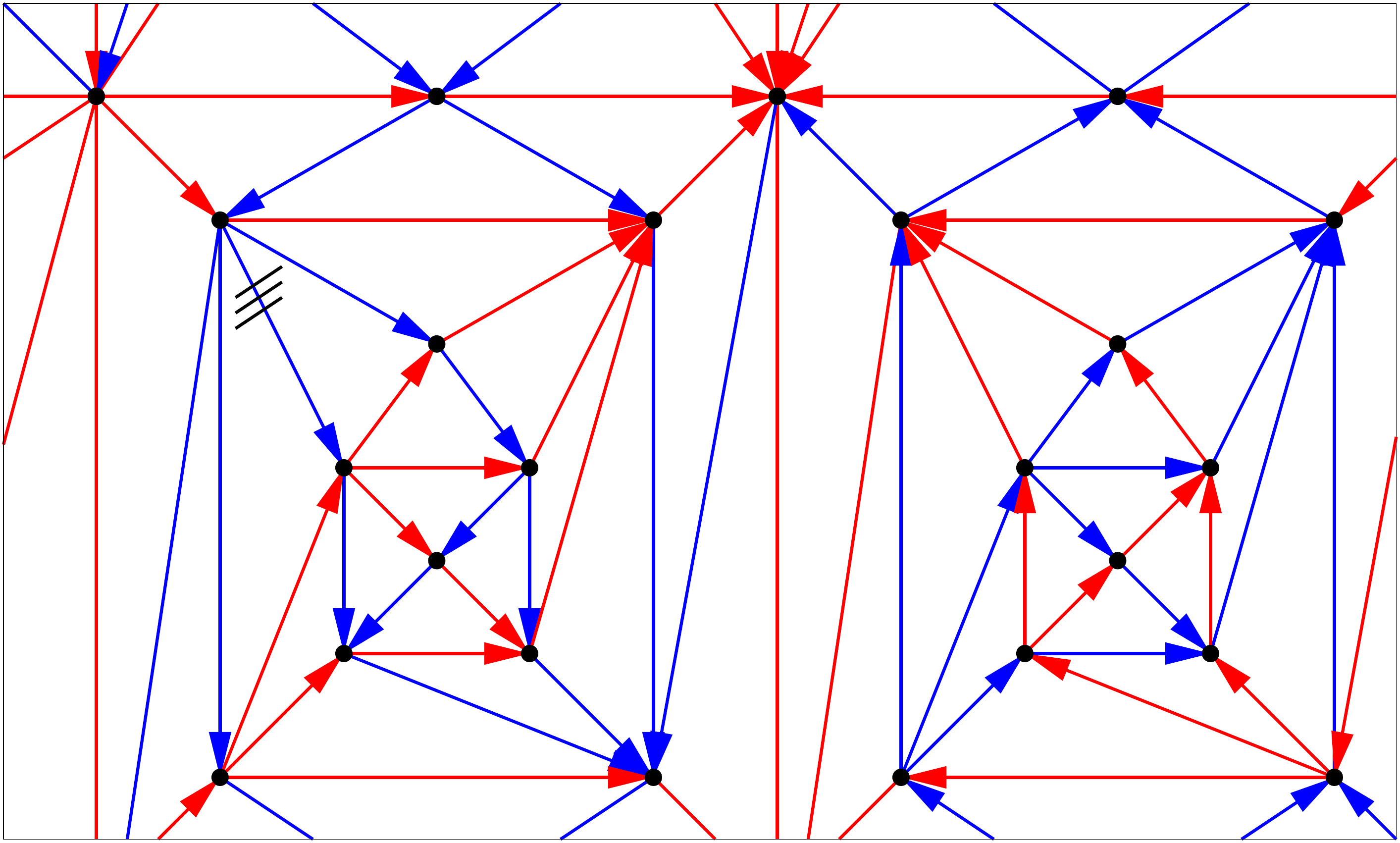} 

\ \\

\includegraphics[scale=0.34]{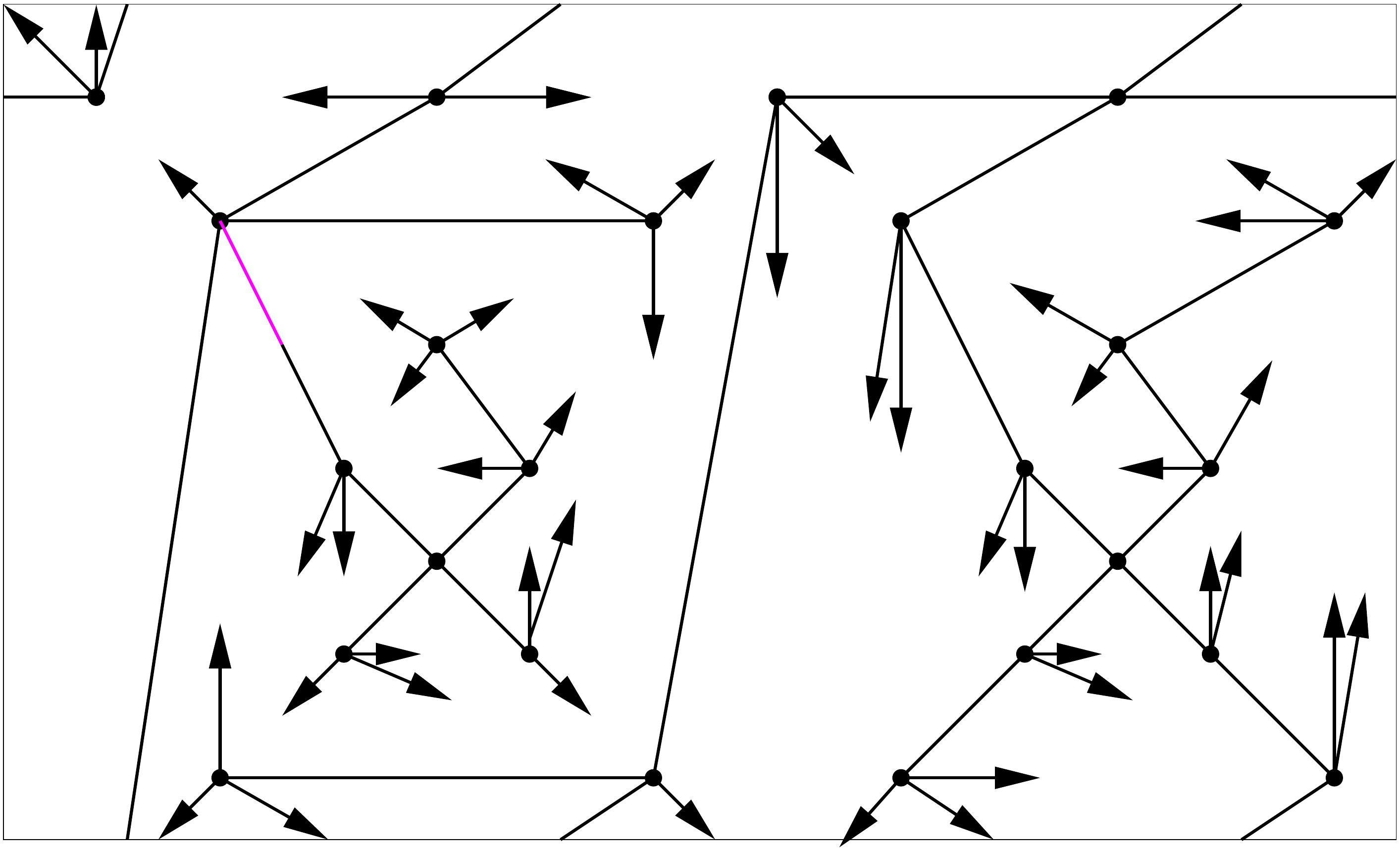}
\caption{Example of a balanced transversal structure of an essentially
  $4$-connected toroidal triangulation with some nested quadrangles
  and the corresponding extended mobile.}
\label{fig:disk8mobile}
\end{figure}

%\comment{A mettre au bon endroit}

A toroidal unicellular map on $n$ vertices has exactly $n+1$ edges.
Since the total number of edges of a triangulation on $n$ vertices is
$3n$, a consequence of Theorem~\ref{th:unicellular} is that the
extended mobile $M^+$ has exactly $n$ vertices, $n+1$ edges and $2n-1$
stems. In total, $M^+$ has $2(n+1)+2n-1=4n+1$ half-edges. So the root
half-edge is not part of the mobile $M$ and is added to $M$ to obtain
$M^+$. So all the vertices of $M^+$ have degree $4$, except the root
vertex that has degree $5$.
% Morover, on the mobile-labeling, the root
%half-edge has no label and is situated between half-edges labeled $0$
%and $1$ in \ccw order.

\subsection{Recovering the original triangulation}
\label{sec:recover}

This section is dedicated to showing how to recover the original
triangulation from the extended mobile.  The recovering process is
described by the following theorem.

\begin{theorem}
\label{th:recover}
Consider an essentially 4-connected toroidal triangulation $G$, and a
root half-edge $h_0$ of $G$, incident to a vertex $v_0$, such that
$h_0$ is in the interior and incident to a maximal quadrangle of $G$.
From the extended mobile $M^+$ associated to the minimal balanced
$4$-orientation of $A(G)$ w.r.t.~$h_0$, one can reattach all the stems
of $M^+$ to obtain $G$ by starting from the angle of $v_0$ just after
$h_0$ in \cw order around $v_0$ and walking along the face of $M^+$ in
\ccw order (according to the interior of this face): each time a stem is met, it is
reattached in order to create a triangular face on its left side.
\end{theorem}

Theorem~\ref{th:recover} is illustrated on
Figure~\ref{fig:k7-mobile-recover} to recover $K_7$ from the extended
mobile of Figure~\ref{fig:k7-mobile}. We have represented only the
first and last two steps of the method. One can also play with the
extended mobile of Figure~\ref{fig:disk8mobile} to recover the
corresponding triangulation.

\begin{figure}[!ht]
\center
\includegraphics[scale=0.34]{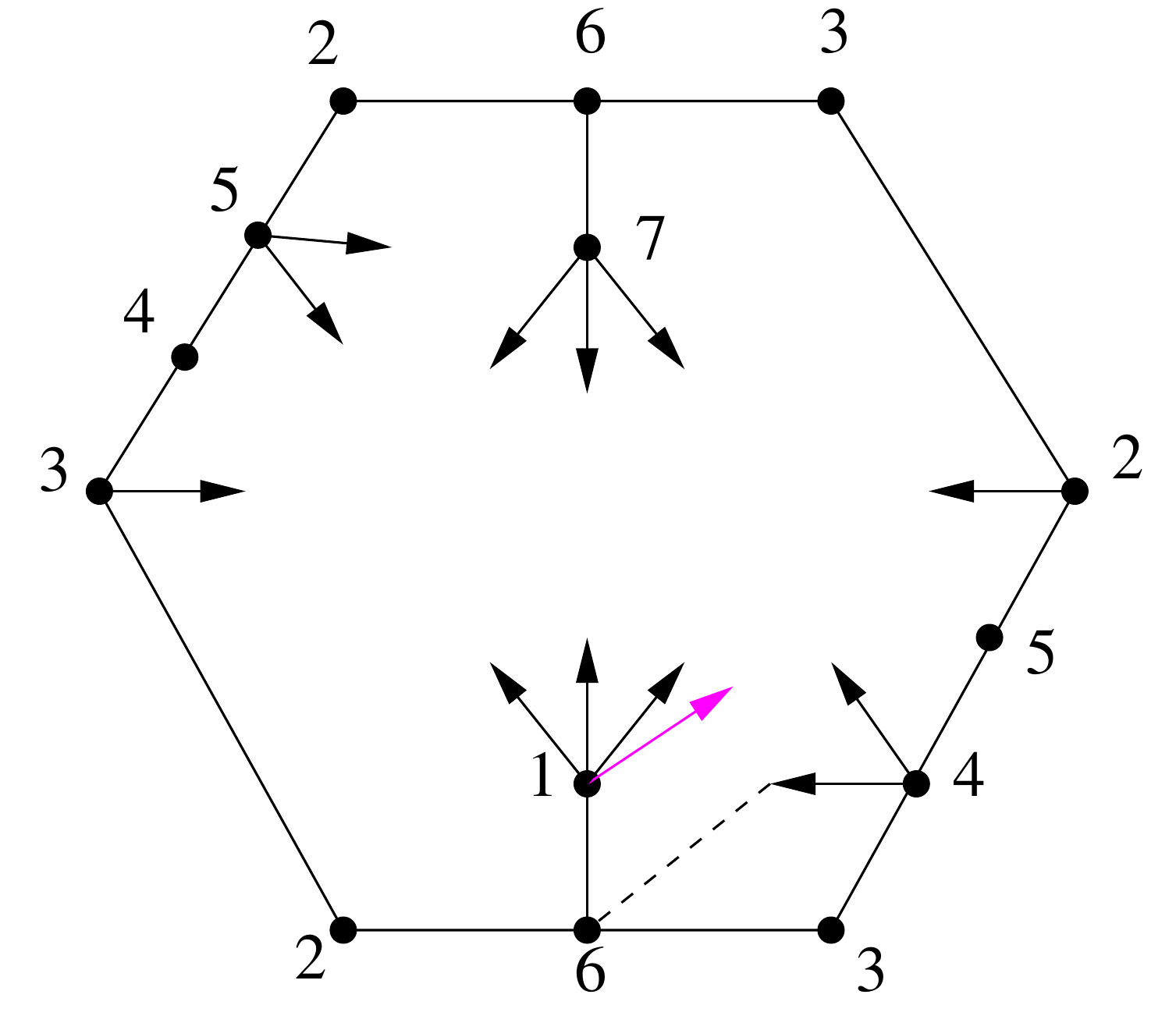} \ \
\includegraphics[scale=0.34]{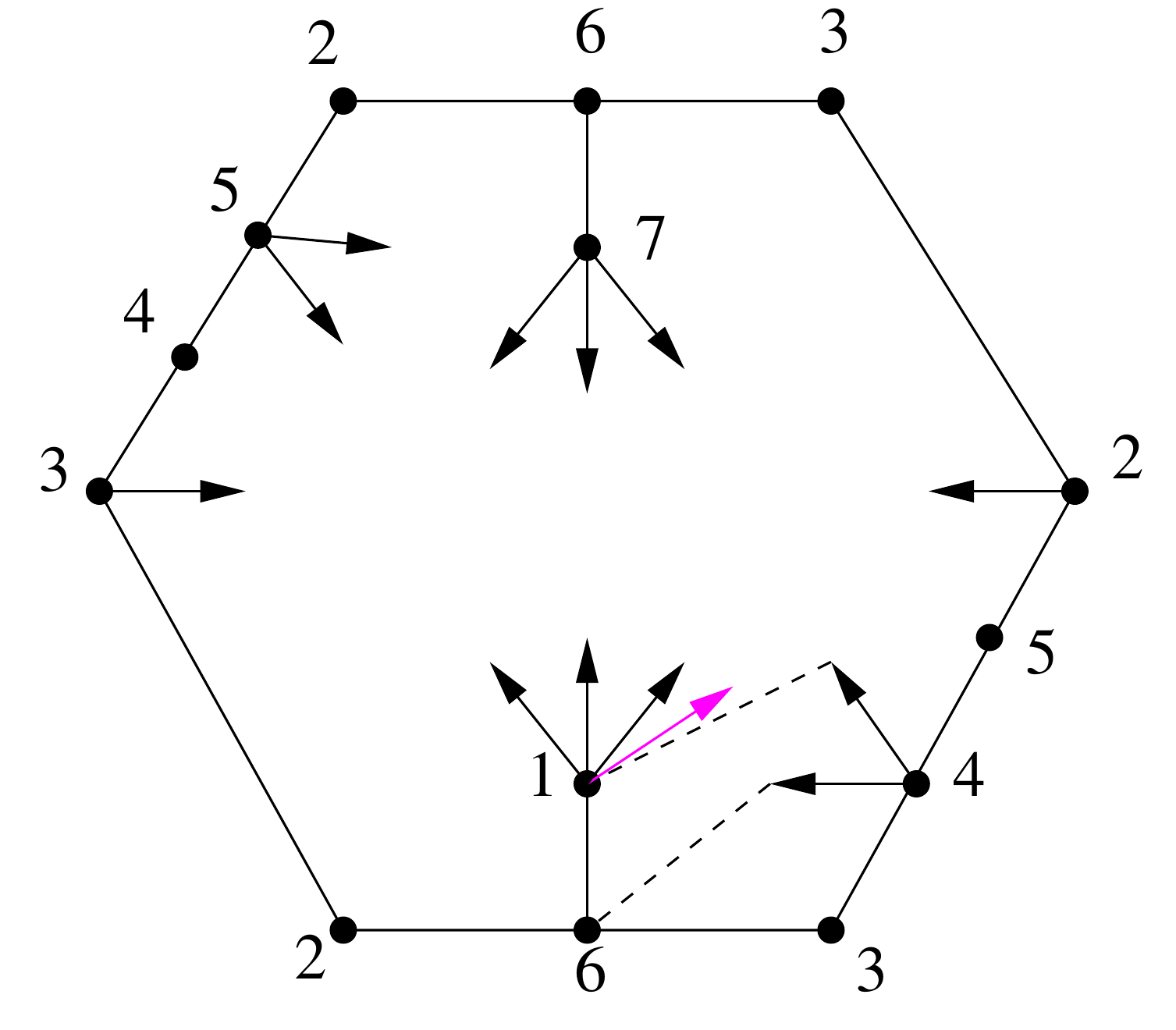} \\ 
\includegraphics[scale=0.34]{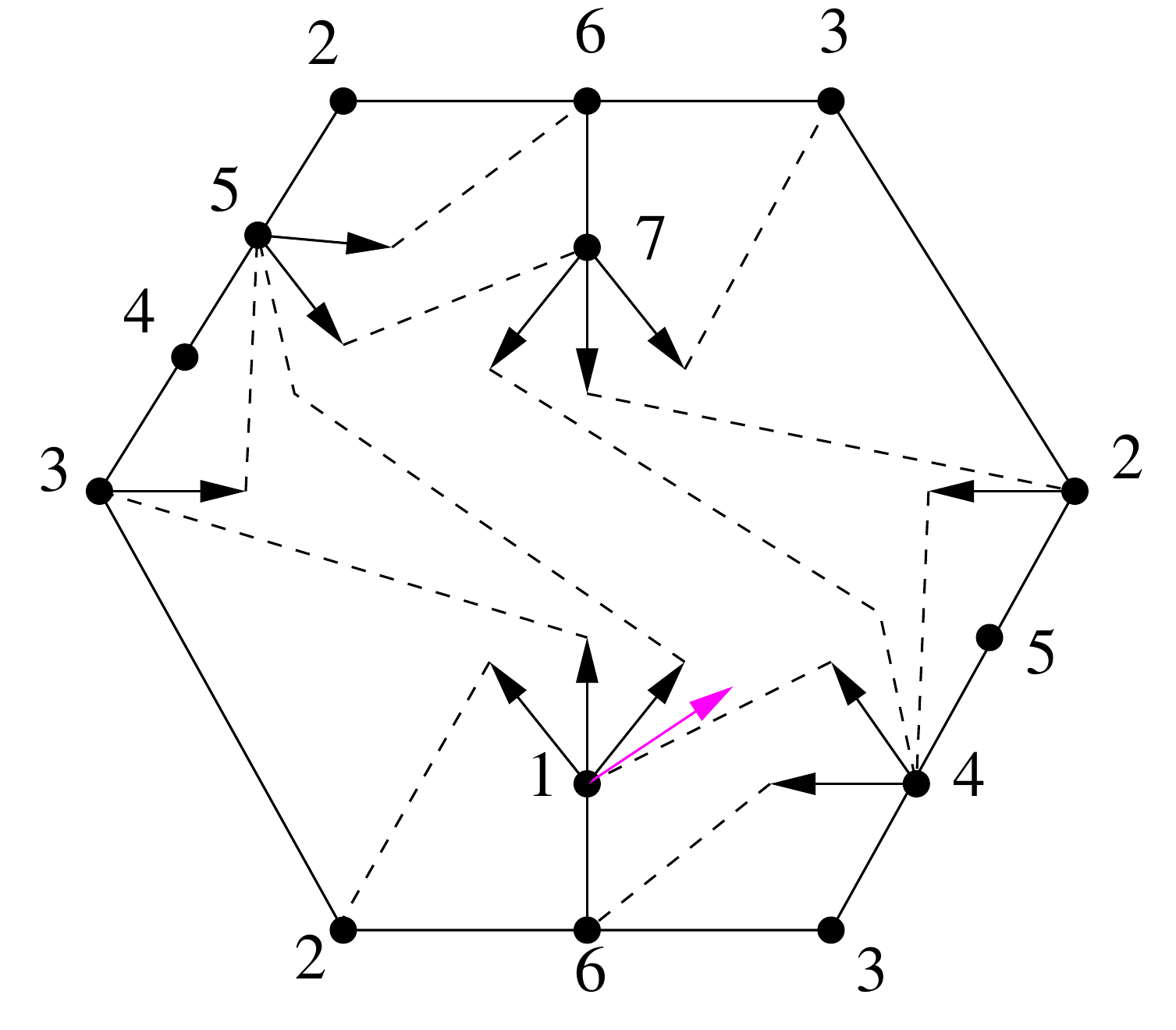} \ \
\includegraphics[scale=0.34]{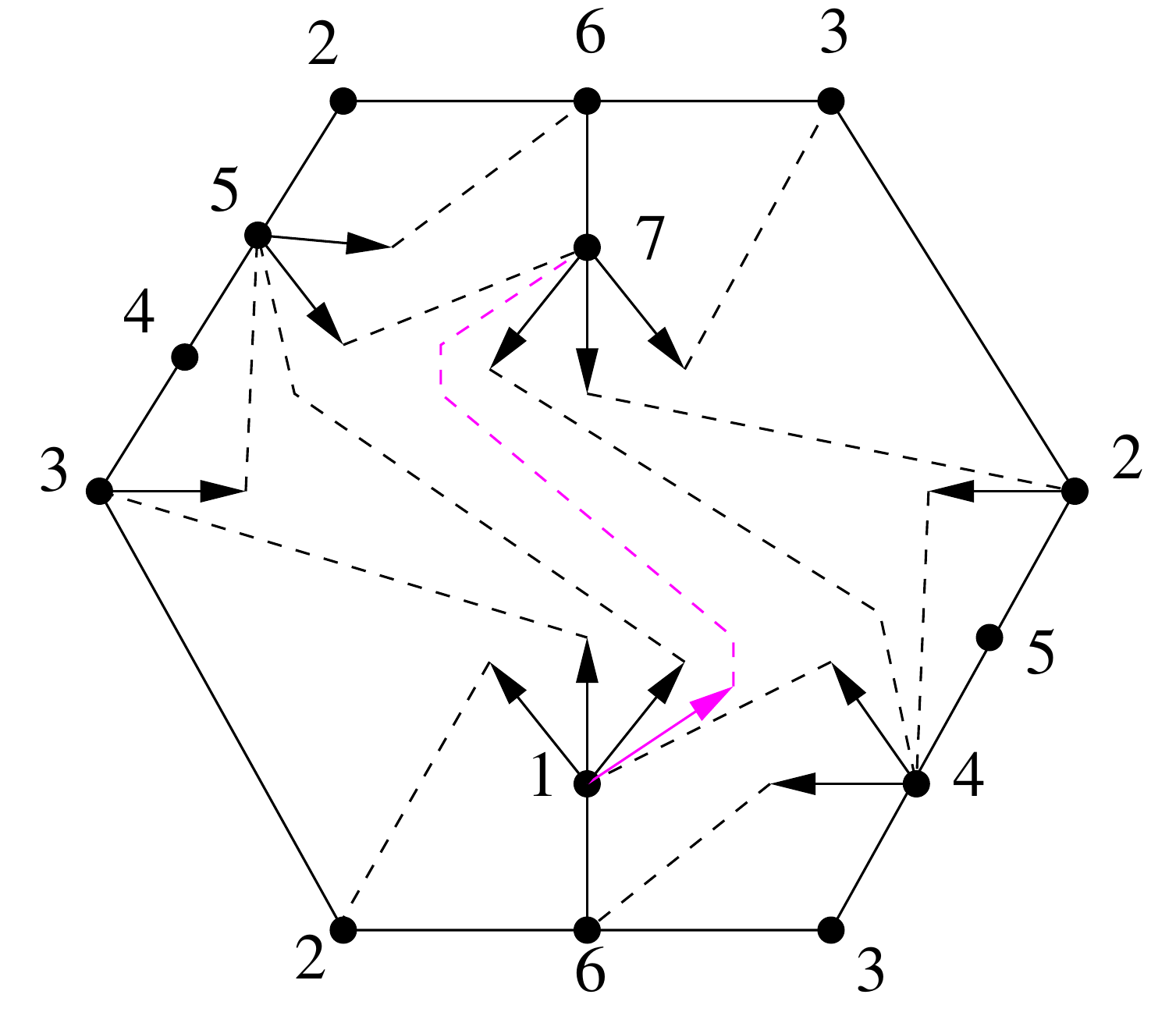}
\caption{Recovering $K_7$ from the extended mobile.}
\label{fig:k7-mobile-recover}
\end{figure}

In fact in this section we define a method, more general than the one
described in Theorem~\ref{th:recover}, that is useful for
Sections~\ref{sec:enumeration}.

Let $\mathcal M_{r}(n)$ denote the set of toroidal unicellular maps
with exactly $n$ vertices, $n+1$ edges and $2n-1$ stems such that all
vertices have degree $4$, except one vertex (called root vertex) that
has degree $5$, moreover the root vertex has a marked incident
half-edge (called the root half-edge) that is either a stem or whose
removal creates two connected components, one of which is a tree.
% \comment{est-ce qu'il faut mettre la fin maintenant avec l'arbre
%  ?}
Note that the extended mobile $M^+$ given  
 by Theorem~\ref{th:unicellular} is an element of
$\mathcal M_r(n)$.
%\comment{faut montrer des choses si on garde l'arbre}

We use the classical closure procedure (see~\cite{Fus09}) to reattach
step by step all the stems of an element $M^+$ of $\mathcal M_r(n)$.
Let $M_0=M^+$, and, for $1\leq k \leq 2n-1$, let $M_{k}$ be the map
obtained from $M_{k-1}$ by reattaching one of its stem (we explicit
below which stem is reattached and how). The \emph{special face of
  $M_0$} is its only face. For $1\leq k \leq 2n-1$, the \emph{special
  face of $M_{k}$} is the face on the right of the stem of $M_{k-1}$
that is reattached to obtain $M_{k}$.  For $0\leq k\leq 2n-1$, the
border of the special face of $M_k$ consists of a sequence of edges
and stems. We define an \emph{admissible triple} as a sequence
$(e_1,e_2,s)$, appearing in \ccw order along the border of the special
face of $M_k$, such that $e_1=\{u,v\}$ and $e_2=\{v,w\}$ are edges of
$M_k$ and $s$ is a stem attached to $w$. The \emph{closure} of the
admissible triple consists in attaching $s$ to $u$, so that it creates
an edge $\{w,u\}$ and so that it creates a triangular face $(u,v,w)$
on its left side (when oriented from $w$ to $u$).  The \emph{complete
  closure} of $U$ consists in closing a sequence of admissible triples,
i.e.  for $1\leq k \leq 2n-1$, the map $M_{k}$ is obtained from
$M_{k-1}$ by closing any admissible triple.

Note that, for $0\leq k\leq 2n-1$, the special face of $M_k$ contains
all the stems of $M_k$. The closure of a stem reduces the number of
edges on the border of the special face and the number of stems by
$1$. At the beginning, the unicellular map $M_0$ has $n+1$ edges and
$2n-1$ stems. So along the border of its special face, there are
$2n+2$ edges and $2n-1$ stems. Thus there is exactly three more edges
than stems on the border of the special face of $M_0$ and this is
preserved while closing stems. So at each step there is necessarily at
least one admissible triple and the sequence $M_k$ is well defined.
Since the difference of three is preserved, the special face of
$M_{2n-2}$ is a quadrangle with exactly one stem. So the reattachment
of the last stem creates two faces that have length three and at the end
$M_{2n-1}$ is a toroidal triangulation.  Note that at a given step
there might be several admissible triples but their closure are
independent and the order in which they are performed does not modify
the obtained triangulation $M_{2n-1}$.

We now apply the closure method to our particular case.  Consider an
essentially 4-connected toroidal triangulation $G$, a root half-edge
$h_0$ of $G$ that is in the interior and incident to a maximal
quadrangle of $G$, and the extended mobile $M^+$ associated to the
minimal balanced $4$-orientation of $A(G)$ w.r.t.~$h_0$.
%  A
%consequence of Theorem~\ref{th:unicellular} (see previous remarks) is
%that $M^+$ is an element of $\mathcal M_{r}(n)$.
Recall that $M^+$ is an element of $\mathcal M_r(n)$ so we can apply on
$M^+$ the complete closure procedure described above.  We use the same
notation as before, i.e. let $M_0=M^+$ and for $1\leq k \leq 2n-1$, the
map $M_{k}$ is obtained from $M_{k-1}$ by closing any admissible
triple.  The following lemma shows that the triangulation obtained by
this method is $G$:

\begin{lemma}
\label{lem:stemstep}
  The complete closure of $M^+$ is $G$,  i.e. $M_{2n-1}=G$.
\end{lemma}

\begin{proof}
  We prove by induction on $k$ that every face of $M_k$ is a face of
  $G$, except for the special face.  This is true for $k=0$ since
  $M_0=M^+$ has only one face, the special face.  Let
  $0\leq k\leq 2n-2$, and suppose by induction that every non-special
  face of $M_k$ is a face of $G$.  Let $(e_1,e_2,s)$ be the admissible
  triple of $M_k$ such that its closure leads to $M_{k+1}$, with
  $e_1=\{u,v\}$ and $e_2=\{v,w\}$. The closure of this triple leads to a
  triangular face $(u,v,w)$ of $M_{k+1}$. This face is the only
  ``new'' non-special face while going from $M_k$ to $M_{k+1}$.

  Suppose, by contradiction, that this face $(u,v,w)$ is not a face of
  $G$.  Let $a_v$ (resp. $a_w$) be the angle of $M_k$ at the special
  face, between $e_1$ and $e_2$ (resp. $e_2$ and $s$).  Since $G$ is a
  triangulation, and $(u,v,w)$ is not a face of $G$, there exists at
  least one stem of $M_k$ that should be attached to $a_v$ or $a_w$ to
  form a proper edge of $G$. Let $s'$ be such a stem that is the
  nearest from $s$. In $G$ the edges corresponding so $s$ and $s'$
  should be incident to the same triangular face $T$. Let $x$ be the
  vertex incident to $s'$.  Let $z\in \{v,w\}$ such that $s'$ should
  be reattached to $z$.  If $z=v$, then $s$ should be reattached to
  $x$ to form a triangular face of $G$. If $z=w$, then $s$ should be
  reattached to a common neighbor of $w$ and $x$ located on the border
  of the special face of $M_k$ in \ccw order between $w$ and $x$. So
  in both cases $s$ should be reattached to a vertex $y$ located on
  the border of the special face of $M_k$ in \ccw order between $w$
  and $x$ (with possibly $y=x$). To summarize $s$ goes from $w$ to $y$
  and $s'$ from $x$ to $z$, and $z,w,y,x$ appear in \ccw around $T$
  with $z=w$ or $y=x$.  The two half-edges $h,h'$ of $T$ that are in
  the same edges with $s,s'$ are not in $M^+$. By the mobile rule (see
  Figure~\ref{fig:mobile-rule}), the two half-edges $h,h'$ that are
  not in $M^+$ corresponds in the orientation of $A(G)$ to two
  distinct outgoing edges for the dual-vertex corresponding to
  $T$. This contradicts the fact that the considered orientation of
  $A(G)$ is a $4$-orientation and that dual-vertices should have
  outdegree $1$.

  So for $0\leq k\leq 2n-2$, all the non-special faces of $M_k$ are
  faces of $G$. In particular every face of $M_{2n-1}$ except one is a
  face of $G$. Then clearly the (triangular) special face of
  $M_{2n-1}$ is also a face of $G$, hence $M_{2n-1}=G$.
\end{proof}

Lemma~\ref{lem:stemstep} shows that one can recover the original
triangulation from $M^+$ with any sequence of admissible triples that
are closed successively. This does not explain how to find the
admissible triples efficiently. In fact the root half-edge $h_0$ can
be used to find a particular admissible triple of $M_k$. We define the
root angle $a_0$ of $G$ as the angle of $v_0$ just after $h_0$ in \cw
order around $v_0$. This definition of $a_0$
naturally extends to $M^+$ or when some admissible triples are
 reattached.

 \begin{lemma}
 \label{lem:firststem}
 For $0\leq k\leq 2n-2$, let $s$ be the first stem met while walking
 \ccw from $a_0$ in the special face of $M_k$. Then before $s$, at
least two edges are met and  the last two of these edges form an
admissible triple with $s$.
\end{lemma}

\begin{proof}
  Since $s$ is the first stem met, there are only edges that are met
  before $s$. Suppose by contradiction that there is only zero or one
  edge met before $s$. Then the reattachment of $s$ to form the
  corresponding edge of $G$ is necessarily such that the triangular
  face $T$ that is formed on the left side of the stem contains the
  root half-edge $h_0$ on its border.  Let $h$ be the half-edge of $T$
  that  is in the same edge with $s$ and not in $M^+$.  Then the
  two half-edges $h,h_0$ that are not in $M$ corresponds in the
  orientation of $A(G)$ to two distinct outgoing edges for the
  dual-vertex corresponding to $T$. This contradicts the fact that the
  considered orientation of $A(G)$ is a $4$-orientation and that
  dual-vertices should have outdegree $1$.
\end{proof}

Lemma~\ref{lem:firststem} shows that one can reattach all the stems by
walking once along the face of $M^+$ in \ccw order starting from
$a_0$. Thus we obtain Theorem~\ref{th:recover}.

Note that $M^+$ is such that the complete closure procedure described
here never \emph{wraps over the root angle}, i.e. when a stem is
reattached, the root angle is always in the face that is on its right
side. The property of never wrapping over the root angle is called
\emph{safe} here. Note that sometimes this property is called ``balanced''
in the literature and here the
word ``balanced'' is already used  with a completely different
meaning. Let $\mathcal M_{r,s}(n)$ denote the set of elements of
$\mathcal M_r(n)$ that are safe. So the extended mobile given by
Theorem~\ref{th:unicellular} is an element of $\mathcal M_{r,s}(n)$.

 We exhibit in Section~\ref{sec:bijection} a
bijection between appropriately rooted essentially 4-connected toroidal triangulations and a
particular subset of $\mathcal M_{r,s}(n)$.

The possibility to close admissible triples in any order to recover
the original triangulation is interesting compared to the simpler
method of Theorem~\ref{th:recover} since it enables to recover the
triangulation even if the root half-edge is not given. 
%\comment{ This property is
%used in Section~\ref{sec:bij2} to obtain a bijection between toroidal
%triangulations and some unrooted unicellular maps.}
Indeed, when the root angle is not given, then one can simply start
from any angle of $M^+$, walk twice around the face of $M^+$ in \ccw
order and reattach all the admissible triples that are encountered
along this walk. Walking twice ensures that at least one complete
round is done from the root angle. Since only admissible triples are
considered, we are sure that no unwanted reattachment is done during
the process and that the final map is $G$. This enables us to reconstruct
$G$ in linear time even if the root angle is not known.  This property
will also be used in Section~\ref{sec:enumeration} for enumeration
purpose.

\subsection{Asymptotically optimal encoding}
\label{sec:coding}

A $4$-connected planar triangulation on $n$ vertices, can be encoded
with a binary word of length
$\sim n\, \log_2(\frac{27}{4})\approx 2.7549\,n$
(see~\cite[Theorem~4.2]{FusThesis}).  This is asymptotically optimal
since, by results of Tutte, the number $P_n$ of $4$-connected planar
triangulations on $n$ vertices satisfies
$\log_2(P_n)\sim n\, \log_2(\frac{27}{4})$.  The results of previous
sections allow us to generalize this optimal encoding to the toroidal
case.

A \emph{ternary tree} is a
  plane tree, rooted at a leaf, such that every inner vertex has
  degree exactly four.  
A ternary tree $T$ on $n$ inner vertices can easily be encoded using a
binary word on $3n$ bits by the following: walk in \ccw order around
$T$ from the root angle, write a ``1'' when an inner vertex is
discovered for the first time, and a ``0'' when a leaf is traversed. A
ternary tree on $n$ inner vertices has $n$ inner vertices and $2n+2$
leaves. So we obtain a binary word of length $3n+2$ with $n$
bits $1$.  Using \cite[Lemma~7]{BGH03}, this word can then be encoded
with a binary word of length
$\log_2\binom{3n+2}{n}+o(n)\sim n\, \log_2(\frac{27}{4})\approx
2.7549\,n$ bits.

Consider an essentially 4-connected toroidal triangulation $G$, a root
half-edge $h_0$ of $G$ that is in the interior and incident to a
maximal quadrangle of $G$, and the extended mobile $M^+$ associated to
the minimal balanced $4$-orientation of $A(G)$ w.r.t.~$h_0$.  By
Theorem~\ref{th:recover} one can retrieve the triangulation $G$ from
$M^+$.  Hence to encode $G$, one just has to encode $M^+$.  The
extended mobile $M^+$ is a toroidal unicellular map with $n$ vertices,
$n+1$ edges, $2n-1$ stems. All its vertices have degree $4$, except
the root vertex that has degree $5$. Either $h_0$ is a stem of $M^+$
or its removal creates two connected components, one of which is a
tree.

Let $k\geq 0$ be the number of (inner) vertices of the tree part
attached to $h_0$, with $k=0$ if $h_0$ is a stem.  We remove the
half-edge $h_0$ from $M^+$, and obtain : a toroidal component $G_1$,
that we root at the angle where $h_0$ is attached, and a tree
component $T_2$, that we root at the half-edge opposite to $h_0$.  So
$G_1$ is a toroidal unicellular map with $n-k$ vertices, $n-k+1$
edges, $2n-2k-2$ stems. And $T_2$ is a planar tree with $k$ vertices,
$k-1$ edges, $2k+2$ stems. Moreover, all the vertices of $G_1$ and
$T_2$ have degree $4$.  Now we choose two edges $e_1,e_2$ of $G_1$,
such that $G_1 \setminus \{e_1,e_2\}$ is acyclic. We transform $G_1$
into a planar tree $T_1$ by cutting $e_1,e_2$ and transforming each of
$e_1,e_2$ into two special stems of $T_1$.  We root $T_1$ on a stem by
keeping the information of where are the special stems, which pairs
should be reattached together, and where is the angle attached to
$h_0$.  This information can be stored with $O(\log(n))$ bits.  One
can recover $G_1$ from $T_1$ by reattaching the special stems in order
to form non-contractible cycles and changing the root.  Thus we are
left to encode two ternary trees $T_1$ and $T_2$ with $n-k$ and $k$
inner vertices, respectively.  

By applying the ternary tree encoding
method on $T_1$ and $T_2$ we obtain the following theorem :

\begin{theorem}
\label{th:encoding}
Any essentially $4$-connected toroidal triangulation on $n$ vertices,
can be encoded with a binary word of length
$\sim n\, \log_2(\frac{27}{4})\approx
2.7549\,n$ and this is asymptotically optimal.
\end{theorem}

The optimality of Theorem~\ref{th:encoding} is due to the fact that the number of essentially
$4$-connected toroidal triangulations is at least the number of 
$4$-connected planar triangulations.

Here is a remark on the complexity of the encoding part. All the
encoding and decoding process is linear as soon as a balanced
transversal structure is given. But even if the proof of
Theorem~\ref{th:existence} is constructive and gives a polynomial
algorithm to find a balanced transversal structure, the obtained algorithm is not
linear. The difficulty is to be able to find contractible edges, and
contract all the graph to a single vertex, in linear time. Currently
this has to be done by Lemma~\ref{lem:contraction} that does not give
linear complexity. So the question to be able to find in linear time a
balanced transversal structure of an essentially 4-connected
triangulation is an interesting and open problem.

In the plane, the proof of the existence of such objects is usually done
quite easily by using a so-called shelling order (or canonical
order). This method consists in starting from the outer face and
removing the vertices one by one. It leads to simple linear time
algorithms.  We do not see how to generalize this kind of method here
since the toroidal objects that we considered are too homogeneous and
there is no special face (and thus no particular starting point)
playing the role of the outer face.

\subsection{Bijective consequences}
\label{sec:bijection}

Consider an essentially 4-connected toroidal triangulation $G$, a root
half-edge $h_0$ of $G$ that is in the interior and incident to a
maximal quadrangle, and the extended mobile $M^+$ associated to the
minimal balanced $4$-orientation of $A(G)$ w.r.t.~$h_0$.
Theorems~\ref{th:unicellular} and~\ref{th:recover} show that $M^+$
gives a toroidal unicellular map with stems from which one can recover
the original triangulation. Thus there is a bijection between
essentially 4-connected toroidal triangulations rooted from an
appropriate half-edge and their corresponding set of extended mobiles. The
goal of this section is to describe exactly the set of these extended
mobiles.

Recall from Section~\ref{sec:recover} that the obtained extended
mobiles are elements of $\mathcal M_{r,s}(n)$.  One may hope that
there is a bijection between essentially 4-connected toroidal
triangulations appropriately rooted and $\mathcal M_{r,s}(n)$. This is
the classic behavior in the planar case since there is a unique
lattice associated to the set of $\alpha$-orientations of a planar map
(for a fixed $\alpha$). But here, things are different since the set
of $4$-orientations of the angle map is now partitioned into several
lattices and there might be several minimal elements, some of which
behave well w.r.t.~the mobile rule.  Indeed, there exists examples of
minimal non-balanced $4$-orientations of angle maps of essentially
4-connected toroidal triangulations appropriately rooted such that the
corresponding extended mobile is in $\mathcal M_{r,s}(n)$. The
balanced property is the property that defines uniquely our considered
minimal element and thus we have to translate this property on the set
of mobiles.

Note that there are two types
of toroidal unicellular maps.
% depicted on
%Figure~\ref{fig:hexasquare}. 
Two cycles of a unicellular map may intersect either on a single
vertex (square case) or on a path (hexagon case). We call such maps
\emph{square unicellular maps} or \emph{hexagon unicellular maps},
respectively.  The square can be seen as a particular case of the
hexagon where one side has length zero and thus the two corners of the
hexagon are identified.  In the square case (resp. hexagon case), the
unicellular map has exactly $2$ (resp. $3$) distinct cycles that are
moreover non-contractible and not weakly homologous to each other.

Recall that given a cycle $C$ of $G$ with a direction of traversal, we
have $\gamma (C)$ equals the number of edges of $A(G)$ leaving $C$ on
its right minus the number of edges of $A(G)$ leaving $C$ on its left.
Recall that the root angle of $G$ is the angle just after the root
half-edge in \cw order.  In each angle of the extended mobile $M^+$,
except the root angle, there is an outgoing edge of $A(G)$ (see rule
of Figure~\ref{fig:mobile-rule}). So for a cycle $C$ of the extended
mobile, one can compute $\gamma (C)$ by considering the angles of
$M^+$ on the left and right side of $C$, except the root angle.  Then,
since we are considering balanced $4$-orientations of the angle map,
for any (non-contractible) cycle $C$ of the extended mobile obtained
by Theorem~\ref{th:unicellular}, we have $\gamma (C)=0$.

Consider an element $M^+$ of $\mathcal M_{r}(n)$.  We say that an
unicellular map of $\mathcal M_{r}(n)$ is \emph{balanced} if every
cycle of the unicellular map has the same number of angles on the left and right
sides, with the special rule that the root angle does not count.  Let
$\mathcal M_{r,s,b}(n)$ denote the subset of elements of
$\mathcal M_{r,s}(n)$ that are balanced.

Let us recall, for the sake of clarity, the complete definition of
$\mathcal M_{r,s,b}(n)$ that is the set of toroidal unicellular maps
with exactly $n$ vertices, $n+1$ edges and $2n-1$ stems such that:
\begin{itemize}
\item \emph{``r'' for root:} All
vertices have degree $4$, except one vertex (called root vertex) that
has degree $5$, moreover the root vertex has a marked incident
half-edge (called the root half-edge) that is either a stem or whose
removal creates two connected components, one of which is a tree.
\item \emph{``s'' for safe:} When the stems of admissible triples are
  reattached (in any order), the angle just after the root half-edge
  in \cw order (called the root angle) is always in the face that is
  on the right side of the stems.
\item \emph{``b'' for balanced:} Every (non-contractible) cycle of the
  map has the same number of angles on its left and right sides, with
  the special rule that the root angle does not count.
\end{itemize}
%%%%%%%%%

Let $\mathcal T_r(n)$ be the set of essentially 4-connected toroidal
triangulations on $n$ vertices rooted at a half-edge that is in the
interior and incident to a maximal quadrangle.

We have the following bijection:

\begin{theorem}
\label{th:bij1}
  There is a bijection between $\mathcal T_r(n)$ and
  $\mathcal M_{r,s,b}(n)$. 
\end{theorem}

\begin{proof}
  Consider the mapping $g$ that associates to an element of
  $\mc T_r(n)$, the extended mobile $M^+$ obtained by
  Theorem~\ref{th:unicellular}. By the above discussion the image of
  $g$ is in $\mathcal M_{r,s,b}(n)$ and $g$ is injective since one can
  recover the original triangulation from its image by
  Theorem~\ref{th:recover}.

  Conversely, given an element $M^+$ of $\mathcal M_{r,s,b}(n)$ with
  root angle $a_0$ (just after the root half-edge in \cw order around
  the root vertex), one can build a toroidal map $G$ by the complete
  closure procedure described in Section~\ref{sec:recover}. The number
  of stems and edges of $M^+$ implies that all faces of $G$ are
  triangles. We explain later why $G$ has no contractible loop nor
  multiple edges and that it is essentially 4-connected.

  While making the complete closure, one can create a 4-orientation
  $D$ of $A(G)$ with the following method.  For each half-edge $h$ of
  $M^+$ distinct from $h_0$, such that $h$ is incident to  vertex
  $v$, add to $D$ an outgoing half-edge incident to $v$ and just after
  $h$ in \cw order around $v$. Note that this is done not only for
  stems of $M^+$ but for all the half-edges of $M^+$, including those
  that are part of full-edges of $M^+$, except $h_0$.

  Consider the moment when an admissible triple $(e_1,e_2,s)$ of $M_k$
  is closed in order to obtain $M_{k+1}$, with $0\leq k\leq 2n-2$. Let
  $e_1=(u,v)$, $e_2=(v,w)$ and $s$ is a stem attached to $w$.  When
  the stem $s$ is reattached to $u$ to form a triangular face $T$ on
  its left side, it is reattached to $u$ in order to leave the
  half-edge of $D$ leaving $u$ (if any) on the right side (see
  Figure~\ref{fig:closestem}). Note that if the angle at $u$ is the
  root angle, then there is no half-edge of $D$ leaving $u$. By doing
  so we maintain the property that for all the angles of the face
  containing the root angle (called the special face in
  Section~\ref{sec:recover}), there is an outgoing half-edge of $D$,
  except for the root angle.

\begin{figure}[!ht]
\center
\includegraphics[scale=0.4]{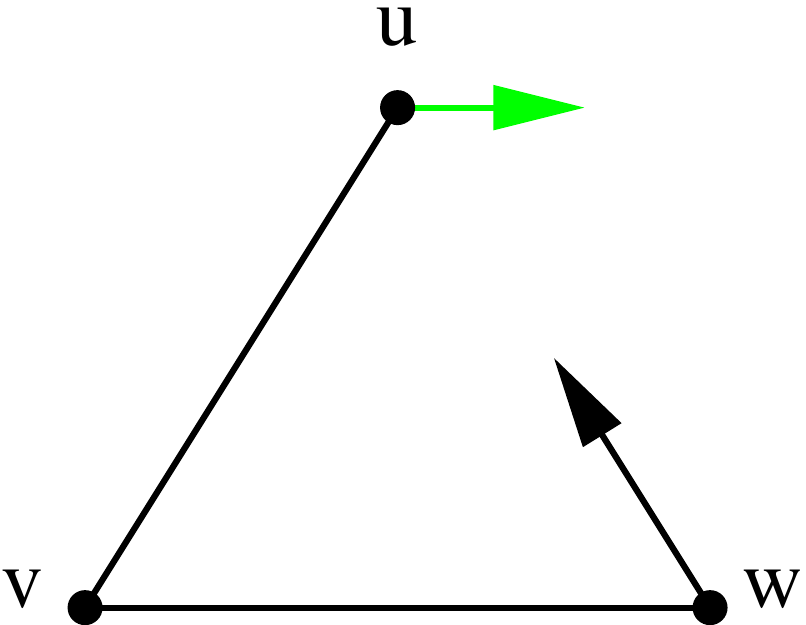} \ \ \ \ 
\includegraphics[scale=0.4]{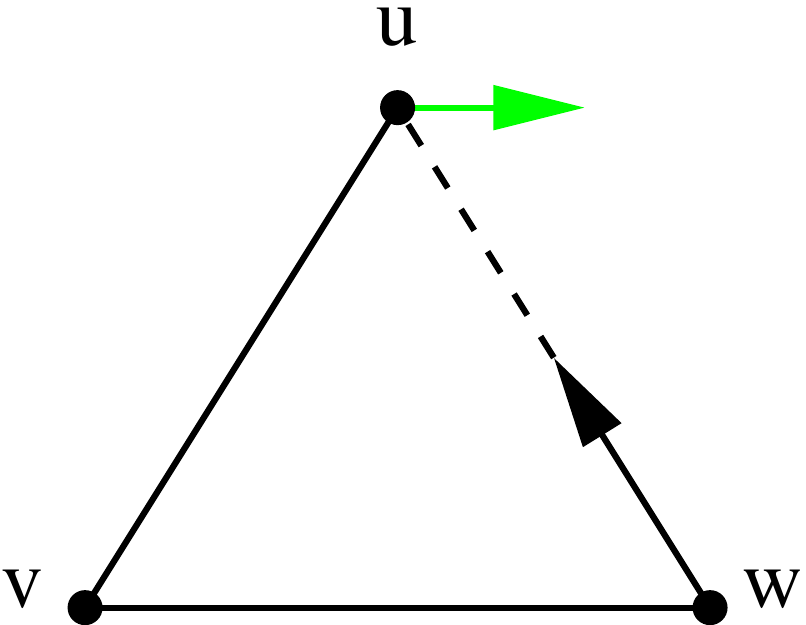}
\caption{Reattachment of a stem.}
\label{fig:closestem}
\end{figure}

In order to describe the edges of $D$ completely, we consider two
cases whether $s$ is the last stem that is reattached or not.

\begin{itemize}
\item \emph{$s$ is not the last reattached stem} 

By the  safe
  property, $s$ has the root angle on its right side when it is
  reattached. So the angle at $v$ (resp. $w$) between $e_1,e_2$
  (resp. $e_2,s$) in \cw order is not the root angle.  So inside the
  triangle $T$, we can reattached the two half-edges of $D$
  incident to $v,w$ to the dual-vertex $f$ of $A(G)$ corresponding to
  $T$ and add an additional edge to $D$ from $f$ to $u$ (see
  Figure~\ref{fig:closestemangle}).

\begin{figure}[!ht]
\center
\includegraphics[scale=0.4]{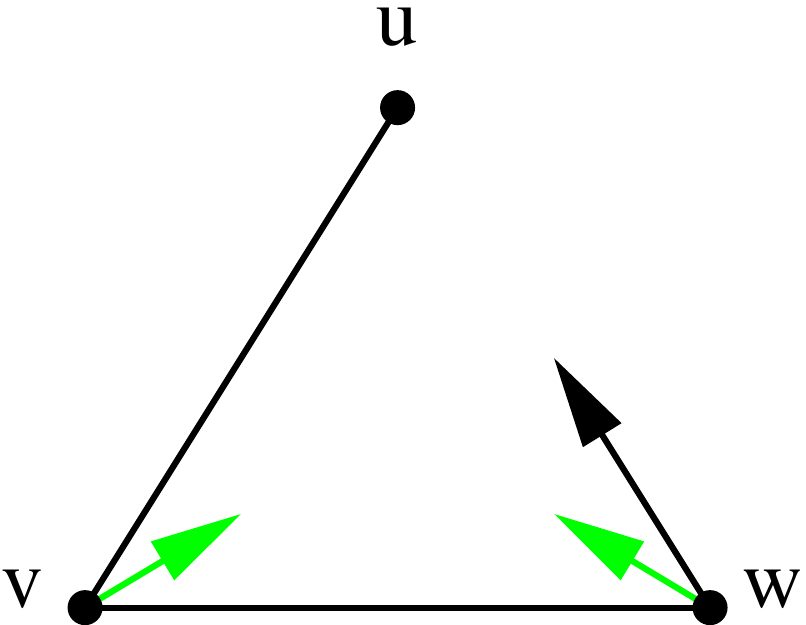} \ \ \ \ 
\includegraphics[scale=0.4]{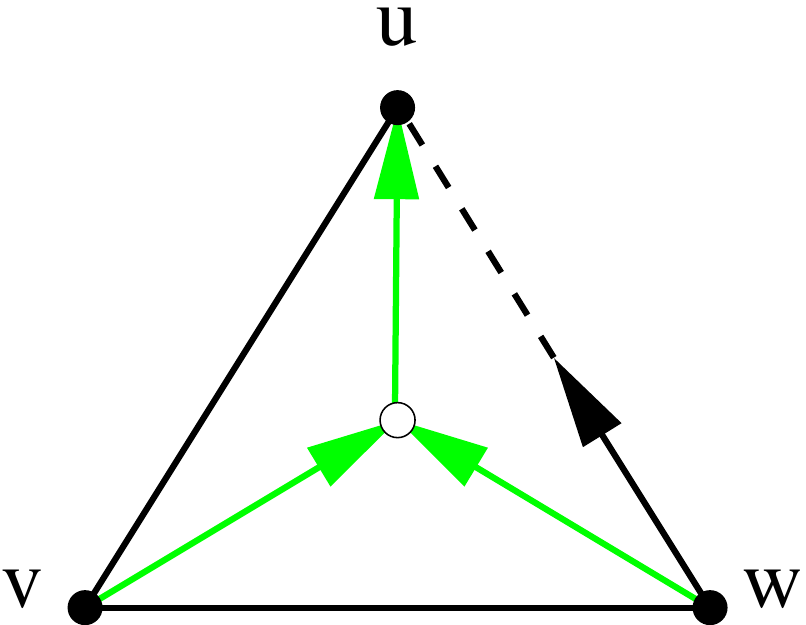}
\caption{Reattachment of a stem and orientation of the angle map.}
\label{fig:closestemangle}
\end{figure}

\item \emph{$s$ is the last reattached stem} 

  By the safe property, the root angle is in the face on the right
  side of $s$.  Thus we are in one of the three case of
  Figure~\ref{fig:finalcloseangle} depending on the position of the
  root half-edge according to $s$ (the root half-edge is represented
  in magenta). In each case we reattached the four depicted half-edges
  of $D$ and add two additional edges to $D$ that are outgoing for
  dual-vertices of $D$ as described on
  Figure~\ref{fig:finalcloseangle}.

\begin{figure}[!ht]
\center
\includegraphics[scale=0.4]{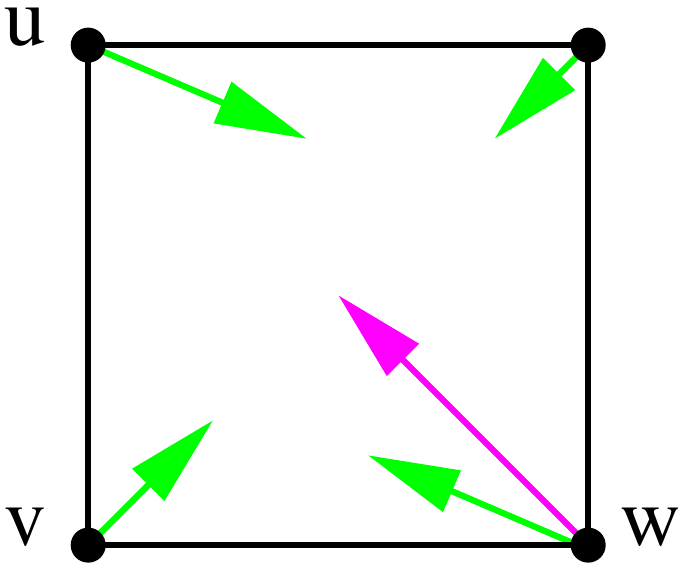} \ \ \ \ 
\includegraphics[scale=0.4]{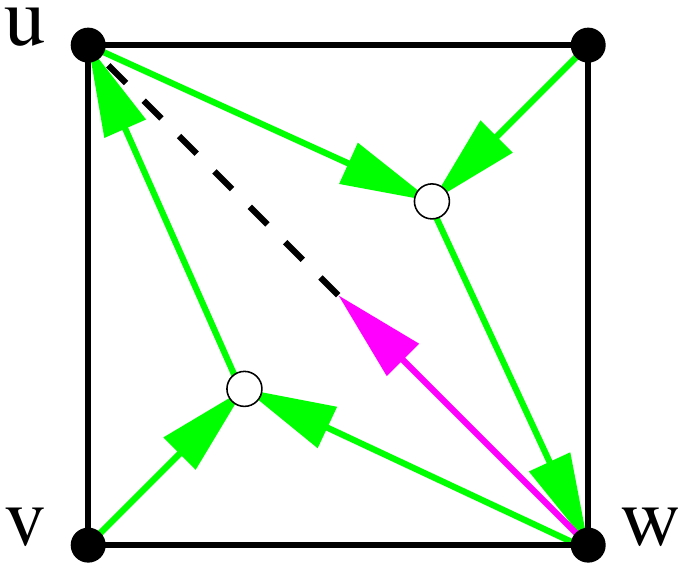} 

\

\includegraphics[scale=0.4]{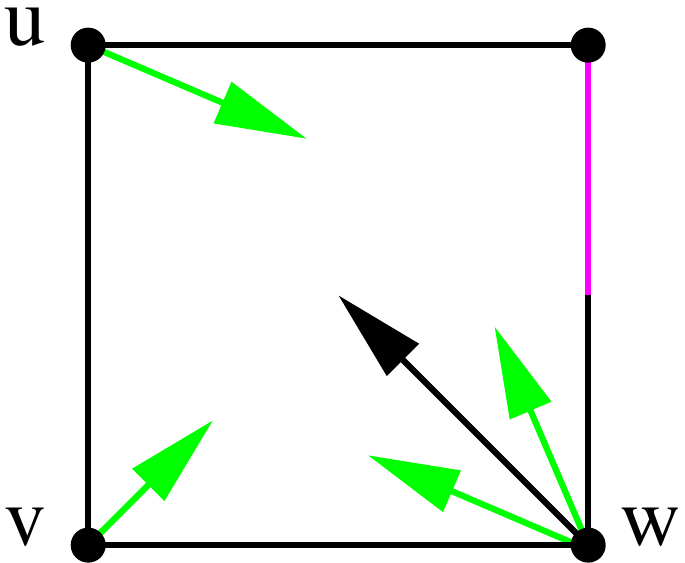} \ \ \ \ 
\includegraphics[scale=0.4]{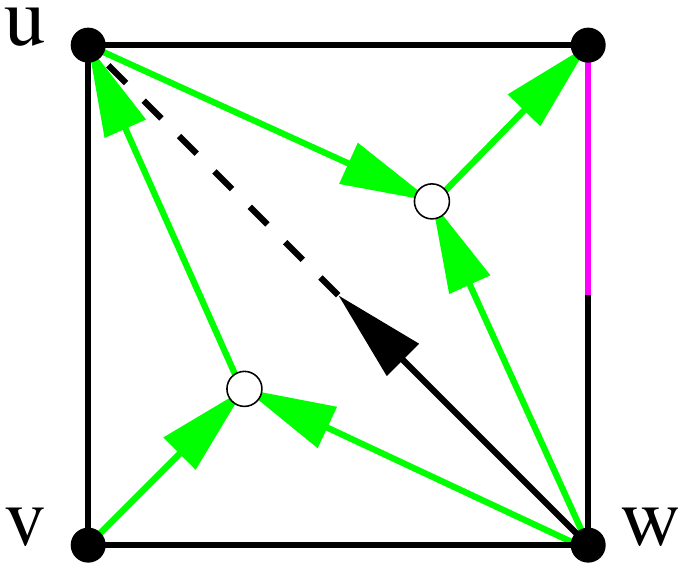} 

\

\includegraphics[scale=0.4]{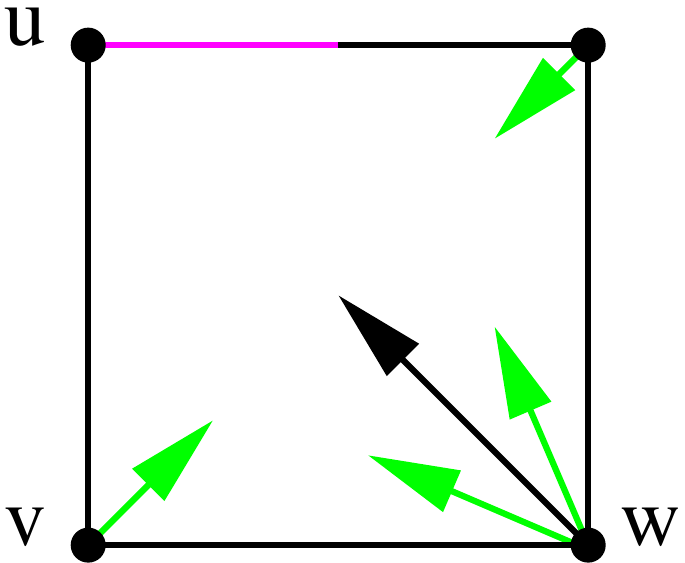} \ \ \ \ 
\includegraphics[scale=0.4]{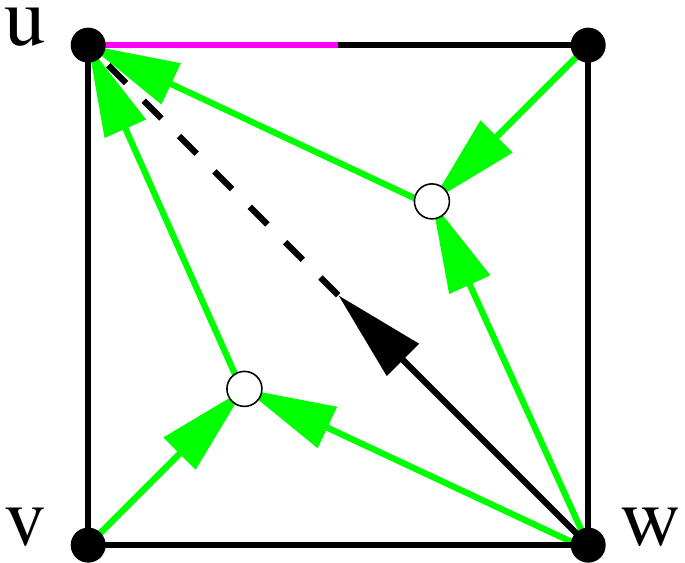}  
\caption{The three possible cases for the reattachment of the last stem.}
\label{fig:finalcloseangle}
\end{figure}

\end{itemize}

By doing so we are sure to reattach all the half-edges of $D$ to
dual-vertices of $A(G)$. In the end, all primal-vertices have
outdegree $4$ and all dual-vertices have outdegree $1$.  So we have
defined a $4$-orientation $D$ of $A(G)$ on which the mobile rule (see
Figure~\ref{fig:mobile-rule}) plus the addition of the root half-edge
gives $M^+$.  Since we are considering a 4-orientation of $A(G)$, the
map $G$ has no contractible loop nor multiple edges and it is
essentially 4-connected, otherwise, there will be a contradiction in a
region homeomorphic to an open disk by a simple counting argument.  It
remains to show that $G$ is appropriately rooted and that $D$
corresponds to the minimal balanced $4$-orientation w.r.t.~this root,
then $g$ will be surjective.

Since $M^+$ is balanced it has at least two non-contractible and not
weakly homologous cycles $C_1,C_2$ with the same number of angles on
their respective left and right sides, with the special rule that the
root angle does not count. All these angles corresponds to exactly one
outgoing edge of $D$ by construction of $D$. So the orientation $D$ of
$A(G)$ satisfies $\gamma (C_1) = \gamma (C_2)=0$.  So by
Lemma~\ref{lem:gamma0all}, the 4-orientation $D$ is balanced.

Suppose by contradiction that $D$ is not minimal w.r.t.~$h_0$. Let
$f_0$ be the face of $A(G)$ containing $h_0$.  We use the terminology
and notations of Section~\ref{sec:lattice}. Then in the Hasse diagram
of the lattice $(\mathcal B(A(G)),\leq_{f_0})$, there is an element
below $D$: Let $D'$ be a balanced $4$-orientation of $A(G)$
such that $D'\leq_{f_0}D$. By Section~\ref{sec:lattice}, we have
$D'\setminus D \in \widetilde{\mathcal{F}}'$.  Let
$\widetilde{F}=D'\setminus D$.  So $\widetilde{F}$ is the \ccw facial
walk of a face of $\widetilde{A(G)}$ not containing $f_0$. So this
facial walk is oriented \ccw (resp. \cww) according to its interior in
$D'$ (resp. $D$).  By Lemma~\ref{lem:contractibletilde},
$\widetilde{F}$ is quasi-contractible and its outer facial walk is a
$\{4,8\}$-disk $W$.  Then, by Lemma~\ref{lem:8diskin}, if $W$ is a
8-disk, the edges that are in the interior of $W$ and incident to it
are entering it.  So in $D$, the orientation of $W$ and of the edges
in its interior and incident to it are as depicted on
Figure~\ref{fig:48disk-mobile}.  Then, by definition of the mobile
(see Figure~\ref{fig:mobile-rule}), there is no half-edge of $M^+$ in
the interior of $W$ and incident to vertices of $W$.  If $W$ is a
$\{4\}$-disk, then the unique edge of $G$ inside $W$ is not covered by
$M^+$.  If $W$ is a $\{8\}$-disk, then either there are some edges of
$G$ inside $W$ that are not covered by $M^+$, or $M^+$ is made of
several connected components. In any cases, this contradicts the fact
that $M^+$ is an element of $\mathcal M_{r,s,b}(n)$ from which $G$ is
obtained by applying the complete closure procedure.  So $D$ is
minimal w.r.t.~$h_0$, and thus it is the minimal balanced
$4$-orientation w.r.t.~$h_0$

Suppose by contradiction that $h_0$ is not ``in the interior and
incident'' to a maximal quadrangle. Then by
Lemma~\ref{lem:uniqueQuadrangle}, there is a unique maximal quadrangle
$Q$ whose interior contains $h_0$. Since $h_0$ is not ``in the
interior and incident'' to $Q$, it is in the strict interior of $Q$.
The quadrangle $Q$ corresponds to a $\{4,8\}$-disk $W$ of $A(G)$ (see
Figure~\ref{fig:48disk-sq}). Note that $W$ is a maximal $\{4,8\}$-disk
containing $h_0$. So, by Lemma~\ref{lem:maxdiskroot}, in $D$, the
$\{4,8\}$-disk $W$ is oriented \cw w.r.t.~its interior.  It is not
possible that $W$ is a $4$-disk since then $h_0$ is not in the strict
interior of $Q$ but incident to it. So $W$ is a $8$-disk. By
Lemma~\ref{lem:8diskin}, the edges that are in the interior of $W$ and
incident to it are entering it.  Then the orientation of $W$ and of
the edges in its interior and incident to it are as depicted on
Figure~\ref{fig:48disk-mobile}.  Then by the definition of the mobile
(see rule of Figure~\ref{fig:mobile-rule}), there is no half-edge of $M^+$ in
the interior of $Q$ and incident to $Q$.  So either there are some
edges of $G$ that are not covered by $M^+$, or $M^+$ is made of
several connected components. In both cases, this contradicts the fact
that $M^+$ is an element of $\mathcal M_{r,s,b}(n)$ from which $G$ is
obtained by applying the complete closure procedure.  So $h_0$ is 
in the interior and incident to a maximal quadrangle of $G$.
\end{proof}

% \subsection{Direct recovering of the minimal transversal structure}

% \begin{theorem}
% \label{th:recoverTS}
% Consider an essentially 4-connected toroidal triangulation $G$, a
% root half-edge $h_0$ of $G$, incident to a vertex $v_0$, such that $h_0$ is not
% in the strict interior of a quadrangle, and the extended mobile $M^+$
% associated to the minimal balanced $4$-orientation $D_{\min}$ of $A(G)$
% w.r.t.~$h_0$.

%  Half-edges of $M^+$, except $h_0$, can be labeled
%     in a unique way with integers $0,1,2,3$ such that the labels that
%     appear around each vertex are exactly $0,1,2,3$ in \ccw order, the
%     two labels that appear on each edge differ exactly by $(2\bmod 4)$,
%     and $h_0$ is between half-edges labeled $0$ and $1$.

%     Moreover while closing admissible triples of $M^+$ \comment{(in
%       any order)} to recover $G$, one can propagate this labeling by
%     keeping the property that the two labels that appear on each edge
%     differs exactly by $2\bmod 4$. By doing so, one obtain the
%     TS-labeling of $G$ corresponding to $D_{\min}$.
% \end{theorem}

\section{Counting essentially 4-connected toroidal triangulations}
\label{sec:enumeration} 

Let $\mathcal{T}_{h}(n)$ be the set of essentially 4-connected
toroidal triangulations on $n$ vertices, rooted at any half-edge.  In
this section we show how to count $\mathcal{T}_{h}(n)$ (see
Theorem~\ref{thm:enumeration}).  The first values of
$|\mathcal{T}_{h}(n)|$, for $n\geq 0$, are
$0, 1, 6, 40, 268, 1801, 12120$ (sequence A289208 in
OEIS~\cite{oeis}).  Figure~\ref{fig:countingR-2} illustrates the six
elements of $\mathcal{T}_{h}(2)$.

\begin{figure}[!ht]
\center
\includegraphics[scale=0.4]{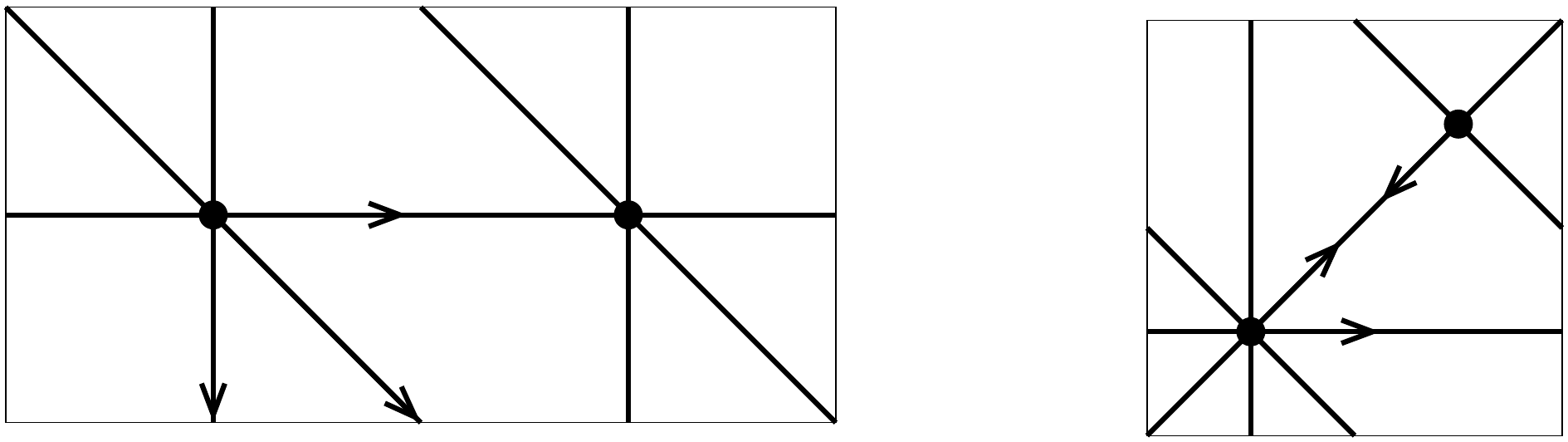}
\caption{The six elements of $\mathcal{T}_{h}(2)$: two different
  underlying graphs, each with three possible roots represented by an
  outgoing half-edge.}
\label{fig:countingR-2}
\end{figure}

\subsection{Decomposition into planar and toroidal parts}
\label{sec:decomp}

%In order to enumerate $\mathcal{T}_{h}(n)$, we need to consider
%intermediate structures. 

%It is not difficult to see that
%given a corner $c$ of a maximal quadrangle, there is a unique maximal
%quadrangle for which $c$ is a corner. Indeed, all quadrangles for
%which $c$ is a corner have their interior that are nested so there is
%a unique maximal one among this set, that is also the only maximal
%quadrangle whose corner is $c$.

Consider an element of $\mathcal{T}_{h}(n)$.  Recall that by
Lemma~\ref{lem:uniqueQuadrangle}, there is a unique maximal quadrangle
containing the root half-edge, that we call the root quadrangle.  We
define the \emph{corners} of a quadrangle of a map as the four angles
that appear in the interior of this quadrangle when its interior is
removed (if non empty).  We define
%$\mathcal{T}_{\alpha,r}(n)$ 
$\mathcal{T}_{h,c}(n)$ as the set of elements of $\mathcal{T}_{h}(n)$
with a marked corner of the root quadrangle. 
% Clearly there is a
%bijection between $\mathcal{T}_{h}(n)\times\{1,2,3,4\}$ and
%$\mathcal{T}_{h,c}(n)$.
%  By above remark 
%the marked corner is given and not
%necessarily the corresponding maximal quadrangle that can be retrieved
%afterwards.
%So we are reduced to the enumeration of $\mathcal{T}_{h,c}(n)$
The elements of $\mathcal{T}_{h,c}$ are decomposed into the toroidal
part that is outside the root quadrangle and the planar part that is
in the interior of the root quadrangle.

We first need the following lemma, which shows that
removing the interior of the root quadrangle does not change the
connectivity of the remaining part. Even if the statement has nothing
to do with transversal structures, the proof is using them as in
Section~\ref{sec:universal}.

\begin{lemma}
\label{lem:e4cquad}
If $G$ is an essentially 4-connected
toroidal triangulations given with  a
maximal quadrangle $Q$, then the map $G'$ obtained by removing all the
vertices and edges that lie in the interior of $Q$ is an essentially 4-connected
toroidal map.
\end{lemma}

\begin{proof} Let $G$ be an essentially 4-connected toroidal
  triangulations given with a maximal quadrangle $Q$ and $G'$ obtained
  by removing all the vertices and edges that lie in the interior of
  $Q$.

  Consider a root half-edge $h_0$ of $G$ that is in the interior and
  incident to $Q$. Consider the minimal balanced $4$-orientation
  $D_{\min}$ of $A(G)$ w.r.t.~$h_0$. By Corollary~\ref{cor:bal4orTS},
  this $4$-orientation corresponds to a transversal structure of $G$.
  Consider the $\{4,8\}$-disk $W$ of $A(G)$ that is inside the maximal
  quadrangle $Q$.  By Lemma~\ref{lem:maxdiskroot}, in $D_{\min}$, the
  $\{4,8\}$-disk $W$ is oriented \cw w.r.t.~its interior.
  Lemma~\ref{lem:8diskin} shows that all the edges of $A(G)$ that are
  in the interior of a $8$-disk of $A(G)$ and incident to it are
  entering it. So the transversal structure of $G$ represented on the
  maximal quadrangle $Q$ is as depicted on one of the three cases of
  Figure~\ref{fig:48disk-ts} (where the outer edges represent the
  quadrangle $Q$).

\begin{figure}[!ht]
\center
\includegraphics[scale=0.4]{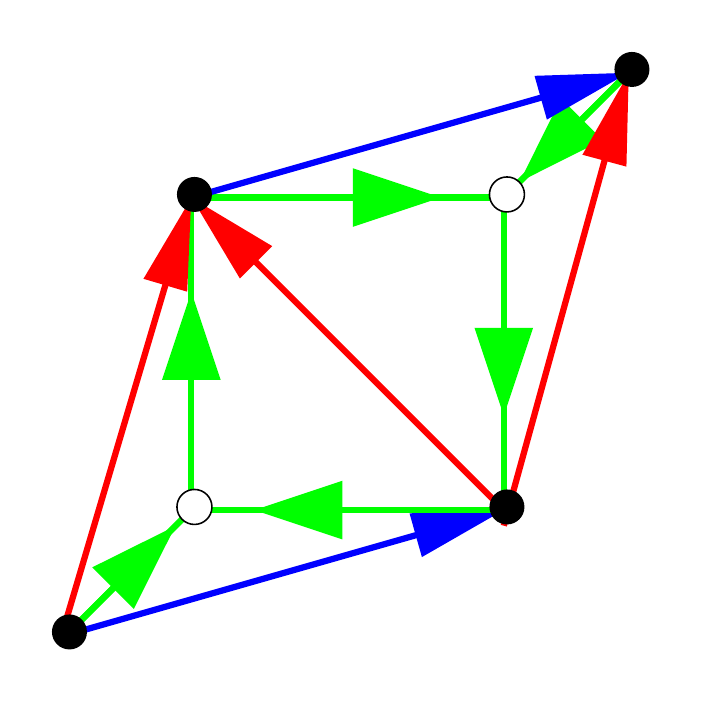} \ \ \ \ 
\includegraphics[scale=0.4]{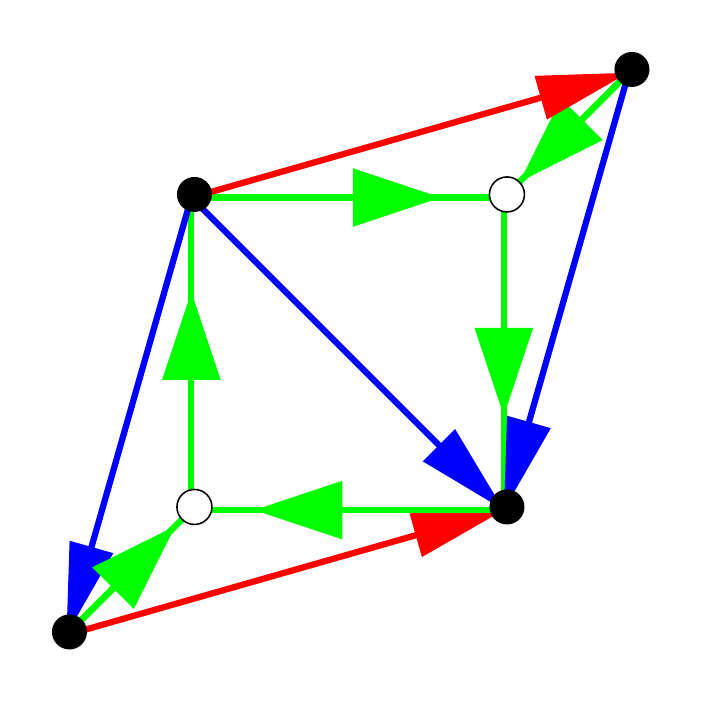} \ \ \ \ 
\includegraphics[scale=0.4]{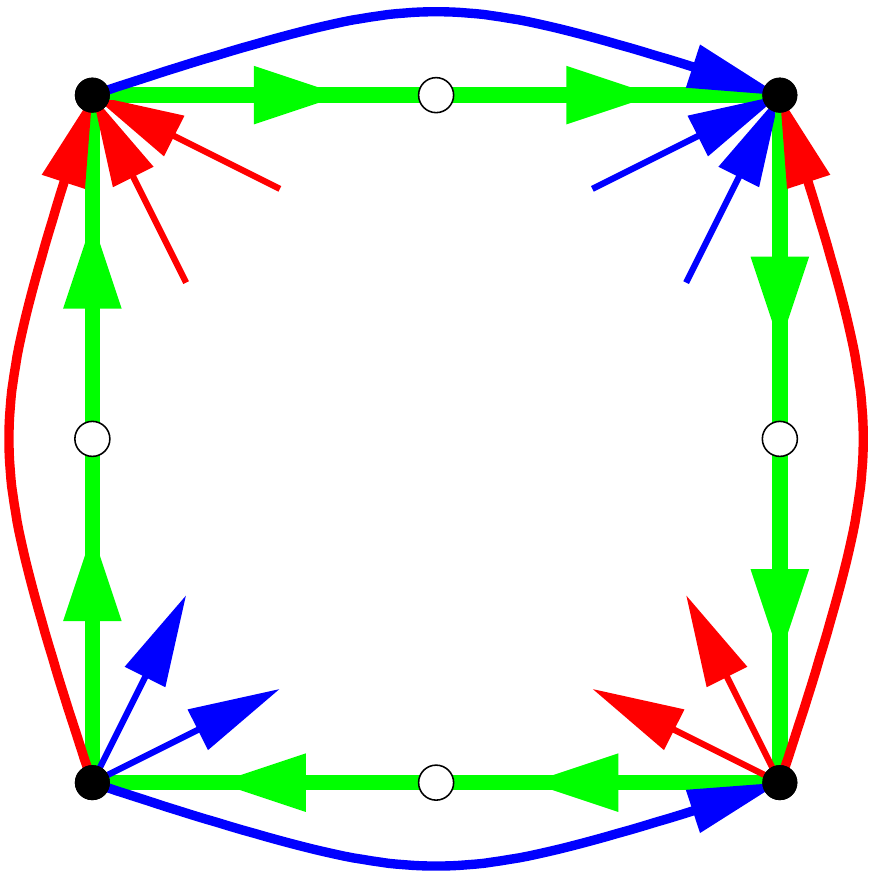}
\caption{Transversal structure on the maximal quadrangle $Q$.}
\label{fig:48disk-ts}
\end{figure}

Recall that from Section~\ref{sec:universal}, that for a vertex $v$ of
$G^\infty$, the subgraph $P_0(v)$ (resp. $P_1(v)$, $P_2(v)$, $P_3(v)$)
of $G^\infty$ is obtained by keeping all the edges that are on an
oriented path of $G^\infty_B$ (resp. $G^\infty_R$,
$(G^\infty_B)^{-1}$, $(G^\infty_R)^{-1}$) starting at $v$.  By
Lemma~\ref{lem:nodirectedcycle}, the subgraphs $P_i(v)$ are
acyclic. Let $P'_i(v)$ be defined similarly but in $G'^\infty$. So
$P'_i(v)$ is a subgraph of $P_i(v)$, thus it is also acyclic.  Note
that removing the interior of the quadrangle $Q$ on the three cases of
Figure~\ref{fig:48disk-ts}, does not change the fact that around every
vertex there are edges that are outgoing blue, outgoing red, incoming
blue and incoming red.  So the $P'_i(v)$ are infinite.

As in the proof of Lemma~\ref{lem:e4c}, suppose by contradiction that
there exists three vertices $x,y,z$ of $G'^\infty$ such that
$G''=G'^\infty\setminus\{x,y,z\}$ is not connected. Then, by
Lemma~\ref{lem:finitecc}, the graph $G''$ has a finite connected
component $R$. Let $v$ be a vertex of $R$. For $i\in\{0,1,2,3\}$, the
infinite and acyclic graph $P'_{i}(v)$ does not lie in $R$ so it
intersects one of $x,y,z$. So for two distinct $i,j$, the two graphs
$P'_{i}(v)$ and $P'_{j}(v)$ intersect in a vertex distinct from $v$.
Thus the two graphs $P_{i}(v)$ and $P_{j}(v)$ intersect in a vertex
distinct from $v$, a contradiction to Lemma~\ref{lem:nocommongeneral}.
\end{proof}

Let $\mathcal T_c^t(n)$ be the set of essentially 4-connected toroidal
maps on $n$ vertices, where all faces are triangles, except one that
is a maximal quadrangle, and, with a marked corner of this quadrangle.
In particular we have $|\mathcal T_c^t(0)|=0$ and
$|\mathcal T_c^t(1)|=1$.  

Let $\mathcal{T}_{h,c}^p(n)$ be the set of
4-connected planar maps on $n$ inner vertices, where all faces are
triangles, except the outer-face that is a quadrangle, with a marked
corner of this quadrangle, and rooted at an
inner half-edge. Note
that here, $n$ counts the number of inner vertices, so there are $n+4$
vertices in an element of $\mathcal{T}_{h,c}^p(n)$. In particular we
have $|\mathcal T_{h,c}^p(0)|=4$ and $|\mathcal
T_{h,c}^p(1)|=8$.

%By Lemma~\ref{lem:e4cquad} there is a bijection between
%$\mathcal{T}_{c}(n)$ and $\bigcup_{1\leq k\leq
%  n}(\mathcal{T}_{c}^p(n-k) \times \mathcal T_{c}^t(k))$. 
Then we
have the following bijection:

\begin{lemma}
\label{lem:bijrootcorner}
There is a bijection between $\mathcal{T}_{h}(n)\times\{1,2,3,4\}$ and 
$\bigcup_{1\leq k\leq n}(\mathcal{T}_{h,c}^p(n-k) \times \mathcal
T_c^t(k))$.
\end{lemma}

\begin{proof}
By Lemma~\ref{lem:uniqueQuadrangle}, there is a bijection between $\mathcal{T}_{h}(n)\times\{1,2,3,4\}$  and $\mathcal{T}_{h,c}(n)$. 
By Lemma~\ref{lem:e4cquad}, there is a bijection between $\mathcal{T}_{h,c}(n)$ and $\bigcup_{1\leq k\leq n}(\mathcal{T}_{h,c}^p(n-k) \times \mathcal T_c^t(k))$. The composition of these two bijections gives the result.
\end{proof}

By Lemma~\ref{lem:bijrootcorner}, the enumeration of the elements of
$\mathcal{T}_{h}$ is reduced to the enumeration of their planar part
$\mathcal{T}_{h,c}^p$  and their toroidal part $\mathcal{T}_c^t$.

Recall from Section~\ref{sec:coding}, that a {ternary tree} is a plane
tree, rooted at a leaf, such that every inner vertex has degree
exactly four.  For $n\geq 1$, let $\mathcal A(n)$ denote the set of
ternary trees with $n$ inner vertices.  By convention we consider that
the tree composed of a single vertex is the unique element of
$\mathcal A(0)$.  The associated generating function satisfies:
\begin{equation} \label{eq:az3}
A(z)=\sum_n |\mathcal A(n)|z^n=1+zA(z)^3.
\end{equation}

% We deduce directly from~\cite{Fus09}[Theorem 4 and Proposition 5] that
% $\mathcal{T}_c^p$ satisfies the  following:
% \begin{lemma}[\cite{Fus09}]
% \label{lem:planarpart} 
% $$|\mathcal{T}_c^p(n)| = 4\frac{(3n)!}{n! (2n+2)!}$$
% and the associate generating function satisfies:
% $$T_c^p(z) = \sum_{n\geq 0} \mathcal{T}_c^p(n) z^n= 3A(z)-A(z)^2$$
% \end{lemma}

The enumeration of $\mathcal{T}_{h,c}^p$ is given by the following lemma:

\begin{lemma}
\label{lem:planarpart2} 
$$|\mathcal{T}_{h,c}^p(n)| = \frac4{n+1}\binom{3n+1}{n}$$
and the associate generating function satisfies:
$$T_{h,c}^p(z) = \sum_{n\geq 0} |\mathcal{T}_{h,c}^p(n)| z^n= 4A(z)^2.$$
\end{lemma}
\begin{proof}
%  Fusy obtain in~
  \cite[Theorem~3]{Fus09} is a bijection between the set  $\mathcal{P}(n)$ of
  (unrooted) plane tree such that every inner vertex has degree
  exactly four and the set $\mathcal{T}^p(n)$ of (unrooted)
  4-connected planar maps on $n$ inner vertices, where all faces are
  triangles, except the outer-face that is a quadrangle.

  Let $\mathcal{T}_{h}^p(n)$ be the set of elements of
  $\mathcal{T}^p(n)$ rooted at an inner half-edge.  
Let $\mathcal{P}_o(n)$ be the elements of
  $\mathcal{P}(n)$ with one oriented edge.
In the bijection
  of~\cite[Theorem~3]{Fus09}, each inner edge of the map corresponds
  to an edge of the corresponding tree. So rooting the elements of
  $\mathcal{T}^p(n)$ on a particular inner half-edge, corresponds to
  orienting an edge of the tree.  Thus we have a bijection between
  $\mathcal{T}_{h}^p(n)$ and $\mathcal{P}_o(n)$.

  Cutting an element of $\mathcal{P}_o(n)$ at the oriented edge creates
  bijectively a couple of  ternary trees with respectively $k$
  and $n-k$ inner vertices, with $0 \leq k \leq n$. Hence:
  $$P_o(z)=\sum_n |\mathcal P_o (n)| = A(z)^2.$$

  As shown
  in~\cite[p~11]{GX06}, the coefficients of $A(z)^2$ admit a simple
  expression:
$$|\mathcal P_o (n)|= \frac1{n+1}\binom{3n+1}{n}.$$
%By the mapping $g$, 

An element of $\mathcal{T}_{h,c}^p(n)$ is obtained from an element of
$\mathcal{T}_{h}^p(n)$ by marking one corner of the outer face. There
are four such choice so
$\mathcal{T}_{h,c}^p(n)=4\mathcal{T}_{h}^p(n)=4\mathcal P_o (n)$ and we
obtain the lemma.
\end{proof}

\subsection{Bijection with square and
  hexagonal unicellular maps}

Given an element $M^+$ of $\mathcal M_{r,s,b}(n)$ (see
Section~\ref{sec:bijection} for the definition), we define the \emph{unrooted mobile $M$
  associated to $M^+$} as the toroidal unicellular map obtained from
$M^+$ by removing the root half-edge and the tree part attached to the
root half-edge (if any).  Figure~\ref{fig:disk8mobilenorootbij2} gives
the unrooted mobile associated to the extended mobile of
Figure~\ref{fig:disk8mobile}.

\begin{figure}[!ht]
\center
\includegraphics[scale=0.34]{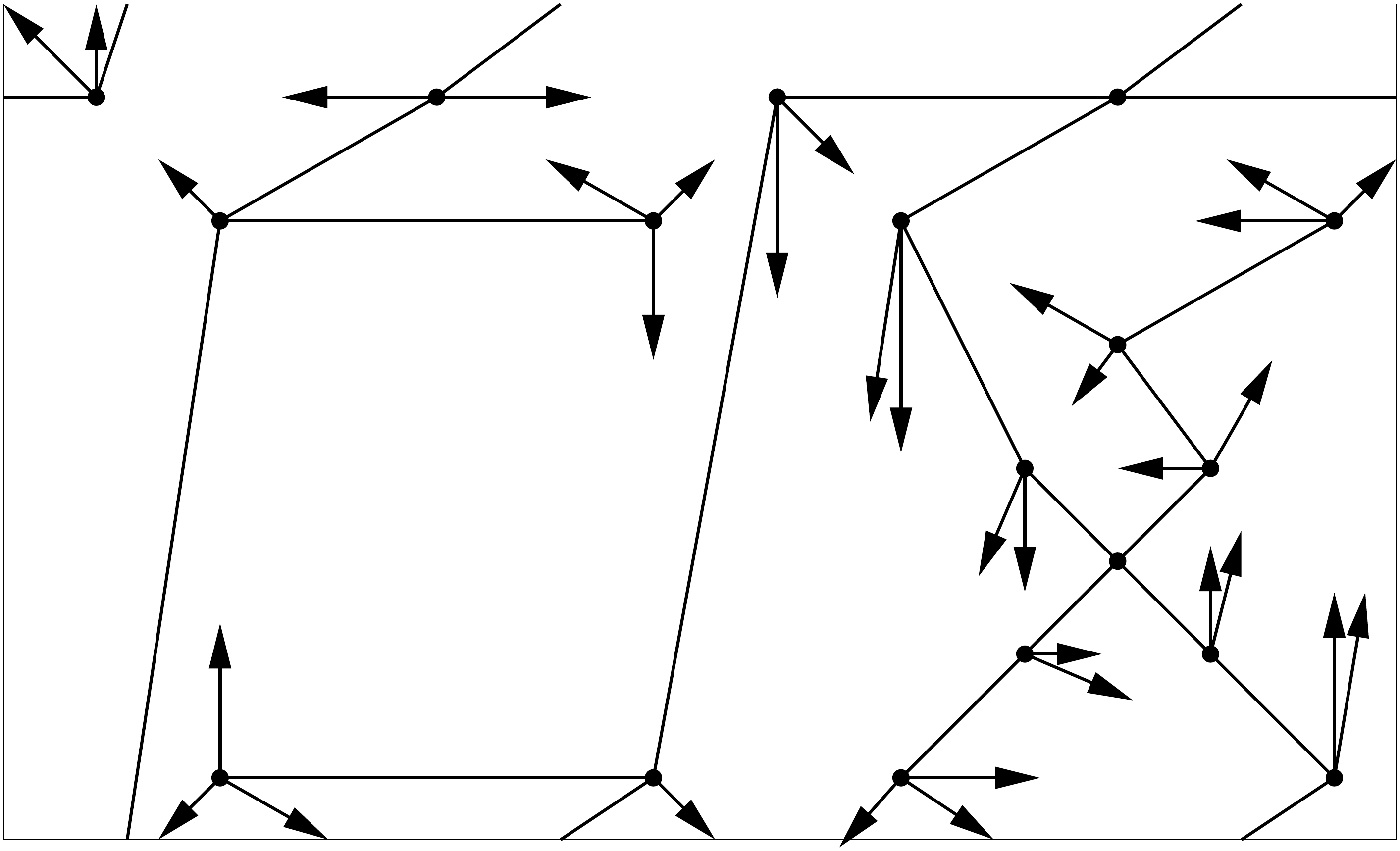}
\caption{The unrooted mobile obtained from the extended mobile of
  Figure~\ref{fig:disk8mobile}.}
\label{fig:disk8mobilenorootbij2}
\end{figure}

Let $\mathcal M_b(n)$ denote the set of (non-rooted) toroidal
unicellular maps with exactly $n$ vertices, $n+1$ edges and $2n-2$
stems such that all vertices have degree $4$, and every cycle of the
map has the same number of angles on its left and right sides
(balanced property).

Consider an element $M^+$ of $\mathcal M_{r,s,b}(n)$. Let $k\geq 0$ be
the number of vertices in the tree part attached to the root half-edge
(if any), with $k=0$ if the root half-edge is a stem.  Then one can
see that the unrooted mobile $M$ associated to $M^+$ is an element of
$\mathcal M_b(n-k)$.

Consider an element $M$ of $\mathcal M_{b}(n)$. A
\emph{mobile-labeling} of $M$ is a labeling $\ell$ of the half-edges
of $M$ with integers $0,1,2,3$ such that the labels that appear around
each vertex are exactly $0,1,2,3$ in \ccw order and the two labels
that appear on each edge differ exactly by $(2\bmod 4)$, see
right of Figure~\ref{fig:labelmobile} for an example.  Let $G$ be
the graph obtained from $M$ by closing all its admissible
triples. Since, $M$ has $2n-2$ stems, we have that $G$ is a ``toroidal
triangulation minus one edge'', i.e. a toroidal map whose all faces
are triangles except one that is a quadrangle. The \emph{extension} of
$\ell$ to $G$ is the labeling of all the half-edges of $G$ obtained
from $\ell$ by keeping the property that the two labels that appear on
each edge differs exactly by $2\bmod 4$. Next lemma shows that the
quadrangle of $G$ is labeled as on Figure~\ref{fig:labelquadrangle} in
the extension of the mobile-labeling.

\begin{figure}[!ht]
\center
\includegraphics[scale=0.4]{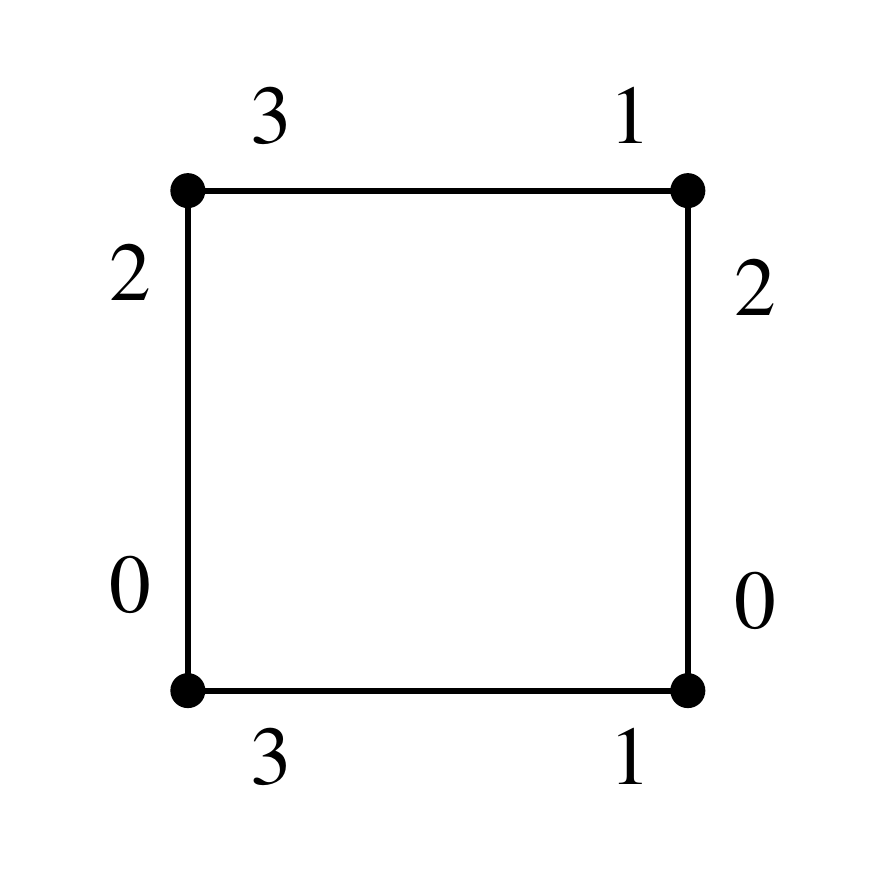}
\caption{Labeling of the remaining quadrangle after extending a mobile-labeling.}
\label{fig:labelquadrangle}
\end{figure}

Each angle of $M$ corresponds to consecutive angles of $G$
(reattaching a stem into an angle, splits this angle in
two). Conversely, to each angle of $G$ we can associate the unique
corresponding angle of $M$ from which it comes from.  Then we have the
following:

  \begin{lemma}
    \label{lem:anglemobile}
    Consider an element $M$ of $\mathcal M_{b}(n)$ given with a
    particular angle $\alpha$. Then $M$ admits a unique
    mobile-labeling, noted $\ell(\alpha)$, such that the angle
    $\alpha$ is between half-edges labeled $0$ and $1$. Moreover,
    after closing the admissible triples of $M$ to obtain $G$, the
    extension of $\ell(\alpha)$ to $G$ is such that the 
    quadrangle $Q$ of $G$ is labeled as on Figure~\ref{fig:labelquadrangle}.
    And the four angles of $M$ that corresponds to the angles of
    $Q$ are incident to half-edges with exactly the same labels in $M$
    and in $Q$.
  \end{lemma}

  \begin{proof}
The toroidal unicellular map  $M$ has $n$ vertices, $n+1$ edges and $2n-2$
stems such that all vertices have degree $4$, and every cycle of the
map has the same number of angles on its left and right sides.

Let $h$ be the half-edge of $M$ that is incident to $\alpha$ and just
after $\alpha$ in \cw order around its incident vertex. Let
$\ell(\alpha)$ be the labeling of the half-edges of $M$ with integers
$0,1,2,3$ obtained by the following: Label $h$ with $0$ and then
extend the labeling to all the half-edges of $M$ by keeping the
property that the labels that appear around each vertex are exactly
$0,1,2,3$ in \ccw order and the two labels that appear on each edge
differ exactly by $(2\bmod 4)$. This is possible and consistent since
every cycle of the map has the same number of angles on its left and
right sides. Indeed, given a cycle $C$ of length $k$, there are $2k$
angles on each sides, so the modification of the labels while starting
from a half-edge of $C$, walking along $C$ and going back to the
starting half-edge is the following: the number of edges of $C$ times
$(2\bmod 4)$, i.e. $(2k \bmod 4)$, plus the number of angles on the
right side of $C$, i.e. $(2k \bmod 4)$, so $(4k \bmod 4)=0$ in
total. Thus, this definition of the mobile-labeling $\ell(\alpha)$ is
consistent and moreover it is the unique such labeling. So we have the
first part of the lemma.

Let $M_0=M$.  For $1\leq k \leq 2n-2$, let $M_{k}$ be the map obtained
from $M_{k-1}$ by closing an admissible triple of $M_{k-1}$.  Extend
the labeling $\ell(\alpha)$ while closing admissible triples by
keeping the property that the two labels that appear on each edge
differs exactly by $2\bmod 4$.  We prove by induction on $k$, that
each map $M_{k}$, for $0\leq k \leq 2n-2$, satisfies the following:
each angle of the special face of $M_k$ is between half-edges whose
labels are distinct (and precisely the same as for the corresponding
angle of $M_{k-1}$ if $k\geq 1$), moreover the labels that appear in
\ccw order around each vertex of $M_k$ form four non-empty intervals
of $0,1,2,3$. Indeed, $M_0$ satisfies the property and suppose that
for $1\leq k \leq 2n-2$, we have $M_{k-1}$ that satisfies the
property.  Consider the admissible triple $(e_1,e_2,s)$ of $M_{k-1}$
that is closed to obtain $M_k$. Let $e_1=\{u,v\}$ and $e_2=\{v,w\}$ with
$s$ is a stem attached to $w$.  Let $i\in\{0,1,2,3\}$ such that $s$ is
labeled $i$. Then since $M_{k-1}$ satisfies the property on the labels
we have that:
\begin{itemize}
\item 
the half-edge of
$e_2$ incident to $w$ is labeled $(i+1) \bmod 4$
\item  the half-edge of
$e_2$ incident to $v$ is labeled $(i+3) \bmod 4$
\item the half-edge of
$e_1$ incident to $v$ is labeled $i$ 
\item the half-edge $h_u$ of $e_1$ incident
to $u$ is labeled $(i+2) \bmod 4$
\item  the half-edge $h'_u$ incident to $u$
and just after $h_u$ in \ccw order around $u$ is labeled $(i+3) \bmod 4$
\end{itemize}
When the admissible triple is closed, a half-edge $h_s$, opposite to
$s$ is created and receive the label $(i+2) \bmod 4$. So the
half-edges $h_u,h_s,h'_u$ appear consecutively in \ccw order around
$u$. Moreover they are labeled $(i+2) \bmod 4$, $(i+2) \bmod 4$ and
$(i+3) \bmod 4$ respectively. So  all the induction properties are
preserved.  Finally, $M_{2n-2}$ satisfies the property and its special
face, that is a quadrangle, is labeled as on
Figure~\ref{fig:labelquadrangle} and the four angles of $M$ that
corresponds to the angles of $Q$ are incident to half-edges with
exactly the same labels in $M$ and in $Q$.
  \end{proof}

  Recall that there are two types of toroidal unicellular maps.  Two
  distinct cycles of a toroidal unicellular map may intersect either
  on a single vertex (square case) or on a path (hexagon case).  In a
  square (resp. hexagon) unicellular map, there are exactly $2$
  (resp. $3$) distinct cycles. A vertex of a toroidal unicellular map
  is called \emph{special} if is contained in all the cycles of the
  map. Note that there is exactly one special vertex in a square
  unicellular map, and exactly two special vertices in a hexagon
  unicellular map.

Let $\mathcal M^s_{b}(n)$ (resp. $\mathcal M^h_{b}(n)$) denote the set
of elements of $\mathcal M_{b}(n)$ that are square (resp. hexagon)
unicellular maps. Moreover we denote the sets
$\mathcal M_{b,a}(n),\mathcal M^s_{b,a}(n), \mathcal
M^h_{b,a}(n)$
 the sets of elements of
$M_{b}(n),\mathcal M^s_{b}(n), \mathcal M^h_{b}(n)$, respectively,
that are rooted at an angle of a special vertex.

We have the following bijection:

\begin{lemma}
\label{lem:bijsquarehexa}
  There is a bijection between $\mathcal T_c^t(n) \times \{1,2\}$ and
  $\mathcal (M^s_{b,a}(n)\times\{1,2\}) \cup (\mathcal M^h_{b,a}(n)\times\{0\})$. 
\end{lemma}

\begin{proof}
  We define a bijective function from
  $\mathcal T_c^t(n) \times \{1,2\}$ to
  $(\mathcal M^s_{b,a}(n)\times\{1,2\}) \cup (\mathcal
  M^h_{b,a}(n)\times\{0\})$.
  This function is defined via three intermediate functions $a$, $g'$ and
  $r$ defined below.

  Let $a$ (for ``add'') be the mapping defined on the elements $G$ of
  $\mathcal T_c^t(n)$ that adds to $G$ a diagonal $e_0$ in the
  interior of the (maximal) quadrangle $Q$ of $G$, incident to
  the marked corner $\alpha$ of $Q$, and returns the obtained map $Z$ rooted at
  the half-edge $h_0$ of $e_0$ incident to $\alpha$.  Let
  $\mc T'_r(n)$ be the subset of $\mc T_r(n)$ (see
  definition in Section~\ref{sec:bijection}) such that the two faces incident to
  the root half-edge form a maximal quadrangle. We claim the following:

  \begin{claim}
\label{cl:bija}
    $a$ is a bijection from $\mathcal T_c^t(n)$ to $\mc T'_r(n)$.
  \end{claim}
  \begin{proofclaim}
    Let $G$ be an element of $\mathcal T_c^t(n)$ and $Z$ its
    image by $a$. Consider the notations of the definition of $a$.  Since
    $G$ is essentially 4-connected, the added edge $e_0$ cannot
    create a contractible loop in $Z$.  If adding $e_0$ creates a pair
    of homotopic multiple edges in $Z$ with an edge $e'_0$, then there
    are two edges of the quadrangle $Q$ of $G$ plus edge $e'_0$ that form
    a separating triangle of $G^\infty$ contradicting the
    $4$-connectedness of $G^\infty$.  So the obtained map $Z$ is a
    toroidal triangulation with no contractible loop nor homotopic
    multiple edges.  Moreover since $Q$ is a maximal quadrangle, the
    edge $e_0$ cannot create a separating triangle of $Z^\infty$.  So
    by Lemma~\ref{lem:e4ciffnst}, the toroidal triangulation $Z$ is
    essentially $4$-connected.  Moreover $Z$ has the particularity
    that the two faces incident to the root half-edge $h_0$ form a
    maximal quadrangle of $Z$.  So $Z$ is is in $\mc T'_r(n)$.

    Let $\overline{a}$ be the mapping defined on the elements $Z$ of
    $\mc T'_r(n)$ that removes from $Z$ the edge containing the root
    half-edge and mark the obtained quadrangle at the corner incident
    to $h_0$. Then clearly $\overline{a}\circ a = Id$.

    Conversely, let $Z$ be an element of $\mc T'_r(n)$ and consider
    its image $G$ by $\overline{a}$.  We have $Z$ is an essentially
    4-connected toroidal triangulation on $n$ vertices rooted at a
    half-edge that is in the interior and incident to a maximal quadrangle, and
    such that the two faces incident to the root half-edge form a
    maximal quadrangle.  So $G$ is a toroidal map on $n$ vertices,
    where all faces are triangles, except one that is a maximal
    quadrangle and with a marked corner of this quadrangle.  The map $G$
    is obtained from $Z$ by removing the interior of a maximal
    quadrangle so, by Lemma~\ref{lem:e4cquad}, $G$ is essentially
    4-connected.  So $G$ is in $\mathcal T_c^t(n)$.

    We clearly have $a \circ \overline{a} = Id$.  So $a$ is a
    bijection from $\mathcal T_c^t(n)$ to $\mc T'_r(n)$.
  \end{proofclaim}

  Let $g'$ be the restriction of the bijection $g$, defined in the
  proof of Theorem~\ref{th:bij1}, to the elements of $\mc T'_r(n)$.
  Let $\mathcal M'_{r,s,b}(n)$ be the subset of
  $\mathcal M_{r,s,b}(n)$ (see definition in
  Section~\ref{sec:bijection}) such that the root half-edge is a
  stem. We claim the following:

 \begin{claim}
\label{cl:bijh}
    $g'$ is a bijection from $\mc T'_r(n)$ to $\mathcal M'_{r,s,b}(n)$.
  \end{claim}
  \begin{proofclaim}
    Let $Z$ be an element  of $\mc T'_r(n)$ and
    $M^+\in \mathcal M_{r,s,b}(n)$ its image by $g$.  By definition of
    $\mc T'_r(n)$, the two faces incident to the root half-edge $h_0$ form
    a maximal quadrangle $Q$. In the minimal balanced
    $4$-orientation of $A(Z)$ w.r.t.~$h_0$, the $4$-disk $W$ inside
    $Q$ is oriented \cw by Lemma~\ref{lem:maxdiskroot} (see left of
    Figure~\ref{fig:48disk-mobile}).  Then by the definition of the
    mobile (see rule of Figure~\ref{fig:mobile-rule}), there is no half-edge
    of $M^+$ in the interior of $Q$ except $h_0$. So $h_0$ is a stem
    of $M^+$. So $M^+$ is in $\mathcal M'_{r,s,b}(n)$.

    Let $\overline{g'}$ be the restriction of $g^{-1}$ to the elements
    of $\mathcal M'_{r,s,b}(n)$. Since $g$ is a bijection,
    we have $\overline{g'}\circ g'= Id$.

    Conversely, let $M^+$ be an element of $\mathcal M'_{r,s,b}(n)$
    and consider its image $Z$ by $g^{-1}$.  By the proof of
    Theorem~\ref{th:bij1}, the complete closure procedure on $M^+$
    gives an essentially 4-connected toroidal triangulation $Z$ of
    $\mc T_r(n)$ rooted at $h_0$ and such that $h_0$ is in the
    interior and incident to a maximal quadrangle $Q'$.  Moreover,
    $M^+$ is the extended mobile associated to the minimal balanced
    4-orientation $D_{\min}$ of $A(Z)$ w.r.t.~$h_0$.  The quadrangle $Q'$
    corresponds to a $\{4,8\}$-disk $W$ of $A(Z)$ (see
    Figure~\ref{fig:48disk-sq}). Note that $W$ is a maximal
    $\{4,8\}$-disk containing $h_0$. So, by
    Lemma~\ref{lem:maxdiskroot}, $W$ is oriented \cw w.r.t.~its
    interior in $D_{\min}$.  Then, by Lemma~\ref{lem:8diskin}, the
    orientation of $W$ and of the edges in its interior and incident
    to it are as depicted on Figure~\ref{fig:48disk-mobile}. Then by
    the definition of the mobile (see rule of Figure~\ref{fig:mobile-rule}),
    there is no half-edge of $M^+$, distinct from $h_0$, in the
    interior of $Q'$ and incident to $Q'$.  Since $M^+$ is covering
    all the edges of $Z$, we have that $Q'$ has no edges in its
    interior, except the one containing $h_0$.  So $Q'$ is the
    quadrangle formed by the two faces incident to $h_0$. So
    the two faces incident to $h_0$ form a maximal quadrangle and so $Z$
    is an element of $\mc T'_r(n)$.

    Since $g$ is a bijection, we have $g' \circ \overline{g'}= Id$. So
    $g'$ is a bijection from $\mc T'_r(n)$ to
    $\mathcal M'_{r,s,b}(n)$.
 \end{proofclaim}

 Let $r$ (for ``remove'') be the mapping defined on the elements
 $(M^+,x)$ of $\mathcal M'_{r,s,b}(n)\times \{1,2\}$, that removes the
 root half-edge $h_0$ of $M^+$ that is a stem and roots the obtained
 mobile $M$ at an angle of a special vertex by the following rule. Let
 $\alpha$ be the angle of $M$ such that $M^+$ is obtained from $M$ by
 adding $h_0$ in the angle $\alpha$.  Consider the unique
 mobile-labeling $\ell(\alpha)$ of $M$, given by
 Lemma~\ref{lem:anglemobile}.  If $M$ is square, let $\beta$ be the
 angle of the special vertex of $M$ that is between the half-edges
 labeled $0$ and $1$.  In this case, $r$ returns $(M,x)$ rooted at
 $\beta$. If $M$ is hexagon, let $v_1$ (resp. $v_2$) be the first
 (resp. second) special vertex of $M$ that is encountered while
 walking \ccw along the border of the unique face of $M$, starting
 from $\alpha$. For $i\in\{1,2\}$, let $\beta_i$ be the angle of $v_i$
 that is between half-edges labeled $0$ and $1$. In this case, $r$
 returns $(M,0)$ rooted at $\beta_x$. We claim the following:

 \begin{claim}
\label{cl:bijr}
    $r$ is a bijection from $\mathcal M'_{r,s,b}(n)\times \{1,2\}$ to 
$(\mathcal M^s_{b,a}(n)\times\{1,2\}) \cup (\mathcal
 M^h_{b,a}(n)\times\{0\})$.
  \end{claim}
  \begin{proofclaim}
    It is clear by the definition of $r$ that the image by $r$ of an
    element of $\mathcal M'_{r,s,b}(n)\times \{1,2\}$ is in
    $(\mathcal M^s_{b,a}(n)\times\{1,2\}) \cup (\mathcal
    M^h_{b,a}(n)\times\{0\})$.

    Let $\overline{r}$ be the mapping defined on the elements $(M,y)$
    of
    $(\mathcal M^s_{b,a}(n)\times\{1,2\}) \cup (\mathcal
    M^h_{b,a}(n)\times\{0\})$
    by the following.  Let $\beta$ be the root angle of $M$ and
    consider the unique labeling $\ell(\beta)$ of $M$, given by
    Lemma~\ref{lem:anglemobile}.  Close all the admissible triples of
    $M$ to obtain a map $G$ whose special face is a quadrangle
    $Q$. Propagate the labeling $\ell(\beta)$ to $G$ by keeping the
    property that the two labels that appear on an edge has to differ
    exactly by $(2\bmod 4)$.  Then by Lemma~\ref{lem:anglemobile} the
    quadrangle $Q$ of $G$ is labeled as on
    Figure~\ref{fig:labelquadrangle}. So $Q$ has a unique angle
    $\alpha$ between half-edges labeled $0$ and $1$.  We also denote
    $\alpha$ the angle of $M$ that corresponds to the angle $\alpha$
    of $Q$.  Let $M^+$ be the map obtained from $M$ by forgetting its
    root angle $\beta$ and adding a root half-edge $h_0$ incident to
    $\alpha$. If $M$ is square, then let $x=y$. If $M$ is hexagon, let
    $x$ be such that $\beta$ is an angle incident to the $x$-th
    special vertex of $M$ encountered while walking \ccw along the
    face of $M$, starting from $\alpha$.  Then $\overline{r}$ returns
    $(M^+,x)$.

We claim that:

    \begin{sclaim}
      $\overline{r}\circ r=Id$
    \end{sclaim}

    \begin{proofsclaim}
      Let $(M^+,x)$ be an element of
      $\mathcal M'_{r,s,b}(n)\times \{1,2\}$ and $(M,y)$ its image by
      $r$. We use the notation of the definition of $r$, i.e. the map
      $M$ is obtained from $M^+$ by removing the root half-edge $h_0$
      of $M^+$, incident to the angle $\alpha$ of $M$ and rooting $M$
      according to the labeling $\ell(\alpha)$ at an angle $\beta$
      between half-edges labeled $0$ and $1$. By
      Lemma~\ref{lem:anglemobile}, there is a unique mobile-labeling
      of $M$ such that $\beta$ is between half-edges
      labeled $0$ and $1$. So the labeling $\ell(\beta)$ used in the
      definition of $\overline{r}$ is exactly the same as
      $\ell(\alpha)$.  So $M^+$ is obtained from $M$ by adding a
      half-edge $h_0$ at an angle between half-edges labeled $0$ and
      $1$ of $\ell(\beta)$.

From $M^+$, one can build the graph
$Z=g^{-1}(M^+)\in \mc T'_r(n)$ by closing admissible triples in
any order.  The recovering method of Theorem~\ref{th:recover} says
that the root half-edge $h_0$ can be the last stem that is reattached
by this procedure.  % Let $M_0=M^+$. For $1\leq k \leq 2n-2$, let
% $M_{k}$ be the map obtained from $M_{k-1}$ by closing an admissible
% triple of $M_{k-1}$ that does not involve $h_0$.  Then the special
% face of $M_{2n-2}$ is a quadrangle with exactly one stem $h_0$. For
% $1\leq k \leq 2n-2$, since $h_0$ is not contained in the considered
% admissible triple of $M_{k-1}$ that is closed to obtain $M_{k}$, the
% presence or not of $h_0$ in $M_{k-1}$ has no effect on the way the
% admissible triple of $M_{k-1}$ is closed. This means that 
So the graph $G$ defined in the definition of $\overline{r}$ is
obtained from $Z$ by removing the edge containing the root half-edge.
So $M^+$ is obtained from $M$ by adding a half-edge $h_0$ at an angle
of $M$ corresponding to one of the angle of the quadrangle $Q$ of $G$.

By Lemma~\ref{lem:anglemobile}, the extension of the labeling
$\ell(\beta)$ to $G$ shows that the quadrangle $Q$ of $G$, is labeled
as on Figure~\ref{fig:labelquadrangle}. Moreover the angles of $M$
that corresponds to the angles of $Q$ are incident to half-edges with
exactly the same labels in $M$ and in $Q$. So there is a unique such
angle $\alpha$ of $M$ between half-edges labeled $0$ and $1$.  So
$M^+$ is obtained from $M$ by adding a half-edge $h_0$ at the angle of
$M$ corresponding to $\alpha$ and $\overline{r}\circ r=Id$.

% If $M$ is square, then clearly $(\overline{r}\circ r)_2=Id$.  If $M$
% is hexagon, let $j$ be such that $\beta$ is incident to the $j$-th special vertex of
% $M$ that is encountered while walking \ccw along the border of the
% unique face of $M$, starting from $\alpha$. By definition of $r$, we
% have that $x=j$.
% \comment{Why $x=j$ ?
%   Probleme potentiel $M^+,1$ et $M^+,2$ donne le meme $M,0$ ??? Je
%   vois pas autrement que de prouver et utiliser le lemme 29}
    \end{proofsclaim}

Conversely, let  $(M,y)$ be an element of $(\mathcal M^s_{b,a}(n)\times\{1,2\}) \cup (\mathcal
    M^h_{b,a}(n)\times\{0\})$ and $(M^+,x)$ its image by
    $\overline{r}$. Since $M$ is balanced, we have that $M^+$ is also
    balanced. Moreover, the root half-edge $h_0$ is added
    to $M$ in an angle of the special face obtained after
    reattaching all the admissible triples of $M$. So $M^+$ is
    safe and $(M^+,x)$ is in $\mathcal M'_{r,s,b}(n) \times \{1,2\}$.

It is clear that
$r \circ \overline{r}=Id$, so $r$ is a bijection.
\end{proofclaim}

By Claim~\ref{cl:bija} to~\ref{cl:bijr}, we have $r\circ(g',Id)\circ (a,Id)$ is a bijection from $\mathcal T_c^t(n) \times \{1,2\}$ to
$(\mathcal M^s_{b,a}(n)\times\{1,2\}) \cup (\mathcal M^h_{b,a}(n)\times\{0\})$. 
\end{proof}

\subsection{Enumeration of skeletons}

A \emph{skeleton} is a toroidal unicellular map such that every inner
vertex, i.e. every vertex of degree at least two, belongs to its
cycles. A skeleton is \emph{balanced} if every cycle of the map has
the same number of angles on its left and right sides.  A skeleton is
\emph{square} (resp. \emph{hexagon}) if it is a square (resp. hexagon)
unicellular map.  Let $\mathcal{S}_a(n)$ be the set of skeletons
on $n$ inner vertices, such that all inner vertices have degree $4$,
and rooted at an angle of a special vertex. Let
$\mathcal{S}_{b,a}(n)$ be the set of balanced element of
$\mathcal{S}_a(n)$.  Let $\mathcal{S}_a^s(n)$,
$\mathcal{S}_a^h(n)$, $\mathcal{S}_{b,a}^s(n)$ and
$\mathcal{S}_{b,a}^h(n)$ be the sets of square and hexagon
elements of $\mathcal{S}_a(n)$ and $\mathcal{S}_{b,a}(n)$,
respectively.

Given an element $M$ of $\mathcal M_{b}(n)$, the \emph{skeleton} of
$M$ is obtained from $M$ by removing all the vertices that are not
vertices of the cycles of $M$ nor in their neighborhood. It is direct to see
that the skeletons of elements of  $\mathcal M_{b,a}(n)$
are in
$\mathcal{S}_{b,a}(n)$.

An element of $\mathcal M_{b,a}(n)$ can be uniquely decomposed
into an element of $\mathcal{S}_{b,a}(k)$, for some $k\geq 1$,
and a $(2k-2)$-uplet of  ternary trees (each ternary tree being
attached to a leaf of the skeleton) such that the total number of
inner vertices of the trees is $n-k$.

Let $\mathcal{F}(n,k)$ be the set of $k$-uplets of rooted ternary
trees with total number of inner vertices $n$. Its associated generating function satisfies 
$F(z,u)=\sum_{n,k} |\mathcal{F}(n,k)|z^nu^k=\sum_{k} A(z)^ku^k=
 \frac{1}{1-uA(z)}$. Moreover, it is known that
$|\mathcal{F}(n,k)| = \frac{k}{2n-k}\binom{3n-k-1}{n}$ (see~\cite[Theorem 5.3.10]{Sta99}).

Let $S^s_{b,a}$ and $S^h_{b,a}$ be the generating functions
associated to $\mathcal S^s_{b,a}$ and $\mathcal S^h_{b,a}$,
respectively, i.e. $S^s_{b,a}(z) =\sum_{n} |\mathcal
S^s_{b,a}(n)|z^n$ and $S^h_{b,a}(z) =\sum_{n}|\mathcal S^h_{b,a}(n)|z^n$.

 Then we have the following lemma:

\begin{lemma}
\label{lem:forestdecomp}
$$|\mathcal{M}_{b,a}^s(n)| =
\sum_{k=1}^n|\mathcal{S}_{b,a}^s(k)|.|\mathcal{F}(n-k,2k-2)|$$

$$|\mathcal{M}_{b,a}^h(n)| =  \sum_{k=1}^n|\mathcal{S}_{b,a}^h(k)|.|\mathcal{F}(n-k,2k-2)|$$

and the associated generating functions satisfy:

  $$M^s_{b,a}(z) =\sum_{n\geq 1} |\mathcal M^s_{b,a}(n)|z^n=
  S^s_{b,a}(zA(z)^2)/A(z)^2$$

  $$M^h_{b,a}(z) =\sum_{n\geq 1}|\mathcal M^h_{b,a}(n)|z^n=
  S^h_{b,a}(zA(z)^2)/A(z)^2.$$
\end{lemma}

\begin{proof}
  The first two formulas are clear by above decomposition.  Moreover,
  each element of $\mathcal{M}_{b,a}^s(n)$ is obtained by substituting
  each of the $2k-2$ leaves of an element of $\mathcal{S}_{b,a}^s(k)$ by a ternary
  tree. So we have:
\begin{align*}
M^s_{b,a}(z) 
       &=& \sum_{k\geq 1} |\mathcal{S}_{b,a}^s(k)| A(z)^{2k-2} z^k\\
       &=&\frac{1}{A(z)^2} \sum_{k\geq 1} |\mathcal{S}_{b,a}^s(k)|(A(z)^2 z)^k\\
       &=&\frac{S^s_{b,a}(zA(z)^2)}{A(z)^2}.
\end{align*}

Similarly, we have $M^h_{b,a}(z)= \frac{S^h_{b,a}(zA(z)^2)}{A(z)^2}$.
\end{proof}

By Lemma~\ref{lem:forestdecomp}, we are reduced to the enumeration of
$\mathcal{S}_{b,a}^s$ and $\mathcal{S}_{b,a}^h$.

Consider an element $S$ of $\mathcal{S}_a(n)$. If $S$ is square,
consider the two edge-disjoint closed walks of $S$ started from the
special vertex, noted $W_1$ and $W_2$. We assume that $W_1$ and $W_2$
are chosen such that the two half-edges $h_1,h_2$ that are incident
to the root angle of $S$ are traversed from the special vertex in
$W_1$, $W_2$, respectively, and that $h_1$ and $h_2$ appear
consecutively in \ccw order around the special vertex.  If $S$ is
hexagon, then consider the three walks $W_1$, $W_2$ and $W_3$ of $S$
starting from the special vertex $v_1$ containing the root angle,
ending at the second special vertex $v_2$, such that the three paths
$W_1,W_2,W_3$ appears consecutively in \ccw order around $v_1$,
starting from the leaf attached to $v_1$.  Note that, for the square
or hexagonal case, the $W_i$ are uniquely defined and oriented.  Along
each walk $W_i$, the inner vertices that are encountered may have both
leaves on the right, both leaves on the left, or one leaf on each
side. In next two lemmas, we encode this by using Grand Motzkin
prefix/paths defined below.

%In each case, for
%$i\in\{1,2\}$ or $i\in\{1,2,3\}$, let $r_i$ (resp. $l_i$) be the
%number of leaves of $S$ incident to an inner vertex of $W_i$ that are
%on the right (resp. left) side of $W_i$. Let
%$\delta_i= (r_i - l_i)/2$.  

% Along each walk $W_i$, the inner vertices
% that are encountered may have both leaves on the right, both leaves on
% the left, or one leaf on each side. Hence if a given walk $W_i$
% contains $k$ inner vertices, this walk can be encoded by a \emph{Grand
%   Motzkin prefix} (or \emph{GM prefix} for short) of length $k$,
% starting at the point $(0,0)$ in $\mathbb R^2$, ending at the point
% $(k,\delta_i)$ and composed of $k$ steps $(1,1)$, $(1,-1)$ and
% $(1,0)$, where each step encode the variation of $\delta_i$ along the
% walk.  Let $gm(n,\delta)$ be the number of Grand Motzkin prefix of
% length $n$ starting at $(0,0)$ and ending at $(n,\delta)$ and
% $GM(z,u) = \sum_{n,\delta} gm(n,\delta) z^n u^\delta$. There is one GM
% prefix of length $0$ and a GM prefix of length $n>0$ is obtained by
% adding one step $(1,1)$, $(1,-1)$ or $(1,0)$ to a GM prefix of length
% $(n-1)$. This decomposition leads to the following equation:

% $$GM(z,u)=1 + z(u+1/u+1)GM(z,u) = \sum_{n\geq0} (1+1/u+u)^n z^n.$$

A \emph{Grand
  Motzkin prefix} (or \emph{GM prefix} for short) of length $n$, is a
path in $\mathbb Z^2$,
starting at the point $(0,0)$, ending at the point
$(n,\delta)$, with $\delta\in\mathbb Z$, and composed of $k$ steps $(1,1)$, $(1,-1)$ and
$(1,0)$.  Let $gm(n,\delta)$ be the number of GM prefix of
length $n$ starting at $(0,0)$ and ending at $(n,\delta)$ and
$GM(z,u) = \sum_{n,\delta} gm(n,\delta) z^n u^\delta$. There is one GM
prefix of length $0$ and a GM prefix of length $n>0$ is obtained by
adding one step $(1,1)$, $(1,-1)$ or $(1,0)$ to a GM prefix of length
$(n-1)$. This decomposition leads to the following equation:

$$GM(z,u)=1 + z(u+1/u+1)GM(z,u) =\frac{1} {1-z(u+1/u+1)} = \sum_{n\geq0} (1+1/u+u)^n z^n.$$

%One can recognize here the sum of the trinomial coefficients so
%$gm(n,\delta)= \binom{n}{\delta}_2$. \comment{BEN CHECK}

Let $[z^n]f$ denote the coefficient of $z^n$ of function $f(z)$,
i.e. if $f=\sum_n f_n z^n$  then $[z^n] =f_n$.

A \emph{Grand Motzkin path} (or \emph{GM path} for short) of length
$k$ is a \emph{Grand Motzkin prefix} of length $k$, ending at the
point $(k,0)$.  The generating function associated to GM paths
satisfies $GM(z)=[u^0]GM(z,u)=\frac{1}{\sqrt{1-2z- 3z^2}}$
(see~\cite{FM14}).

The square skeletons satisfy:

\begin{lemma} \label{lem:Ss}
  $$|\mathcal{S}_{b,a}^s(n)| = \frac{3^n - (-1)^n}4$$
  $$S^s_{b,a}(z) = \frac{z}{1-2z-3z^2}.$$ %% sequence A015518
% 1, 2, 7, 20, 61, 182, 547, 1640, 4921, 14762
\end{lemma}
\begin{proof} 
  With above notations, an element of $\mathcal{S}_{b,a}^s(n)$ is
  uniquely decomposed into a special vertex and the two closed walks
  $W_1$ and $W_2$.  for $i\in\{1,2\}$, let $r_i$ (resp. $l_i$) be the
  number of leaves of $S$ incident to an inner vertex of $W_i$ that
  are on the right (resp. left) side of $W_i$. Let
  $\delta_i= (r_i - l_i)/2$.  Since the special node as no leaves
  attached to it, the balanced property implies that
  $\delta_1=\delta_2=0$.  So if $W_i$ contains $k_i$ inner vertices,
  then $W_i$ can be encoded by a GM path of length $k_i$.  This
  decomposition results in the product of respective generating
  series: $S^s_{b,a}(z) =z.GM(z)^2= \frac{z}{1-2z-3z^2}$.

  Now observe that $S^s_{b,a}(z) - 2zS^s_{b,a}(z) - 3 z^2 S^s_{b,a}(z) =z$. We deduce
  the following recurrence: for $n\geq 2$, we have
  $|\mathcal{S}_{b,a}^s(n)| = 2|\mathcal{S}_{b,a}^s(n-1)| +
  3|\mathcal{S}_{b,a}^s(n-2)|$.
  Moreover, we have $|\mathcal{S}_{b,a}^s(0)|=0$ and
  $|\mathcal{S}_{b,a}^s(1)|=1$. Since $\frac{3^n - (-1)^n}4$
  satisfies  the same conditions $|\mathcal{S}_{b,a}^s(n)|$, the two are identical.
\end{proof}

The hexagon skeletons satisfy:

\begin{lemma}
\label{lem:Sh}
$$|\mathcal{S}_{b,a}^h(n)| = (n-2).3^{n-1}+\frac{5.3^{n-1}+(-1)^n}4$$
$$S^h_{b,a}(z) = \frac{4z^2}{(z+1)(3z-1)^2}.$$
%% sequence A191008*4 et décalé de 2
% 0, 0, 1, 5, 22, 86, 319, 1139, 3964, 13532, 45517, 151313, 498226
\end{lemma}
\begin{proof}
Let $S$ be an  element of $\mathcal{S}_a^h(n)$. Note that $S$ is
not assumed to be  balanced here. Considering this larger class, we
are able extract the series $S^h_{b,a}(z)$ by following the \emph{standard diagonal method}~\cite[Section~6.3]{Sta99}.
  %Let $S$ be an element of $\mathcal{S}_{b,a}^h(n)$. 
As with above notations, consider the three walks $W_1$, $W_2$ and
$W_3$ of $S$ starting from the special vertex $v_1$ containing the
root angle, ending at the second special vertex $v_2$, such that the
three paths $W_1,W_2,W_3$ appear consecutively  in \ccw order around
$v_1$, starting from the leaf attached to $v_1$.

%Then $S$ is uniquely decomposed into the two special vertices, the
%three path between them, a leaf attached to each special vertex, plus
%the root in an angle of the first special vertex. 

There are different cases to consider depending on the position of the
leaves.  We say that $S$ is of Type $i$ if the leaf of $v_2$ is after
$W_i$ in the \ccw order around $v_2$ (see the top line of
Figure~\ref{fig:hexagon-type}).  In order to ease the upcoming
computation, let rename these walks depending of the type.  For types 1
and 3, let $W_c=W_1, W_x=-W_2$ and $W_y=-W_3$ (see the bottom left and
bottom right of Figure~\ref{fig:hexagon-type}). For type 2, let
$W_c = W_3, W_x=-W_1$ and $W_y=-W_2$ (see the bottom center of
Figure~\ref{fig:hexagon-type}, where edges have been redrawn
differently than in top-center figure).

\begin{figure}[!ht]
  \center 
  \begin{tabular}{ccc}
   \scalebox{0.5}{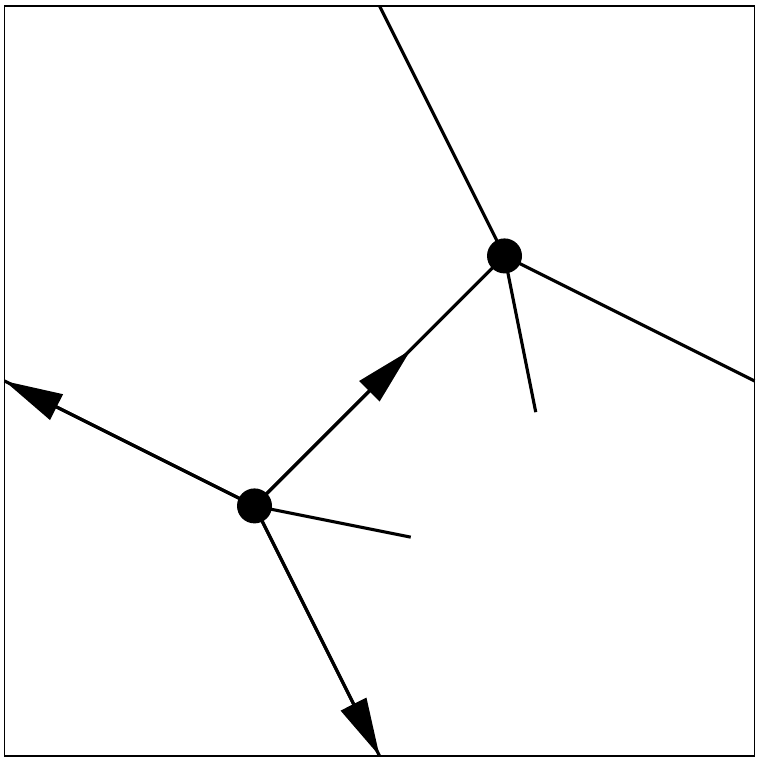} \ \ & \ \
\scalebox{0.5}{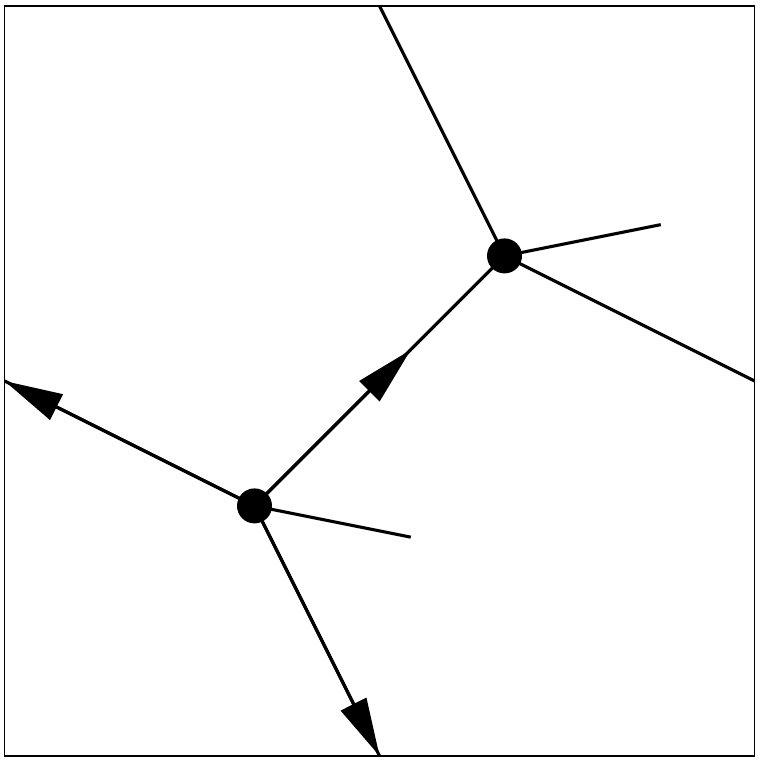} \ \ & \ \
\scalebox{0.5}{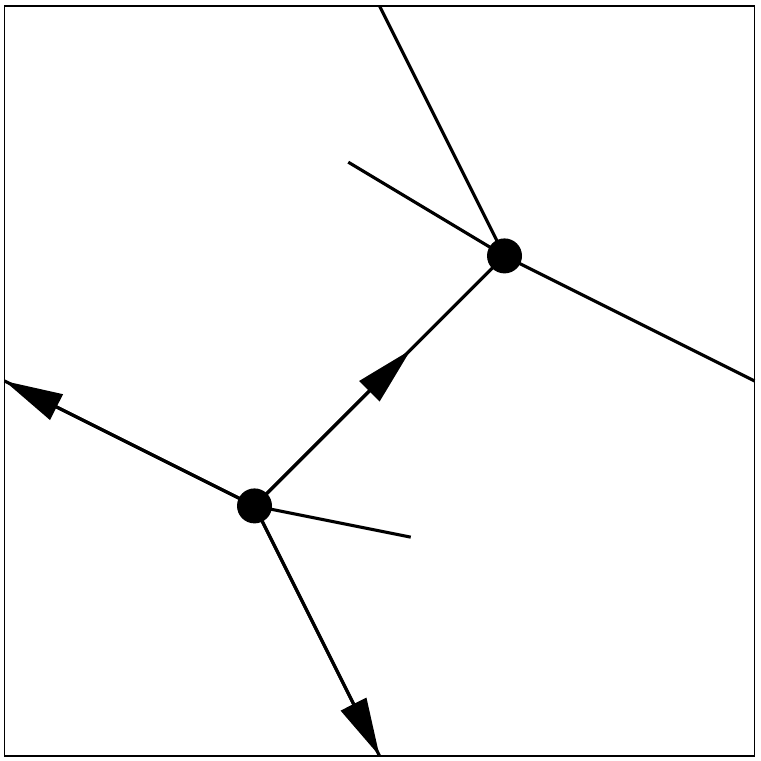}\\
Type 1  \ \ & \ \
Type 2 \ \ & \ \
Type 3\\
   \scalebox{0.5}{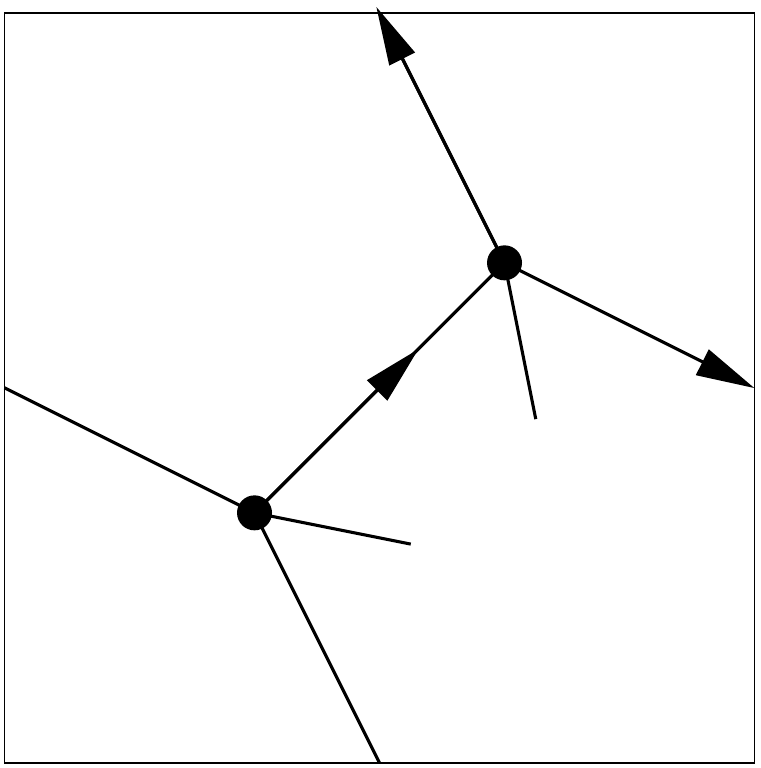} \ \ & \ \
\scalebox{0.5}{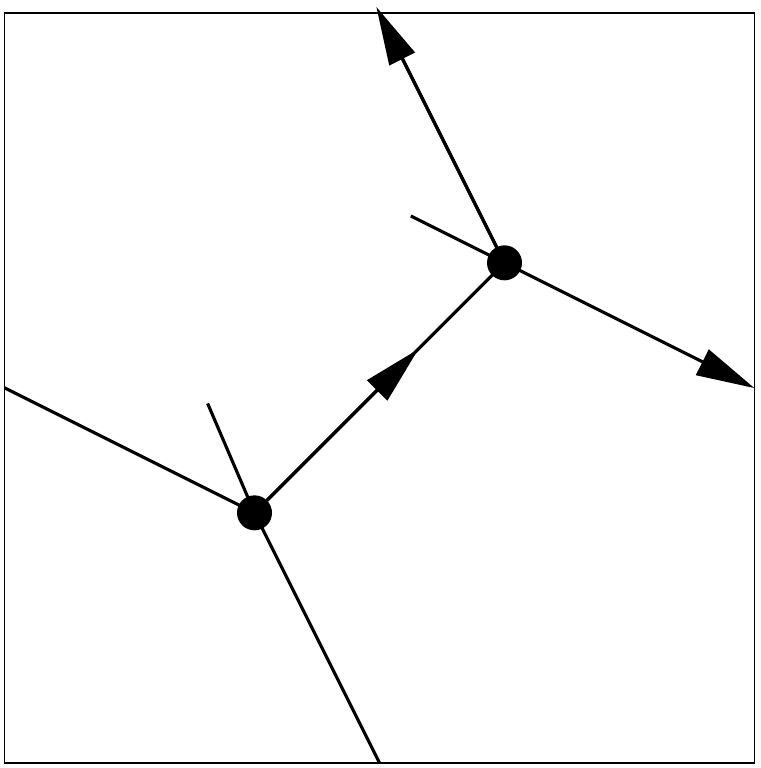} \ \ & \ \
\scalebox{0.5}{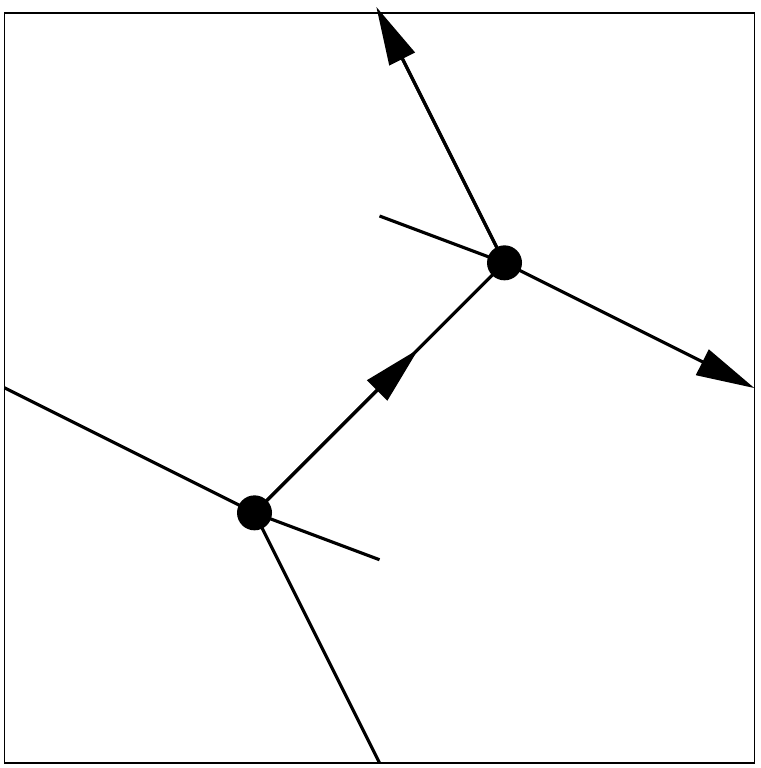}\\

  \end{tabular}
\caption{Different types of  hexagon skeletons.}
\label{fig:hexagon-type}
\end{figure}

Let $C_1,C_2$ be the cycles of $S$ made of $W_c+W_x$ and $W_c+W_y$
respectively, with the direction of traversal corresponding to the
orientation of  $W_c,W_x,W_y$.
For $i\in\{1,2\}$, let $\delta_i$ be the
number of leaves of $S$ incident to $C_i$ that  are
on its right side minus the
number of leaves of $S$ incident to $C_i$ that  are
on its left side, divided by two.

Let $\mathcal{S}_a^h(k,\ell,n)$ denote the elements of
$\mathcal{S}_a^h(n)$ such that $\delta_1=k$ and $\delta_2=\ell$,
with $(k,\ell)\in \mathbb Z^2$.  Let $S^h_a(u,v,z)$ be the
associated  generating function, i.e.
$S^h_a(u,v,z) = \sum_{(k,\ell,n)\in \mathbb Z^2\times \mathbb N}
|\mathcal{S}_a^h(k,\ell,n)| u^{k}v^{\ell} z^n$.
This generating function can be computed with the following method.
An hexagonal skeleton can be decomposed into 2 special vertices
(contributing for $z^2$) plus tetra-valent caterpillars $C_c$, $C_x$
and $C_y$ respectively contributing for $GM(z,uv)$, $GM(z,u)$ and
$GM(z,v)$. Depending of the type of hexagon skeleton, the special
vertices are contributing for $+1$, $0$ or $-1$ to $\delta_1$ and
$\delta_2$, this is translated by a factor $(\frac{1}{uv}+uv+1)$ on
the generating function.  There are four possible root angles around
$v_1$, so we have:
\begin{eqnarray*}
S^h_a(u,v,z)   &=& 4z^2(uv + \frac{1}{uv} + 1)(GM(z,uv).GM(z,u).GM(z,v))\\
  &=& -\frac{4z^2(u^2v^2+uv+1)uv}{(u^2z+uz-u+z)(v^2z+vz-v+z)(u^2v^2z+uvz-uv+z)}.
\end{eqnarray*}

Observe that $S^h_{b,a}(z)=[v^0][u^0]S^h_a(u,v,z)$.

The denominator of $S^h_a(u,v,z)$, seen as a polynomial of $u$
admits four roots: $U, \frac{1}{U}, \frac{U}{v}$ and $\frac{1}{Uv}$
where $U = \frac{1-z-\sqrt{-3z^2-2z+1}}{2z}$, and we have:

$$S^h_a(u,v,z) =-\frac{4(u^2v^2+uv+1)u}{(u-U)(u-\frac{1}{U})(u-\frac{U}{v})(u-\frac{1}{Uv})v(v^2z+vz-v+z)}.$$

Hence $S^h_a(u,v,z)$ can be converted into partial fractions of $u$:

\begin{align*}
S^h_a(u,v,z) = A.\biggl(&\frac{U^2+Uv+v^2}{(1-Uu)U^2v^2(U^2-v)}+\frac{U^2v^2+Uv+1}{(1-\frac{U}{u})uUv^2(1-U^2v)}\\
                       &+\frac{U^2+U+1}{(1-\frac{U}{uv})(U^2-v)Uuv^2}+\frac{U^2+U+1}{(1-Uuv)U^2v(1-U^2v)}\biggl),
\end{align*}

where $A=\frac{4v^3U^3}{(v-1)(U^2-1)(v^2z+vz-v+z)}$.

As $A=\mc O(z^3)$ and $U=\mc O(z)$ (when $z$ tends to $0$), this identity
splits  into a sum of four power series in $z$ with coefficients in
$\mathbb{Q}[u,\frac{1}{u},v,\frac{1}{v}]$, two with only negative
powers of $u$ and two with only nonnegative powers of $u$. 

\begin{align*}
S^h_a(u,v,z) = A.\biggl( &\frac{U^2+Uv+v^2}{U^2v^2(U^2-v)}\sum_{n\geq 0}(Uu)^n +\frac{U^2v^2+Uv+1}{uUv^2(1-U^2v)}\sum_{n\geq 0} \left(\frac{U}{u}\right)^n\\
                       &+\frac{U^2+U+1}{(U^2-v)Uuv^2}\sum_{n\geq 0} \left(\frac{U}{uv}\right)^n+\frac{U^2+U+1}{U^2v(1-U^2v)}\sum_{n\geq 0} \left(Uuv\right)^n\biggl).
\end{align*}

Hence the coefficient $[u^0]$ can be directly extracted:

\begin{eqnarray*}
[u^0]S^h_a(u,v,z) &= & A.\left( \frac{U^2+Uv+v^2}{U^2v^2(U^2-v)}+\frac{U^2+U+1}{U^2v(1-U^2v)}\right)\\
                 & = & \frac{4(v^2z+vz+v+z)z^2v}{\sqrt{-(z+1)(3z-1)}(v^2z^2+vz^2+2vz+z^2-v)(v^2z+vz-v+z)}.
\end{eqnarray*}

Again, the denominator of $[u^0] S^h_a(u,v,z)$, seen as a polynomial
of $v$ admits four roots: $V_0, \frac{1}{V_0}, V_1, \frac{1}{V_1}$
where $V_0 = U$ and $V_1 =-\frac{z^2+\sqrt{-(z+1)(3z-1
)}(1-z)+2z-1}{2z^2}$.

Hence $[u^0]S^h_a(u,v,z)$ can be converted into partial fractions of $v$:

\begin{align*}
[u^0]S^h_a(u,v,z) = B .\biggl(&  \frac{V_0(V_0^2z+V_0z+V_0+z)}{(1-V_0/v)v(V_0^2-1)(V_0V_1-1)}+\frac{V_1(V_1^2z+V_1z+V_1+z)}{(1-V_1/v)v(-V_1^2+1)(V_0V_1-1)}\\
                            & +\frac{V_0^2z+V_0z+V_0+z}{(1-V_0v)(V_0^2-1)(V_0V_1-1)}+\frac{V_1^2z+V_1z+V_1+z}{(1-V_1v)(1-V_1^2)(V_0V_1-1)} \biggl),
\end{align*}

where $B = \frac{4V_0V1}{z\,\sqrt{-(z+1)(3z-1)}(V_0-V_1)}$.

As $B=\mc O(z)$, $V_0=\mc O(z)$ and $V_1=\mc O(z^2)$, this identity splits  into a sum of four power series in $z$ with coefficients in $\mathbb{Q}[v,\frac{1}{v}]$, two with only negative powers of $v$ and two with only nonnegative powers of $v$.

\begin{align*}
[u^0]S^h_a(u,v,z) = B .\biggl(&  \frac{V_0(V_0^2z+V_0z+V_0+z)}{v(V_0^2-1)(V_0V_1-1)}\sum_{n\geq 0} \left(\frac{V_0}{v}\right)^n   +\frac{V_1(V_1^2z+V_1z+V_1+z)}{v(-V_1^2+1)(V_0V_1-1)}\sum_{n\geq 0} \left(\frac{V_1}{v}\right)^n\\
                            & +\frac{V_0^2z+V_0z+V_0+z}{(V_0^2-1)(V_0V_1-1)}\sum_{n\geq 0} \left(V_0v\right)^n+\frac{V_1^2z+V_1z+V_1+z}{(1-V_1^2)(V_0V_1-1)}\sum_{n\geq 0} \left(V_1v\right)^n \biggl),
\end{align*}

Hence the coefficient $[v^0]$ can be directly extracted:

$$[v^0][u^0]S^h_a(u,v,z) = B\left(\frac{V_0^2z+V_0z+V_0+z}{(V_0^2-1)(V_0V_1-1)}+\frac{V_1^2z+V_1z+V_1+z}{(V_1^2-1)(V_0V_1-1)}\right).$$

Which simplifies into the second part of the lemma:
$$S^h_{b,a}(z) = [v^0][u^0]S^h_a(u,v,z) = \frac{4z^2}{(z+1)(3z-1)^2}.$$

We can observe that $S^h_{b,a}(z) - 5zS^h_{b,a}(z) + 3z^2S^h_{b,a}(z)
+ 9z^3S^h_{b,a}(z) = 4z$. We deduce the following recurrence for
$n>3$, $|\mathcal{S}_{b,a}^h(n)|  =
5|\mathcal{S}_{b,a}^h(n-1)|-3|\mathcal{S}_{b,a}^h(n-2)| -
9|\mathcal{S}_{b,a}^h(n-3)|$. Moreover, we have
$|\mathcal{S}_{b,a}^h(0)| = |\mathcal{S}_{b,a}^h(1)|=0$ and
$|\mathcal{S}_{b,a}^h(2)|=4$. Since
$(n-2).3^{n-1}+\frac{5.3^{n-1}+(-1)^n}4$ satisfies  the same
conditions as $|\mathcal{S}_{b,a}^h(n)|$, the two are identical.

\end{proof}

\subsection{Enumeration theorem}

The enumeration theorem that we obtain is the following:

\begin{theorem}
\label{thm:enumeration}
The generating function associated to the number $\mathcal {T}_{h}(n)$ of essentially
4-connected toroidal triangulations on $n$ vertices, rooted on any half-edge, is:

$$T_{h}(z)  = \sum_{n\geq 0} |\mathcal{T}_{h}(n)| z^n =  \frac{zA(z)}{7zA(z)^2-21zA(z)+9z+1}$$

 where $A(z)$ is the generating function of (leaf-rooted) ternary trees satisfying $A(z)=1+zA(z)^3$.

\

Moreover, the values of $|\mathcal{T}_{h}(n)|$ are given by the following formulas:

$$ |\mathcal{T}_{h}(n)|= \frac{1}{4}\sum_{k=1}^{n} |\mathcal{T}_{h,c}^p(n-k)|.|\mathcal{T}_c^t(k)|$$

$$|\mathcal{T}_{h,c}^p(n)| =\frac4{n+1}\binom{3n+1}{n}$$

$$|\mathcal T_c^t(n)|=
\sum_{k=1}^n\mathcal{S}(k).|\mathcal{F}(n-k,2k-2)|$$

$$ \mathcal S(n) =\frac{(-1)^{n-1}+(3+4n)3^{n-1}}8$$

$$|\mathcal{F}(n,k)| = \frac{k}{2n-k}\binom{3n-k-1}{n}.$$
 \end{theorem}

\begin{proof}
By Lemma~\ref{lem:bijrootcorner}, we have:
$$ |\mathcal{T}_{h}(n)|= \frac{1}{4}\sum_{k=0}^{n} |\mathcal{T}_{h,c}^p(n-k)|.|\mathcal{T}_c^t(k)|$$
and so: 
\begin{equation}\label{eq:tptt}T_{h}(z)=\frac{1}{4}T_{h,c}^p(z)T_c^t(z).
\end{equation}

By Lemma~\ref{lem:planarpart2}, we have:
$$|\mathcal{T}_{h,c}^p(n)| = \frac4{n+1}\binom{3n+1}{n}$$
\begin{equation}\label{eq:tpA}
T_{h,c}^p(z) = 4 A(z)^2.
\end{equation}

By Lemma~\ref{lem:bijsquarehexa}, we have:
\begin{equation}\label{eq:gf-Ttblabla1}|\mathcal T_c^t(n)|=
|\mathcal M^s_{b,a}(n)|+ \frac{1}{2}|\mathcal
M^h_{b,a}(n)|\end{equation}

\begin{equation}\label{eq:gf-Ttblabla2}
T_c^t(z) = M^s_{b,a}(z) + \frac{1}{2}M^h_{b,a}(z).
\end{equation}

% By Lemma~\ref{lem:forestdecomp}, we have:
% $$|\mathcal{M}_{b,a}^s(n)| =
% \sum_{k=1}^n|\mathcal{S}_{b,a}^s(k)|.|\mathcal{F}(n-k,2k-2)|$$

% $$|\mathcal{M}_{b,a}^h(n)| =  \sum_{k=1}^n|\mathcal{S}_{b,a}^h(k)|.|\mathcal{F}(n-k,2k-2)|$$

%   $$M^s(z) =
%   S^s(zA(z)^2)/A(z)^2$$

%   $$M^h(z) =
%   S^h(zA(z)^2)/A(z)^2$$.

Let
$\mathcal S(n) = |\mathcal{S}_{b,a}^s(n)| +
\frac{1}{2}|\mathcal{S}_{b,a}^h(n)|$ and $S(z) = S^s_{b,a}(z) + \frac{1}{2} S^h_{b,a}(z)$.
So, by Lemma~\ref{lem:forestdecomp}, Equations~(\ref{eq:gf-Ttblabla1}) and~(\ref{eq:gf-Ttblabla2})  become:

$$|\mathcal T_c^t(n)|=  \sum_{k=1}^n\mathcal{S}(k).|\mathcal{F}(n-k,2k-2)|$$

\begin{equation} \label{eq:gf-Ttbis}
T_c^t(z) = \frac{S(zA(z)^2)}{A(z)^2}
.\end{equation}

From Lemmas~\ref{lem:Ss} and~\ref{lem:Sh}, we obtain:
$$ \mathcal S(n) =\frac{(-1)^{n-1}+(3+4n)3^{n-1}}8$$

\begin{equation} \label{eq:szblabla}
S(z) = \frac{z(1-z)}{(1+z)(1-3z)^2}.
\end{equation}

From~(\ref{eq:gf-Ttbis}) and~(\ref{eq:szblabla}), we obtain:

\begin{equation} \label{eq:tzblabla}
T_c^t(z) = \frac{z-z^2.A(z)^2}{(z.A(z)^2+1).(3z.A(z)^2-1)^2}
.\end{equation}

% We also
%set $|\mathcal{F}(0,0)|=1$. Rk $F(z,u) = \sum_{z,u} |\mathcal{F}(n,k)| z^n
%u^k$ satisfies : $F(z,u) = 1+ uA(z) + u^2A(z)^2 + ... = \frac{1}{1-uA(z)}$.

Combining~(\ref{eq:tptt}),~(\ref{eq:tpA})  and~(\ref{eq:tzblabla}) gives :
$$T_{h}(z)  = \frac{(z-z^2A(z)^2)A(z)^2}{(zA(z)^2+1)(3zA(z)^2-1)^2}.$$

By (\ref{eq:az3}) (see Section~\ref{sec:decomp}), one can replace $zA^3$ by $A-1$ in above formula
and obtain: 
$$T_{h}(z)   =  \frac{zA(z)}{7zA(z)^2-21zA(z)+9z+1}.$$

\end{proof}

%\comment{$T_{h}$ vérifie la  relation suivant ; $(729*z^3+2700*z^2-848*z+64)*T(z)^3+(756*z^2-112*z)*T(z)^2+(54*z^2-z)*T(z)+z^2 = 0$}

%In the last section of this paper, we have been able to provide
%enumerative formulas for essentially 4-connected toroidal
%triangulations. Even if these results are based on a nice bijection,

Part of the proof of Theorem~\ref{thm:enumeration} relies on
generating function analysis and do not completely explain the
simplicity of some expressions. For instance, sequences of the number
of different kinds of skeletons ($S^s_{b,a}(n)$ and $S^h_{b,a}(n)$)
have nice simple formulas (see Lemmas~\ref{lem:Ss} and~\ref{lem:Sh})
that deserve clean bijective interpretations. Note that these
sequences  already appear in OEIS~\cite{oeis} (resp. A015518, A191008).
Having a bijective proof of the enumeration of skeletons could be
essential to provide an efficient (e.g. sub-quadratic) random
generation algorithm for essentially $4$-connected toroidal maps.

\section{Conclusion}

In this paper, we have generalized transversal structures and some of
its applications to the toroidal case.  Using only a local property in
the definition, as in the planar case, is not enough to obtain
interesting properties.  Indeed, the set of toroidal transversal
structures of a given toroidal map is partitioned into several
distributive lattices.  The main point of this paper is to be able to
find a global property, called ``balanced'', that such an object may
have or may not have. Then, the set of balanced objects defines a
unique lattice whose minimal element has properties useful to apply
techniques devised for the planar case. This is very similar to what
happens for Schnyder woods for which it is also possible to define an
analogous balanced property (see~\cite{LevHDR}). Then, one can apply
on these minimal balanced objects, the ``mobile method'' as here, or
alternatively ``Poulalhon and Schaeffer's method'' as in \cite{DGL15},
to obtain bijections between the considered classes of toroidal maps
and particular toroidal unicellular maps.  This seems to be a general
framework that can be applied to various classes of toroidal maps and
their corresponding set of $\alpha$-orientations. In a paper to come
of Eric Fusy and the second author, a similar balanced property is
found for toroidal fractional $\frac{d}{d-2}$-orientations with
similar bijective consequences.  A challenging question is to see if
one can go further and found a generalization of the balanced property
in higher genus. Currently we have no idea of what could be the answer
even for the double torus.

Schnyder woods and transversal structures are also used in the planar
case to compute straightline grid drawings of plane
triangulations~\cite{Sch89, Fus09}. Recently toroidal Schnyder woods
have been used to obtain drawings of toroidal graphs
(see~\cite{GL13}). Another natural application of toroidal transversal
structures would be a straightline grid drawing algorithm of
essentially 4-connected toroidal triangulations.

% We have presented an asymptotically optimal encoding scheme for minimal balanced transversal structure (see Section~\ref{sec:coding}). In order to apply this result on essentially 4-connected toroidal triangulations, finding a linear time algorithm to compute such transversal structure would be useful in that context. In the planar case, there exist two linear time algorithm to compute a transversal structure~\cite{KH94}: one base on edge contraction and one base on canonical ordering. The algorithm we proposed is a generalisation of the first one and is not linear. The best candidate of a \emph{toroidal canonical ordering} would be the conquest order presented in~\cite{CFL09}, but it is not clear how to use this approach to compute transversal structure.

%The generalization of transversal structure to essentially 4-connected triangulations of higher genus (following what was done for generalized Schnyder woods~\cite{GKL15}) is also a very natural perspective of this work.

% Finally technique presented in this paper could also be applied to obtain similar results on 3-connected toroidal maps. \comment{A développer (cf. mail d’Eric)}.

\vspace{12pt}
\noindent
{\bf Acknowledgments.} We would like to thank Mireille Bousquet-Mélou,  Eric Fusy and Daniel Gonçalves
%for help and advises. We would also like to thank Eric Fusy and Daniel
%Gonçalves 
for fruitful discussions and suggestions.

\bibliographystyle{alpha} 
\bibliography{tts} 
\end{document}